%% file: ms_jasa.tex
\PassOptionsToPackage{unicode}{hyperref}
\PassOptionsToPackage{hyphens}{url}
\PassOptionsToPackage{dvipsnames,svgnames,x11names}{xcolor}
\documentclass[
  12pt]{article}

\usepackage{amsmath,amssymb}
\usepackage{iftex}
\ifPDFTeX
  \usepackage[T1]{fontenc}
  \usepackage[utf8]{inputenc}
  \usepackage{textcomp}
\else
  \usepackage{unicode-math}
  \defaultfontfeatures{Scale=MatchLowercase}
  \defaultfontfeatures[\rmfamily]{Ligatures=TeX,Scale=1}
\fi
\usepackage{lmodern}
\IfFileExists{upquote.sty}{\usepackage{upquote}}{}
\IfFileExists{microtype.sty}{%
  \usepackage[]{microtype}
  \UseMicrotypeSet[protrusion]{basicmath}
}{}
\makeatletter
\@ifundefined{KOMAClassName}{%
  \IfFileExists{parskip.sty}{%
    \usepackage{parskip}
  }{%
    \setlength{\parindent}{0pt}
    \setlength{\parskip}{6pt plus 2pt minus 1pt}}
}{%
  \KOMAoptions{parskip=half}}
\makeatother
\usepackage{xcolor}
\makeatletter
\ifx\paragraph\undefined\else
  \let\oldparagraph\paragraph
  \renewcommand{\paragraph}{
    \@ifstar
      \xxxParagraphStar
      \xxxParagraphNoStar
  }
  \newcommand{\xxxParagraphStar}[1]{\oldparagraph*{#1}\mbox{}}
  \newcommand{\xxxParagraphNoStar}[1]{\oldparagraph{#1}\mbox{}}
\fi
\ifx\subparagraph\undefined\else
  \let\oldsubparagraph\subparagraph
  \renewcommand{\subparagraph}{
    \@ifstar
      \xxxSubParagraphStar
      \xxxSubParagraphNoStar
  }
  \newcommand{\xxxSubParagraphStar}[1]{\oldsubparagraph*{#1}\mbox{}}
  \newcommand{\xxxSubParagraphNoStar}[1]{\oldsubparagraph{#1}\mbox{}}
\fi
\makeatother

\usepackage{longtable,booktabs,array}
\usepackage{calc}
\usepackage{etoolbox}
\makeatletter
\patchcmd\longtable{\par}{\if@noskipsec\mbox{}\fi\par}{}{}
\makeatother
\IfFileExists{footnotehyper.sty}{\usepackage{footnotehyper}}{\usepackage{footnote}}
\makesavenoteenv{longtable}
\usepackage{graphicx}
\makeatletter
\def\maxwidth{\ifdim\Gin@nat@width>\linewidth\linewidth\else\Gin@nat@width\fi}
\def\maxheight{\ifdim\Gin@nat@height>\textheight\textheight\else\Gin@nat@height\fi}
\makeatother
\setkeys{Gin}{width=\maxwidth,height=\maxheight,keepaspectratio}
\makeatletter
\def\fps@figure{htbp}
\makeatother

\makeatletter
\@ifpackageloaded{caption}{}{\usepackage{caption}}
\AtBeginDocument{%
\ifdefined\contentsname
  \renewcommand*\contentsname{Table of contents}
\else
  \newcommand\contentsname{Table of contents}
\fi
\ifdefined\listfigurename
  \renewcommand*\listfigurename{List of Figures}
\else
  \newcommand\listfigurename{List of Figures}
\fi
\ifdefined\listtablename
  \renewcommand*\listtablename{List of Tables}
\else
  \newcommand\listtablename{List of Tables}
\fi
\ifdefined\figurename
  \renewcommand*\figurename{Figure}
\else
  \newcommand\figurename{Figure}
\fi
\ifdefined\tablename
  \renewcommand*\tablename{Table}
\else
  \newcommand\tablename{Table}
\fi
}
\@ifpackageloaded{float}{}{\usepackage{float}}
\floatstyle{ruled}
\@ifundefined{c@chapter}{\newfloat{codelisting}{h}{lop}}{\newfloat{codelisting}{h}{lop}[chapter]}
\floatname{codelisting}{Listing}

\makeatother
\makeatletter
\@ifpackageloaded{caption}{}{\usepackage{caption}}
\@ifpackageloaded{subcaption}{}{\usepackage{subcaption}}
\makeatother

\ifLuaTeX
  \usepackage{selnolig}
\fi
\usepackage[]{natbib}
\usepackage{bookmark}

\IfFileExists{xurl.sty}{\usepackage{xurl}}{}
\hypersetup{
  pdftitle={High-dimensional Markov-switching Ordinary Differential Processes},
  colorlinks=true,
  linkcolor={blue},
  filecolor={Maroon},
  citecolor={Blue},
  urlcolor={Blue}}

\newcommand{\anon}{1}

\if1\anon
  \hypersetup{pdfauthor={Katherine Tsai; Mladen Kolar; Sanmi Koyejo}}
\fi

\graphicspath{{Figs/}}

\usepackage{mkolar_definitions}
\usepackage{mathabx}
\usepackage{mathrsfs}
\usepackage{rotating}
\usepackage{bibunits}
\usepackage{caption}
\usepackage{pifont}
\usepackage{placeins}
\usepackage{booktabs}
\usepackage{adjustbox}
\usepackage{algorithm}
\usepackage[noend]{algpseudocode}
\usepackage{multirow} 
 \usepackage{multicol}

\DeclareMathOperator{\dist}{dist}

\renewcommand{\baselinestretch}{1.05}

\usepackage{xargs}
\usepackage[colorinlistoftodos,prependcaption,textsize=tiny]{todonotes}
\newcommandx{\unsure}[2][1=]{\todo[linecolor=red,backgroundcolor=red!25,bordercolor=red,#1]{#2}}
\newcommandx{\change}[2][1=]{\todo[linecolor=blue,backgroundcolor=blue!25,bordercolor=blue,#1]{#2}}
\newcommandx{\info}[2][1=]{\todo[linecolor=OliveGreen,backgroundcolor=OliveGreen!25,bordercolor=OliveGreen,#1]{#2}}
\newcommandx{\improvement}[2][1=]{\todo[linecolor=Plum,backgroundcolor=Plum!25,bordercolor=Plum,#1]{#2}}

\usetikzlibrary{arrows, automata, positioning}

\usepackage{subcaption}
\usepackage[symbols,nogroupskip,sort=none]{glossaries-extra}
\input{myglossary}

\allowdisplaybreaks

\begin{document}
\begin{bibunit}[my-plainnat]

\def\spacingset#1{\renewcommand{\baselinestretch}%
{#1}\small\normalsize} \spacingset{1}

\if1\anon
{
  \title{\bf High-dimensional Markov-switching Ordinary Differential Processes}
  \author{Katherine Tsai\\
    Department of Electrical and Computer Engineering,\\ University of Illinois at Urbana-Champaign\\
    Mladen Kolar \\
    Marshall School of Business, University of Southern California\\
    Sanmi Koyejo\\
    Department of Computer Science,
    Stanford University\\
    }
  \maketitle
} \fi

\if0\anon
{
  \bigskip
  \bigskip
  \bigskip
  \begin{center}
    {\LARGE\bf High-dimensional Markov-switching Ordinary Differential Processes}
\end{center}
  \medskip
} \fi

\bigskip

\begin{abstract}
Complex systems often evolve in continuous time but are observed only at discrete time points with measurement noise. When the interactions can change across unobserved states, recovering both the state transitions and the differential equations is important for mechanistic insight
and for identifying time-varying directed graphs. However, this recovery is challenging in high dimensions. In this paper, we study estimation and graph recovery for high-dimensional Markov-switching ordinary differential equation models from noisy discrete observations. We model the switching between states by a latent continuous-time Markov chain, and within each state the drift is a sparse additive function with unknown nonlinear components.  We approximate each component by a finite basis expansion and estimate the parameters under group sparsity. Our procedure estimates the latent trajectories by wavelet smoothing and then runs a penalized surrogate expectation-maximization algorithm. Under stationarity and geometric $\beta$-mixing, given a sufficiently accurate estimated trajectory and other local regularity conditions, we show that the algorithm converges geometrically to a neighborhood of the true parameters. Finally, in an application to resting-state fMRI data, the procedure identifies multiple connectivity states, whose occupation times we compare descriptively
between the children with combined-type ADHD and the typically developing children used in the fit.
\end{abstract}

\noindent%
\vfill

\newpage
\spacingset{1.75}

\input{main/00main}
\input{main/8_acknowledgement}
\putbib[references]

\end{bibunit}

\appendix
\pagebreak

\begin{bibunit}[my-plainnat]
\input{main/01appendix}

\putbib[references]
\end{bibunit}

\end{document}

%% file: myglossary.tex
\glsxtrnewsymbol[description={number of nodes}]{p}{$p$}
\glsxtrnewsymbol[description={number of states}]{k}{$k$}
\glsxtrnewsymbol[description={number of basis functions}]{m}{$m$}
\glsxtrnewsymbol[description={number of sampling intervals; the observations are $Y_0,\ldots,Y_N$; in Section~\ref{sec:simulation} and Appendices~\ref{appendix:simulation-details} and~\ref{sec:baseline}, a sample size $N$ counts the observations}]{N}{$N$}
\glsxtrnewsymbol[description={maximum sparsity level, $s=\max_{i,\ell}\abr{\{(j,i):\|\theta_{ij}^{\ell\star}\|_2>0, j\in[p]\}}$}]{s}{$s$}

\glsxtrnewsymbol[description={coarsest resolution level of the wavelet expansion, defined in Section~\ref{ssec:waveletsmooth}}]{j0}{$j_0$}

\glsxtrnewsymbol[description={contraction constant of the population update, defined in Assumption~\ref{assumption:population}}]{kappa}{$\kappa$}
\glsxtrnewsymbol[description={$\bar\kappa=C_\kappa\kappa$, the contraction rate of Theorem~\ref{theorem:graphrecovery}}]{kappa-bar}{$\bar\kappa$}
\glsxtrnewsymbol[description={the error level of Theorem~\ref{theorem:graphrecovery:full}, bounded in Theorem~\ref{theorem:graphrecovery}; the iterates approach a neighborhood of $\Theta^\star$ of radius $e_N/(1-\bar\kappa)$}]{e-N}{$e_N$}
\glsxtrnewsymbol[description={lower bound on the curvature of the population Gram matrix weighted by the truncated smoothing probabilities, from which the restricted eigenvalue condition is derived, as described in Theorem~\ref{theorem:graphrecovery:full}}]{c-RE}{$c_{\mathrm{RE}}$}
\glsxtrnewsymbol[description={bound on the bias $\dist(M(\Theta^\star),\Theta^\star)$ of the surrogate population update at $\Theta^\star$, defined in Lemma~\ref{lemma:aggregate-surrogate-bias}}]{b-surr}{$b_{\mathrm{surr}}(\Theta^\star)$}
\glsxtrnewsymbol[description={bound on the true and estimated trajectories, $|X_i(t)|,|\hat X_i(t)|\leq R_X$, under Assumption~\ref{assumption:local}}]{RX}{$R_X$}
\glsxtrnewsymbol[description={misclassification mass (an expected soft assignment, not a classification error probability) of the posterior at $\Theta^\star$, defined in Lemma~\ref{lem:bias-sigma}}]{alphamis}{$\alpha_{\ell\ell'}$}
\glsxtrnewsymbol[description={posterior uncertainty at $\Theta^\star$, defined in Lemma~\ref{lem:bias-theta}}]{chipost}{$\chi_\ell$}
\glsxtrnewsymbol[description={root mean square of the switching remainder, $\bar\rho_i^2=N^{-1}\sum_{n=1}^N\EE[\rho_{n,i}^2]$, defined in Lemma~\ref{lem:bias-sigma}; $\bar\rho^2=p^{-1}\sum_{i=1}^p\bar\rho_i^2$}]{rhobar}{$\bar\rho_i$}
\glsxtrnewsymbol[description={curvature constant of the restricted eigenvalue condition, defined in Assumption~\ref{assumption:RE}}]{a}{$a$}
\glsxtrnewsymbol[description={tolerance of the restricted eigenvalue condition, defined in Assumption~\ref{assumption:RE}}]{nu}{$\nu$}
\glsxtrnewsymbol[description={transition-comparability constant, defined in Assumption~\ref{assumption:mixing}}]{zeta}{$\zeta$}
\glsxtrnewsymbol[description={two-sided spectral bound on the population Gram matrices $K_{\Psi_j}$, defined in Assumption~\ref{assumption:eigenvalue}}]{cpsi}{$c_{\Psi}$}
\glsxtrnewsymbol[description={the stationary distribution $\pi(Q)$ of an irreducible transition rate matrix $Q$, satisfying $\pi Q=0$}]{pi}{$\pi$}
\glsxtrnewsymbol[description={$\pi^\star=\pi(Q^\star)$ is the stationary distribution of the true latent Markov chain, satisfying $\pi^\star Q^\star=0$}]{pi-star}{$\pi^\star$}

\glsxtrnewsymbol[description={$\pi_{\min}=\inf_{Q\in\Omega_Q}\min_{\ell\in[k]}\pi_\ell(Q)$}]{pimin}{$\pi_{\min}$}
\glsxtrnewsymbol[description={$\delta_{\min}=\inf_{Q\in\Omega_Q}\min_{\ell,\ell'\in[k]}[\exp(Qh)]_{\ell\ell'}$, defined in Assumption~\ref{assumption:mixing}}]{delta-min}{$\delta_{\min}$}
\glsxtrnewsymbol[description={penalty parameter of~\eqref{eq:M_ntheta}; the experiments use $\lambda_i^\ell\equiv\lambda$}]{lambda}{$\lambda$}
\glsxtrnewsymbol[description={implementation value of the penalty parameter, $\lambda_{\mathrm{code}}=\lambda/(2\sqrt m)$; the values of $\lambda$ reported for the proposed method are values of $\lambda_{\mathrm{code}}$, as described in Appendix~\ref{sses:expsetup}}]{lambda-code}{$\lambda_{\mathrm{code}}$}
\glsxtrnewsymbol[description={threshold on $\norm{\hat{\theta}_{ij,\mathrm o}^{\ell}}_2$ that defines the estimated edge set}]{epsilon-t}{$\epsilon_{\mathrm{edge}}$}
\glsxtrnewsymbol[description={tolerance of the stopping rule of Algorithm~\ref{alg:update:ode}}]{epsilon-tol}{$\epsilon_{\mathrm{tol}}$}
\glsxtrnewsymbol[description={the permutation of $[k]$ that matches estimated to true states, defined by~\eqref{eq:permutation} (ties go to the first permutation in lexicographic order) when the minimizers based on $\{\theta_{ij}^\ell\}$ and on $Q$ agree}]{Xi}{$\Xi$}

\glsxtrnewsymbol[description={$X(t)=(X_1(t),\ldots,X_p(t))^\top\in\RR^p$ is the latent continuous trajectory, a stochastic process solving~\eqref{eq:definex}}]{X}{$X(t)$}
\glsxtrnewsymbol[description={$\hat{X}(t)=(\hat{X}_1(t),\ldots,\hat{X}_p(t))^\top\in\RR^p$ is the estimated trajectory defined in~\eqref{eq:wavelet_estimator} in Section~\ref{ssec:waveletsmooth}}]{hat-X}{$\hat{X}(t)$}
\glsxtrnewsymbol[description={a bound on $\max_{i\in[p],\,n\in[N]}\opnorm{\hat{X}_i-X_i}{[t_{n-1},t_n]}$ that holds on the event $\mathcal E_{\mathrm{rec}}$}]{delta-1}{$\delta_1$}
\glsxtrnewsymbol[description={the reconstruction event $\{\max_{i\in[p],\,n\in[N]}\opnorm{\hat X_i-X_i}{[t_{n-1},t_n]}\leq\delta_1\}$, defined in Section~\ref{ssec:main:results}; the probability bounds of Theorem~\ref{theorem:graphrecovery} and Corollaries~\ref{corollary:main} and~\ref{corollary:full} include $\PP(\mathcal E_{\mathrm{rec}}^c)$}]{E-rec}{$\mathcal E_{\mathrm{rec}}$}
\glsxtrnewsymbol[description={reconstruction radius: $\hat\Psi(t_n)$ depends only on the observations within $r_{\mathrm{rec}}$ sampling intervals of $t_n$, as assumed in Assumption~\ref{assumption:local}}]{r-rec}{$r_{\mathrm{rec}}$}
\glsxtrnewsymbol[description={period in Assumption~\ref{assumption:local}: away from the ends, the map from the observations to $\hat\Psi(t_n)$ depends on $n$ only through $n$ modulo $d_{\mathrm{rec}}$}]{d-rec}{$d_{\mathrm{rec}}$}
\glsxtrnewsymbol[description={$r_+=r+r_{\mathrm{rec}}$, the radius of the observation window on which a truncated smoothing probability computed with $\hat\Psi$ depends, as established in Lemma~\ref{lemma:local:blocks}}]{r-plus}{$r_+$}
\glsxtrnewsymbol[description={truncation radius of the truncated smoothing probabilities, defined in Section~\ref{ssec:approximateEM}}]{r}{$r$}
\glsxtrnewsymbol[description={$C_{\mathrm{col}}=\max_{a\in[k],\,i\in[p]}\sum_{j=1}^p\norm{\theta_{ji}^{a\star}}_1$, the largest column $\ell_1$ norm of the true parameters $\theta_{ji}^{a\star}$, defined before Lemma~\ref{lem:posterior-sensitivity}}]{C-col}{$C_{\mathrm{col}}$}
\glsxtrnewsymbol[description={bound on the influence of one observation on the estimated trajectory, $\int_0^T|\partial\hat X_i(u)/\partial Y_{n,i}|\,\mathrm du\leq hC_{\hat X}$ for $i\in[p]$ and $n\in\{0,\ldots,N\}$, as in Lemma~\ref{lem:posterior-sensitivity}}]{C-Xhat}{$C_{\hat X}$}
\glsxtrnewsymbol[description={$\underline\tau:=\inf_{Q\in\Omega_Q}\min_{\ell,a,b\in[k]}\EE[\tau_\ell(h)\mid Z(0)=a,Z(h)=b;Q]$, the smallest expected occupation time of a state during one sampling interval, given the states at its endpoints, in Theorem~\ref{theorem:graphrecovery:full}}]{tau-lower}{$\underline\tau$}

\glsxtrnewsymbol[description={$Z(t)\in\{1,\ldots,k\}$ is the latent continuous-time Markov chain with transition rate matrix $Q^\star$}]{Z}{$Z(t)$}
\glsxtrnewsymbol[description={independent $\mathcal N(0,I_p)$ noise vectors for $n=0,\ldots,N$, independent of $(X,Z)$}]{varepsilon-n}{$\varepsilon_n$}

\glsxtrnewsymbol[description={$Y_n=X(t_n)+\sigma^\star\varepsilon_n\in\RR^p$}]{Y}{$Y_n$}
\glsxtrnewsymbol[description={$Y_0^N=\{Y_n;n=0,\ldots,N\}$}]{Y-set}{$Y_0^N$}
\glsxtrnewsymbol[description={$Y_{n,i}^{\Delta}=Y_{n,i}-Y_{n-1,i}$, written $Y_n^\Delta$ when $i$ is fixed}]{Y-n-Delta}{$Y_{n,i}^\Delta$}

\glsxtrnewsymbol[description={$Z_0^N=(Z_0,\ldots,Z_N)$, where $Z_n:=Z(t_n)$}]{Z-set}{$Z_0^N$}
\glsxtrnewsymbol[description={drift of node $i$ in a general system $\dot X_i(t)=f_i(X(t))$, $i=1,\ldots,p$}]{f}{$f_i$}
\glsxtrnewsymbol[description={$g(\cdot)=(g_1(\cdot),\ldots,g_m(\cdot))^\top\in\RR^m$ is the vector of $m$ basis functions}]{g}{$g(\cdot)$}
\glsxtrnewsymbol[description={$\Psi_i(t_n)=\int_{t_{n-1}}^{t_n}g(X_i(u))\mathrm{d}u$}]{Psi-i}{$\Psi_i(t_n)$}
\glsxtrnewsymbol[description={$\hat{\Psi}_i(t_n)=\int_{t_{n-1}}^{t_n}g(\hat{X}_i(u))\mathrm{d}u$}]{hat-Psi-i}{$\hat{\Psi}_i(t_n)$}
\glsxtrnewsymbol[description={$\theta_{ij}^{\ell\star}\in\RR^{m}$ is the true parameter vector of the basis expansion of $f_{ij}^\ell$, for predictor node $j$ and response node $i$ in state $\ell$}]{theta-ij-l}{$\theta_{ij}^{\ell\star}$}
\glsxtrnewsymbol[description={$\delta_{ij}^\ell(\cdot)\in\RR$ is the residual function defined in~\eqref{eq:approx_f}}]{delta-ij-l}{$\delta_{ij}^\ell(\cdot)$}
\glsxtrnewsymbol[description={the set of transition rate matrices whose Markov chain is irreducible; in the theoretical analysis, the transition rate matrices whose off-diagonal entries lie within $\ell_1$-distance $r_0$ of those of $Q^\star$}]{Omega-Q}{$\Omega_Q$}
\glsxtrnewsymbol[description={the space of parameters of interest $\Omega = \Omega_Q \times \RR^{kp^2m}\times(0,\infty)$, where $\RR^{kp^2m}$ collects the $\theta_{ij}^\ell\in\RR^m$}]{Omega}{$\Omega$}
\glsxtrnewsymbol[description={$Q^{\star}\in\RR^{k\times k}$ is the true transition rate matrix}]{Q-star}{$Q^\star$}
\glsxtrnewsymbol[description={true noise standard deviation}]{sigma-star}{$\sigma^\star$}
\glsxtrnewsymbol[description={$\Theta^\star=(Q^\star,\{\theta_{ij}^{\ell\star}:\ell\in[k],\,i,j\in[p]\},\sigma^{\star2})$ is the true parameter; in vector form it is ordered as $(\{\theta_i^{\ell\star}\},Q^\star_\backslash,\sigma^{\star2})$}]{Theta-star}{$\Theta^\star$}
\glsxtrnewsymbol[description={$\tilde{E}^\ell=\{(j,i):f_{ij}^\ell\not\equiv0\}$ is the edge set of state $\ell$}]{tilde-E-ell}{$\tilde{E}^\ell$}
\glsxtrnewsymbol[description={${E}^\ell=\{(j,i): \norm{\theta_{ij}^{\ell\star}}_2>0\}$ is the edge set given by the support of $\theta_{ij}^{\ell\star}$; it equals $\tilde E^\ell$ under Assumption~\ref{assumption:edgeset}}]{E-ell}{${E}^\ell$}
\glsxtrnewsymbol[description={the deviation term, defined in~\eqref{eq:define:deltaepsilon}}]{Delta-varepsilon}{$\Delta_\varepsilon$}
\glsxtrnewsymbol[description={a quantity defined in~\eqref{eq:define:deltapsi}}]{Delta-Psi}{$\Delta_\Psi$}
\glsxtrnewsymbol[description={a quantity defined in~\eqref{eq:define:deltax}}]{Delta-X}{$\Delta_X$}
\glsxtrnewsymbol[description={a quantity defined in~\eqref{eq:define:deltaw}}]{Delta-w}{$\Delta_w$}
\glsxtrnewsymbol[description={the observed-data objective $\log p(Y_0^N;\tilde\Theta)$ of model~\eqref{eq:definey}--\eqref{eq:definex}, conditional on $X(0)$ and with $Z(t_0)\sim\pi(\tilde Q)$, as defined in Section~\ref{ssec:graphest}}]{Lcal-0N}{$\Lcal_{0,N}$}
\glsxtrnewsymbol[description={the expected complete-data log-likelihood over the discretized latent states, defined in Section~\ref{ssec:graphest}}]{Lcal-1N}{$\Lcal_{1,N}$}
\glsxtrnewsymbol[description={the surrogate objective of Definition~\ref{def:surrogate_function}}]{Lcal-star-1N}{$\Lcal^*_{1,N}$}
\glsxtrnewsymbol[description={log-likelihood of $Y_0^N$ under the surrogate model with parameter $\Theta$, computed by the forward recursion with the initial distribution in use and a Gaussian term for $Y_0$, and monitored by the stopping rule of Algorithm~\ref{alg:update:ode}}]{Lcal-star-0N}{$\Lcal^{*}_{0,N}$}
\glsxtrnewsymbol[description={the surrogate population objective defined in~\eqref{eq:poprisk}, for a record of $N$ sampling intervals}]{Lcal}{$\Lcal$}
\glsxtrnewsymbol[description={the covariance-corrected EM-type auxiliary objective defined in Appendix~\ref{ssec:corrected-loglikelihood}}]{Lcal-cc-N}{$\Lcal^{\mathrm{cc}}_N$}
\glsxtrnewsymbol[description={$\Lambda_N=\operatorname{tridiag}(-1,2,-1)\in\RR^{N\times N}$; $\sigma^{\star2}\Lambda_N$ is the covariance of the differenced measurement noise $(\sigma^\star(\varepsilon_{n,i}-\varepsilon_{n-1,i}))_{n=1}^N$ of node $i$}]{Lambda-N}{$\Lambda_N$}
\glsxtrnewsymbol[description={the lower-bidiagonal Cholesky factor of $\Lambda_N$, so that $\Lambda_N=\Gamma\Gamma^\top$; see Lemma~\ref{lem:cholesky}}]{Gamma-chol}{$\Gamma$}
\glsxtrnewsymbol[description={$\beta$-mixing coefficient; see Definition~\ref{definition:discretebetamixing} for sequences and Definition~\ref{definition:betamixing} for continuous-time processes}]{beta}{$\beta(\cdot)$}
\glsxtrnewsymbol[description={father wavelet in the unit-interval coordinate $u$}]{phi}{$\phi(u)$}
\glsxtrnewsymbol[description={mother wavelet in the unit-interval coordinate $u$}]{psi}{$\psi(u)$}
\glsxtrnewsymbol[description={$\phi_{jl}(t/T)=2^{j/2}\phi(2^jt/T-l)$ for physical time $t\in[0,T]$, with boundary adjustment as in Appendix~\ref{appendix:wavelet-reconstruction}}]{phi-jell}{$\phi_{jl}(t/T)$}
\glsxtrnewsymbol[description={$\psi_{jl}(t/T)=2^{j/2}\psi(2^jt/T-l)$ for physical time $t\in[0,T]$, with boundary adjustment as in Appendix~\ref{appendix:wavelet-reconstruction}}]{psi=jell}{$\psi_{jl}(t/T)$}
\glsxtrnewsymbol[description={$\tilde{X}_i(t)=N^{-1/2}\sum_{n=1}^NY_{n,i}\phi_{J,n-1}(t/T)$ for $t\in[0,T]$, where $2^J=N$}]{tilde-X}{$\tilde{X}_i(t)$}
\glsxtrnewsymbol[description={$\hat{\xi}_{i,j_0l}=T^{-1}\int_0^T\tilde{X}_i(t)\phi_{j_0l}(t/T)\,\mathrm{d}t$ for $i\in[p]$ and $l=0,\ldots,2^{j_0}-1$}]{xi-ijl}{$\hat{\xi}_{i,j_0l}$}
\glsxtrnewsymbol[description={$\tilde{\eta}_{i,jl}=T^{-1}\int_0^T\tilde{X}_i(t)\psi_{jl}(t/T)\,\mathrm{d}t$}]{eta-ijl}{$\tilde{\eta}_{i,jl}$}
\glsxtrnewsymbol[description={$\hat{\eta}_{i,jl}=\sgn(\tilde{\eta}_{i,jl})(\abr{\tilde{\eta}_{i,jl}}-\lambda_{jl})_+$}]{hat-eta-ijell}{$\hat{\eta}_{i,jl}$}
\glsxtrnewsymbol[description={$m_{\ell\ell'}(Y_0^N;\Theta)=\EE[m_{\ell\ell'}\mid Y_0^N;\Theta]$ is the posterior expected number of transitions of $Z$ from state $\ell$ to $\ell'$ in $[t_0,t_N]$, where $m_{\ell\ell'}$ is the number of these transitions}]{m-ellel}{$m_{\ell\ell'}(Y_0^N;\Theta)$}
\glsxtrnewsymbol[description={$\tau_{\ell}(Y_0^N;\Theta)=\EE[\tau_\ell\mid Y_0^N;\Theta]$ is the posterior expected occupation time of state $\ell$ in $[t_0,t_N]$, where $\tau_\ell$ is this occupation time}]{tau-ell}{$\tau_\ell(Y_0^N;\Theta)$}
\glsxtrnewsymbol[description={the truncated smoothing probability $p(Z(t_n)=\ell\mid Y_{(n-r)\vee 0},\ldots, Y_{(n+r)\wedge N};\Theta)$, computed under the surrogate model with $\hat\Psi$ and started from $\pi(Q)$, as defined in Section~\ref{ssec:approximateEM}}]{hat-w-elltn}{$\hat{w}_{\ell,\Theta}(t_n)$}
\glsxtrnewsymbol[description={the smoothing probability $p(Z(t_n)=\ell\mid Y_0^N;\Theta)$ under the surrogate model with $\hat\Psi$, started from $\pi(Q)$, as defined in Section~\ref{ssec:approximateEM}; $w_{\ell,n}:=w_{\ell,\Theta^\star}(t_n)$}]{w-elltn}{$w_{\ell,\Theta}(t_n)$}
\glsxtrnewsymbol[description={$\bar\eta_r=C_{\mathrm{sm}}k\delta_{\min}^{-1}\zeta^{-8}\pi_{\min}^{-2}\{1-(\zeta\pi_{\min})^2\}^{r-1}$, a bound on the differences between the full and the truncated marginal and joint smoothing probabilities computed with $\hat\Psi$ and started from $\pi(Q)$, defined before Corollary~\ref{corollary:full}}]{eta-bar-r}{$\bar\eta_r$}
\glsxtrnewsymbol[description={$K_{\Psi_j}=N^{-1}\sum_{n=1}^N\EE[\Psi_j(t_n)\Psi_j^\top(t_n)]$}]{K-psi}{$K_{\Psi_j}$}
\glsxtrnewsymbol[description={$\hat{K}_{\Psi_j}=N^{-1}\sum_{n=1}^N\hat\Psi_j(t_n)\hat\Psi_j^\top(t_n)$}]{hat-K-psi}{$\hat{K}_{\Psi_j}$}
\glsxtrnewsymbol[description={$\hat{\Kcal}_\Psi=\{\hat{K}_{\Psi_j};j=1,\ldots,p\}$}]{Kcal}{$\hat{\Kcal}_\Psi$}

%% file: main/00main.tex
\section{Introduction}

Ordinary differential equations (ODEs) have long been used to model the dynamics
of complex systems, including disease
progression~\citep{wu2005statistical} and neuronal
activity~\citep{friston2003dynamic}. Classical approaches
estimate the parameters of a known ODE by two-step
collocation~\citep{varah1982spline, liang2008parameter} or by generalized
smoothing~\citep{ramsay2007parameter}. More recently, a number of data-driven procedures
have been proposed to estimate such ODEs from
observations~\citep{dattner2015optimal, chen2017network, pfister2019learning}.
These procedures typically model a $p$-dimensional trajectory
$X(t) = (X_1(t), \ldots, X_p(t))^\top$ with initial condition $X(0) = x_0$
through an ODE of the form $\dot{X}(t) = f(X(t))$, and aim to recover
the unknown vector field $f$ or its parameters. However, in practice the
trajectory $X(t)$ is not observed continuously. Instead, it is
sampled at discrete time points with additive noise,
\[
Y_n = X(t_n) + \sigma^\star\varepsilon_n, \quad n = 0, 1, \ldots, N,
\]
where $t_n =nh$, $h>0$ is the sampling period, $\sigma^\star>0$ is the noise level and $\varepsilon_n$ are independent noise terms.

Beyond parameter recovery, the primary scientific interest often lies in interpreting these parameters to understand the interactions among system components, or ``nodes'', $X_1, \ldots, X_p$. If $\mathrm{d}X_i(t)/\mathrm{d}t$ depends on $X_j(t)$, we say that node $j$ influences node $i$, which corresponds to a directed edge $j\to i$ in the underlying graph. Identifying such relationships gives insight into mechanisms such as brain connectivity~\citep{friston2003dynamic} or gene regulation~\citep{chen1999modeling}. However, \citet{dattner2015optimal} estimate the parameters of a specified ODE, and \citet{chen2017network} and \citet{pfister2019learning} target directed dependencies that do not change over time. In practice, the interactions between nodes may evolve over time. For example, brain networks are widely thought to exhibit time-varying connectivity, although its statistical validity and physiological origins remain debated~\citep{lurie2020questions}. This motivates ODE models whose graph can switch over time.

Our approach to modeling time-varying ODEs is motivated by resting-state brain dynamics measured by functional magnetic resonance imaging (fMRI). Empirical studies suggest that brain activity exhibits recurring patterns, which are often modeled by switching between latent states~\citep{vidaurre2017brain, park2021state}. In this framework, each state $\ell$ is governed by its own system of differential equations, which encodes a distinct directed graph on the nodes. The external factors that drive the switches between states are typically unobserved. Therefore, we model these switches by a continuous-time hidden Markov model, in which a latent Markov chain selects the active state-specific ODE and the resulting process is observed at discrete time points. Our formulation differs from most prior work in neuroscience, which focuses on estimating undirected functional connectivity~\citep{vidaurre2017brain, tsai2022nonconvex, tsai2024latent}. We instead aim to recover directed dynamical dependencies among the measured regional signals under an estimated ODE model. Our modeling perspective is similar to that of~\citet{park2021state}, who used a Bayesian approach in a low-dimensional setting. In comparison, our framework allows for high-dimensional systems and nonlinear differential equations, which makes it applicable to more complex biological dynamics.

We focus on high-dimensional settings. In such settings, sparsity-inducing penalties are widely used to improve both interpretability and statistical efficiency. This strategy has been successful in a variety of problems, including graphical model selection~\citep{friedman2008sparse} and high-dimensional vector autoregression~\citep{basu2015regularized}. However, incorporating sparsity into ODE models is challenging, primarily because the functional forms of the state-specific dynamics are unknown and potentially nonlinear. To address this problem, \citet{henderson2014network} and \citet{chen2017network} proposed to approximate each component function by an additive model with a fixed dictionary of basis functions. Estimation then reduces to a group-sparse regression on the basis coefficients, which allows direct regularization of the model parameters. In this paper, we adopt the additive modeling strategy of~\citet{chen2017network} and extend it to a continuous-time hidden Markov setting, which allows for switching dynamics and time-varying graphs.

Our procedure has two steps. First, we estimate each latent trajectory $X_i(t)$ from its noisy samples by wavelet smoothing, which accommodates the non-smooth points that switches between states create. Second, we estimate the state-specific graphs and the transition rate matrix by a surrogate expectation-maximization (EM) algorithm. Note that the exact likelihood is intractable, since both $X(t)$ and the continuous-time path of the latent Markov chain are unobserved. Therefore, we plug the estimated trajectories into a surrogate complete-data objective for the increments $Y_n-Y_{n-1}$, which freezes the state within each sampling interval. The E-step computes the smoothing probabilities by the forward--backward algorithm, and the M-step consists of a group lasso for each node and state together with closed-form updates of the noise variance and the transition rate matrix.

\subsection{Related work and our contributions}

A closely related class of models is the Markov-switching vector autoregressive (MS-VAR) model, which has been used in a variety of applications~\citep{krolzig2013markov}. In an MS-VAR model, the observed time series follows a vector autoregressive process whose coefficients are governed by an unobserved finite-state Markov chain. The parameters are typically estimated by the EM algorithm~\citep{dempster1977maximum}. Several recent works have extended the MS-VAR model to high-dimensional settings by incorporating sparsity constraints on the autoregressive parameters~\citep{monbet2017sparse, chavez2023penalized}. In particular, \citet{li2022estimation} established convergence guarantees and finite-sample bounds on the statistical error, which lay the groundwork for the analysis of high-dimensional hidden Markov models. Our model shares with MS-VAR models the latent state transitions and the estimation by EM. However, it differs fundamentally in that it models continuous-time nonlinear dynamics by ordinary differential equations, which introduces additional challenges such as trajectory estimation and continuous-time integration from noisy discrete observations.

Switching state-space models have also been developed in signal
processing and machine learning, including Bayesian nonparametric switching
linear dynamical systems~\citep{fox2011bayesian} and recurrent switching linear
dynamical systems, whose discrete-state transitions depend on the continuous
state~\citep{linderman2017bayesian}. In control, related problems are studied as
identification of hybrid systems~\citep{paoletti2007identification}. The
switching state-space models of~\citet{fox2011bayesian}
and~\citet{linderman2017bayesian} use discrete-time dynamics, whereas we estimate
sparse directed graphs in high-dimensional, continuous-time nonlinear systems.

An EM algorithm can converge to local optima or saddle points.
\citet{balakrishnan2017statistical} and \citet{yi2015regularized} give
conditions under which EM and its regularized version, initialized near the true
parameter, converge geometrically to a statistically small neighborhood of it.
However, this line of work largely focuses on i.i.d.\ data, and its analysis
does not apply to temporally dependent observations. More
recently,~\citet{yang2017statistical} studied the convergence of the Baum--Welch
algorithm for hidden Markov models with isotropic Gaussian emissions,
and~\citet{li2022estimation} studied the convergence of an approximate EM
algorithm for MS-VAR models. To handle dependent
samples, both works analyze a truncated EM algorithm, which approximates the
original EM algorithm by replacing the posterior of each latent state given all
observations with its posterior given only the observations in a window around
it. We adopt this idea in our analysis. Establishing theoretical guarantees in
our setting remains challenging, requiring additional steps and techniques to
show that the updates of the truncated surrogate EM algorithm closely approximate the original updates under the
proposed surrogate objective.

 Our contributions are threefold. First, we propose a Markov-switching additive
ODE model whose interactions and directed graph can change across latent
states, and a two-step estimator that combines wavelet
smoothing of the trajectories with a surrogate EM algorithm with group-sparse
M-steps. Second, assuming local contraction of the population update, uniform
concentration near the true parameters, an accurate estimated trajectory, and a
good initialization, we show that the iterates of the surrogate EM algorithm converge geometrically to a neighborhood of the
true parameters. Our finite-sample bounds separate four sources of error:
statistical, reconstruction, truncation, and surrogate bias. 
As a corollary, we give conditions under which thresholding recovers the
state-specific graphs. Third, we illustrate the method in simulations and on
resting-state fMRI data from the NYU ADHD dataset, where it estimates
differences in state occupation times between two groups of
children.

\section{Background and problem setup}
\label{sec:background}

In this section, we introduce the model and the assumptions used in our analysis. We start with the model, its discretization under slow switching, and the estimation target.

\textbf{Notation.}
\label{ssec:notation}
For a positive integer $i$, let $[i]:=\{1,\ldots,i\}$. For $T>0$, we let $L_2([0,T])$ be the space of square-integrable functions on $[0,T]$, with the inner product $\dotp{f}{g}=\int_0^Tf(t)g(t)\mathrm{d}t$ and the norm $\opnorm{f}{}=(\dotp{f}{f})^{1/2}$ for $f,g\in L_2([0,T])$.
For an interval $I\subseteq[0,T]$, we write
$\opnorm{f}{I}:=\bigl(\int_I |f(t)|^2\,\mathrm{d}t\bigr)^{1/2}$
for the $L_2$ norm on $I$.
For $x\in\RR^p$, let $\norm{x}_2=(\sum_{i=1}^p x_i^2)^{1/2}$, $\norm{x}_1=\sum_{i=1}^p\abr{x_i}$, and $\norm{x}_\infty = \max_{i=1,\ldots,p} \abr{x_i}$. For a positive semidefinite $A\in\RR^{d\times d}$ and $x\in \RR^d$, let $\norm{x}_A=(x^\top A x)^{1/2}$, a seminorm that is a norm when $A\succ0$. For a stacked vector $x=(x_1^\top,\ldots,x^\top_p)^\top\in\RR^{pm}$ with $x_i\in \RR^m$, and a collection $\Acal=\{A_i\}_{i=1}^p$ of positive semidefinite matrices $A_i\in \RR^{m\times m}$, let $\norm{x}_{\infty, \Acal}=\max_i\norm{x_i}_{A_i}$ and $\norm{x}_{1,\Acal}=\sum_i\norm{x_i}_{A_i}$; these are norms when every $A_i\succ0$. For a matrix $B$, let $\sigma_{\min}(B)$ and $\sigma_{\max}(B)$ denote the minimum and the maximum singular values of $B$, respectively. We write $a\lesssim b$, or equivalently $b\gtrsim a$, if $a\leq cb$ for a constant $c>0$ that does not depend on $N$ or $p$, and $a\asymp b$ if $a\lesssim b$ and $b\lesssim a$. A maximum over an empty index set is taken to be $0$; for $k=1$, this applies to maxima over pairs $\ell'\neq\ell''$. A list of the main symbols is provided in the Appendix.

\subsection{Setup}
\label{ssec:setup}
We observe $Y_n=(Y_{n,1},\ldots,Y_{n,p})^\top\in\RR^p$ at times $t_n=nh$, $n=0,1,\ldots,N$, where $h>0$ is the sampling period and $T=Nh$. The observations are driven by two unobserved processes on $[0,T]$: a latent continuous-time Markov chain $(Z(t))_{t\in[0,T]}$ with $Z(t)\in[k]$ and almost surely right-continuous paths, whose value we call the state of the process, and the trajectory $X(t)\in\RR^p$ of an ODE system whose drift depends on $Z(t)$.

We assume that $(Z(t))_{t\in[0,T]}$ is time-homogeneous and irreducible. Since its state space is finite, it is positive recurrent with a unique stationary distribution. Let $Q^\star=(q^\star_{\ell\ell'})\in\RR^{k\times k}$ denote its transition rate matrix, so that $q_{\ell\ell'}^\star\ge 0$ for $\ell\ne \ell'$ and $q_{\ell\ell}^\star=-\sum_{\ell'\ne \ell}q_{\ell\ell'}^\star$.
For any irreducible transition rate matrix $Q$, we denote its stationary distribution by $\pi(Q)$, and we write $\pi^\star:=\pi(Q^\star)$ for the stationary distribution of the true latent Markov chain, so that $\pi^\star Q^\star=0$.
We study the following Markov-switching additive ODE model:
\begin{align}
    Y_n &= {{X(t_n)}}+\sigma^\star \varepsilon_n,\label{eq:definey}\\
    \dot{X}_i(t) &{}= \sum_{j=1}^pf^{Z(t)}_{ij}(X_j(t)),\quad i\in[p],\label{eq:definex}
\end{align}
where $f_{ij}^\ell:\RR\rightarrow\RR$ for $i,j\in[p]$, $\ell\in[k]$, and $\varepsilon_0,\ldots,\varepsilon_N$ are i.i.d.\ $\mathcal{N}(0, I_p)$, independent of $(X, Z)$, and $\sigma^\star>0$ is the noise level. Since the drift, namely $\dot{X}(t)$, can jump when $Z$ switches, $X$ is absolutely continuous and \eqref{eq:definex} holds for almost every $t$; equivalently, $X_i(t)=X_i(0)+\int_0^t\sum_{j=1}^pf_{ij}^{Z(u)}(X_j(u))\,\mathrm{d}u$ for all $t$.

In the neuroscience application, $\{Y_n\}_{n\in[N]\cup\{0\}}$ is the observed fMRI time course, $X(t)$ is a latent continuous trajectory of the blood-oxygen-level-dependent (BOLD) signal, of which $Y_n$ is a noisy sample, and $Z(t)$ is the brain state that governs its dynamics~\citep{vidaurre2017brain}. The pattern of the functions $f_{ij}^{\ell}$ encodes the directed dynamical dependencies among the regional signals. In particular, if $f_{ij}^{\ell}\not\equiv0$, then node $j$ influences the dynamics of node $i$ under brain state $\ell$, and we write $j\rightarrow i$. Therefore, estimating these dependencies is a graph recovery problem.

Following~\citet{henderson2014network} and~\citet{chen2017network}, who proposed additive drifts as a tractable approximation to a general system $\dot{X}_i(t)=f_i(X(t))$, we approximate each $f_{ij}^{\ell}$ by a finite basis expansion.
Let $g(\cdot)=(g_1(\cdot),\ldots,g_m(\cdot))^\top\in\RR^m$ be a finite-dimensional basis of linearly independent functions $g_1,\ldots,g_m:\RR\to\RR$ whose span contains no nonzero constant function, and write
\begin{equation}\label{eq:approx_f}
f_{ij}^\ell(X_j(t)) = g(X_j(t))^\top \theta_{ij}^{\ell\star}+\delta_{ij}^{\ell}(X_j(t))\quad i,j=1,\ldots, p,\;\ell=1,\ldots, k,
\end{equation}
where $\theta_{ij}^{\ell\star}\in \RR^{m}$ is a column vector of parameters and $\delta_{ij}^{\ell}(X_j(t))$ is the residual function, which captures the approximation error. Since~\eqref{eq:approx_f} holds for any $\theta_{ij}^{\ell\star}$ with a suitable residual, it determines $\theta_{ij}^{\ell\star}$ only under Assumption~\ref{assumption:edgeset} below.

\textbf{Approximate additive ODEs under slow switching.}
\label{ssec:approximateODE}
Next, we write the observations in terms of the parameters $\theta_{ij}^{\ell\star}$. Integrating~\eqref{eq:definex} over $[t_{n-1},t_n]$, substituting~\eqref{eq:approx_f}, and using~\eqref{eq:definey}, we have that, for each $n\in[N]$ and each node $i\in[p]$,
\begin{align*}
    Y_{n,i}
    = X_i(t_{n-1}) + \sum_{j=1}^p\int_{t_{n-1}}^{t_n} g(X_j(u))^\top \theta_{ij}^{Z(u)\star}\mathrm{d}u
    + \sum_{j=1}^p\int_{t_{n-1}}^{t_n}\delta_{ij}^{Z(u)}(X_j(u))\mathrm{d}u
    +\sigma^\star\varepsilon_{n,i}.
\end{align*}
Let ${\Psi}_j(t_n)=\int_{t_{n-1}}^{t_n}g({X}_j(u))\,\mathrm{d}u$. Adding and subtracting $\sum_{j=1}^p\Psi_j(t_n)^\top\theta_{ij}^{Z(t_n)\star}$, we obtain
\begin{align}
    Y_{n, i}&= X_i(t_{n-1}) + \sum_{j=1}^p{\Psi}_j(t_n)^\top \theta_{ij}^{Z(t_n)\star}
    +\rho_{n,i}
    +
r_{n,i}+\sigma^\star\varepsilon_{n,i}\label{eq:approxY},
\end{align}
where
\begin{align*}
\rho_{n,i}=\sum_{j=1}^p\int_{t_{n-1}}^{t_n} g(X_j(u))^\top\theta_{ij}^{Z(u)\star}\,\mathrm{d}u - \sum_{j=1}^p{\Psi}_j(t_n)^\top\theta_{ij}^{Z(t_n)\star},\;
r_{n,i}=\sum_{j=1}^p\int_{t_{n-1}}^{t_n}\delta_{ij}^{Z(u)}(X_j(u))\,\mathrm{d}u.
\end{align*}
Note that the switching remainder $\rho_{n,i}$ vanishes unless $Z$ jumps in $(t_{n-1},t_n]$, an event of probability at most $h\max_{\ell\in[k]}|q^\star_{\ell\ell}|$, so that $\rho_{n,i}=0$ on most intervals when the switching rate is slow relative to the sampling rate $1/h$, which is the setting we consider. The basis residual $r_{n,i}$ vanishes under Assumption~\ref{assumption:edgeset} below, which we maintain throughout the theoretical analysis.

\textbf{Estimation target.}
For each state $\ell\in[k]$, we define the edge set $\tilde{E}^\ell = \{(j,i): f_{ij}^\ell\not\equiv0\}$, and the graph of state $\ell$ as the directed graph with node set $[p]$ and edge set $\tilde{E}^\ell$. Our goal is to recover the edge sets $\tilde{E}^\ell$, the transition rate matrix $Q^\star$, and the parameters $\{\theta_{ij}^{\ell\star}:\ell\in[k], i,j\in[p]\}$ from $\{Y_n\}_{n\in[N]\cup\{0\}}$. We estimate the edge sets through the supports of these parameters,
\[
E^\ell = \left\{(j,i):\norm{\theta_{ij}^{\ell\star}}_2>0\right\}.
\]
Note that $E^\ell$ need not equal $\tilde E^\ell$ when~\eqref{eq:approx_f} holds only approximately, since a small approximation error does not rule out that a weak component $f_{ij}^\ell$ is missing from $E^\ell$. Relating the two edge sets requires a minimum-signal condition on the nonzero components. To keep the analysis tractable, we assume that the basis represents each $f_{ij}^\ell$ exactly.
\begin{assumption}\label{assumption:edgeset}
    For all $i,j\in[p]$ and $\ell\in[k]$, the function $f_{ij}^\ell$ lies in the span of $g_1,\ldots,g_m$, that is, $\delta_{ij}^\ell\equiv0$ in~\eqref{eq:approx_f}.
\end{assumption}
Under Assumption~\ref{assumption:edgeset}, $r_{n,i}=0$ in~\eqref{eq:approxY}. Furthermore, since $g_1,\ldots,g_m$ are linearly independent, $\theta_{ij}^{\ell\star}$ is uniquely determined by $f_{ij}^\ell$, and $\tilde E^\ell=E^\ell$ for every $\ell\in[k]$. Moreover, the states are identified only up to a common permutation of their labels; the local convergence analysis of Section~\ref{sec:theory} fixes the labelling of $\Theta^\star$.
We write $\Theta=(Q, \{\theta_{ij}^\ell:\ell\in[k], i,j\in[p]\}, \sigma^2)$ for a generic parameter and $\Theta^\star=(Q^\star, \{\theta_{ij}^{\ell\star}:\ell\in[k], i,j\in[p]\}, \sigma^{\star2})$ for the true one. Let $\Omega_Q$ be the set of transition rate matrices whose Markov chain is irreducible, so that $Q^\star\in\Omega_Q$, and define the parameter space $\Omega = \Omega_Q \times \RR^{kp^2m}\times(0,\infty)$, where $\RR^{kp^2m}$ collects the parameters $\theta_{ij}^{\ell}\in\RR^{m}$, $\ell\in[k]$, $i,j\in[p]$. In the theoretical analysis of Section~\ref{sec:theory}, we further restrict $\Omega_Q$ to a neighborhood of $Q^\star$, defined in Section~\ref{ssec:popEM}.

\section{Methodology}
\label{sec:method}

In this section, we describe the procedure for estimating $\Theta^\star$, which follows the two-step collocation framework that has been widely used to estimate ODEs~\citep{dattner2015optimal, chen2017network}. In the first step, we estimate the continuous trajectory $X(t)$ from the noisy discrete observations $\{Y_n\}_{n\in[N]\cup\{0\}}$ by wavelet smoothing~\citep{brown1998wavelet}. In the second step, we plug the estimated trajectory into a surrogate EM algorithm that estimates $\Theta^\star$. Finally, we describe how to select the tuning parameters.

\subsection{Step 1: Wavelet smoothing}
\label{ssec:waveletsmooth}
For each dimension $i\in[p]$, we estimate the univariate function $X_i(t)$ from $\{Y_{n,i}\}_{n\in[N]}$ by the wavelet shrinkage estimator of~\citet{brown1998wavelet}. We use wavelets because switches between states make the trajectory non-differentiable at the switching times.

Let $\phi$ and $\psi$ be the father and mother wavelets, respectively, and write $\phi_{jl}(u)=2^{j/2}\phi(2^ju-l)$ and $\psi_{jl}(u)=2^{j/2}\psi(2^ju-l)$ for $u\in[0,1]$, with boundary adjustments as described in Appendix~\ref{appendix:wavelet-reconstruction}.
For simplicity of presentation, let $J$ be an integer such that $N=2^J$. Throughout, $t\in[0,T]$ denotes physical time, and wavelet functions defined on $[0,1]$ are evaluated at $t/T$.
For a coarsest level $j_0\geq0$, the resulting estimated trajectory $\hat X_i(t)$ is
\begin{equation}
\hat{X}_i(t)=\sum_{l=0}^{2^{j_0}-1}\hat{\xi}_{i,j_0l}\phi_{j_0l}(t/T)
+\sum_{j=j_0}^{J-1}\sum_{l=0}^{2^j-1}\hat{\eta}_{i,jl}\psi_{jl}(t/T).
\label{eq:wavelet_estimator}
\end{equation}
Here $\hat\xi_{i,j_0l}$ and $\hat\eta_{i,jl}$ are the estimated scaling and soft-thresholded detail coefficients, respectively, as defined in Appendix~\ref{appendix:wavelet-reconstruction}.

\subsection{Step 2: Graph estimation via the surrogate EM algorithm}\label{ssec:graphest}

Next, we describe the surrogate EM algorithm for model~\eqref{eq:definey}--\eqref{eq:definex}, whose hidden Markov process $(X,Z)$ is observed at discrete times through $Y_0^N=\{Y_n;n\in[N]\cup\{0\}\}$.
Throughout, $\tilde\Theta=(\tilde Q,\{\tilde\theta_{ij}^\ell\},\tilde\sigma^2)$ denotes the parameter over which we optimize, and $\Theta$ denotes the current iterate. Let $Z_n:=Z(t_n)$ and $Z_0^N:=(Z_0,\ldots,Z_N)$. Under model~\eqref{eq:definey}--\eqref{eq:definex}, conditional on $X(0)$ and with $Z(t_0)\sim\pi(\tilde Q)$, let $\Lcal_{0,N}(\tilde{\Theta})=\log p(Y_0^N;\tilde{\Theta})$ denote the observed-data objective, and let $
\Lcal_{1,N}(\tilde{\Theta}\mid \Theta) = \int \log P(Y_0^N, Z_0^N;\tilde{\Theta})\,\mathrm{d}P(Z_0^N\mid Y_0^N;\Theta)$ denote the expected complete-data log-likelihood over the discretized latent states $Z_0^N$.
By Jensen's inequality, $\Lcal_{0,N}(\tilde{\Theta})\geq\Lcal_{1,N}(\tilde{\Theta}\mid\Theta)-\int\log P(Z_0^N\mid Y_0^N;\Theta)\,\mathrm{d}P(Z_0^N\mid Y_0^N;\Theta)$. However, the exact $\Lcal_{1,N}(\tilde{\Theta}\mid \Theta)$ is computationally intractable. Therefore, we use the following surrogate.
\begin{definition}\label{def:surrogate_function}
The surrogate objective is defined as
\begin{multline*}
  {\Lcal}_{1,N}^*(\tilde{\Theta}\mid \Theta)=\sum_{\ell\ne\ell'} m_{\ell\ell'}(Y_0^N;\Theta)\log \tilde{q}_{\ell\ell'}+\sum_{\ell=1}^k\tilde{q}_\ell\tau_{\ell}(Y_0^N;\Theta)-\sum_{n=1}^N \frac{p}{2}(\log2\tilde{\sigma}^2+\log 2\pi)\\
    -\frac{1}{4\tilde{\sigma}^2}\sum_{n,\ell,i=1}^{N,k,p}
p(Z(t_n)=\ell\mid Y_0^N ; \Theta)\rbr{Y_{n,i}-Y_{n-1,i} - \sum_{j=1}^p\hat{\Psi}_j(t_n)^\top\tilde{\theta}_{ij}^{\ell}}^2.
\end{multline*}
where $\hat{\Psi}_j(t_n)=\int_{t_{n-1}}^{t_n}g(\hat{X}_j(u))\,\mathrm{d}u$, $\tilde{q}_\ell=-\sum_{\ell'\ne\ell}\tilde{q}_{\ell\ell'}$, and $m_{\ell\ell'}(Y_0^N;\Theta)$ and $\tau_\ell(Y_0^N;\Theta)$ are the posterior expected number of transitions from $\ell$ to $\ell'$ and the posterior expected occupation time of state $\ell$, respectively, where the occupation time of $\ell$ is the time that $Z(t)$ spends in $\ell$ during $[t_0,t_N]$.
\end{definition}
Appendix~\ref{sec:surrogate-explanation} derives and discusses this surrogate objective. The surrogate replaces $\Psi_j(t_n)$ by $\hat{\Psi}_j(t_n)$, which is computed from the estimated trajectory $\hat X$, and makes two further approximations. First, it replaces the exact integrated drift by $\sum_{j=1}^p\Psi_j(t_n)^\top\tilde{\theta}_{ij}^{\ell}$, which freezes the state within each sampling interval and omits $\rho_{n,i}$ in~\eqref{eq:approxY}. Second, it treats the differenced noise terms $\sigma^\star(\varepsilon_{n,i}-\varepsilon_{n-1,i})$ as independent across $n$, which gives the diagonal covariance $2\sigma^2 I$, whereas their covariance is $\sigma^{\star2}\Lambda_N$, where $\Lambda_N$ is tridiagonal with $2$ on the diagonal and $-1$ on the two adjacent diagonals. Appendix~\ref{ssec:corrected-loglikelihood} gives the covariance-corrected objective and explains why its dense inverse precludes the standard forward--backward algorithm in the E-step. These approximations make the surrogate computationally tractable, but they introduce a bias, which we analyze in Section~\ref{sec:theory}. Finally, note that the Markov-chain part of $\Lcal_{1,N}^*$ conditions on the initial state $Z(t_0)$ and omits the term $\log\pi_{Z(t_0)}(\tilde Q)$.

Algorithm~\ref{alg:update:ode} summarizes the surrogate EM algorithm. At each M-step, we maximize $\Lcal_{1,N}^*$ over $Q$ and $\sigma^2$ separately, with the other parameters fixed. To encourage sparse graphs, we update each $\theta_i^\ell$ by the group lasso problem below. We denote the resulting M-step map by $M_N(\Theta)$, whose transition-rate and variance components are
\begin{align*}
    M_{N, q_{\ell\ell'}}(\Theta) &= \frac{m_{\ell\ell'}(Y_0^N;\Theta)}{\tau_\ell(Y_0^N;\Theta)}\quad \ell\neq\ell', \quad M_{N,q_{\ell\ell}}(\Theta) = - \sum_{\ell'\neq\ell}M_{N, q_{\ell\ell'}}(\Theta);\\
    M_{N,\sigma^2}(\Theta) &= \frac{1}{2pN}\sum_{n=1}^N\sum_{\ell=1}^k\sum_{i=1}^p p(Z(t_n)=\ell\mid Y_0^N;\Theta)\rbr{Y_{n,i}-Y_{n-1,i}-\sum_{j=1}^p\hat{\Psi}_j(t_n)^\top \theta_{ij}^\ell}^2.
\end{align*}
For $\theta_i^\ell=(\theta_{i1}^{\ell\top},\ldots,\theta_{ip}^{\ell\top})^\top$, we define the update of $\theta_i^\ell$ by the following group lasso problem, which we solve with the accelerated proximal-gradient method FISTA~\citep{beck2009fast}:
\begin{multline}\label{eq:M_ntheta}
M_{N,\theta_i^\ell}(\Theta)=\argmax_{\tilde{\theta}_{i}^\ell}-\frac{1}{N}\sum_{n=1}^{N}
p(Z(t_n)=\ell\mid Y_0^N ; \Theta)\rbr{Y_{n,i}-Y_{n-1,i} - \sum_{j=1}^p\hat{\Psi}_j(t_n)^\top\tilde{\theta}_{ij}^{\ell}}^2\\
-\displaystyle\lambda_i^\ell\sum_{j=1}^p\cbr{N^{-1}\sum_{n=1}^N(\hat{\Psi}_j(t_n)^\top\tilde{\theta}^\ell_{ij})^2}^{1/2}.
\end{multline}
For convenience in computation and analysis, we use $1/N$ in place of $1/(4\sigma^2)$ to express the loss as an average weighted squared error, with $\lambda_i^\ell$ defined on this scale.

Here $\lambda_i^\ell>0$ is the penalty parameter for node $i$ in state $\ell$. The penalty encourages entire blocks $\tilde\theta^\ell_{ij}$ to be zero, so that $\lambda_i^\ell$ controls the sparsity of the estimated graph. The penalty parameters $\lambda_i^\ell$ in Section~\ref{sec:theory} refer to~\eqref{eq:M_ntheta}.
After obtaining $M_N(\Theta)$, the E-step recomputes the smoothing probabilities at $M_N(\Theta)$ for the next M-step. In particular, the smoothing probabilities $p(Z(t_n)=\ell\mid Y_0^N;\Theta)$ and the joint smoothing probabilities $p(Z(t_{n-1})=\ell, Z(t_n)=\ell'\mid Y_0^N;\Theta)$ are computed by the forward--backward algorithm~\citep{baum1970maximization}. Using these probabilities, we then compute $m_{\ell\ell'}(Y_0^N;\Theta)$ and $\tau_\ell(Y_0^N;\Theta)$ in Definition~\ref{def:surrogate_function}, as is standard in the estimation of continuous-time Markov chains~\citep{bladt2005statistical,hobolth2005statistical}. Appendix~\ref{appendix:algorithm:details} gives the details.

Finally, given the output $\hat{\Theta}_{\mathrm o}$ and a threshold $\epsilon_{\mathrm{edge}}\geq0$, we estimate the edge sets by
$\hat{E}^\ell = \left\{(j,i):\norm{\hat{\theta}_{ij,\mathrm o}^{\ell}}_2>\epsilon_{\mathrm{edge}}\right\}, \ell\in[k]
$.

\begingroup
\makeatletter
\let\originalfloatmakebox\float@makebox
\def\float@makebox#1{\color@vbox\normalcolor\originalfloatmakebox{#1}\color@endbox}
\makeatother
\begin{algorithm}[t!]
 \caption{Surrogate EM Algorithm}\label{alg:update:ode}
\begin{algorithmic}
\State \textbf{Input:} data $Y_0^N$ and $\{\hat{\Psi}_i(t_n);i\in[p],n\in[N]\}$; initial parameter $\Theta^{(0)}$; penalty parameter $\lambda$; tolerance $\epsilon_{\mathrm{tol}}$; iteration limit $t_{\max}$.
\For{$t=0,1,\dots,t_{\max}-1$}
 \State \textbf{E-step:} compute $p(Z(t_n)\mid Y_0^N;\Theta^{(t)})$ and $p(Z(t_{n-1}),Z(t_n)\mid Y_0^N;\Theta^{(t)})$ by the forward--backward algorithm, and then $m_{\ell\ell'}(Y_0^N;\Theta^{(t)})$ and $\tau_\ell(Y_0^N;\Theta^{(t)})$ by~\eqref{eq:computemll} and~\eqref{eq:computetaul}.
 \State \textbf{M-step:} set $Q^{(t+1)}=M_{N,Q}(\Theta^{(t)})$, then $\theta_i^{\ell(t+1)}=M_{N,\theta_i^\ell}(\Theta^{(t)})$ by~\eqref{eq:M_ntheta} for all $i,\ell$, and then $(\sigma^2)^{(t+1)}=M_{N,\sigma^2}(\Theta^{(t)})$.
 \State \textbf{Stopping rule:} compute $\Lcal^{*}_{0,N}(\Theta^{(t+1)})$ by the forward recursion; if $t\geq2$ and $\Lcal^{*}_{0,N}(\Theta^{(t+1)})-\Lcal^{*}_{0,N}(\Theta^{(t)})<\epsilon_{\mathrm{tol}}$, stop.
\EndFor
\State \textbf{Output:} the last iterate $\hat{\Theta}_{\mathrm o}$.
\end{algorithmic}
\end{algorithm}
\endgroup
\subsection{Selection of tuning parameters}
\label{ssec:parameterselect}

In this subsection, we describe how to select the tuning parameters. We assume that the wavelets $\phi,\psi$ and the family of basis functions $\{g_i(\cdot):i\in\NN\}$ are given. There are four parameters to select: the wavelet threshold $\lambda_{jl}$, the number of hidden states $k$, the number of basis functions $m$, and the penalty parameter $\lambda$.

Our analysis uses the universal wavelet threshold $\lambda_{jl}=\sigma^\star\{2\log N/N\}^{1/2}$ described in Appendix~\ref{appendix:wavelet-reconstruction}. In our experiments, we use the adaptive BayesShrink thresholds implemented in scikit-image~\citep{van2014scikit}, with implementation details given in Appendix~\ref{sses:expsetup}.

We select the remaining parameters $k,m,\lambda_i^\ell$ by a grid search with a BIC-style criterion, where, instead of searching over all combinations of $\lambda_i^\ell$, we tune a single shared $\lambda$. For fixed $k,m,\lambda$, let $\hat{\Theta}_{\mathrm o}$ be the output of Algorithm~\ref{alg:update:ode}. We compute the surrogate BIC as
$\rbr{k^2-k+m\sum_{\ell=1}^k\sum_{i,j=1}^{p}\one\{\norm{\hat{\theta}_{ij,\mathrm o}^{\ell}}_2\neq 0\}}\log N
    - 2\Lcal_{1,N}^*(\hat{\Theta}_{\mathrm o}\mid\hat{\Theta}_{\mathrm o})$.
The first term counts the free parameters, namely the $k^2-k$ off-diagonal entries of the transition rate matrix and the $m\sum_{\ell=1}^k\sum_{i,j=1}^{p}\one\{\norm{\hat{\theta}_{ij,\mathrm o}^{\ell}}_2\neq 0\}$ nonzero parameters in $\{\theta_{ij}^\ell\}$. The second term uses the surrogate objective $\Lcal_{1,N}^*$ of Definition~\ref{def:surrogate_function} in place of the observed-data log-likelihood $\Lcal_{0,N}$ that the classical BIC uses; hence the name surrogate BIC. We select the parameters in two stages. First, for each $k$ in the grid, we find the minimum surrogate BIC over all candidates of $m$ and $\lambda$, and we select the number of states $k$ at the first elbow point of the resulting curve, identified by visual inspection. Second, given $k$, we select the $m$ and $\lambda$ with the minimal surrogate BIC, breaking ties toward the largest $\lambda$ and then the smallest $m$.

\section{Theory}
\label{sec:theory}

Our goal is to bound the distance between the iterates of a truncated version of the surrogate EM algorithm and the true parameter $\Theta^\star$, for an initialization in a neighborhood of $\Theta^\star$. In Section~\ref{ssec:popEM}, we first study the contraction of the update of the surrogate EM algorithm for the surrogate population objective. We then ask whether a similar guarantee holds for the updates $M_N$ under a proper choice of the penalty parameter $\lambda$, and we break this analysis into two parts. In Section~\ref{ssec:approximateEM}, we define the truncated update map $\tilde M_N$ and outline the one-step analysis under restricted eigenvalue and deviation bound conditions, developed in Appendix~\ref{sec:proof:onestep}. Then, in Section~\ref{ssec:main:results}, we combine the population and sample analyses into finite-sample error bounds for the truncated iteration. Finally, we extend these bounds to full smoothing probabilities started from $\pi(Q)$, where $Q$ is the current transition-rate block.

Our concentration arguments require stationarity and weak dependence across sampling intervals. We state both for the interval-block process
\[
W_n:=
\left(
\bigl((Z(t_{n-1}+u),X(t_{n-1}+u))\bigr)_{0\leq u\leq h},
\varepsilon_n
\right),
\]
which collects the latent path over the $n$th sampling interval and the noise of the $n$th observation. Assumptions~\ref{assumption:stationary} and~\ref{assumption:mixingbeta} below refer to this sequence. Note that each observation
$Y_n=X(t_n)+\sigma^\star\varepsilon_n$ is a
measurable function of $W_n$, and so is each $\Psi_j(t_n)$.

The definition of $\beta$-mixing is given in Definition~\ref{definition:discretebetamixing} in Appendix~\ref{appendix:mixing:properties}.
Our analysis requires the following two assumptions.

\begin{assumption}
\label{assumption:stationary}
Let $\stackrel{d.}{=}$ denote equality in distribution. The process $\{W_n\}_{n\in\mathbb{N}}$ is strictly stationary, that is, for any $n\in\mathbb{N}$ and $n',n''\in\NN\cup\{0\}$, $(W_n,\ldots,W_{n+n'})\stackrel{d.}{=}(W_{n+n''},\ldots,W_{n+n'+n''})$.
\end{assumption}

\begin{assumption}
\label{assumption:mixingbeta}
The process $\{W_n\}_{n\in\mathbb{N}}$ is geometrically $\beta$-mixing, that is, there exists a constant $c_W>0$ such that $\beta(j)\leq 2\exp(-c_Wj)$, $j\in\NN$.
\end{assumption}
Propositions~\ref{prop:general_example} and~\ref{prop:cute_example} in Appendix~\ref{appendix:mixing:properties} provide two examples for which these assumptions hold under the stated conditions.

\subsection{Population update of the surrogate EM algorithm}
\label{ssec:popEM}

We begin with the surrogate population objective and then generalize the results to the updates $M_N$, which is a common approach in the analysis of EM algorithms~\citep{yi2015regularized, balakrishnan2017statistical}. We define the surrogate population objective as
\begin{multline}
    {\Lcal}(\tilde{\Theta}\mid \Theta)
    = \EE\Bigg[\sum_{\ell\ne\ell'} m_{\ell\ell'}(Y_0^N;\Theta)\log \tilde{q}_{\ell\ell'}+\sum_{\ell=1}^k\tilde{q}_\ell\tau_{\ell}(Y_0^N;\Theta)-\sum_{n=1}^N \frac{p}{2}(\log 2\tilde{\sigma}^2+\log 2\pi)\\
    -\frac{1}{4\tilde{\sigma}^2}\sum_{n=1}^N \sum_{\ell=1}^k\sum_{i=1}^p p(Z(t_n)=\ell\mid Y_0^N; \Theta)\rbr{Y_{n,i}-Y_{n-1,i} - \sum_{j=1}^p\Psi_j(t_n)^\top \tilde{\theta}_{ij}^{\ell}}^2\Bigg],\label{eq:poprisk}
\end{multline}
where $\Psi_j(t_n)$ is defined before~\eqref{eq:approxY}.
Recall that $\theta_i^\ell=(\theta_{i1}^{\ell\top},\ldots,\theta_{ip}^{\ell\top})^\top\in\RR^{pm}$, and write $\Psi(t_n)=(\Psi_1^\top(t_n),\ldots,\Psi_p^\top(t_n))^\top\in\RR^{pm}$, let $Q_{\backslash}$ be the vector of off-diagonal entries of $Q$, and let $s=\max_{i,\ell}\abr{\{(j,i):\|\theta_{ij}^{\ell\star}\|_2>0, j\in[p]\}}$ denote the maximum sparsity level.  We first define a distance between parameters.
\begin{definition}
\label{def:distance}
Write the parameter in vector form as $\Theta=(\{\theta_i^\ell\}_{i\in[p],\ell\in[k]},Q_\backslash,\sigma^2)$, with $\theta_i^\ell\in\RR^{pm}$, $Q_\backslash\in\RR^{k(k-1)}$ and $\sigma^2\in\RR$. The diagonal of $Q$ is determined by $Q_\backslash$, and vector norms and Jacobians use this coordinate order.
Define
$
\|\Theta\|_{\mathrm d}
:=
\max_{\ell\in[k],\,i\in[p]}\|\theta_i^\ell\|_2
+\|Q_\backslash\|_1+|\sigma^2|$.
It is straightforward to verify that $\|\cdot\|_{\mathrm d}$ is a norm on this space.
For two parameters $\Theta,\Theta'\in\Omega$, define
$
\dist(\Theta,\Theta')
:=\|\Theta-\Theta'\|_{\mathrm d}.
$
\end{definition}
Assumption~\ref{assumption:population} and Proposition~\ref{prop:mvt} below measure the contraction of the population update in this norm.
 Let $M(\Theta)=(M_{\theta_{i}^\ell}(\Theta), M_Q(\Theta), M_{\sigma^2}(\Theta))$, where each component maximizes $\Lcal(\tilde{\Theta}\mid \Theta)$ over the corresponding block of $\tilde\Theta$, with the other two blocks fixed at their values in $\Theta$. That is, $M_{Q}(\Theta)=\argmax_{\tilde{Q}}\Lcal(\tilde{Q}\mid \Theta)$, where $\tilde Q$ ranges over all valid transition rate matrices, so that $M_{q_{\ell\ell'}}(\Theta)=\EE[m_{\ell\ell'}(Y_0^N;\Theta)]/\EE[\tau_\ell(Y_0^N;\Theta)]$ for $\ell\neq\ell'$. Furthermore, $M_{\sigma^2}(\Theta)=\argmax_{\tilde{\sigma}^2}\Lcal(\tilde{\sigma}^2\mid \Theta)$, and, as in~\eqref{eq:M_ntheta}, we define $M_{\theta_i^\ell}(\Theta)$ as
\begin{align*}
M_{\theta_{i}^\ell}(\Theta)=&\argmax_{\tilde{\theta}_{i}^\ell}-\EE\sbr{\frac{1}{N}\sum_{n=1}^{N}
p(Z(t_n)=\ell\mid Y_0^N ; \Theta)\rbr{Y_{n,i}-Y_{n-1,i} - {\Psi}(t_n)^\top\tilde{\theta}_{i}^{\ell}}^2}.
\end{align*}
Here, and throughout Section~\ref{sec:theory}, the smoothing probabilities $p(Z(t_n)=\ell\mid Y_0^N;\Theta)$ are those of the surrogate model of Section~\ref{ssec:graphest}, computed with $\hat\Psi_j(t_n)$ as in Algorithm~\ref{alg:update:ode}, and started from $Z(t_0)\sim\pi(Q)$, where $Q$ is the transition-rate block of $\Theta$. Only the regression term uses $\Psi(t_n)$.

From now on, $\Omega_Q$ denotes the set of transition rate matrices $Q$ with $\|Q_\backslash-Q^\star_\backslash\|_1\leq r_0$, where $r_0>0$ is smaller than both $\sigma^{\star2}$ and every positive off-diagonal entry of $Q^\star$. Every $Q\in\Omega_Q$ has positive rates wherever $Q^\star$ does, so it is irreducible and hence positive recurrent. Furthermore, the set $\Omega_Q$ is compact and convex, and hence $\Omega$ is convex. The set $\Omega_Q$ restricts only the parameter at which the updates are evaluated. 
Define the local region of $\Theta^\star$ as
\[
B(r_0,\Theta^\star)=\Bigl\{\Theta\in\Omega:\dist(\Theta,\Theta^\star)\leq r_0,\ \max_{\ell\in[k],\,i\in[p]}\|\theta_i^\ell-\theta_i^{\ell\star}\|_1\leq 10c_\Psi^{-1}\sqrt{ms}\,r_0\Bigr\},
\]
where $c_\Psi\in(0,1]$ is the constant of Assumption~\ref{assumption:eigenvalue} below. The second condition restricts the $\ell_1$-norm of each error $\theta_i^\ell-\theta_i^{\ell\star}$. It holds whenever $\dist(\Theta,\Theta^\star)\leq r_0$ and every $\theta_i^\ell$ has at most $s$ nonzero blocks $\theta_{ij}^\ell$. For a linear map $\mathcal J$ on the parameter space, let
$
\|\mathcal J\|_{\mathrm d}
:=\sup_{\|\Phi\|_{\mathrm d}=1}\|\mathcal J\Phi\|_{\mathrm d}
$
denote the operator norm induced by $\|\cdot\|_{\mathrm d}$.
\begin{assumption}[Local population contraction]
\label{assumption:population}
Every off-diagonal entry of $Q^\star$ is positive, and the population update map $M$ is continuously differentiable on an open
neighborhood of $B(r_0,\Theta^\star)$. 
For $\Theta'$ in this neighborhood, let $\mathcal{J}_M(\Theta')$
be the Jacobian of $M$ at $\Theta'$, excluding the diagonal terms of $Q$, in the coordinates
$(\{\theta_i^\ell\}_{i\in[p],\ell\in[k]},Q_\backslash,\sigma^2)$. There
exists $\kappa\in[0,1)$ such that, for every
$\Theta,\Theta'\in B(r_0,\Theta^\star)$,
$
\sup_{u\in[0,1]}
\left\|\mathcal{J}_M\bigl(\Theta'+u(\Theta-\Theta')\bigr)\right\|_{\mathrm d}
\leq\kappa
$.
\end{assumption}
Note that every off-diagonal entry of $Q^\star$ exceeds $r_0$ and $\sigma^{\star2}>r_0$. By our choice of $r_0$ and the assumption that all off-diagonal entries of $Q^\star$ are positive, every parameter in $B(r_0,\Theta^\star)$ has strictly positive off-diagonal transition rates and strictly positive variance.

\begin{proposition}
\label{prop:mvt}
Under Assumption~\ref{assumption:population}, for every
$\Theta, \Theta'\in B(r_0,\Theta^\star)$, $\dist\bigl(M(\Theta),M(\Theta')\bigr)
\leq\kappa\dist(\Theta,\Theta')$.
\end{proposition}

To bound the error of one population update, we combine Proposition~\ref{prop:mvt}, which we prove in Appendix~\ref{sec:surrogate-bias}, with the triangle inequality and obtain
$
\dist(M(\Theta),\Theta^\star)
\leq
\kappa\dist(\Theta,\Theta^\star)
{}+\dist(M(\Theta^\star),\Theta^\star)
$
for every $\Theta\in B(r_0,\Theta^\star)$.
Since the surrogate objective is biased, in general
$\Theta^\star\neq M(\Theta^\star)$, so the second term need not vanish.
For the switching remainder in~\eqref{eq:approxY}, let $\rho$ satisfy $\max_{n\in[N],\,i\in[p]}|\rho_{n,i}|\leq\rho$ almost surely. 
The following lemma, whose proof is in Appendix~\ref{sec:surrogate-bias}, bounds the population bias $\dist(M(\Theta^\star),\Theta^\star)$.

\begin{lemma}[Aggregate population surrogate bias]
\label{lemma:aggregate-surrogate-bias} For $\Theta\in\Omega$ and $\ell\in[k]$, define the weighted Gram matrices
$K_\Psi^{w,\ell}(\Theta):=\frac{1}{N}\sum_{n=1}^N
\EE\sbr{p(Z(t_n)=\ell\mid Y_0^N;\Theta)\Psi(t_n)\Psi(t_n)^\top}$. 
Suppose that $\Theta^\star\in\Omega$ and, for every $\ell\in[k]$, $K_\Psi^{w,\ell}(\Theta^\star)\succ0$, $\EE[\tau_\ell(Y_0^N;\Theta^\star)]>0$. 
Assume that $|\rho_{n,i}|\leq\rho$,
$\|\Psi(t_n)\|_\infty\leq Bh$ and $\|\hat\Psi(t_n)\|_\infty\leq Bh$
almost surely for all $n\in[N]$ and $i\in[p]$.
Assume also that each $\theta_i^{\ell\star}$ is $ms$-sparse
and that the estimated trajectory $\hat X$ satisfies the conditions of
Lemma~\ref{lem:posterior-sensitivity}.
Let $b_{\theta,i}^{\ell}$, $b_{q_{\ell\ell'}}$ and $b_{\sigma^2}$
denote the explicit upper bounds in Lemmas~\ref{lem:bias-theta},
\ref{lem:bias-q} and~\ref{lem:bias-sigma}, respectively. Then
\[
\begin{aligned}
\dist\bigl(M(\Theta^\star),\Theta^\star\bigr)
&\leq b_{\mathrm{surr}}(\Theta^\star):=\max_{\ell\in[k],\,i\in[p]}b_{\theta,i}^{\ell}
+\sum_{\ell\neq\ell'}b_{q_{\ell\ell'}}+b_{\sigma^2}.
\end{aligned}
\]
\end{lemma}

\subsection{Truncated surrogate EM algorithm}
\label{ssec:approximateEM}

We now study the convergence properties of the update $M_N$ defined in Section~\ref{ssec:graphest}. The main challenge in extending the population result is that the data are not \emph{i.i.d.}, and the conditional probability $p(Z(t_n)\mid Y_0^N;\Theta)$ depends on all observations, so that the weighted loss terms in~\eqref{eq:M_ntheta} are dependent across all $n$, making concentration bounds difficult to establish. To obtain concentration bounds, we replace the full smoothing probabilities in the updates by truncated smoothing probabilities, an approach inspired by~\citet{yang2017statistical} and~\citet{li2022estimation}. For $n\in[N]$, let $w_{\ell,\Theta}(t_n)=p\rbr{Z(t_n)=\ell\mid Y_0^N;\Theta}$ denote the smoothing probability, and, for an integer truncation radius $r\geq1$, define its truncated version, the truncated smoothing probability, by
$\hat{w}_{\ell,\Theta}(t_n)=p\rbr{Z(t_n)=\ell\mid Y_{(n-r)\vee 0}^{(n+r)\wedge N};\Theta}$. 
Similarly, for $a,b\in[k]$, let $\hat w_{ab,\Theta}(t_n)=p\rbr{Z(t_{n-1})=a,Z(t_n)=b\mid Y_{(n-r)\vee 0}^{(n+r)\wedge N};\Theta}$ be the truncated joint smoothing probabilities, and define the truncated posterior expected numbers of transitions and occupation times
\begin{align*}
\hat m_{\ell\ell'}(Y_0^N;\Theta)&=\sum_{n=1}^N\sum_{a,b=1}^k\hat w_{ab,\Theta}(t_n)\,\EE[m_{\ell\ell'}(h)\mid Z(0)=a,Z(h)=b;\Theta],\\
\hat\tau_\ell(Y_0^N;\Theta)&=\sum_{n=1}^N\sum_{a,b=1}^k\hat w_{ab,\Theta}(t_n)\,\EE[\tau_\ell(h)\mid Z(0)=a,Z(h)=b;\Theta],
\end{align*}
which are~\eqref{eq:computemll} and~\eqref{eq:computetaul} with the joint smoothing probabilities replaced by their truncated versions; here $m_{\ell\ell'}(h)$ and $\tau_\ell(h)$ are the number of transitions from $\ell$ to $\ell'$ and the occupation time of $\ell$ during $[0,h]$.

We define the truncated update map $\tilde M_N$ by replacing $p(Z(t_n)=\ell\mid Y_0^N;\Theta)$, $m_{\ell\ell'}(Y_0^N;\Theta)$ and $\tau_\ell(Y_0^N;\Theta)$ in the block updates of Section~\ref{ssec:graphest} with $\hat{w}_{\ell,\Theta}(t_n)$, $\hat m_{\ell\ell'}(Y_0^N;\Theta)$ and $\hat\tau_\ell(Y_0^N;\Theta)$, respectively. Each block is evaluated with the other parameters fixed at their values in $\Theta$. In particular, $\tilde{M}_{N,\theta_{i}^\ell}(\Theta)$ solves~\eqref{eq:M_ntheta} with $p(Z(t_n)=\ell\mid Y_0^N;\Theta)$ replaced by $\hat{w}_{\ell,\Theta}(t_n)$ and the same penalty parameter $\lambda_i^\ell$.

\subsection{Main results}
\label{ssec:main:results}

We analyze the truncated iteration $\Theta^{(t+1)}=\tilde M_N(\Theta^{(t)})$. Corollary~\ref{corollary:full} below extends the results to the iteration $\Theta^{(t+1)}=M_N(\Theta^{(t)})$, which uses full smoothing probabilities started from $\pi(Q)$, as specified in Section~\ref{ssec:approximateEM}.

In addition to Assumption~\ref{assumption:mixingbeta}, we need bounds on the transition probabilities of the sampled chain that hold uniformly over $\Omega_Q$, following the conditions introduced in~\citet{van2008hidden} and~\citet{yang2017statistical}. For a transition rate matrix $Q$ and the sampling period $h$, write $P=\exp(Qh)$ and $\pi=\pi(Q)$, and let $\{Z(t_n)\}_{n=0,\ldots,N}$ be the sampled Markov chain of $(Z(t))_{t\in[0,T]}$ associated with the transition probability matrix $P$.

\begin{assumption}[Uniform transition bounds]
\label{assumption:mixing}
    There exists a constant $\zeta=\zeta(h)\in(0,1]$ such that for any $Q\in\Omega_Q$ and for all $\ell,\ell'\in[k]$,
    \begin{equation}
    \label{eq:mixing}
    \zeta \leq \frac{P_{\ell\ell'}}{\pi_{\ell'}}\leq \zeta^{-1}.
    \end{equation}
    We further assume the following uniform nondegeneracy bounds over $\Omega_Q$:
    \[
    \pi_{\min}:=\inf_{Q\in\Omega_Q}\min_{\ell\in[k]}\pi_\ell(Q)>0,
    \qquad
    \delta_{\min}:=\inf_{Q\in\Omega_Q}\min_{\ell,\ell'\in[k]}[\exp(Qh)]_{\ell\ell'}>0.
    \]
\end{assumption}
The nontrivial lower and upper bounds in~\eqref{eq:mixing} guarantee that (i) the Markov chain does not get stuck in any single state and (ii) the dependence on the initial state decays geometrically~\citep{yang2017statistical}. Assumption~\ref{assumption:mixing} is common in the study of Markov chains, and it ensures that the filtering and smoothing probabilities forget the distant past and future geometrically. The condition $\pi_{\min}>0$ means that, for every $Q\in\Omega_Q$, every state has stationary probability at least $\pi_{\min}$. Note that if the continuous-time Markov chain is irreducible, then $P_{\ell\ell'}>0$ for any $h>0$ and each fixed $Q$, as discussed in Proposition~6.1 of~\citet{lalley2012continuous}. Assumption~\ref{assumption:mixing} requires the same $\zeta$, $\pi_{\min}$, and $\delta_{\min}$ for all $Q\in\Omega_Q$, and this holds automatically for the set $\Omega_Q$ of Section~\ref{ssec:popEM}. That set is a compact set of irreducible transition rate matrices, on which $\exp(Qh)$ and $\pi(Q)$ are continuous with positive entries, so $\delta_{\min}>0$ and $\pi_{\min}>0$. Furthermore, $\zeta=\min\{\delta_{\min},\pi_{\min}\}$ satisfies~\eqref{eq:mixing}, because $\delta_{\min}\leq P_{\ell\ell'}\leq P_{\ell\ell'}/\pi_{\ell'}\leq1/\pi_{\min}$.

The following lemma bounds the difference between the truncated and the full smoothing probabilities.

\begin{lemma}
\label{lemma:truncated_smoooth}
Suppose that Assumption~\ref{assumption:mixing} holds, let $\Theta\in\Omega$, and let $r\geq1$ be an integer. All probabilities below are computed under the surrogate model with parameter $\Theta$, using the convention in Appendix~\ref{ssec:truncatedforward}, and the bounds hold for every realization of $Y_0^N$. There exists an absolute constant $C>0$ such that, for all $0\leq n\leq N-1$ and $\ell,\ell'\in[k]$,
\begin{multline*}
    \abr{P\bigg(Z_n=\ell, Z_{n+1}=\ell'\mid Y_{(n-r)\vee 0}^{(n+r)\wedge N};\Theta\bigg)
        -
        P\bigg(Z_n=\ell, Z_{n+1}=\ell'\mid Y_{0}^{N};\Theta\bigg)
        }\\
        \leq Ck\delta_{\min}^{-1}\zeta^{-8}\pi_{\min}^{-2}\cbr{1-(\zeta\pi_{\min})^2}^{r-1}.
\end{multline*}
Furthermore, for all $0\leq n\leq N$ and $\ell\in[k]$, we have
\[
    \bigg|P(Z_n=\ell\mid Y_0^N;\Theta)-P\rbr{Z_n=\ell\mid Y_{(n-r)\vee 0}^{(n+r)\wedge N};\Theta}\bigg|\leq \frac{10}{\delta_{\min}}k\zeta^{-8}\pi_{\min}^{-2}\cbr{1-(\zeta\pi_{\min})^2}^{r-1}.
    \]
\end{lemma}

By Lemma~\ref{lemma:truncated_smoooth}, the difference between the truncated and the full smoothing probabilities decays geometrically as $r$ increases. Since the latent Markov chain has a finite state space, the total variation distance between the two distributions also decays geometrically. The rate of decay in Lemma~\ref{lemma:truncated_smoooth} depends on the sampling period. Since $P_{ab}(h)=q_{ab}h+O(h^2)$ for $a\neq b$, we have $\zeta=O(h)$ and $\delta_{\min}=O(h)$ as $h\to0$, so the bound of Lemma~\ref{lemma:truncated_smoooth} is small only when $r$ is at least of order $h^{-2}\log(1/h)$. Therefore, the truncation radius must grow as the sampling period decreases, and our results are most informative for a fixed $h$. 

For each $j\in[p]$, let $K_{\Psi_j}=N^{-1}\sum_{n=1}^N\EE[\Psi_j(t_n)\Psi_j^\top(t_n)]$ be the population Gram matrix of the features $\Psi_j(t_n)$. The next assumption requires these matrices to be uniformly well conditioned.
\begin{assumption}
\label{assumption:eigenvalue}
    There exists $c_{\Psi}\in(0,1]$ such that, for every $j\in[p]$,
    $
    c_{\Psi}\leq \sigma_{\min}(K_{\Psi_j})\leq \sigma_{\max}(K_{\Psi_j})\leq c_{\Psi}^{-1}.
    $
\end{assumption}

  Finally, the analysis uses two properties of $\hat X$. First, $\hat X$ is close to $X$: our bounds hold on the reconstruction event
  \[
  \mathcal E_{\mathrm{rec}}:=\Bigl\{\max_{i\in[p],\,n\in[N]}\opnorm{\hat X_i-X_i}{[t_{n-1},t_n]}\leq\delta_1\Bigr\},
  \]
  where $\delta_1$ is the reconstruction error, and their failure probabilities include $\PP(\mathcal E_{\mathrm{rec}}^c)$. Second, $\hat X$ is bounded and local, as stated in the
  following assumption.

\begin{assumption}[Local reconstruction]
\label{assumption:local}
There is $R_X>0$ such that $|X_i(t)|\leq R_X$ and $|\hat X_i(t)|\leq R_X$ for all $i\in[p]$ and $t\in[0,T]$, almost surely. Furthermore, there are integers $r_{\mathrm{rec}}\geq0$ and $d_{\mathrm{rec}}\geq1$, which do not depend on $N$ or $p$, such that, for every $n\in[N]$, $\hat\Psi(t_n)$ is a measurable function of the observations within $r_{\mathrm{rec}}$ sampling intervals of $t_n$, $Y_{(n-r_{\mathrm{rec}})\vee0}^{(n+r_{\mathrm{rec}})\wedge N}$. For each $N$, there are measurable functions $H_0,\ldots,H_{d_{\mathrm{rec}}-1}$ such that
\[
\hat\Psi(t_n)
=H_c\!\left(Y_{n-r_{\mathrm{rec}}},\ldots,Y_{n+r_{\mathrm{rec}}}\right),
\qquad c=n\bmod d_{\mathrm{rec}},
\]
for every interior index $r_{\mathrm{rec}}<n\leq N-r_{\mathrm{rec}}$. Thus, indices with the same remainder modulo $d_{\mathrm{rec}}$ use the same function, evaluated on their respective windows of observations.
\end{assumption}
These locality and repetition properties describe the procedure used to compute $\hat\Psi$. Lemma~\ref{lemma:wavelet:local} establishes them for the wavelet smoothing of Section~\ref{ssec:waveletsmooth}, with $d_{\mathrm{rec}}=2^{J-j_0}$, under the stated conditions on the decomposition depth, thresholds, and boundary treatment. Shifting by $d_{\mathrm{rec}}$ sampling intervals preserves the alignment of every wavelet grid used in the construction.

The estimated trajectory of
  Section~\ref{ssec:waveletsmooth}, projected onto $[-R_X,R_X]$, satisfies Assumption~\ref{assumption:local} with $d_{\mathrm{rec}}=2^{J-j_0}$ when the number of decomposition levels
  $J-j_0$ is fixed and the threshold is deterministic, under the support and boundary conditions of Lemma~\ref{lemma:wavelet:local}. 

We are now ready to state the main result. This is a shorter version of Theorem~\ref{theorem:graphrecovery:full} in Appendix~\ref{ssec:mainproof}, which gives the explicit constants, conditions and penalty.

\begin{theorem}
\label{theorem:graphrecovery}
Suppose that Assumptions~\ref{assumption:edgeset}--\ref{assumption:population}, \ref{assumption:mixing}, \ref{assumption:eigenvalue} and~\ref{assumption:local} hold, that $m\geq1$ and $1\leq s\leq p/2$, that the regularity conditions of Theorem~\ref{theorem:graphrecovery:full} hold, and that $r\geq1$ is a sufficiently large truncation radius that does not depend on $N$. Assume further that the concentration bounds used in the proof hold uniformly over $B(r_0,\Theta^\star)$, as stated in Theorem~\ref{theorem:graphrecovery:full}. Let $\Theta^{(0)}\in B(r_0,\Theta^\star)$, let $d_0:=\dist(\Theta^{(0)},\Theta^\star)$, and let $\Theta^{(t+1)}=\tilde M_N(\Theta^{(t)};\lambda^{(t)})$, with the iteration-dependent penalty $\lambda^{(t)}$ of Theorem~\ref{theorem:graphrecovery:full}. There are constants $C_{\mathrm{ss}},C_e,C_\kappa\geq1$ and $c_\delta>0$, specified in Appendix~\ref{ssec:mainproof}, such that the following holds. They depend on $k$, $m$, $s$, $h$, $r$, $r_{\mathrm{rec}}$, $d_{\mathrm{rec}}$, $r_0$, $\Theta^\star$, the $\beta$-mixing rate and the constants in the assumptions and regularity conditions, but not on $N$ or $\delta_1$. Suppose that
\[
N\geq C_{\mathrm{ss}}\max\cbr{\frac{m^4\{\log(Np)\}^5}{\log p},\,(sm)^3\log p,\,k^4,\,r+r_{\mathrm{rec}}}
\qquad\text{and}\qquad
\delta_1\leq c_\delta.
\]
Let $\bar\kappa:=C_\kappa\kappa$, and let $e_N$ be the error level of Theorem~\ref{theorem:graphrecovery:full}, which satisfies
\[
e_N\leq C_e\sbr{
\sqrt{\frac{\log(Np)}{N}}
+\delta_1
+\cbr{1-(\zeta\pi_{\min})^2}^{r-1}
+b_{\mathrm{surr}}(\Theta^\star)}.
\]
If $\bar\kappa<1$ and $e_N\leq(1-\bar\kappa)r_0$, then, with probability at least $1-144d_{\mathrm{rec}}k(p^{-1}+N^{-4})-\PP(\mathcal E_{\mathrm{rec}}^c)$, the iterates satisfy $\Theta^{(t)}\in B(r_0,\Theta^\star)$ for every $t\geq0$ and, for every $L\geq0$,
\[
\dist\bigl(\Theta^{(L)},\Theta^\star\bigr)
\leq\bar\kappa^Ld_0+\frac{1-\bar\kappa^L}{1-\bar\kappa}\,e_N.
\]
\end{theorem}

Theorem~\ref{theorem:graphrecovery}, which we prove in Appendix~\ref{ssec:mainproof}, shows that the iterates converge geometrically, at rate $\bar\kappa$, to a neighborhood of $\Theta^\star$ of radius $e_N/(1-\bar\kappa)$. The radius has four sources: the statistical error, of order $\sqrt{\log(Np)/N}$; the reconstruction error $\delta_1$; the truncation error $\{1-(\zeta\pi_{\min})^2\}^{r-1}$, which does not vanish as $N\to\infty$ with $r$ fixed, and the surrogate bias $b_{\mathrm{surr}}(\Theta^\star)$. Only the statistical error is guaranteed to vanish as $N\to\infty$: the truncation error does not vanish with $r$ fixed, $b_{\mathrm{surr}}(\Theta^\star)$ is a bias of the surrogate, and $\delta_1$ depends on the accuracy of the estimated trajectory. The radius does not depend on the initial error $d_0$, because the penalty $\lambda^{(t)}$ tracks the current error bound, as in~\citet{yi2015regularized}. A fixed penalty would have to dominate the population discrepancy at the first iteration, and every later update would then carry an error of order $\kappa d_0$. The penalty $\lambda^{(t)}$ depends on $\kappa$, $d_0$ and $b_{\mathrm{surr}}(\Theta^\star)$, which are unknown, so it is a device of the analysis rather than a tuning rule. In practice, we use a fixed $\lambda$.

Combined with a threshold on the estimates of $\theta_{ij}^{\ell}$, Theorem~\ref{theorem:graphrecovery} gives graph recovery under sufficient signal strength.
\begin{corollary}
\label{corollary:main}
Under the conditions of Theorem~\ref{theorem:graphrecovery}, including $\bar\kappa<1$ and $e_N\leq(1-\bar\kappa)r_0$, let
$\hat E^\ell=\{(j,i):\norm{\theta_{ij}^{\ell(L)}}_2>\epsilon_{\mathrm{edge}}\}$ for $\ell\in[k]$, where $\theta_{ij}^{\ell(L)}$ is the corresponding block of $\Theta^{(L)}$, and suppose that the
threshold parameter $\epsilon_{\mathrm{edge}}$ of Section~\ref{ssec:graphest} is selected as $
\epsilon_{\mathrm{edge}}
=\frac{2}{3}\min_{(j,i)\in E^\ell,\,\ell\in[k]}
\norm{\theta_{ij}^{\ell\star}}_2$. 
If
\[
\bar\kappa^Ld_0+\frac{1-\bar\kappa^L}{1-\bar\kappa}\,e_N
<\frac{\epsilon_{\mathrm{edge}}}{2},
\]
then $\hat E^\ell=E^\ell$ for every $\ell\in[k]$ simultaneously, with probability at least
$1-144d_{\mathrm{rec}}k(p^{-1}+N^{-4})-\PP(\mathcal E_{\mathrm{rec}}^c)$. Since $\tilde E^\ell=E^\ell$ under Assumption~\ref{assumption:edgeset}, we also recover
$\tilde E^\ell$ with the same probability.
\end{corollary}
The proof of Corollary~\ref{corollary:main} is given in Appendix~\ref{ssec:corollary:main}. The condition of Corollary~\ref{corollary:main} limits the switching remainder. Since the expression for $e_N$ in Theorem~\ref{theorem:graphrecovery:full} contains the term $C_\kappa b_{\mathrm{surr}}(\Theta^\star)$, the condition requires the term of $b^\ell_{\theta,i}$ involving the switching remainder, given in Lemma~\ref{lem:bias-theta}, to be small compared with $\epsilon_{\mathrm{edge}}$.

Theorem~\ref{theorem:graphrecovery} and Corollary~\ref{corollary:main} concern the truncated iteration. Before extending both results to full smoothing probabilities, we introduce the bound
\[
\bar\eta_r:=C_{\mathrm{sm}}k\delta_{\min}^{-1}\zeta^{-8}\pi_{\min}^{-2}\cbr{1-(\zeta\pi_{\min})^2}^{r-1},
\]
where $C_{\mathrm{sm}}\geq10$ is a constant. By Lemma~\ref{lemma:truncated_smoooth}, for every $\Theta\in\Omega$ and every realization of the data, $\bar\eta_r$ bounds the differences between the full and the truncated marginal and joint smoothing probabilities. The following shows that the conclusions of Theorem~\ref{theorem:graphrecovery} and Corollary~\ref{corollary:main} hold for the iteration with the full smoothing probabilities, with larger constants.
\begin{corollary}
\label{corollary:full}
Suppose that the conditions of Theorem~\ref{theorem:graphrecovery} hold, with the constants of Theorem~\ref{theorem:graphrecovery:full} enlarged as specified in Appendix~\ref{ssec:corollary:full}, and with $\bar\kappa<1$ and $e_N\leq(1-\bar\kappa)r_0$ for these constants. Suppose further that $r$ is large enough that
\[
\bar\eta_r\leq\min\cbr{\frac{c_\Psi c_{\mathrm{RE}}}{1600s},\;\frac{\underline\tau}{8hk^2}}.
\]
Here $c_{\mathrm{RE}}>0$ is the lower bound on the curvature of the population Gram matrix weighted by the truncated smoothing probabilities $\hat w_{\ell,\Theta}(t_n)$, as defined in Theorem~\ref{theorem:graphrecovery:full}. The constant $\underline\tau>0$, defined in the same theorem, is the smallest expected occupation time of a state during one sampling interval, given the states at its endpoints, over all $Q\in\Omega_Q$.
Let $\Theta^{(t+1)}=M_N\bigl(\Theta^{(t)};\lambda^{(t)}\bigr)$ for $t\geq0$, where $M_N$ is the M-step map of Section~\ref{ssec:graphest}, with the full posterior started from $\pi(Q^{(t)})$ at iteration $t$, and $\lambda^{(t)}$ is the penalty of Theorem~\ref{theorem:graphrecovery}. Then, with probability at least $1-144d_{\mathrm{rec}}k(p^{-1}+N^{-4})-\PP(\mathcal E_{\mathrm{rec}}^c)$, the conclusions of Theorem~\ref{theorem:graphrecovery} hold for these iterates. Under the additional threshold and signal conditions of Corollary~\ref{corollary:main}, its graph-recovery conclusion holds with the same probability.
\end{corollary}
We prove Corollary~\ref{corollary:full} in Appendix~\ref{ssec:corollary:full}. The proof compares one full update with one truncated update at the same parameter, both using the initial distribution $\pi(Q)$. Since $\bar\eta_r$ bounds the differences between the full and the truncated marginal and joint smoothing probabilities for every realization, the one-step bounds of Appendix~\ref{sec:proof:onestep} extend to this full update, with larger constants and additional terms of order $\bar\eta_r$. Corollary~\ref{corollary:full} uses the current stationary distribution $\pi(Q^{(t)})$ at each iteration. Our implementation of Algorithm~\ref{alg:update:ode} starts from $\pi(Q^{(0)})$ and then replaces the initial distribution at each M-step defined in Section~\ref{ssec:graphest}. We do not provide guarantees for this update of the initial distribution, the fixed penalty, or the initialization used in our experiments.

\section{Simulations}
\label{sec:simulation}

We illustrate the finite-sample performance of the proposed method on simulated data. We evaluate graph recovery through ROC curves and report the model selection and the quality of the parameter estimates in Appendix~\ref{ssec:modelselection:experiment}.
We implement the proposed method in Python, using the package \emph{scikit-image}~\citep{van2014scikit} for the wavelet smoothing of Section~\ref{ssec:waveletsmooth} and the package \emph{group\_lasso} for the M-step updates~\eqref{eq:M_ntheta} of $\theta_i^\ell$. Appendix~\ref{sses:expsetup} gives the details of the experimental setup, and here we describe the data-generating processes and the evaluation.

We consider two data-generating processes, one nonlinear and one linear, and in both we set the number of states to $k=2$ and
\[
Q^\star=\begin{pmatrix}
-0.27& 0.27\\
0.18& -0.18
\end{pmatrix}, \; \sigma^{\star2} = 0.01.
\]

\textbf{Data-generating process 1.} We consider additive ODEs similar to those in Section~5.1 of~\citet{chen2017network}, with $p=10$, $t\in[0,40]$, $i=1,\ldots,5$ and $\ell=1,2$:
\begin{align}
    \dot{X}_{2i-1}(t) &=  g(X_{2i-1}(t))^\top\theta_{2i-1,2i-1}^{\ell\star} +  g(X_{2i}(t))^\top \theta_{2i-1,2i}^{\ell\star};\label{eq:dgp1:1}\\
    \dot{X}_{2i}(t) &= g(X_{2i-1}(t))^\top\theta_{2i,2i-1}^{\ell\star}  +  g(X_{2i}(t))^\top \theta_{2i, 2i}^{\ell\star},\label{eq:dgp1:2}
\end{align}
 with initial condition $X(0)=(
    -2,2,2,-2,-1.5,1.5,-1,1,1,-1)^\top$ and basis $g(x)=(x,x^2,x^3)^\top$. Appendix~\ref{ssec:details:dgp1} lists the parameters. Figures~\ref{fig:dgp1:1} and~\ref{fig:dgp1:2} show the graphs of the two states, and Figure~\ref{fig:dgp1_xy} shows the trajectories of $X(t)$ and $\{Y_n\}$.

\textbf{Data-generating process 2.} The second process has $p=20$ nodes, and its two states have different graphs: star graphs in state $\ell=1$ and a ring graph in state $\ell=2$. In state $\ell=1$, the system is
\begin{align}
    \dot{X}_{5i+j}(t)&=g(X_{5i+1}(t))^\top \theta_{5i+j,5i+1}^{1\star}, \quad i=0,1,2,3,\;j=2,3,4,5;\notag\\
    \dot{X}_{5i+1}(t)&=\sum_{j=2}^5 g(X_{5i+j}(t))^\top \theta_{5i+1,5i+j}^{1\star}, \quad i = 0,1,2,3,\label{eq:dgp2:1}
\end{align}
for $t\in[0,40]$ with $g(x)=x$, where $\theta_{5i+j,5i+1}^{1\star}=-0.8\pi$ and $\theta_{5i+1,5i+j}^{1\star}=0.8\pi$ for $i=0,1,2,3$ and $j=2,3,4,5$ (here $\pi=3.14159\ldots$). In state $\ell=2$, the system is
\begin{align}
    \dot{X}_i(t) &= g(X_{i-1}(t)) ^\top \theta_{i,i-1}^{2\star}+ g(X_{i+1}(t))^\top \theta_{i,i+1}^{2\star}, \quad i = 2,\ldots, 19;\notag\\
    \dot{X}_1(t) &= g(X_{20}(t))^\top \theta_{1,20}^{2\star} + g(X_2(t))^\top \theta_{1,2}^{2\star};\notag\\
    \dot{X}_{20}(t) &= g(X_{1}(t))^\top \theta_{20,1}^{2\star} + g(X_{19}(t))^\top \theta_{20,19}^{2\star},\label{eq:dgp2:2}
\end{align}
for $t\in[0,40]$ with $g(x)=x$, where $\theta_{i,i-1}^{2\star}=-0.8\pi$ and $\theta_{i,i+1}^{2\star}=0.8\pi$ for $i=2,\ldots,19$, and $\theta_{1,20}^{2\star}=\theta_{20,19}^{2\star}=-0.8\pi$ and $\theta_{20,1}^{2\star}=\theta_{1,2}^{2\star}=0.8\pi$. In this process, we draw the entries of $X(0)$ independently from the standard normal distribution. Figures~\ref{fig:dgp2:1} and~\ref{fig:dgp2:2} show the graphs of the two states, and Figure~\ref{fig:dgp2_xy} shows the trajectories of $X(t)$ and $\{Y_n\}$.

\textbf{ROC curves with varying sample size.} To assess graph recovery, we assume that the number of states $k$ and the number of basis functions $m$ are given and compute the pooled ROC curve over $10$ runs (measurement-noise replicates). In each run, we observe $\{Y_n\}$ with $Y_n=X(t_n)+\sigma^\star\varepsilon_n$ and $\varepsilon_n\sim\Ncal(0, I_p)$, where the noise $\{\varepsilon_n\}$ is drawn independently in each run, while the path of the latent Markov chain $Z(t)$, the initial value $X(0)$ and the trajectory $X(t)$ are the same in all runs. A run with sample size $N$ has $N$ observations at equally spaced times from $0$ to $T$. Next, we trace the ROC curve by varying the penalty parameter $\lambda$ over the natural exponentials of $100$ equally spaced points in $[-7,-1]$. Unlike Theorem~\ref{theorem:graphrecovery}, which uses the iteration-dependent penalty $\lambda^{(t)}$, we run Algorithm~\ref{alg:update:ode} with a fixed $\lambda$ for each grid value, warm-started from the solution at the previous, larger value, as described in Appendix~\ref{sses:expsetup}.
Appendix~\ref{ssec:roc:computation} gives the state-matching rule, the definitions of TPR and FPR, and the computation and aggregation of ROC curves and AUCs.

Throughout, we fix the time horizon $T=40$ and vary the sampling period $h=T/(N-1)$, which gives the sample sizes $N=40,80,120,160,200$. We compare the proposed method with an oracle method that knows the latent Markov chain $Z(t)$ and hence does not need the E-step. Figure~\ref{fig:roc_plot_hatgl} and Table~\ref{tab:roc_aux} show that the oracle method generally has a higher AUC. For the proposed method on the first process, the AUCs range from $0.52$ to $0.64$ for $N=80$ and $N=120$ and increase to $0.92$ and $0.87$ for states $1$ and $2$ at $N=160$. Furthermore, Appendix~\ref{ssec:simulation:withX} reports the results when $X(t)$ is known, and Appendix~\ref{sec:baseline} compares the proposed method with an MS-VAR baseline.

\begin{figure}[!t]
    \centering    \includegraphics[width=.9\textwidth]{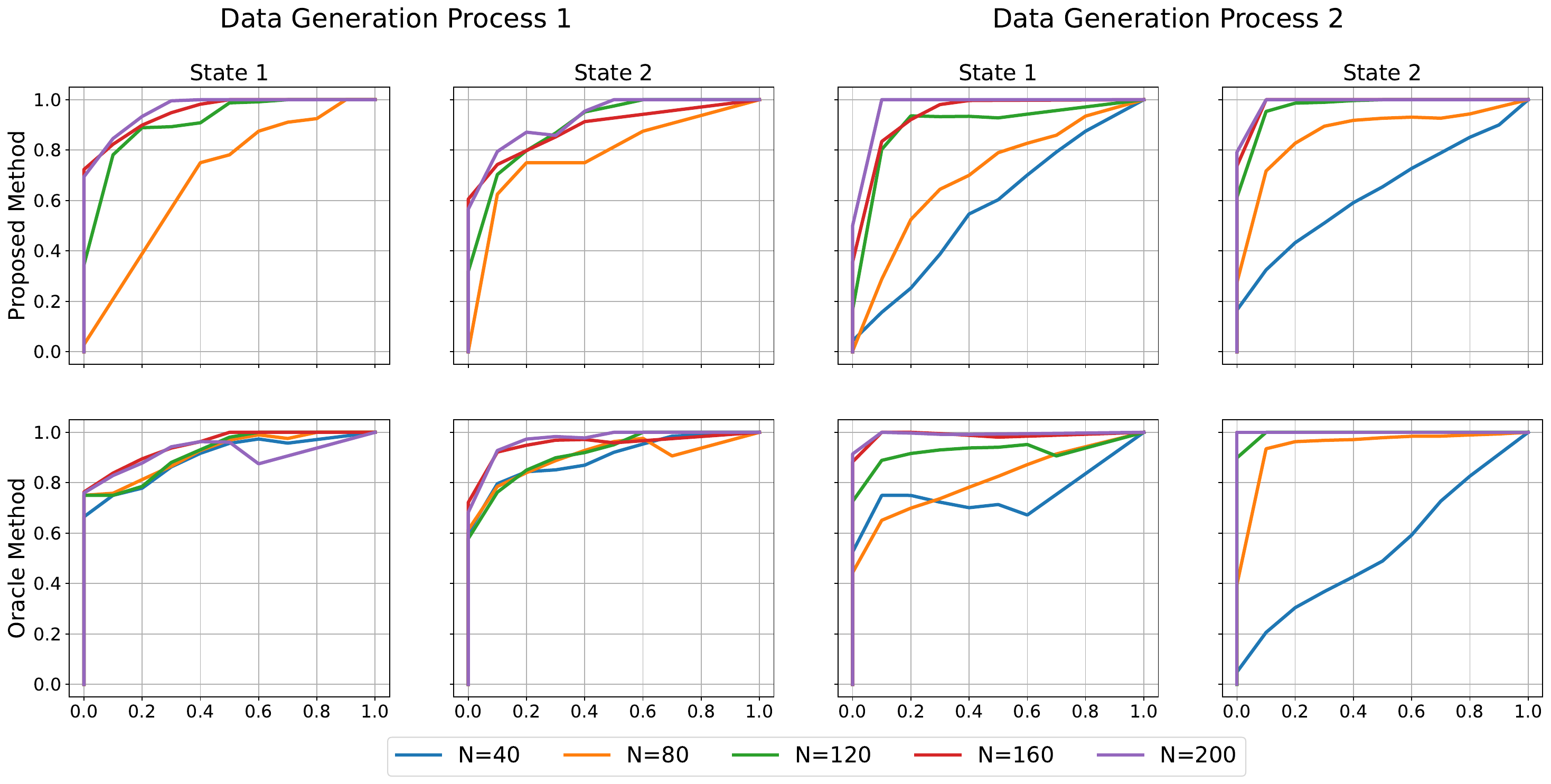}
    \caption{Pooled ROC curves for the two data-generating processes over $10$ runs (independent measurement-noise replicates conditional on one fixed latent path), using FPR bins of width $0.1$ as described in Appendix~\ref{ssec:roc:computation}, with $X(t)$ estimated by the wavelet smoothing of Section~\ref{ssec:waveletsmooth}. The area under a pooled curve can differ substantially from the mean AUC in Table~\ref{tab:roc_aux}. Each panel plots the true positive rate (y-axis) against the false positive rate (x-axis). \textbf{Top row.} Results of the proposed method. The AUC generally improves with increasing sample size.
    \textbf{Bottom row.} Results of the oracle method, where the latent Markov chain $Z(t)$ is assumed to be known. The oracle method typically achieves a higher AUC at the same sample size.
    }
    \label{fig:roc_plot_hatgl}
\end{figure}

\begin{table}[t]
\centering
    \caption{The AUCs of the two data-generating processes in Figure~\ref{fig:roc_plot_hatgl}. Each numerical cell gives the mean AUC and, in parentheses, its standard deviation over all $10$ runs, which are independent measurement-noise replicates conditional on one fixed latent path. A dash indicates that some runs have no valid ROC curve under the state-matching rule, defined in Appendix~\ref{ssec:roc:computation}. The oracle method assumes that the latent Markov chain is known.}
    \label{tab:roc_aux}
    \resizebox{\textwidth}{!}{
    \begin{tabular}{lllllllll}\toprule{} & 
    \multicolumn{4}{l}{Data-Generating Process 1}
    &
    \multicolumn{4}{l}{Data-Generating Process 2}\\
    &\multicolumn{2}{l}{Proposed Method} & \multicolumn{2}{l}{Oracle Method} & 
 \multicolumn{2}{l}{Proposed Method} & \multicolumn{2}{l}{Oracle Method} \\
    State &           1 &           2 &           1 &           2 &           1 &           2 &           1 &           2 \\
    N   &             &             &             &             &             &             &             &             \\\midrule
    40 & {\textemdash} & {\textemdash} & {0.89(0.025)} & {0.88(0.020)} & {0.53(0.025)} & {0.60(0.016)} & {0.77(0.021)} & {0.50(0.037)} \\
80 & {0.53(0.048)} & {0.52(0.073)} & {0.90(0.025)} & {0.90(0.039)} & {0.62(0.068)} & {0.83(0.027)} & {0.79(0.026)} & {0.92(0.012)} \\
120 & {0.63(0.173)} & {0.64(0.177)} & {0.91(0.017)} & {0.90(0.023)} & {0.81(0.128)} & {0.95(0.005)} & {0.93(0.018)} & {1.00(0.002)} \\
160 & {0.92(0.029)} & {0.87(0.037)} & {0.94(0.024)} & {0.95(0.052)} & {0.91(0.037)} & {0.98(0.002)} & {0.99(0.007)} & 1.00(0.000) \\
200 & {0.90(0.101)} & {0.86(0.060)} & {0.93(0.032)} & {0.96(0.018)} & {0.96(0.010)} & 0.98(0.001) & {0.99(0.009)} & 1.00(0.000) \\
\bottomrule
    \end{tabular}}
\end{table}

\section{Analysis of the NYU ADHD data}
\label{sec:experiment}

We analyze resting-state fMRI from the NYU ADHD dataset~\citep{castellanos2008cingulate}, released through the ADHD-200 Initiative~\citep{bellec2017neuro}. Motivated by the group differences in resting-state network dynamics reported by~\citet{shappell2021children}, we investigate differences between typically developing children (TDC) and children with combined-type ADHD (ADHD-C), focusing on the mean occupation time of each latent network state. We fit a three-state linear Markov-switching ODE model using $116$ brain regions and $15$ children per group. The state-specific parameters $\{\theta_{ij}^\ell\}$ are shared across groups, while the transition rate matrices $Q_{\text{ADHD-C}}$ and $Q_{\text{TDC}}$ and noise variances are group-specific. Details of data preprocessing, subject selection, model fitting, and occupation-time computation are provided in Appendix~\ref{appendix:adhd-analysis}.

Figure~\ref{fig:adhdgraph} shows the estimated graphs, transition rate matrices, and smoothing probabilities.
State~$1$ has the most nonzero off-diagonal estimates, as reported in Appendix~\ref{appendix:adhd-connectivity}.

We write $\tau_{\ell,\text{group}}$ for the average posterior expected occupation time of state $\ell$ over the $15$ subjects in the group, as defined in Appendix~\ref{appendix:adhd-occupation}.
Table~\ref{tab:sample_tau} shows that, under the estimated model, the occupation time of state $1$ is longer in the ADHD-C group and that of state $2$ is longer in the TDC group, while the occupation times of state $3$ are comparable. This pattern is in line with previous observations that the two groups differ in the time they spend in network states~\citep{shappell2021children}, although this comparison is descriptive and based on the $30$ subjects used to fit the model. However, these differences need not reflect disorder-specific network dynamics alone. In resting-state fMRI, connectivity states also track arousal and vigilance~\citep{lurie2020questions}, which may differ between the two groups during the scan.

\begin{figure}[!tbp]
    \centering
    \begin{minipage}{0.4\textwidth}
        \centering
        \subcaptionbox{Connectome of state $1$.\label{fig:sub1}}{
            \includegraphics[width=.9\linewidth]{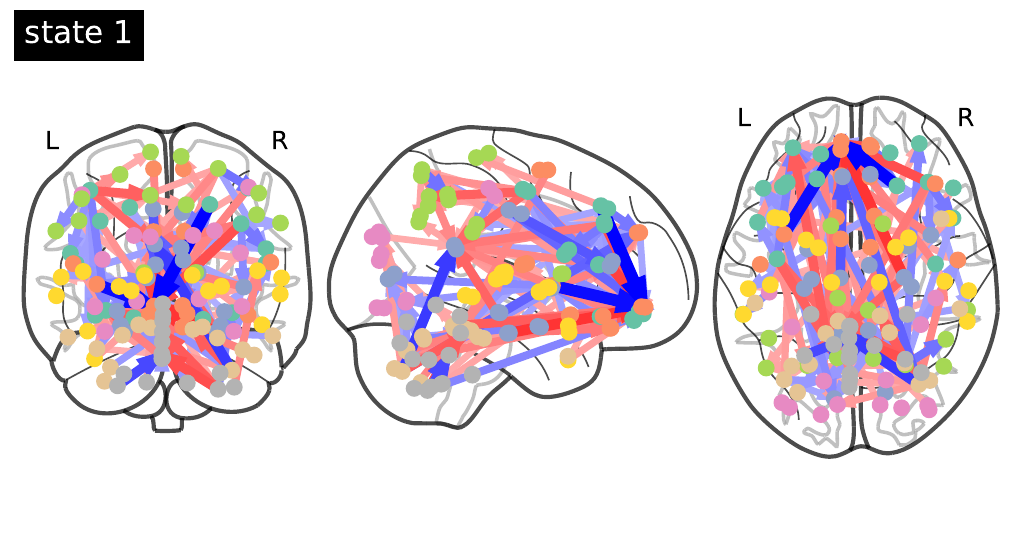}
        }\\
        \subcaptionbox{Connectome of state $2$.\label{fig:sub2}}{
            \includegraphics[width=.9\linewidth]{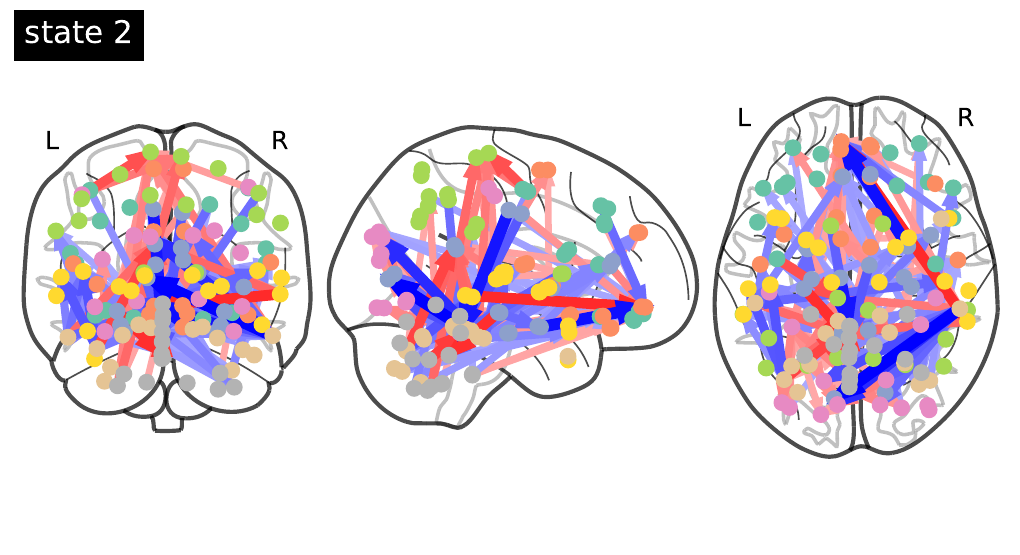}
        }\\
        \subcaptionbox{Connectome of state $3$.\label{fig:sub3}}{
            \includegraphics[width=.9\linewidth]{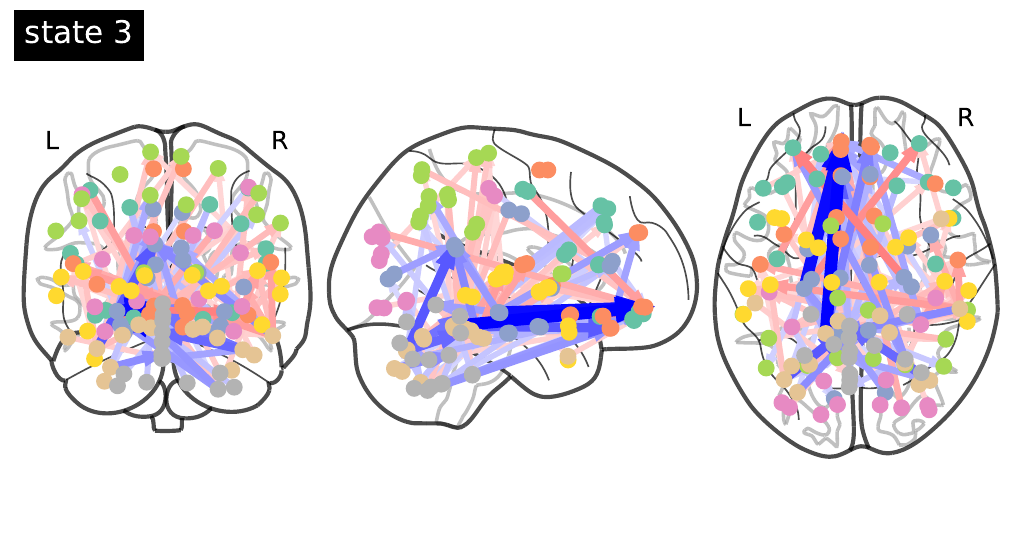}
        }
    \end{minipage}
    \begin{minipage}{0.55\textwidth}
        \centering
        \subcaptionbox{Transition rate matrices. \textbf{Left}: $Q_{\text{TDC}}$. \textbf{Right}: $Q_{\text{ADHD-C}}$.\label{fig:sub4}}{
    \includegraphics[width=.9\linewidth]{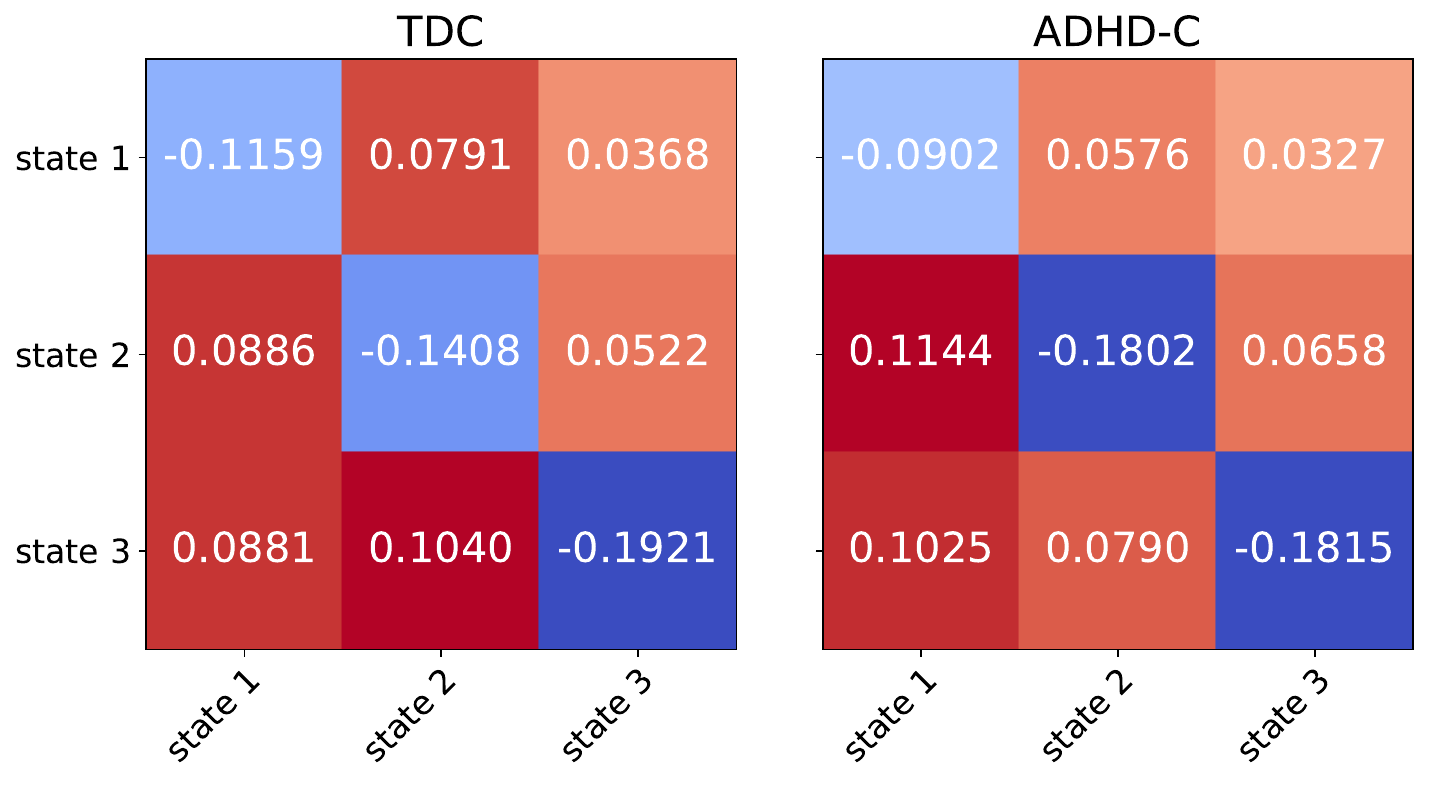}
        }\\
        \subcaptionbox{Estimated smoothing probabilities for each subject.\label{fig:sub5}}{
            \includegraphics[width=.9\linewidth]{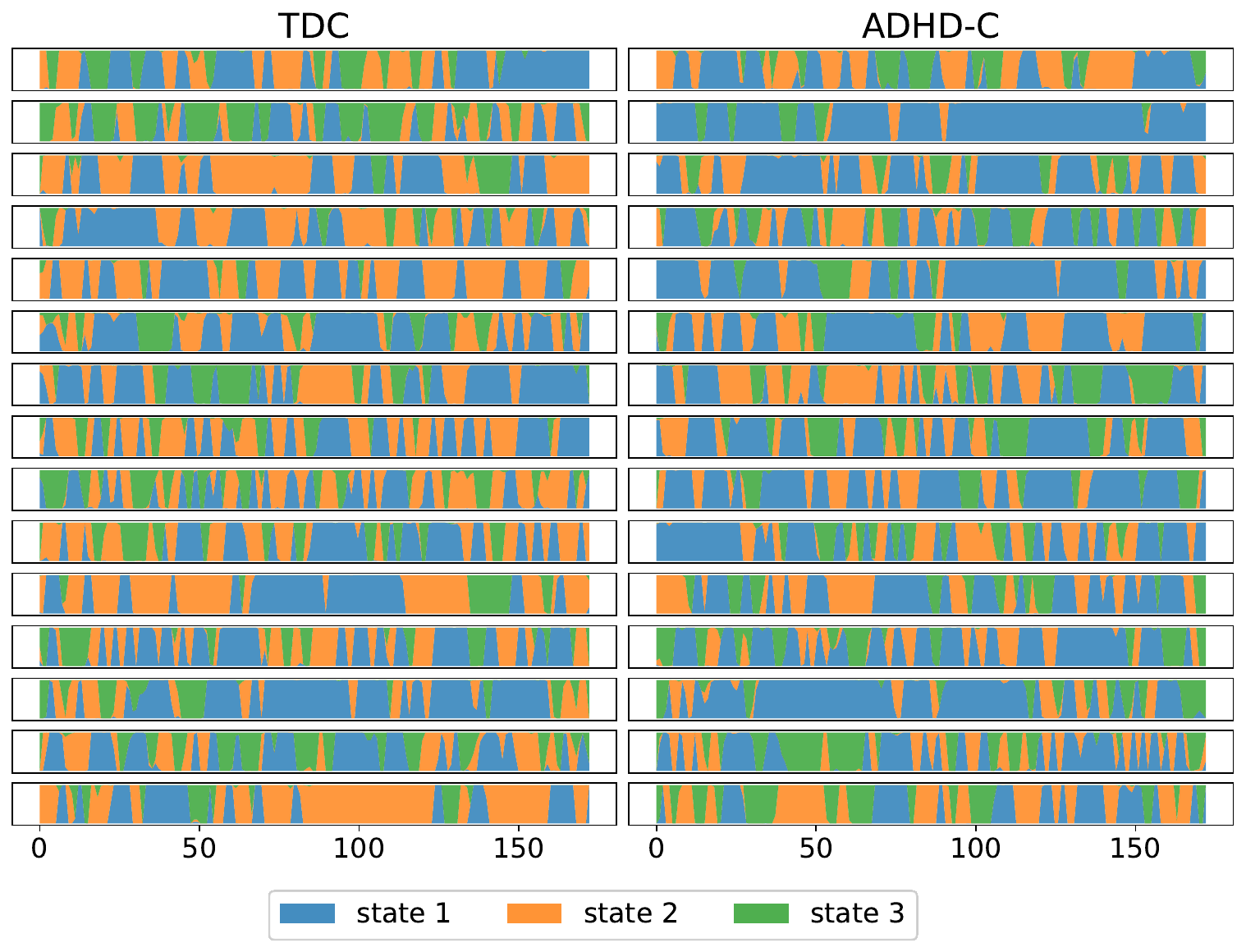}
        }
    \end{minipage}
    \caption{Panels (a)--(c) show the estimated graphs of states $1$--$3$ as connectomes on the $116$ AAL regions. Appendix~\ref{appendix:adhd-connectivity} gives the full summaries. Red and blue arrows indicate positive and negative estimates of $\theta_{ij}^\ell$; darker colors indicate larger magnitudes within each panel. Panel (d) shows the estimated transition rate matrices. Panel (e) shows the estimated smoothing probabilities $P(Z(t_n)=\ell\mid Y_0^N;\hat\Theta_{\text{group}})$ for each subject.}
    \label{fig:adhdgraph}
\end{figure}

\begin{table}[ht!]
\centering
\caption{Posterior expected occupation time of each state by group, in seconds, computed with a sampling period of $h=2.0$ seconds. The rows labeled ``total occupation time'' sum over the $15$ subjects of each group, and the rows labeled ``average occupation time'' give $\tau_{\ell,\text{group}}$, with the standard deviation across the $15$ subjects in parentheses.}
\label{tab:sample_tau}
\begin{tabular}{lcccc}
\toprule
& group & state $1$ & state $2$ & state $3$ \\
\midrule
\multirow{2}{*}{total occupation time} &
TDC & $2216.10$ & $1960.19$ & $953.71$\\
& ADHD-C & $2802.16$ & $1338.80$ & $989.04$\\
\midrule
\multirow{2}{*}{\begin{tabular}[c]{@{}l@{}}average occupation time\\($\tau_{\ell,\text{group}}$)\end{tabular}} & TDC &
$147.74(30.068)$ &
$130.68(36.884)$ &
$63.58(31.534)$\\
& ADHD-C &
$186.81(48.278)$ &
$89.25(32.223)$ &
$65.94(21.425)$\\
\bottomrule
\end{tabular}
\end{table}

\section{Discussion}
\label{sec:discussion}

There are two directions in which this work should be extended. It
would be interesting to study how the transition rates, that is, the entries of
$Q$, relative to the sampling rate affect the estimation quality. Under a high
switching rate, one potential remedy is to integrate a second data modality with
a much higher sampling rate. Second, we model an ordinary differential
process with additive observation noise. A natural question is under what
conditions graph recovery is feasible when the latent process $X(t)$ follows a
stochastic differential equation. We believe that such results would give a
broader understanding of the directed graphs underlying data generated by
dynamical systems.

\textbf{Data and Code Availability.} The ADHD-200 data are available at \url{https://fcon_1000.projects.nitrc.org/indi/adhd200/index.html}; we use the Athena-preprocessed version~\citep{bellec2017neuro}. Anonymized code for the method, simulations, and data analysis will be provided to reviewers and deposited in a public archive before final acceptance.

\textbf{Declaration of Generative AI Use.} The authors formulated the problem, designed the algorithm, derived the theoretical results, designed the experiments, and wrote the initial draft without generative AI. Claude Opus 5.5 and ChatGPT Astra-6 were then used to correct grammar, polish the exposition, check and correct calculations in the proofs, and review code hygiene and the implementation of the experiments. The authors reviewed all AI-assisted edits, verified every proof, and take full responsibility for the content.

%% file: main/8_acknowledgement.tex
\section*{Acknowledgement}
Katherine Tsai is supported in part by NSF Graduate Research Fellowship. Mladen Kolar is supported in part by NSF ECCS-2216912. 
Additionally, this work is partially supported by the National Science Foundation under grants No. 2046795, 1934986, 2205329, NIH 1R01MH116226-01A, NIFA award 2020-67021-32799, the Alfred P. Sloan Foundation.

%% file: main/01appendix.tex
\printunsrtglossary[type=symbols,style=long]\label{sec:notation_list}
\addtocounter{table}{-1}%
\newpage

\section{Wavelet smoothing}\label{appendix:wavelet-reconstruction}

Throughout, $t\in[0,T]$ denotes physical time, and wavelet functions defined on $[0,1]$ are evaluated at $t/T$.
Let $\phi$ and $\psi$ be the father and mother wavelets, and define
$\phi_{jl}(u)=2^{j/2}\phi(2^ju-l)$ and
$\psi_{jl}(u)=2^{j/2}\psi(2^ju-l)$ for $u\in[0,1]$ and $l=0,\ldots,2^j-1$.
For a coarsest level $j_0\geq0$ and with a boundary adjustment near $0$ and $1$ that satisfies the support and boundary conditions of Lemma~\ref{lemma:wavelet:local} in the Appendix, the collection $\{\phi_{j_0l};l=0,\ldots,2^{j_0}-1; \psi_{jl}, j\geq j_0,l=0,\ldots,2^j-1\}$ is an orthonormal basis of $L_2([0,1])$. Our implementation treats the boundary differently, as described in Appendix~\ref{sses:expsetup}. The wavelet coefficients of $X_i$ are
\begin{align*}
\xi_{i,j_0l}&=\frac1T\int_0^T X_i(t)\phi_{j_0l}(t/T)\,\mathrm{d}t,
\quad l=0,\ldots,2^{j_0}-1,\\
\eta_{i,jl}&=\frac1T\int_0^T X_i(t)\psi_{jl}(t/T)\,\mathrm{d}t,
\quad l=0,\ldots,2^j-1,\ j\geq j_0.
\end{align*}
The factor $1/T$ comes from the change of variable $u=t/T$ and preserves the normalization of the coefficients on the unit interval. Therefore, $X_i$ has the wavelet expansion
\[
X_i(t)=\sum_{l=0}^{2^{j_0}-1}\xi_{i,j_0l}\phi_{j_0l}(t/T)
+\sum_{j=j_0}^\infty\sum_{l=0}^{2^j-1}\eta_{i,jl}\psi_{jl}(t/T),
\qquad t\in[0,T].
\]
Next, we estimate the coefficients $\xi_{i,j_0l}, \eta_{i,jl}$ following~\citet{brown1998wavelet}. For simplicity of presentation, let $J$ be an integer such that $N=2^J$; in our implementation, wavelet smoothing uses all $N+1$ observations $Y_{0,i},\ldots,Y_{N,i}$. The dyadic restriction $N=2^J$ is used only to describe the construction, and the numerical implementation also handles non-dyadic sample sizes, through a symmetric extension of the series, as described in Appendix~\ref{sses:expsetup}. Define
$\tilde{X}_i(t)=N^{-1/2}\sum_{n=1}^N Y_{n,i}\phi_{J,n-1}(t/T)$
for $t\in[0,T]$, and let
\begin{align*}
\hat\xi_{i,j_0l}&=\frac1T\int_0^T\tilde X_i(t)\phi_{j_0l}(t/T)\,\mathrm{d}t,
\quad l=0,\ldots,2^{j_0}-1,\\
\tilde\eta_{i,jl}&=\frac1T\int_0^T\tilde X_i(t)\psi_{jl}(t/T)\,\mathrm{d}t,
\quad l=0,\ldots,2^j-1,\ j=j_0,\ldots,J-1.
\end{align*}
These coefficients are computed by a wavelet transform of $Y_{n,i}$ for $n\in[N]$. To remove small coefficients that are likely due to noise, we apply soft-thresholding~\citep{donoho1994ideal, brown1998wavelet}, $
\hat{\eta}_{i,jl}=\sgn(\tilde{\eta}_{i,jl})(|\tilde{\eta}_{i,jl}|-\lambda_{jl})_+
$, with the threshold $\lambda_{jl}=\hat\sigma_{\mathrm{wav}}(2N^{-1}\log N)^{1/2}$, where $\hat\sigma_{\mathrm{wav}}$ is the noise estimate described in Appendix~\ref{sses:expsetup}, which also describes the level-dependent thresholds that our implementation uses instead. The resulting estimated trajectory $\hat{X}_i(t)$ is
\[
\hat{X}_i(t)=\sum_{l=0}^{2^{j_0}-1}\hat{\xi}_{i,j_0l}\phi_{j_0l}(t/T)
+\sum_{j=j_0}^{J-1}\sum_{l=0}^{2^j-1}\hat{\eta}_{i,jl}\psi_{jl}(t/T).
\]
Repeating this procedure for $i\in[p]$ gives the
estimated trajectory $\hat X(t)$ directly on the physical time interval $[0,T]$.
The number of decomposition levels $J-j_0$ determines how far in time the estimate reaches: under the support conditions of Lemma~\ref{lemma:wavelet:local}, a coefficient at level $j$ depends only on the observations in a window whose length is proportional to $2^{J-j}$ sampling intervals. In the analysis of Section~\ref{sec:theory}, we fix $J-j_0$, so that, under the same conditions, $\hat X_i(t)$ depends only on the observations near $t$, and we replace $\hat\sigma_{\mathrm{wav}}$ in the threshold by $\sigma^\star$, so that the threshold is deterministic. This is the construction that we analyze; our implementation differs from it in the boundary treatment and the thresholds, as described in Appendix~\ref{sses:expsetup}.

\section{\texorpdfstring{$\beta$}{Beta}-mixing of Markov processes}\label{appendix:mixing:properties}

Our analysis concerns $\beta$-mixing Markov processes. Therefore, we need sufficient conditions under which a Markov process is strictly stationary and $\beta$-mixing. Establishing $\beta$-mixing is often nontrivial and typically involves multiple steps. Our approach combines results of~\citet{meyn1992stability, meyn1993stability, meyn1993stabilityII, meyn2012markov}. Rather than verifying $\beta$-mixing directly, we first study ergodicity, that is, the convergence of the distribution of the process to a unique stationary distribution in total variation distance. For Markov processes, ergodicity is closely connected to $\beta$-mixing, which captures the asymptotic independence of events separated by increasingly long time intervals. The reason is that the influence of the initial condition vanishes as the process approaches stationarity, so that events far apart in time are nearly independent. \citet{davydov1974mixing} established a formal relationship between mixing and ergodic coefficients. Therefore, our proof verifies a Foster--Lyapunov condition for the generator of the process, which, together with a petite-set condition, implies geometric ergodicity~\citep{meyn1993stability}. Under additional conditions, geometric ergodicity then gives $\beta$-mixing.

We begin by reviewing the main theoretical tools. Section~\ref{ssec:generator} introduces the generator of a Markov process, which is the starting point for the theory of mixing. Section~\ref{sec:mixing_ergodic} defines mixing and relates it to ergodicity. To establish geometric ergodicity, we combine two conditions. The first, the Foster--Lyapunov condition, is discussed in Section~\ref{ssec:fosterlyapunov}. The second, that every compact set is petite for some skeleton chain, is addressed in Section~\ref{ssec:petiteskeleton}. With these tools in hand, Section~\ref{ssec:msd_examples} verifies the Foster--Lyapunov condition for linear and nonlinear Markov-switching ODEs, and the subsequent propositions combine these results with the structural assumptions to establish $\beta$-mixing.

\subsection{Stationarity and mixing}\label{ssec:mixing_stationary}
Recall that mixing coefficients measure how far $W_n$ and $W_{n'}$ are from independence as $|n-n'|$ grows~\citep{bradley2005basic}. We use the $\beta$-mixing coefficient, which compares the joint distribution of the past and the future with the product of their marginals.

\begin{definition}[$\beta$-mixing]
\label{definition:discretebetamixing}
Given a discrete-time process $\{W_n\}_{n\in\mathbb{Z}}$ and $j\in\NN\cup\{0\}$, the $\beta$-mixing coefficient is defined as
$
\beta(j)=\sup_n\beta(\Fcal_{-\infty}^n, \Fcal^{\infty}_{n+j})
=
2^{-1}\sup_n\norm{P_{n,j}-P_{-\infty}^n\otimes P_{n+j}^\infty}_{\text{TV}}$,
where $\Fcal_{-\infty}^n=\sigma(\{W_{n'}:n'\leq n\})$ and $\Fcal_{n+j}^{\infty}=\sigma(\{W_{n'}:n'\geq n+j\})$. Here $P_{n,j}$ is the joint distribution of $(W_{n'})_{n'\leq n}$ and $(W_{n'})_{n'\geq n+j}$, and $P_{-\infty}^n$ and $P_{n+j}^\infty$ are its marginals. The total variation norm is $\|\mu\|_{\mathrm{TV}}=\sup_{|f|\leq1}|\int f\,\mathrm{d}\mu|$, so the factor $2^{-1}$ gives the standard normalization $0\leq\beta(j)\leq1$~\citep[Eqs.~(1.5) and~(1.9)]{bradley2005basic}. A stochastic process is said to be $\beta$-mixing if $\beta(j)\to 0$ as $j\to\infty$. A process is geometrically $\beta$-mixing if the coefficient decays at an exponential rate: there exists a constant $c>0$ such that
$\beta(j)\leq 2\exp(-cj)$.
\end{definition}
Assumption~\ref{assumption:stationary} states that the joint distribution of consecutive interval blocks does not change over time. Assumption~\ref{assumption:mixingbeta} states that the dependence between the past and the future of the sequence decays geometrically in the number of sampling intervals that separate them. Our concentration arguments use the latent process $(X,Z)$ only through these two assumptions. The following proposition gives sufficient conditions for them.

\begin{proposition}\label{prop:general_example}
    Suppose that Assumption~\ref{assumption:edgeset} holds, and define
    \[
    b_\ell(x)=\rbr{\sum_{j=1}^pg(x_j)^\top\theta_{1j}^{\ell\star},\ldots,\sum_{j=1}^pg(x_j)^\top\theta_{pj}^{\ell\star}}^\top\in\RR^p,\quad\ell\in[k],
    \]
    so that $\dot X(t)=b_{Z(t)}(X(t))$ by~\eqref{eq:definex}. Suppose that the following conditions hold:
    \begin{enumerate}
        \item $g(\cdot)$ is locally Lipschitz;
        \item there exists a constant $K_0>0$ such that for each state $\ell\in[k]$,
    $\left\|b_\ell(x)\right\|_2\leq K_0(1+\|x\|_2)$;
    \item for all $x\in\RR^p$ and $\ell\in[k]$,
    $
    x^\top b_\ell(x)\leq c_\ell\norm{x}_2^2+c',$
    for some constants $c_\ell\in\RR$ and $c'>0$.
    \end{enumerate}
    Define $A=-2\mathrm{diag}(c_1,\ldots,c_k)-Q^\star$. Suppose that $A$ is a nonsingular M-matrix (its off-diagonal entries $-q^\star_{\ell\ell'}$ are nonpositive, so this holds if and only if $Av$ has positive entries for some $v\in(0,\infty)^k$), and that Assumptions~\ref{assumption:tchain_accessibility} and~\ref{assumption:A4_1} in Appendix~\ref{appendix:mixing:properties} hold (the $T$-chain property and a reachable point for a skeleton of $(Z(t),X(t))$ on a closed set $\mathcal X$ that is invariant under the flow of each state $\ell\in[k]$, that is, every solution of $\dot x=b_\ell(x)$ that starts in $\mathcal X$ stays in $\mathcal X$, and the moment bound $\sup_{s\geq 0}\EE\norm{X(s)}_2^2<\infty$), and that $X(0)\in\mathcal X$ almost surely.
    Then the joint process $(Z(t),X(t))$ admits a stationary distribution $\Pi$ and is geometrically $\beta$-mixing in the sense of Definition~\ref{definition:betamixing} in Appendix~\ref{appendix:mixing:properties}. If, moreover, $(Z(0),X(0))\sim\Pi$, then the interval-block process $\{W_n\}_{n\in\mathbb Z}$ is strictly stationary and geometrically $\beta$-mixing.
\end{proposition}
The proof of Proposition~\ref{prop:general_example} is given in Appendix~\ref{appendix:mixing:properties}. The first two conditions in Proposition~\ref{prop:general_example} are standard in the theory of Markov-switching differential equations and guarantee existence and uniqueness of the solution~\citep{yin2010hybrid}. The third condition, together with the assumption that $A$ is a nonsingular M-matrix, yields the Foster--Lyapunov condition of Lemma~\ref{lemma:alt2}, via the Lyapunov function used in the proof of Theorem~5.1 of~\citet{yuan2003asymptotic}. To obtain asymptotic stability in distribution, that theorem also assumes the incremental condition $(x-y)^\top\{b_\ell(x)-b_\ell(y)\}\leq c_\ell\norm{x-y}_2^2$. We do not assume this condition, and instead combine the Foster--Lyapunov condition with Assumption~\ref{assumption:tchain_accessibility}. Finally, Proposition~\ref{prop:cute_example} in Appendix~\ref{appendix:mixing:properties} treats skew-symmetric linear dynamics on a sphere, which Proposition~\ref{prop:general_example} does not cover; see also the remark after Assumption~\ref{assumption:tchain_accessibility} on drifts $b_\ell$ that share an equilibrium. In particular, the bases of Sections~\ref{sec:simulation} and~\ref{sec:experiment} satisfy $g(0)=0$, so the origin is an equilibrium of every drift $b_\ell$, and Assumption~\ref{assumption:tchain_accessibility} can hold only on a set $\mathcal X$ in which the origin is absent or isolated.

\subsection{Markov processes and their generators}
\label{ssec:generator}

In this section, we review fundamental properties of Markov processes, with particular emphasis on their construction. We begin by formally defining a Markov process.
\begin{definition}[Markov Process]
\label{definition:MarkovProcess} 
Let $\mathscr{B}$ be the Borel $\sigma$-field in $\RR^p$.
A stochastic process $X(t)\in\RR^p$, defined for $t\geq 0$ on a probability space with probability measure $P$, is a Markov process if for all $A\in\mathscr{B}$ and $0\leq s<t$,
\[
P(X(t)\in A \mid \sigma(\{X(u), u\leq s\}))=P(X(t)\in A\mid X(s)),
\]
where $\sigma(\{X(u), u\leq s\})$ is the $\sigma$-field generated by $\{X(u), u\leq s\}$. 
\end{definition}
For example, the ordinary differential equation (ODE)
\begin{equation}\label{eq:example:ode}
    \mathrm{d}x = f(x_t)\mathrm{d}t,
\end{equation}
can be shown to define a Markov process~\citep{khasminskii2011stochastic}. Furthermore, it is straightforward to verify that the joint process $(X(t),Z(t))$ of the model~\eqref{eq:definey}--\eqref{eq:definex} is also a Markov process~\citep{yin2010hybrid}. In this appendix, $X$ denotes a generic Markov process; in our application it is the joint process $(X(t),Z(t))$, since $X$ alone is generally not Markov. Similarly, $\pi$ denotes a generic stationary distribution, not the stationary distribution $\pi(Q)$ of the latent Markov chain, and $\Pi$ denotes that of $(X,Z)$; $f$ is a generic drift or norm weight, $g$ a test function and $T$ a kernel ($T_t$ a semigroup), unrelated to the dictionary $g$ and the horizon $T$ of the main text.

We first consider the simpler Markov process~\eqref{eq:example:ode}. Following Definition~\ref{definition:MarkovProcess}, we define its \emph{transition kernel} as
\[
P(X(t)\in A\mid X(s)=x) = p(s, x, t, A),
\]
which satisfies the Chapman--Kolmogorov equation:
\begin{equation}
\label{eq:ChaKol}
p(s, x, t, A)=\int_{\RR^p}p(s, x, u, \mathrm{d}y)p(u, y, t, A),\quad s<u<t.
\end{equation}
Given a transition kernel, one can construct a Markov process for any initial distribution. A Markov process is time-homogeneous if its transition kernel is invariant under shifts in time, that is, $p(s, x, t, A)=p(s+u, x, t+u, A)$ for any $u>0$. In this case, we write $p(s,x,t,A)=p(x,t-s,A)$.

Now, suppose that $(X(t))_{t \geq 0}$ is a time-homogeneous Markov process with transition kernel $p(x, t, A)$. We define the operator $T_t$ acting on a measurable test function $V : \RR^p \to \RR$ as
\begin{align}
\label{eq:generatorfunc}
    T_t V(x) = \int p(x, t, \mathrm{d}y)\, V(y) = \EE^x[V(X(t))].
\end{align}  
By the Chapman--Kolmogorov equation~\eqref{eq:ChaKol}, the family $\{T_t\}_{t \geq 0}$ satisfies the semigroup property $T_{t+s} = T_t T_s$, and is therefore referred to as a \emph{(time-homogeneous) semigroup}. The \emph{generator} of the semigroup is then defined as
\[
\mathscr{L} V(x) = \lim_{t \to 0^+} \frac{T_t V(x) - V(x)}{t}.
\]  
We write $\mathrm{Dom}(\mathscr{L})$ for the set of functions $V$ for which this limit exists for every $x$. Under appropriate regularity conditions, the generator uniquely determines the associated transition kernel~\citep[Chapter~3]{khasminskii2011stochastic}.

\noindent\textbf{Generator of the Markov-switching ODE process.} 
We now consider the switching ODE process studied in this manuscript:
\[
\mathrm{d}x = f(x_t, z_t)\mathrm{d}t,
\]
where $z_t \in \{1, \ldots, k\}$ is a discrete-valued process. In particular, $(Z(t))_{t \geq 0}$ is a continuous-time Markov chain with transition rate matrix $Q$. The joint process $(X(t), Z(t))$ is Markovian, and its transition kernel
\[
P(X(t) \in A, Z(t) \in B \mid X(0) = x, Z(0) = z) = p((x, z), t, (A, B))
\]
is time-homogeneous. The associated semigroup operator is defined as
\[
T_t V(x, z) = \int p((x, z), t, (\mathrm{d}y , \mathrm{d}w))\, V(y, w),
\]
and the generator is given by
\[
\mathscr{L} V(x, z) = \lim_{t \to 0^+} \frac{T_t V(x, z) - V(x, z)}{t}.
\]

If $V(\cdot, z)$ is a sufficiently smooth function for each $z \in \{1, \ldots, k\}$, then the generator of the Markov-switching ODE process is given by~\citep[Chapter~2]{yin2010hybrid}
\begin{equation}\label{eq:generatorfunc_ms}
\mathscr{L} V(x, z) = \mathscr{A} V(\cdot, z)(x) + \mathscr{Q} V(x, \cdot)(z),
\end{equation}
where the differential and switching components are defined as
\[
\mathscr{A} V(\cdot, z)(x) = \nabla_x V(\cdot, z)^\top f(x, z), \quad
\mathscr{Q} V(x, \cdot)(z) = \sum_{z'=1}^k q_{z z'} V(x, z').
\]
Here, $\mathscr{A}$ is the generator of the ODE flow of state $z$, and $\mathscr{Q}$ is the generator of the latent Markov chain that governs the switching.

With the generator in hand, we can study the stability, ergodicity and mixing of the Markov process. In what follows, we review the results most relevant to our analysis, and we refer the reader to~\citet{khasminskii2011stochastic} and \citet{meyn2012markov} for a complete treatment.

\subsection{Mixing and ergodicity of Markov processes}
\label{sec:mixing_ergodic}

In this section, we relate the mixing and ergodicity properties of a Markov process. We begin with the definition of $\beta$-mixing. Informally, a process is mixing if $X(s)$ and $X(t)$ become asymptotically independent as the time separation $|t - s|$ tends to infinity. Definition~\ref{definition:discretebetamixing} defines $\beta$-mixing for discrete-time processes, and we now give the corresponding definition for continuous-time processes.

\begin{definition}[$\beta$-mixing for continuous-time processes]
\label{definition:betamixing}
Let $(X(t))_{t\geq0}$ be a continuous-time process. Given $t \geq 0$, the \emph{$\beta$-mixing coefficient} is defined as
\[
\beta(t) = \sup_{s\geq0} \beta(\mathcal{F}_{0}^s, \mathcal{F}_{s+t}^{\infty})
= {\frac12}\sup_{s\geq0} \left\| P_{s, t} - P_{0}^s \otimes P_{s+t}^{\infty} \right\|_{\mathrm{TV}},
\]
where $\mathcal{F}_{0}^s = \sigma\big(\{X(u) : u \in [0, s]\}\big)$ and $\mathcal{F}_{s+t}^{\infty} = \sigma\big(\{X(u) : u \in [s+t, \infty)\}\big)$. Here $P_{s,t}$ is the joint distribution of $(X(u))_{0\leq u\leq s}$ and $(X(u))_{u\geq s+t}$, and $P_{0}^s$ and $P_{s+t}^{\infty}$ are its marginals. For a strictly stationary process indexed by $t\in\RR$, we take the supremum over $s\in\RR$ with $\mathcal F_{-\infty}^s=\sigma(\{X(u):u\leq s\})$ in place of $\mathcal F_0^s$. By stationarity, and since $\beta(\mathcal F_{-u}^s,\mathcal F_{s+t}^\infty)$ increases to $\beta(\mathcal F_{-\infty}^s,\mathcal F_{s+t}^\infty)$ as $u\to\infty$, this two-sided coefficient equals the coefficient of the restriction of the process to $t\geq0$. A stochastic process is said to be \emph{absolutely regular}, or \emph{$\beta$-mixing}, if $\beta(t) \to 0$ as $t \to \infty$. It is said to be \emph{geometrically $\beta$-mixing} if there exist constants $C,c>0$ such that $\beta(t)\leq C\exp(-ct)$ for all $t\geq0$.
\end{definition}
In practice, $\beta$-mixing is often hard to verify directly. Instead, we can first check whether the Markov process is ergodic, which is more tractable and can be established through the Foster--Lyapunov condition. We now define ergodicity of a Markov process.

\begin{definition}
\label{define:ergodic}
A Markov process is said to be \emph{ergodic} if there exists a stationary distribution $\pi$ such that
\[
\lim_{t \to \infty} \left\| p(x, t, \cdot) - \pi \right\|_{\mathrm{TV}} = 0, \quad \text{for all } x \in \mathcal{X}.
\]

\end{definition}
If a Markov process is positive Harris recurrent and some skeleton chain $\{X(n\tilde h)\}_{n\geq0}$ with $\tilde h>0$ is irreducible, then the process is ergodic~\citep[Theorem~6.1]{meyn1993stabilityII}. As mentioned earlier, an ergodic Markov process is also mixing under additional mild conditions. The following lemma formalizes the connection between ergodicity and $\beta$-mixing~\citep{davydov1974mixing, masuda2007ergodicity}.

\begin{lemma}[Lemma~3.9 in~\citet{masuda2007ergodicity}]\label{lemma:ergodicity}
Let $(X(t))_{t \geq 0}$ be a Markov process. Let $\eta$, $(p(\cdot, t, \cdot))_{t \geq 0}$, and $\beta_X(t)$ denote, respectively, the initial distribution, the transition kernel, and the $\beta$-mixing coefficient of $(X(t))_{t \geq 0}$. Suppose that there exist a probability measure $\pi$ on $(\mathcal{X}, \mathscr{B}(\mathcal{X}))$, a measurable function $B : \mathcal{X} \to \RR_+$, and a deterministic sequence $(\delta(t))_{t \geq 0}$ with $\delta(t) \to 0$ as $t \to \infty$, such that:
\begin{enumerate}
    \item \label{cd1} $\|p(x, t, \cdot) - \pi\|_{\mathrm{TV}} \leq B(x)\, \delta(t)$ for all $t \geq 0$ and $x \in \mathcal{X}$;
    \item \label{cd2} $\bar B := \sup_{s \geq 0} \int B(x)\, \eta p(\cdot, s, \mathrm{d}x) < \infty$.
\end{enumerate}
Then $\beta_X(t) \leq 2\bar B \delta(t)$ for all $t \geq 0$; that is, $X$ is $\beta$-mixing at rate $\delta(t)$.
\end{lemma}

Hence, by Lemma~\ref{lemma:ergodicity}, a process that is geometrically ergodic, as defined below, is geometrically $\beta$-mixing whenever the integrability condition of that lemma holds.

\begin{definition}[Geometric Ergodicity]\label{definition:geoergodicity}
Let $(X(t))_{t \geq 0}$ be a Markov process that is positive recurrent and admits a unique stationary distribution $\pi$. The process is said to be \emph{geometrically ergodic} if there exist a constant $\gamma > 0$ and a real-valued function $B : \RR^p \to \RR_+$ such that, for all $t > 0$ and $x \in \RR^p$,
\[
\left\| p(x, t, \cdot) - \pi \right\|_{\mathrm{TV}} \leq B(x) \exp(-\gamma t).
\]
\end{definition}

\subsection{Foster--Lyapunov condition and geometric ergodicity}
\label{ssec:fosterlyapunov}

Having connected ergodicity with mixing, we next give conditions under which a stochastic process is geometrically ergodic. \citet{meyn1993stability} showed that the Foster--Lyapunov condition is a practical way to verify geometric ergodicity, and we state their result in Lemma~\ref{lemma:expergodic}. We first introduce this condition. Recall from Section~\ref{ssec:generator} that $\mathscr{L}$ denotes the generator of the Markov process $(X(t))_{t \geq 0}$.

\begin{assumption}\label{assumption:expErgodic}
There exists a function $V \in \mathrm{Dom}(\mathscr{L})$ with $V\geq0$ such that $V(x) \to \infty$ as $\|x\| \to \infty$, and constants $c > 0$ and $d < \infty$ for which
\[
\mathscr{L} V(x) \leq -c V(x) + d.
\]
\end{assumption}
Assumption~\ref{assumption:expErgodic} is a special case of condition (CD3) in~\citet{meyn1993stability}. Condition (CD3) is stated in terms of the extended generator $\mathscr{L}_m$ of the stopped process (see~\citet{meyn1993stability} for the definition), whereas Assumption~\ref{assumption:expErgodic} uses the generator $\mathscr{L}$ of the original process. However, the two coincide for the Markov-switching ODE when $V(\cdot,\ell)\in\mathcal C^1(\RR^p)$ for each $\ell$, as for the functions in Section~\ref{ssec:msd_examples}. In that case, the chain rule along the flow between jumps gives $\mathscr{L}_m V = \mathscr{L} V$ on the open set at whose exit the process is stopped, with $\mathscr{L} V$ given by~\eqref{eq:generatorfunc_ms}. Note that Assumption~\ref{assumption:expErgodic} alone does not guarantee geometric ergodicity. In addition, we need a discretely sampled version of the continuous-time process that behaves well on every compact set in $\mathscr{B}(\mathcal{X})$, a condition formalized through the notion of a petite set~\citep{meyn1992stability}. Before stating the main result, we introduce three additional concepts.

\begin{definition}[Skeleton]
The \emph{$\tilde{h}$-skeleton chain} of the process $(X(t))_{t \geq 0}$ is defined as $X_n^{(\tilde{h})} := X(n\tilde{h})$ for some fixed $\tilde{h} > 0$ and $n \in \mathbb{N} \cup \{0\}$.
\end{definition}
By definition, $\{X_n^{(\tilde{h})}\}_{n \in \mathbb{N}}$ is a discrete-time Markov chain. To simplify notation, we write $\{X_n\}_{n \in \mathbb{N}}$ for a discrete-time Markov chain with state space $\mathcal{X}$, and $\mathscr{B}(\mathcal{X})$ for the Borel $\sigma$-algebra on $\mathcal{X}$. Its one-step transition kernel is $p(x, A) := p(x, \tilde{h}, A)$ for $x \in \mathcal{X}$ and $A \in \mathscr{B}(\mathcal{X})$. Let $a = \{a(n)\}_{n \geq 0}$ be a probability distribution on the non-negative integers, which serves as a sampling distribution for the times at which the chain is observed. The resulting sampled chain $\{X^a_n\}$ has transition kernel
\begin{equation}\label{eq:define:Ka}
    K_a(x, A) = \sum_{n=0}^\infty p^n(x, A)\, a(n),
\end{equation}
where $p^n(x, A)$ is the $n$-step transition probability.

\begin{definition}[Petite Set, adapted from Section~5.5.2 in~\citet{meyn2012markov}]
A set $C \in \mathscr{B}(\mathcal{X})$ is said to be \emph{$\nu_a$-petite} if the transition kernel $K_a$ of the sampled chain satisfies
\[
K_a(x, B) \geq \nu_a(B)
\]
for all $x \in C$ and $B \in \mathscr{B}(\mathcal{X})$, where $\nu_a$ is a non-trivial measure on $\mathscr{B}(\mathcal{X})$.
\end{definition}

Finally, we define the $f$-norm, which generalizes the total variation norm and plays a central role in the analysis of ergodicity.

\begin{definition}[$f$-norm]
Let $f \geq 1$ be a positive measurable function. For any signed measure $\mu$, the \emph{$f$-norm} of $\mu$ is defined as
\[
\|\mu\|_f = \sup_{|g| \leq f} |\mu(g)|.
\]
\end{definition}

The total variation norm is the $f$-norm with $f \equiv 1$. With the skeleton chain, the petite set and the $f$-norm in place, we can now state the main result on geometric ergodicity.

\begin{lemma}[Theorem~6.1 of~\citet{meyn1993stability}]\label{lemma:expergodic}
Suppose that $(X(t))_{t \geq 0}$ is a right process, and that all compact sets are petite for some skeleton chain. If Assumption~\ref{assumption:expErgodic} holds, then $(X(t))_{t \geq 0}$ has a unique stationary distribution $\pi$, and there exist constants $\gamma > 0$ and $C < \infty$ such that
\begin{equation}\label{eq:gercondition}
\|p(x, t, \cdot) - \pi\|_f \leq C f(x) \exp(-\gamma t), \quad t \geq 0,\ x \in \mathcal{X},
\end{equation}
where $f = V + 1$ and $V$ is the function from Assumption~\ref{assumption:expErgodic}.
\end{lemma}

 Since $V\geq0$, we have $\|p(x,t,\cdot)-\pi\|_{\mathrm{TV}}\leq\|p(x,t,\cdot)-\pi\|_{V+1}$, so inequality~\eqref{eq:gercondition} implies geometric ergodicity with $B(x)=C\{V(x)+1\}$. Besides the right-process property, which we verify for the switching ODE in the proof of Proposition~\ref{prop:cute_example}, Lemma~\ref{lemma:expergodic} requires Assumption~\ref{assumption:expErgodic} and the compact-petite condition for a skeleton chain. The next section shows how the $T$-chain property and a reachable point imply the latter condition.

\subsection{Petite set and skeleton chain}
\label{ssec:petiteskeleton}

 In this section, we review the $T$-chain property and reachability, which together imply that all compact sets are petite, as required in Lemma~\ref{lemma:expergodic}. Throughout, the state space is a locally compact separable metric space with its Borel $\sigma$-algebra.

Let $p(\cdot, \cdot)$ denote the transition kernel of a discrete-time Markov chain. For example, the transition kernel of an $\tilde{h}$-skeleton chain is given by $p(x, A) := p(x, \tilde{h}, A)$ for $x \in \mathcal{X}$ and $A \in \mathscr{B}(\mathcal{X})$. Define the quantity
\[
G(x, A) = \sum_{n=1}^\infty p^n(x, A),
\]
 which is the expected number of visits to $A$ at positive integer times when the chain starts from $x$. In particular, $G(x,A)>0$ means that $A$ can be reached from $x$ with positive probability.

We now introduce the concept of irreducibility, following~\citet{meyn1992stability}.

\begin{definition}[$\varphi$-irreducibility]
A Markov chain $\{X_n\}_{n \in \mathbb{N}}$ is said to be \emph{$\varphi$-irreducible} if there exists a non-trivial finite measure $\varphi$ on $\mathscr{B}(\mathcal{X})$ such that $G(x, A) > 0$ for all $x \in \mathcal{X}$ whenever $\varphi(A) > 0$. The measure $\varphi$ is called an \emph{irreducibility measure}.
\end{definition}

 We use the $T$-chain property and a reachable point to establish $\varphi$-irreducibility and the compact-petite condition. The following definitions use the terminology of Chapter~6 of~\citet{meyn2012markov}.

\begin{definition}[$T$-chain]\label{definition:tchain}
A kernel $T$ is said to be a continuous component of the sampled chain kernel $K_a$ defined in~\eqref{eq:define:Ka} if
\[
K_a(x, A) \geq T(x, A), \quad \text{for all } x \in \mathcal{X}, A \in \mathscr{B}(\mathcal{X}),
\]
where $T(\cdot, A)$ is a lower semicontinuous function for every $A \in \mathscr{B}(\mathcal{X})$. If in addition $T(x, \mathcal{X}) > 0$ for all $x \in \mathcal{X}$, then the chain $\{X_n\}_{n \in \mathbb{N}}$ is called a \emph{$T$-chain}.
\end{definition}

 The next ingredient is \emph{reachability}, the existence of a \emph{reachable point}: every neighborhood of this point can be reached with positive probability from every initial state.

\begin{definition}[Reachable Point]
A point $x^* \in \mathcal{X}$ is said to be \emph{reachable} if for every open neighborhood $O \in \mathscr{B}(\mathcal{X})$ of $x^*$,
\[
\sum_{n=1}^\infty p^n(x, O) > 0, \quad \text{for all } x \in \mathcal{X}.
\]
\end{definition}

The following lemma shows that a $T$-chain with a reachable point is $\varphi$-irreducible, which is an important requirement for verifying geometric ergodicity.

\begin{lemma}[Proposition~6.2.1 in~\citet{meyn2012markov}]
\label{lemma:tchain}
If the Markov chain $\{X_n\}_{n \in \mathbb{N}}$ is a $T$-chain and the state space $\mathcal{X}$ contains a reachable point $x^*$, then $\{X_n\}$ is $\varphi$-irreducible with $\varphi = T(x^*, \cdot)$.
\end{lemma}

The next lemma then gives the compact-petite condition.

\begin{lemma}[Theorem~6.2.5(ii) of~\citet{meyn2012markov}]
\label{lemma:tchain_compact_petite}
Let $\{X_n\}_{n\geq0}$ be a time-homogeneous Markov chain on a locally compact separable metric space $\mathcal X$, equipped with its Borel $\sigma$-algebra. If the chain is a $\varphi$-irreducible $T$-chain, then every compact subset of $\mathcal X$ is petite.
\end{lemma}

We summarize the argument. By Lemma~\ref{lemma:tchain}, a skeleton that is a $T$-chain and has a reachable point is $\varphi$-irreducible, so that every compact set is petite by Lemma~\ref{lemma:tchain_compact_petite}. Together with Assumption~\ref{assumption:expErgodic}, the Foster--Lyapunov condition, this yields both conditions of Lemma~\ref{lemma:expergodic}. Under the integrability condition of Lemma~\ref{lemma:ergodicity}, the resulting geometric ergodicity bound gives geometric $\beta$-mixing. Section~\ref{ssec:msd_examples} verifies Assumption~\ref{assumption:expErgodic} for switching ODEs.

\subsection{Markov-switching ODEs}
\label{ssec:msd_examples}

In this section, we verify the Foster--Lyapunov condition for Markov-switching ODEs. In particular, we present two examples in which there exist functions $V$ satisfying this condition
\[
\mathscr{L} V \leq -a V + d
\]
for some constants $a > 0$ and $d < \infty$.

\subsubsection{Linear model}

 We first consider a linear switching system whose trajectories stay on a sphere. The next lemma establishes the Foster--Lyapunov condition on a fixed sphere.

\begin{lemma}
\label{lemma:betamixing_x}
 Let $A_\ell\in\RR^{p\times p}$ for $\ell\in[k]$, and consider
\[
\dot X(t)=A_{Z(t)}X(t),
\]
where $Z(t)$ is a finite-state continuous-time Markov chain with transition rate matrix $Q$. Suppose $A_\ell^\top=-A_\ell$ for every $\ell$ and $\mathcal X=\{x\in\RR^p:\|x\|_2=R_{\mathrm{sph}}\}$ for a fixed $R_{\mathrm{sph}}>0$. Then there exist a positive function $V:\RR^p\times[k]\to(0,\infty)$ with $V(\cdot,\ell)\in\mathcal C^1(\RR^p)$ and $V(x,\ell)\to\infty$ as $\|x\|_2\to\infty$, and constants $a,d>0$, such that
\[
\mathscr LV(x,\ell)\leq-aV(x,\ell)+d,
\qquad x\in\mathcal X,\quad \ell\in[k].
\]
\end{lemma}

\begin{proof}[Proof of Lemma~\ref{lemma:betamixing_x}]
 Since $A_\ell^\top=-A_\ell$ and $x^\top A_\ell x$ is a scalar,
\[
x^\top A_\ell x=(x^\top A_\ell x)^\top=x^\top A_\ell^\top x=-x^\top A_\ell x,
\]
so $x^\top A_\ell x=0$. Let $V(x,\ell)=1+\|x\|_2^2$ for every $\ell\in[k]$. This function is positive, continuously differentiable on $\RR^p$, and tends to infinity as $\|x\|_2\to\infty$. Since $V$ does not depend on $\ell$ and the rows of $Q$ sum to zero,
\[
\mathscr LV(x,\ell)
=2x^\top A_\ell x+(1+\|x\|_2^2)\sum_{\ell'=1}^kq_{\ell\ell'}
=0,
\qquad x\in\RR^p,\quad\ell\in[k].
\]
Fix $a>0$ and let $d=a(1+R_{\mathrm{sph}}^2)$. Since $V(x,\ell)=1+R_{\mathrm{sph}}^2$ for $x\in\mathcal X$, we obtain $\mathscr LV(x,\ell)=0=-aV(x,\ell)+d$ for $x\in\mathcal X$ and $\ell\in[k]$.
\end{proof}

\subsubsection{General setting}

We now consider a more general nonlinear setting. The following two assumptions are standard in the study of nonlinear dynamical systems with switching behavior~\citep{skorokhod2009asymptotic, yin2010hybrid}.
Recall from Section~\ref{ssec:setup} that $g:\RR\to\RR^m$, with $g(\cdot)=(g_1(\cdot),\ldots,g_m(\cdot))^\top\in\RR^m$, is the vector of basis functions used to construct the differential equation. By~\eqref{eq:approx_f} and Assumption~\ref{assumption:edgeset}, the effect of node $j$ on node $i$ under latent state $\ell$ is
\[
f^\ell_{ij}(X_j(t))= g(X_j(t))^\top \theta_{ij}^{\ell\star},
\]
where $\theta_{ij}^{\ell\star}\in\RR^{ m}$ is an $m$-dimensional column vector.
\begin{assumption}\label{assumption:A1}
Let $g : \RR \to \RR^m$ be a vector-valued function applied component-wise to the state $x = (x_1, \ldots, x_p)^\top \in \RR^p$. Assume that $g$ is locally Lipschitz continuous.
\end{assumption}
Define 
\[
b_\ell(x)=\rbr{\sum_{j=1}^pg(x_j)^\top \theta_{1j}^{\ell\star},\ldots, \sum_{j=1}^pg(x_j)^\top \theta_{pj}^{\ell\star}}^\top\in\RR^p,\quad \ell=1,\ldots,k,
\]
so that $\dot X(t)=b_{Z(t)}(X(t))$ by~\eqref{eq:definex}.
\begin{assumption}\label{assumption:A2}
There exists a constant $K_0 > 0$ such that for each state $\ell = 1, \ldots, k$, and all $x \in \RR^p$,
\[
\left\| b_\ell(x)\right\|_2 \leq K_0 (1 + \|x\|_2),
\]
where $x = (x_1, \ldots, x_p)^\top \in \RR^p$.
\end{assumption}

Assumptions~\ref{assumption:A1}--\ref{assumption:A2} guarantee the existence and uniqueness of solutions to the switching system~\citep{skorokhod2009asymptotic, yin2010hybrid}.

\begin{lemma}\label{lemma:alt2}
Suppose Assumptions~\ref{assumption:A1}--\ref{assumption:A2} hold. For each state $\ell = 1, \ldots, k$, assume that there exist constants $c_\ell \in \RR$ and $c' > 0$ such that
\[
x^\top b_\ell(x) \leq c_\ell \|x\|_2^2 + c', \quad \text{for all } x \in \RR^p.
\]
Let
\[
A := -2 \operatorname{diag}(c_1, \ldots, c_k) - Q^\star,
\]
and suppose that \( A \) is a nonsingular M-matrix; since the off-diagonal entries $-q^\star_{\ell\ell'}$ of $A$ are nonpositive, this holds if and only if $Av$ has positive entries for some $v\in(0,\infty)^k$. Then there exist a vector $v\in(0,\infty)^k$ and constants $a>0$ and $d>0$ such that $V(x,\ell):=v_\ell\|x\|_2^2$ satisfies
\[
\mathscr{L} V(x, \ell) \leq -a V(x, \ell) + d, \quad \text{for all } x \in \RR^p \text{ and } \ell\in[k].
\]
\end{lemma}

\begin{proof}
We adapt the Lyapunov function used in the proof of Theorem~5.1 of~\citet{yuan2003asymptotic}, which studies the stability of Markov-switching SDEs.

Since \( A \) is a nonsingular M-matrix, there exists a vector \( v = (v_1, \ldots, v_k)^\top \in \RR^k \) with strictly positive entries such that
\[
\tilde{c} := A v,
\]
and all entries of \( \tilde{c} \) are strictly positive. Define the candidate Lyapunov function $V(x, \ell) := v_\ell \|x\|_2^2$.
Applying the generator to \( V \), we obtain
\begin{align*}
\mathscr{L} V(x, \ell) 
&= 2 v_\ell x^\top b_\ell(x)
+ \sum_{\ell'=1}^k q_{\ell \ell'}^\star v_{\ell'} \|x\|_2^2 \\
&\leq \left( 2 v_\ell c_\ell + \sum_{\ell'=1}^k q_{\ell \ell'}^\star v_{\ell'} \right) \|x\|_2^2 + 2 c' v_\ell \\
&= -\tilde{c}_\ell \|x\|_2^2 + 2 c' v_\ell \\
&= -\frac{\tilde{c}_\ell}{v_\ell} V(x, \ell) + 2 c' v_\ell.
\end{align*}
Let \( a := \min_\ell \left( \tilde{c}_\ell / v_\ell \right) \) and \( d := \max_\ell (2 c' v_\ell) \). Then for all \( x \in \RR^p \) and \( \ell\in[k] \),
\[
\mathscr{L} V(x, \ell) \leq -a V(x, \ell) + d.
\]
\end{proof}

\subsection{Sufficient conditions for petite sets}\label{ssec:irreduciblechain}

 Lemmas~\ref{lemma:betamixing_x} and~\ref{lemma:alt2} address the Foster--Lyapunov condition in Lemma~\ref{lemma:expergodic}. We now impose a structural assumption that supplies its compact-petite condition for a skeleton of the joint process.

\begin{assumption}[$T$-chain and reachable point]
\label{assumption:tchain_accessibility}
Assume that $\mathcal X\subset\RR^p$ is closed, and every trajectory starting in $\mathcal X$ remains in $\mathcal X$ under the flow of each state $\ell\in[k]$. There exists $\tilde h>0$ such that the skeleton
\[
\{(X(n\tilde h),Z(n\tilde h))\}_{n\geq0}
\]
is a $T$-chain on $\mathcal X\times[k]$ and admits a reachable point $(x^*,\ell^*)\in\mathcal X\times[k]$.
\end{assumption}
The choice of $\mathcal X$ matters for this assumption. For switching linear systems $b_\ell(x)=A_\ell x$, the origin is an equilibrium of every drift $b_\ell$, which prevents the $T$-chain condition from holding on $\RR^p$. The assumption may nevertheless hold on a closed invariant set $\mathcal X$ that excludes the origin, such as a sphere of positive radius when the matrices $A_\ell$ are skew-symmetric, as considered in Proposition~\ref{prop:cute_example}. The $T$-chain and reachability conditions must then be verified on that restricted set. More generally, the assumption fails on any $\mathcal X$ in which all drifts $b_\ell$ share an equilibrium that is not an isolated point of $\mathcal X$, for example the origin for every basis with $g(0)=0$.

By Lemma~\ref{lemma:tchain}, the skeleton in Assumption~\ref{assumption:tchain_accessibility} is $\varphi$-irreducible. Lemma~\ref{lemma:tchain_compact_petite} therefore implies that every compact subset of $\mathcal X\times[k]$ is petite for this skeleton.

\subsection{Proof of Proposition~\ref{prop:general_example}}\label{ssec:proofprop:general}
Sections~\ref{ssec:msd_examples} and~\ref{ssec:irreduciblechain} address the Foster--Lyapunov condition and petite sets, respectively. We now apply these results to show that the process is $\beta$-mixing.

Recall that $p(\cdot, s, \cdot)$ denotes the transition kernel of the process $(X(t), Z(t))$, and $\eta$ is its initial distribution.

The following proposition uses Lemma~\ref{lemma:betamixing_x} to give a stationary linear example on a fixed sphere, under Assumption~\ref{assumption:tchain_accessibility}.

\begin{proposition}[Linear Model]\label{prop:cute_example}
Let $p\geq2$, $R_{\mathrm{sph}}>0$, and $\mathcal X=\{x\in\RR^p:\|x\|_2=R_{\mathrm{sph}}\}$, and let $X(0)\in\mathcal X$. Let $Z(t)$ have the constant transition rate matrix $Q^\star$ and unique stationary distribution $\pi^\star$. Consider
\[
\dot X(t)=A_{Z(t)}X(t),
\qquad A_\ell^\top=-A_\ell,\quad \ell\in[k].
\]
Under Assumption~\ref{assumption:tchain_accessibility}, the joint process $(X(t),Z(t))$ admits a stationary distribution $\Pi$ and is geometrically $\beta$-mixing. If $(X(0),Z(0))\sim\Pi$, then the interval-block process $\{W_n\}_{n\in\mathbb Z}$ is strictly stationary and geometrically $\beta$-mixing.
\end{proposition}

Starting from $X(0)=x\in\mathcal X$, the trajectory $X(t)$ moves along the sphere of radius $R_{\mathrm{sph}}=\|x\|_2$, since the flow $e^{tA_\ell}$ of each state $\ell$ is orthogonal and switching changes the active vector field while leaving $X$ continuous. Therefore, $\|X(t)\|_2=R_{\mathrm{sph}}$ for all $t\geq0$. The $T$-chain and reachability properties remain explicit assumptions on the joint process. Since the state space is compact, the Lyapunov integrability condition needed in Lemma~\ref{lemma:ergodicity} holds automatically.

\begin{proof}[Proof of Proposition~\ref{prop:cute_example}]
We proceed in three steps. The first two steps verify the conditions of Lemma~\ref{lemma:expergodic} under the assumptions of Proposition~\ref{prop:cute_example}, and the third step passes from the joint process to the interval-block process.

\emph{Step 1}.
Since the switching ODE is linear, we can take $g(u)=u$ and $b_\ell(x)=A_\ell x$, so that Assumptions~\ref{assumption:A1}--\ref{assumption:A2} are satisfied. Under these assumptions, $(X(t),Z(t))$ is a piecewise-deterministic Markov process whose flows, one for each state $\ell\in[k]$, are generated by the locally Lipschitz vector fields $b_\ell$ and do not explode, and whose switches between states occur at the bounded rates $-q^\star_{\ell\ell}$. Therefore, it is a Borel right process~\citep[Theorem~27.8]{davis1993markov}, and so is its restriction to the closed invariant set $\mathcal X\times[k]$~\citep[Section~2.1]{meyn1993stabilityII}, as required in Lemma~\ref{lemma:expergodic}. Under Assumption~\ref{assumption:tchain_accessibility}, Lemmas~\ref{lemma:tchain} and~\ref{lemma:tchain_compact_petite} show that all compact sets are petite for a skeleton chain, which is the compact-petite condition of Lemma~\ref{lemma:expergodic}.

\emph{Step 2}.
Using Lemma~\ref{lemma:betamixing_x} with $Q=Q^\star$ and Step 1, we apply Lemma~\ref{lemma:expergodic} and obtain a stationary distribution $\Pi$ such that
\begin{equation}\label{eq:pxl}
\|p((x,\ell), t, \cdot) - \Pi\|_{\mathrm{TV}} \leq \|p((x,\ell), t, \cdot) - \Pi\|_{V + 1} \leq C (V(x, \ell) + 1) e^{-c t}, \quad t \geq 0,
\end{equation}
for some constants $C > 0$ and $c > 0$. The bound~\eqref{eq:pxl}, together with the boundedness of $V+1$ on the compact state space $\mathcal X\times[k]$, verifies the conditions of Lemma~\ref{lemma:ergodicity}. Therefore, the joint process $(Z(t), X(t))$ is geometrically $\beta$-mixing.

\emph{Step 3}.
Define the latent interval block
\[
\mathcal U_n:=\bigl((Z(t_{n-1}+u),X(t_{n-1}+u))\bigr)_{0\leq u\leq h}.
\]
For every $\tau\geq1$,
\[
\sigma(\mathcal U_j:j\leq n)
\subseteq\mathcal F_{-\infty}^{t_n},
\qquad
\sigma(\mathcal U_j:j\geq n+\tau)
\subseteq\mathcal F_{t_{n+\tau-1}}^\infty.
\]
By these inclusions, and since $(Z(t),X(t))$ is geometrically $\beta$-mixing, there exist constants $C_\beta,c_\beta>0$ such that
\[
\beta_{\mathcal U}(\tau)
\leq\beta_{(Z,X)}((\tau-1)h)
\leq C_\beta\exp\{-c_\beta(\tau-1)h\}.
\]
Under the stationary initialization, Lemma~\ref{lemma:strictstationary} shows that $(Z(t),X(t))$ is
strictly stationary, and hence so is $\{\mathcal U_n\}_{n\in\mathbb Z}$. Since
\[
W_n=(\mathcal U_n,\varepsilon_n),
\]
and $\{\varepsilon_n\}_{n\in\mathbb Z}$ is i.i.d.\ and independent
of $\{\mathcal U_n\}_{n\in\mathbb Z}$, the process
$\{W_n\}_{n\in\mathbb Z}$ is strictly stationary. Furthermore,
Lemma~\ref{lemma:betamix:upperbound} gives, for $j\in\NN$,
\[
\beta_W(j)
\leq\beta_{\mathcal U}(j)
\leq C_\beta\exp\{-c_\beta(j-1)h\}.
\]
Since $\beta_W(j)\leq1$, and the right-hand side is at most $2\exp(-c_\beta hj/2)$ for all sufficiently large $j$, there exists $c_W>0$ such that $\beta_W(j)\leq2\exp(-c_Wj)$ for all $j\in\NN$. Hence $\{W_n\}_{n\in\mathbb Z}$ is geometrically $\beta$-mixing in the sense of Definition~\ref{definition:discretebetamixing}, as required in Assumption~\ref{assumption:mixingbeta}. This completes the proof.
\end{proof}

Next, we prove the mixing property of the general model in Proposition~\ref{prop:general_example}.
As before, $p(\cdot, s, \cdot)$ denotes the transition kernel of the process $(X(t), Z(t))$, and $\eta$ denotes its initial distribution. The proof uses the following moment condition.
\begin{assumption}\label{assumption:A4_1}
There exist constants $v_\ell > 0$ for $\ell = 1, \ldots, k$ such that
\[
\sup_{s \in \RR_+} \sum_{\ell=1}^k \int (v_\ell \|x\|_2^2 + 1)\, \eta p(\cdot, s, (\mathrm{d}x, \ell)) < \infty.
\]
\end{assumption}

\begin{proof}[Proof of Proposition~\ref{prop:general_example}]
The proof follows the same three steps as the proof of Proposition~\ref{prop:cute_example}. In the first step, conditions~1 and~2 of Proposition~\ref{prop:general_example} are Assumptions~\ref{assumption:A1}--\ref{assumption:A2}, so that, as in Step~1 of that proof, $(X(t),Z(t))$ is a right process on $\mathcal X\times[k]$. Furthermore, Assumption~\ref{assumption:tchain_accessibility} and Lemma~\ref{lemma:tchain} give a $\varphi$-irreducible $T$-chain skeleton, and by Lemma~\ref{lemma:tchain_compact_petite}, every compact subset of $\mathcal X\times[k]$ is petite for this skeleton. In the second step, we apply Lemma~\ref{lemma:expergodic} with the Lyapunov function constructed in Lemma~\ref{lemma:alt2}. Assumption~\ref{assumption:A4_1} gives the integrability of this Lyapunov function under the marginal distributions, as required by Lemma~\ref{lemma:ergodicity}. This establishes geometric $\beta$-mixing of $(Z(t),X(t))$. Under the stationary initialization in Proposition~\ref{prop:general_example}, the final step is identical to Step 3 in the proof of Proposition~\ref{prop:cute_example}, and it establishes strict stationarity and geometric $\beta$-mixing of the interval-block process $\{W_n\}_{n\in\mathbb Z}$. This completes the proof.
\end{proof}

\section{Derivation of the surrogate objective}\label{sec:surrogate-explanation}

In this section, we derive the EM-type surrogate objective ${\Lcal}_{1,N}^*(\tilde{\Theta}\mid\Theta)$, which combines a continuous-time Markov-chain contribution with a model for the observation increments $Y_n-Y_{n-1}$. Since we model only the increments, the initial observation $Y_0$ enters only through $Y_1-Y_0$.
Let $m_{\ell\ell'}$ and $\tau_\ell$ denote the number of transitions of $Z(t)$ from state $\ell$ to state $\ell'$ and the occupation time of state $\ell$ (the time that $Z(t)$ spends in $\ell$), respectively, during $[t_0,t_N]$. Under the current parameter $\Theta$, we write $m_{\ell\ell'}(Y_0^N;\Theta):=\EE[m_{\ell\ell'}\mid Y_0^N;\Theta]$ and $\tau_\ell(Y_0^N;\Theta):=\EE[\tau_\ell\mid Y_0^N;\Theta]$ for the posterior expected number of transitions and the posterior expected occupation time. With the diagonal entries $\tilde{q}_\ell=-\sum_{\ell'\ne\ell}\tilde{q}_{\ell\ell'}$ of the transition rate matrix $\tilde Q$, the expected Markov-chain log-likelihood is
\begin{align}
\sum_{\ell\ne\ell'}m_{\ell\ell'}(Y_0^N;\Theta)\log \tilde{q}_{\ell\ell'}+\sum_{\ell=1}^k\tilde{q}_\ell\tau_{\ell}(Y_0^N;\Theta).\label{eq:llq}
\end{align}
For the observation component, we use an increment model.
Given the estimated trajectory $\hat{X}_i$, we compute $\hat{\Psi}_i(t_n):=\int_{t_{n-1}}^{t_n} g(\hat{X}_i(u))\mathrm{d}u$ as an estimate of the unobserved quantity ${\Psi}_i(t_n):=\int_{t_{n-1}}^{t_n} g({X}_i(u))\mathrm{d}u$. Based on~\eqref{eq:approxY}, we use the increment approximation
\[
Y_{n,i} \approx Y_{n-1,i}+\sigma^\star(\varepsilon_{n,i}-\varepsilon_{n-1,i}) + \sum_{j}\hat{\Psi}_j(t_n)^\top \theta_{ij}^{Z(t_n)\star}.
\]
For each $n$, the measurement-noise increment $\sigma^\star(\varepsilon_{n,i}-\varepsilon_{n-1,i})$ follows $\Ncal(0,2\sigma^{\star2})$, and the diagonal-covariance surrogate treats these increments as independent across time. Under this increment model, the contribution of the observation increments to the surrogate is
\begin{multline}
    -\sum_{n=1}^N \frac{p}{2}(\log2\tilde{\sigma}^2+\log 2\pi)\\
    -\frac{1}{4\tilde{\sigma}^2}\sum_{n,\ell,i=1}^{N,k,p}
p(Z(t_n)=\ell\mid Y_0^N ; \Theta)\rbr{Y_{n,i}-Y_{n-1,i} - \sum_{j=1}^p\hat{\Psi}_j(t_n)^\top \tilde{\theta}_{ij}^{\ell}}^2.\label{eq:lly}
\end{multline}
Adding~\eqref{eq:llq} and~\eqref{eq:lly}, we obtain the surrogate objective of Definition~\ref{def:surrogate_function}
\begin{align}
{\Lcal}_{1,N}^*(\tilde{\Theta}\mid\Theta)&=\sum_{\ell\ne\ell'} m_{\ell\ell'}(Y_0^N;\Theta)\log \tilde{q}_{\ell\ell'}+\sum_{\ell=1}^k\tilde{q}_\ell\tau_{\ell}(Y_0^N;\Theta)-\sum_{n=1}^N \frac{p}{2}(\log2\tilde{\sigma}^2+\log 2\pi)\notag\\
    &\quad-\frac{1}{4\tilde{\sigma}^2}\sum_{n,\ell,i=1}^{N,k,p}
p(Z(t_n)=\ell\mid Y_0^N ; \Theta)\rbr{Y_{n,i}-Y_{n-1,i} - \sum_{j=1}^p\hat{\Psi}_j(t_n)^\top \tilde{\theta}_{ij}^{\ell}}^2\notag.
\end{align}
Since we maximize a surrogate of the log-likelihood, the self-consistency property of $\Theta^\star$ does not hold in general for the population objective. However, Lemma~\ref{lemma:aggregate-surrogate-bias} bounds the bias $\dist(M(\Theta^\star),\Theta^\star)$ of one population update at the true parameter, and Appendix~\ref{sec:surrogate-bias} gives its proof.
 
\subsection{Covariance-corrected log-likelihood of the differenced observations}\label{ssec:corrected-loglikelihood}
In this subsection, we contrast the surrogate objective with a covariance-corrected EM-type auxiliary objective, which accounts for the temporal correlation of the increments, and we argue that the surrogate is computationally simpler. As before, we treat the estimated trajectory $\hat{X}(t)$, and hence the covariates $\hat{\Psi}_j(t_n)=\int_{t_{n-1}}^{t_n}g(\hat{X}_j(u))\,\mathrm{d}u$, as fixed plug-in quantities. We work with the increments
\[
Y^\Delta_{n,i}:=Y_{n,i}-Y_{n-1,i},\qquad
Y^\Delta_i:=(Y^\Delta_{1,i},\ldots,Y^\Delta_{N,i})^\top\in\mathbb{R}^N.
\]
For a sequence $z=(z_0,z_1,\ldots,z_N)$ of latent states at the sampling times and a candidate parameter $\tilde{\Theta}$, the covariance-corrected increment model is
\[
Y^\Delta_i\mid\{Z_0^N=z\},\hat{X}{\sim}\mathcal{N}\bigl(\hat{\mu}_i(z;\tilde{\Theta}),\,\tilde{\sigma}^{2}\Lambda_N\bigr),
\]
where
\[
\hat{\mu}_i(z;\tilde{\Theta}) :=
\begin{pmatrix}
\sum_{j=1}^p\hat{\Psi}_j(t_1)^\top\tilde{\theta}_{ij}^{z_1}\\
\vdots\\
\sum_{j=1}^p\hat{\Psi}_j(t_N)^\top\tilde{\theta}_{ij}^{z_N}
\end{pmatrix}
\in\mathbb{R}^N.
\]
The model freezes the latent state at $Z(t_n)$ over $[t_{n-1},t_n]$, discards the switching remainder $\rho_{n,i}$ and the basis residual $r_{n,i}$ in~\eqref{eq:approxY}, and uses $\hat{\Psi}_j(t_n)$ in place of the unobserved $\Psi_j(t_n)$. It corrects the covariance induced by differencing the observation noise, but it does not correct these drift approximations.
The covariance has the moving-average structure produced by differencing the i.i.d.\
measurement noise,
\[
\Lambda_N=
\begin{pmatrix}
2 & -1 &        &        &    \\
-1 & 2 & -1     &        &    \\
   & \ddots & \ddots & \ddots &  \\
   &   & -1     & 2      & -1 \\
   &   &        & -1     & 2
\end{pmatrix}\in\mathbb{R}^{N\times N},
\]
because $\varepsilon^{\Delta}_{n,i}:=\sigma^\star(\varepsilon_{n,i}-\varepsilon_{n-1,i})$ satisfies
$\operatorname{Var}(\varepsilon^{\Delta}_{n,i})=2\sigma^{\star2}$,
$\operatorname{Cov}(\varepsilon^{\Delta}_{n,i},\varepsilon^{\Delta}_{n\pm1,i})=-\sigma^{\star2}$, and
$\operatorname{Cov}(\varepsilon^{\Delta}_{n,i},\varepsilon^{\Delta}_{m,i})=0$ for $|n-m|\ge2$.
This is the unconditional covariance of the differenced noise. It is not the conditional covariance given $Y_0$ or $\hat X$, both of which depend on the noise, so the covariance-corrected model is also only an approximation.

By Lemma~\ref{lem:cholesky}, we have the Cholesky factorization $\Lambda_N=\Gamma\Gamma^\top$, where $\Gamma$ is invertible and lower bidiagonal with entries
\[
\Gamma_{nn}=\sqrt{\tfrac{n+1}{n}},\qquad
\Gamma_{n,n-1}=-\sqrt{\tfrac{n-1}{n}},\qquad
\det\Lambda_N=\prod_{n=1}^{N}\tfrac{n+1}{n}=N+1.
\]
Therefore, $\Lambda_N^{-1}=\Gamma^{-\top}\Gamma^{-1}$, and the quadratic form becomes the squared $\ell_2$ norm of the whitened residual:
\[
\bigl(Y^\Delta_i-\hat{\mu}_i(z;\tilde{\Theta})\bigr)^\top\Lambda_N^{-1}\bigl(Y^\Delta_i-\hat{\mu}_i(z;\tilde{\Theta})\bigr)=\bigl\lVert \Gamma^{-1}\bigl(Y^\Delta_i-\hat{\mu}_i(z;\tilde{\Theta})\bigr)\bigr\rVert_2^2.
\]
Since the coordinates of the Gaussian measurement noise are independent, the log-likelihood of the covariance-corrected model conditional on $z$ is the sum of the Gaussian log-densities over $i=1,\dots,p$:
\begin{align*}
\sum_{i=1}^{p}\log\mathcal{N}\bigl(Y^\Delta_i;\,\hat{\mu}_i(z;\tilde{\Theta}),\,\tilde{\sigma}^{2}\Lambda_N\bigr)
&=-\frac{Np}{2}\log\bigl(2\pi\tilde{\sigma}^{2}\bigr)
  -\frac{p}{2}\log(N+1)\\
&\quad-\frac{1}{2\tilde{\sigma}^{2}}\sum_{i=1}^{p}
   \bigl\lVert \Gamma^{-1}\bigl(Y^\Delta_i-\hat{\mu}_i(z;\tilde{\Theta})\bigr)\bigr\rVert_2^2 .
\end{align*}

Let $\Ical$ denote the set of realizations of $Z_0^N=(Z(t_0),Z(t_1),\ldots,Z(t_N))$. Averaging this log-likelihood over $Z_0^N$ with respect to its posterior distribution under the current parameter $\Theta$, and adding the Markov-chain part~\eqref{eq:llq}, we obtain, up to an additive constant, the covariance-corrected EM-type auxiliary objective
\begin{multline*}
\Lcal^{\mathrm{cc}}_N(\tilde{\Theta}\mid\Theta)=\sum_{\ell\ne\ell'}m_{\ell\ell'}(Y_0^N;\Theta)\log \tilde{q}_{\ell\ell'}+\sum_{\ell=1}^k\tilde{q}_\ell\tau_{\ell}(Y_0^N;\Theta)
-\frac{Np}{2}\log\bigl(2\pi\tilde{\sigma}^2\bigr)\\
\quad-\frac{1}{2\tilde{\sigma}^2}\sum_{z\in\Ical}\sum_{i=1}^pP(Z_0^N=z\mid Y_0^N;\Theta)
\bigl\lVert \Gamma^{-1}\bigl(Y^\Delta_i-\hat{\mu}_i(z;\tilde{\Theta})\bigr)\bigr\rVert_2^2 .
\end{multline*}
The correction enters only through $\Gamma^{-1}$.
The latent-path sum contains $|\Ical|=k^{N+1}$ terms, and the standard E-step cannot evaluate it, because the dense lower-triangular matrix $\Gamma^{-1}$ couples the increment residuals across all time points and precludes the local factorization that the forward--backward algorithm requires. This motivates the diagonal-covariance surrogate objective of Definition~\ref{def:surrogate_function}.

The following lemma gives the Cholesky factorization used above.

\begin{lemma}[Cholesky factor of $\Lambda_N$]\label{lem:cholesky}
Let $\Lambda_N=\operatorname{tridiag}(-1,2,-1)\in\mathbb R^{N\times N}$, that is,
$(\Lambda_N)_{nn}=2$, $(\Lambda_N)_{n,n-1}=(\Lambda_N)_{n-1,n}=-1$, and $0$ otherwise.
Then $\Lambda_N=\Gamma\Gamma^\top$ for the lower-bidiagonal matrix $\Gamma$ with entries
\[
\Gamma_{nn}=\sqrt{\frac{n+1}{n}}\ \ (n=1,\dots,N),
\qquad
\Gamma_{n,n-1}=-\sqrt{\frac{n-1}{n}}\ \ (n=2,\dots,N),
\]
and all other entries zero. Therefore,
\[
\det\Lambda_N=\prod_{n=1}^{N}\Gamma_{nn}^2=\prod_{n=1}^{N}\frac{n+1}{n}=N+1 .
\]
Since the determinant of $\Gamma$ is nonzero, $\Gamma$ is invertible.
\end{lemma}

\begin{proof}
We first verify the factorization entrywise and then compute the determinant. Since $\Gamma$ is lower bidiagonal, the product $\Gamma\Gamma^\top$ is tridiagonal, and its entries
are determined by the diagonal $\Gamma_{nn}$ and subdiagonal $\Gamma_{n,n-1}$ alone:
\begin{align}
(\Gamma\Gamma^\top)_{nn}       &= \Gamma_{nn}^2+\Gamma_{n,n-1}^2, \label{eq:chol-diag}\\
(\Gamma\Gamma^\top)_{n,n-1}    &= \Gamma_{n,n-1}\,\Gamma_{n-1,n-1}, \label{eq:chol-off}\\
(\Gamma\Gamma^\top)_{nm}       &= 0 \qquad (|n-m|\ge 2), \label{eq:chol-zero}
\end{align}
where the term $\Gamma_{n,n-1}^2$ in~\eqref{eq:chol-diag} is absent for $n=1$. Equation
\eqref{eq:chol-zero} holds because rows $n$ and $m$ of $\Gamma$ have overlapping support only
when $|n-m|\le1$, so that $\Lambda_N$ and $\Gamma\Gamma^\top$ agree off the tridiagonal band.
It remains to verify that~\eqref{eq:chol-diag} and~\eqref{eq:chol-off} match the
entries of $\Lambda_N$.

\emph{Diagonal entry.} For $n=1$, \eqref{eq:chol-diag} gives
$\Gamma_{11}^2=\frac{2}{1}=2=(\Lambda_N)_{11}$. For $n\ge2$,
\[
\Gamma_{nn}^2+\Gamma_{n,n-1}^2
=\frac{n+1}{n}+\frac{n-1}{n}
=\frac{(n+1)+(n-1)}{n}
=\frac{2n}{n}
=2
=(\Lambda_N)_{nn}.
\]

\emph{Subdiagonal entry.} For $n\ge2$, \eqref{eq:chol-off} gives
\[
\begin{aligned}
(\Gamma\Gamma^\top)_{n,n-1}
&=\Gamma_{n,n-1}\,\Gamma_{n-1,n-1}\\
&=\Bigl(-\sqrt{\tfrac{n-1}{n}}\Bigr)\Bigl(\sqrt{\tfrac{n}{n-1}}\Bigr)
=-1=(\Lambda_N)_{n,n-1}.
\end{aligned}
\]

\emph{Superdiagonal entry.} For $n\le N-1$,
\[
(\Gamma\Gamma^\top)_{n,n+1}
=\Gamma_{nn}\,\Gamma_{n+1,n}
=\Bigl(\sqrt{\tfrac{n+1}{n}}\Bigr)\Bigl(-\sqrt{\tfrac{n}{n+1}}\Bigr)
=-1
=(\Lambda_N)_{n,n+1}.
\]
Therefore, $\Gamma\Gamma^\top=\Lambda_N$ entrywise.

\emph{Determinant.} Since $\Gamma$ is triangular, $\det\Gamma=\prod_{n=1}^N \Gamma_{nn}$, so that
$\det\Lambda_N=(\det\Gamma)^2=\prod_{n=1}^N \Gamma_{nn}^2=\prod_{n=1}^N\frac{n+1}{n}$, and this product
telescopes:
\[
\prod_{n=1}^{N}\frac{n+1}{n}
=\frac{2}{1}\cdot\frac{3}{2}\cdots\frac{N+1}{N}
=N+1 .
\]
\end{proof}

\section{Proofs of Proposition~\ref{prop:mvt} and Lemma~\ref{lemma:aggregate-surrogate-bias}}\label{sec:surrogate-bias}
\begin{proof}[Proof of Proposition~\ref{prop:mvt}]
Let $\Theta, \Theta'\in B(r_0,{\Theta}^\star)$ and define
\[
\Theta_u:=\Theta+u(\Theta'-\Theta),\qquad u\in[0,1].
\]
Since $\Omega$ is convex, as noted in Section~\ref{ssec:popEM}, and both conditions in the definition of $B(r_0,\Theta^\star)$ define convex sets, $B(r_0,\Theta^\star)$ is convex. Therefore, $\Theta_u\in B(r_0,\Theta^\star)$
for every $u\in[0,1]$. Since the population update map is continuously differentiable on an open neighborhood of this set by Assumption~\ref{assumption:population}, the fundamental theorem of calculus gives
\[
M(\Theta')-M(\Theta)
=\int_0^1\mathcal{J}_M(\Theta_u)(\Theta'-\Theta)\,\mathrm{d}u.
\]
By the triangle inequality for integrals,
\[
\left\|\int_0^1\mathcal{J}_M(\Theta_u)(\Theta'-\Theta)\,\mathrm{d}u\right\|_{\mathrm d}
\leq \int_0^1
\|\mathcal{J}_M(\Theta_u)(\Theta'-\Theta)\|_{\mathrm d}
\,\mathrm{d}u,
\]
and Assumption~\ref{assumption:population} and the definition of the
induced operator norm give
\begin{align*}
\dist\bigl(M(\Theta'),M({\Theta})\bigr)
&\leq\int_0^1
\|\mathcal{J}_M(\Theta_u)(\Theta'-\Theta)\|_{\mathrm d}\,\mathrm{d}u\\
&\leq\int_0^1
\|\mathcal{J}_M(\Theta_u)\|_{\mathrm d}
\|\Theta'-\Theta\|_{\mathrm d}\,\mathrm{d}u\\
&\leq\kappa\dist(\Theta',\Theta).
\end{align*}
\end{proof}

\subsection{Bias between \texorpdfstring{$M_{\sigma^2}(\Theta^\star)$}{the variance update} and \texorpdfstring{$\sigma^{\star2}$}{the true variance}}
\label{appendix:biassigma}

The following lemma bounds the bias of the variance update at the true parameter.
\begin{lemma}[Bias of the variance update]
\label{lem:bias-sigma}
Suppose that
$\|\Psi(t_n)\|_\infty\leq Bh$ for every $n\in[N]$ almost surely, that $\EE[\rho_{n,i}^2]<\infty$ for every $n\in[N]$ and $i\in[p]$, and that
$\theta_i^{\ell\star}$ is $ms$-sparse for every $i\in[p]$ and
$\ell\in[k]$. Write $w_{\ell,n}:=p(Z(t_n)=\ell\mid Y_0^N;\Theta^\star)$ for the smoothing probabilities at $\Theta^\star$, and define 
\[
\alpha_{\ell\ell'}
:=\frac{1}{N}\sum_{n=1}^N
\EE\sbr{w_{\ell,n}\one\{Z(t_n)=\ell'\}},
\qquad \ell,\ell'\in[k],
\]
and the root mean squares of the switching remainder
\[
\bar\rho_i:=\cbr{\frac{1}{N}\sum_{n=1}^N\EE\sbr{\rho_{n,i}^2}}^{1/2},
\quad i\in[p],
\qquad
\bar\rho:=\cbr{\frac{1}{p}\sum_{i=1}^p\bar\rho_i^2}^{1/2}.
\]
Then
\[
\left|M_{\sigma^2}(\Theta^\star)-\sigma^{\star2}\right|
\leq b_{\sigma^2},
\]
where
\[
b_{\sigma^2}
:=\bar\rho^2+2msB^2h^2A_\theta
+2\sigma^\star Bh\sqrt{msA_\theta},
\qquad
A_\theta:=\frac{1}{p}\sum_{i=1}^p\sum_{\ell'=1}^k\sum_{\ell\neq\ell'}
\alpha_{\ell\ell'}\norm{\theta_i^{\ell'\star}-\theta_i^{\ell\star}}_2^2.
\]
\end{lemma}
\begin{proof}
We write $M_{\sigma^2}(\Theta^\star)$, the maximizer of the surrogate population objective over the variance, as $\sigma^{\star2}$ plus a squared remainder and a cross term, and we bound the last two terms separately. By definition, we have
\[
M_{\sigma^2}(\Theta^\star) = \frac{1}{2pN}\sum_{n=1}^N\sum_{i=1}^p\EE\sbr{\sum_{\ell=1}^k p (Z(t_n)=\ell\mid Y_0^N;\Theta^\star)\rbr{Y_{n,i}-Y_{n-1,i}-\Psi(t_n)^\top\theta_{i}^{\ell\star}}^2}.
\]
By~\eqref{eq:approxY} with $r_{n,i}=0$ under Assumption~\ref{assumption:edgeset}, the residual inside the square is
\begin{equation*}
Y_{n,i}-Y_{n-1,i}-\Psi(t_n)^\top\theta_{i}^{\ell\star}=\sigma^\star\rbr{\varepsilon_{n,i}-\varepsilon_{n-1,i}}+\rho_{n,i}+\Psi(t_n)^\top (\theta_i^{Z(t_n)\star}-\theta_i^{\ell\star}).
\end{equation*}
Therefore, we can rewrite $M_{\sigma^2}(\Theta^\star)$ as
\begin{align}
\frac{1}{2pN}\sum_{n=1}^N\sum_{i=1}^p\EE\Biggl[&\sum_{\ell=1}^k p (Z(t_n)=\ell\mid Y_0^N;\Theta^\star)\notag\\
&\times\cbr{\sigma^\star\rbr{\varepsilon_{n,i}-\varepsilon_{n-1,i}}+\rho_{n,i}+\Psi(t_n)^\top (\theta_i^{Z(t_n)\star}-\theta_i^{\ell\star})}^2\Biggr].\label{eq:loss_expansion}
\end{align}

First, since the smoothing probabilities sum to one, we have
\[
\frac{1}{2pN}\sum_{n=1}^N\sum_{i=1}^p\EE\sbr{\sum_{\ell=1}^k p (Z(t_n)=\ell\mid Y_0^N;\Theta^\star)\sigma^{\star2}(\varepsilon_{n,i}-\varepsilon_{n-1,i})^2} = \sigma^{\star2}.
\]
Recall that $\|\Psi(t_n)\|_\infty\leq Bh$ for every $n\in[N]$, and that each $\theta_i^{\ell\star}$ is $ms$-sparse. Then the vector $\theta_i^{\ell'\star}-\theta_i^{\ell\star}$ is at most $2ms$-sparse, and for $\ell'\neq\ell$ we have
\[
\left|
\Psi(t_n)^\top
\rbr{\theta_i^{\ell'\star}-\theta_i^{\ell\star}}
\right|
\leq
\|\Psi(t_n)\|_\infty
\left\|
\theta_i^{\ell'\star}-\theta_i^{\ell\star}
\right\|_1
\leq
Bh\sqrt{2ms}\,
\left\|
\theta_i^{\ell'\star}-\theta_i^{\ell\star}
\right\|_2.
\]
To average this bound over the posterior, we split according to the value of $Z(t_n)$. The terms with $Z(t_n)=\ell$ vanish, and therefore
\[
\frac{1}{pN}\sum_{n=1}^N\sum_{i=1}^p
\EE\sbr{\sum_{\ell=1}^k w_{\ell,n}
\norm{\theta_i^{Z(t_n)\star}-\theta_i^{\ell\star}}_2^2}
=\frac{1}{p}\sum_{i=1}^p\sum_{\ell'=1}^k\sum_{\ell\neq\ell'}
\alpha_{\ell\ell'}\norm{\theta_i^{\ell'\star}-\theta_i^{\ell\star}}_2^2
=A_\theta.
\]
Using $(a+b)^2\leq 2a^2+2b^2$, the fact that the smoothing probabilities sum to one, the definition of $\bar\rho$, the sparsity bound and the last display, we obtain
\begin{align}
&\frac{1}{2pN}\sum_{n=1}^N\sum_{i=1}^p\EE\sbr{\sum_{\ell=1}^k p (Z(t_n)=\ell\mid Y_0^N;\Theta^\star)\cbr{\rho_{n,i}+\Psi(t_n)^\top (\theta_i^{Z(t_n)\star}-\theta_i^{\ell\star})}^2}\notag\\
&\quad\leq
\bar\rho^2+
\frac{1}{pN}\sum_{n=1}^N\sum_{i=1}^p\EE\sbr{\sum_{\ell=1}^k p (Z(t_n)=\ell\mid Y_0^N;\Theta^\star)\cbr{\Psi(t_n)^\top (\theta_i^{Z(t_n)\star}-\theta_i^{\ell\star})}^2}\notag\\
&\quad\leq
\bar\rho^2+
\frac{2msB^2h^2}{pN}\sum_{n=1}^N\sum_{i=1}^p
\EE\sbr{\sum_{\ell=1}^k w_{\ell,n}
\norm{\theta_i^{Z(t_n)\star}-\theta_i^{\ell\star}}_2^2}\notag\\
&\quad=\bar\rho^2+2msB^2h^2A_\theta.
\label{eq:sigma_expansion_up}
\end{align}
Subtracting the noise-only term from~\eqref{eq:loss_expansion} leaves the squared remainder, which is bounded in~\eqref{eq:sigma_expansion_up}, and the cross term. The part of the cross term that contains $\rho_{n,i}$ vanishes. Indeed, let $\mathcal F_{XZ}$ be the $\sigma$-field generated by $(X(t),Z(t))_{t\in[t_0,t_N]}$. Since the smoothing probabilities sum to one, this part equals $\sigma^\star(pN)^{-1}\sum_{n,i}\EE[(\varepsilon_{n,i}-\varepsilon_{n-1,i})\rho_{n,i}]$. Each $\rho_{n,i}$ is $\mathcal F_{XZ}$-measurable and square integrable, and, by~\eqref{eq:definey}, $\varepsilon_0,\ldots,\varepsilon_N$ are i.i.d.\ $\Ncal(0,I_p)$ and independent of $\mathcal F_{XZ}$, so that $\EE[(\varepsilon_{n,i}-\varepsilon_{n-1,i})\rho_{n,i}]=0$. For the remaining part, since the smoothing probabilities sum to one and $\mathbb{E}(\varepsilon_{n,i}-\varepsilon_{n-1,i})^2=2$, the weighted Cauchy--Schwarz inequality and the middle line of~\eqref{eq:sigma_expansion_up} give
\begin{align*}
&\left|\frac{1}{pN}\sum_{n=1}^N\sum_{i=1}^p\EE\sbr{\sum_{\ell=1}^k w_{\ell,n}\,\sigma^\star(\varepsilon_{n,i}-\varepsilon_{n-1,i})\,\Psi(t_n)^\top(\theta_i^{Z(t_n)\star}-\theta_i^{\ell\star})}\right|\\
&\leq \sqrt2\sigma^\star\left[\frac{1}{pN}\sum_{n=1}^N\sum_{i=1}^p\EE\sbr{\sum_{\ell=1}^k w_{\ell,n}\cbr{\Psi(t_n)^\top(\theta_i^{Z(t_n)\star}-\theta_i^{\ell\star})}^2}\right]^{1/2}
\leq2\sigma^\star Bh\sqrt{msA_\theta}.
\end{align*}
The triangle inequality and~\eqref{eq:sigma_expansion_up} therefore yield
\begin{align*}
&\abr{M_{\sigma^2}(\Theta^\star)-\sigma^{\star2}}\\
&\quad\leq\bar\rho^2+2msB^2h^2A_\theta
+2\sigma^\star Bh\sqrt{msA_\theta}.
\end{align*}

\end{proof}

\subsection{Bias between \texorpdfstring{$M_{Q}(\Theta^\star)$}{the rate update} and \texorpdfstring{$Q^\star$}{the true rates}}
\label{appendix:biasQ}
The following lemma bounds the bias of the transition-rate update at the true parameter.
\begin{lemma}[Bias of the transition-rate update]
\label{lem:bias-q}
Fix $\ell,\ell'\in[k]$ with $\ell\neq\ell'$, and suppose that
$\EE[\tau_\ell(Y_0^N;\Theta^\star)]>0$. For $a,b\in[k]$, let
\[
w_{a,\Theta^\star}(t_n):=P(Z(t_n)=a\mid Y_0^N;\Theta^\star),
\qquad
w_{ab,\Theta^\star}(t_n):=P(Z(t_{n-1})=a,Z(t_n)=b\mid Y_0^N;\Theta^\star)
\]
be the smoothing probabilities and joint smoothing probabilities at $\Theta^\star$, and let $P(t)=\exp(Q^\star t)$ for $t\geq0$. Then
\[
\left|M_{q_{\ell\ell'}}(\Theta^\star)-q_{\ell\ell'}^\star\right|
\leq b_{q_{\ell\ell'}},
\]
where
\[
b_{q_{\ell\ell'}}
:=
\frac{
\sum_{n=1}^N\sum_{a,b=1}^k
\left|\EE[w_{ab,\Theta^\star}(t_n)]
-P_{ab}(h)\EE[w_{a,\Theta^\star}(t_{n-1})]\right|
\left|\bar m_{ab}(Q^\star)\right|
}{
\sum_{n=1}^N\sum_{a,b=1}^k
\EE[w_{ab,\Theta^\star}(t_n)]
\EE[\tau_\ell(h)\mid Z(0)=a,Z(h)=b;Q^\star]
},
\]
and, for $Q\in\Omega_Q$ and $a,b\in[k]$, define
\begin{equation}\label{eq:nudef}
  \bar m_{ab}(Q)
  :=
  \mathbb{E}\big[\,m_{\ell\ell'}(h)\big|Z(0)=a,Z(h)=b;Q\,\big]
  -
  q_{\ell\ell'}\,
  \mathbb{E}\big[\,\tau_{\ell}(h)\big|Z(0)=a,Z(h)=b;Q\,\big].
\end{equation}
Here $q_{\ell\ell'}$ is the $(\ell,\ell')$ entry of $Q$, and $m_{\ell\ell'}(h)$ and $\tau_\ell(h)$ are the number of transitions from $\ell$ to $\ell'$ and the occupation time of $\ell$ during $[0,h]$ by a continuous-time Markov chain $Z$ with transition rate matrix $Q$.
\end{lemma}
\begin{proof}
Recall that
\[
M_{q_{\ell\ell'}}(\Theta^\star)=\frac{\EE[m_{\ell\ell'}(Y_0^N;\Theta^\star)]}{\EE[\tau_\ell(Y_0^N;\Theta^\star)]}.
\]
To compute the bias, we decompose
\[
\abr{M_{q_{\ell\ell'}}(\Theta^\star)-q_{\ell\ell'}^\star} = \frac{1}{\EE[\tau_{\ell}(Y_0^N;\Theta^\star)]}\abr{\EE[m_{\ell\ell'}(Y_0^N;\Theta^\star)]-q_{\ell\ell'}^\star\EE[\tau_{\ell}(Y_0^N;\Theta^\star)]}.
\]

By~\eqref{eq:computemll} and~\eqref{eq:computetaul}, we have
\begin{align*}
\EE[m_{\ell\ell'}(Y_0^N;\Theta^\star)]-q_{\ell\ell'}^\star\EE[\tau_{\ell}(Y_0^N;\Theta^\star)]
=\sum_{n=1}^N\sum_{a,b=1}^k \EE[{w}_{ab,\Theta^\star}(t_n)]\bar m_{ab}(Q^\star).
\end{align*}
By Lemma~\ref{lem:rowidentity}, $\sum_{b=1}^kP_{ab}(h)\bar m_{ab}(Q^\star)=0$ for every $a\in[k]$, so that
\[
\sum_{n=1}^N\sum_{a,b=1}^kP_{ab}(h)\EE[w_{a,\Theta^\star}(t_{n-1})]\bar m_{ab}(Q^\star)=0.
\]
Subtracting this sum and applying the triangle inequality outside the expectations, we obtain
\begin{align}\label{eq:wij}
\left|\sum_{n=1}^N\sum_{a,b=1}^k \EE[{w}_{ab,\Theta^\star}(t_n)]\bar m_{ab}(Q^\star)\right|
&\leq\sum_{n=1}^N\sum_{a,b=1}^k\left|\EE[w_{ab,\Theta^\star}(t_n)]-P_{ab}(h)\EE[w_{a,\Theta^\star}(t_{n-1})]\right||\bar m_{ab}(Q^\star)|.
\end{align}
In addition, we can decompose
\begin{align}\label{eq:tau_ij}
\EE[\tau_\ell(Y_0^N;\Theta^\star)] &= \sum_n\sum_{a,b}\EE[w_{ab,\Theta^\star}(t_n)]\EE[\tau_\ell(h)\mid Z(0)=a,Z(h)=b;Q^\star].
\end{align}
Combining~\eqref{eq:wij} and~\eqref{eq:tau_ij}, we obtain
\[
\abr{M_{q_{\ell\ell'}}(\Theta^\star)-q_{\ell\ell'}^\star}\leq
\frac{\sum_{n=1}^N\sum_{a,b=1}^k\left|\EE[w_{ab,\Theta^\star}(t_n)]-P_{ab}(h)\EE[w_{a,\Theta^\star}(t_{n-1})]\right||\bar m_{ab}(Q^\star)|}{\sum_{n=1}^N\sum_{a,b=1}^k\EE[w_{ab,\Theta^\star}(t_n)]\EE[\tau_\ell(h)\mid Z(0)=a,Z(h)=b;Q^\star]}.
\]
The denominator is the total expected occupation time of state $\ell$ across the $N$ sampling intervals. The numerator aggregates the endpoint-specific discrepancies $\bar m_{ab}(Q^\star)$ between the expected number of $\ell\to\ell'$ transitions and $q_{\ell\ell'}^\star$ times the expected occupation time of state $\ell$. They are weighted by $|\EE[w_{ab,\Theta^\star}(t_n)]-P_{ab}(h)\EE[w_{a,\Theta^\star}(t_{n-1})]|$, which measures how far the joint smoothing probabilities are from being consistent with the transition probabilities of the latent Markov chain. This weight is zero for the exact posterior of the model, for which $\EE[w_{ab,\Theta^\star}(t_n)]=P(Z(t_{n-1})=a,Z(t_n)=b)=P_{ab}(h)\EE[w_{a,\Theta^\star}(t_{n-1})]$, so the numerator measures the effect of the surrogate approximation.
For a simpler but coarser $N$-free bound, let
\[
d_{\ell,\min}:=\min_{a,b\in[k]}\EE[\tau_\ell(h)\mid Z(0)=a,Z(h)=b;Q^\star].
\]
Using $|\EE[w_{ab,\Theta^\star}(t_n)]-P_{ab}(h)\EE[w_{a,\Theta^\star}(t_{n-1})]|\leq1$, which holds because both terms lie in $[0,1]$, and lower-bounding each interval's expected occupation time gives
\begin{align*}
\abr{M_{q_{\ell\ell'}}(\Theta^\star)-q_{\ell\ell'}^\star}
&\leq\frac{N\sum_{a,b=1}^k|\bar m_{ab}(Q^\star)|}{N d_{\ell,\min}}
=\frac{\sum_{a,b=1}^k|\bar m_{ab}(Q^\star)|}{d_{\ell,\min}}.
\end{align*}
The factors of $N$ cancel, so this coarser bound does not grow with $N$, and it is finite because $d_{\ell,\min}$ is positive. However, it is large for small $h$: since $P_{ab}(u)=O(u)$ for $a\neq b$, \eqref{eq:decomposetau} gives $d_{\ell,\min}\leq\EE[\tau_\ell(h)\mid Z(0)=Z(h)=a;Q^\star]=O(h^3)$ for $a\neq\ell$, while $|\bar m_{\ell\ell'}(Q^\star)|\to1$ as $h\to0$, so that the coarser bound grows at least like $h^{-3}$.

\end{proof}

The following lemma gives the row identity used in the proof above.

\begin{lemma}[Row identity for the centered bridge statistic]\label{lem:rowidentity}

Fix $\ell,\ell'\in[k]$ with $\ell\neq\ell'$. Let $Q\in\Omega_Q$ be the transition rate matrix of an irreducible continuous-time Markov chain, let
$P(t)=\exp(Qt)$, and let $\bar m_{ab}(Q)$ be as in~\eqref{eq:nudef}. Then, for every $a\in[k]$,
\begin{equation*}
  \sum_{b=1}^{k} P_{ab}(h)\,\bar m_{ab}(Q)=0 .
\end{equation*}
\end{lemma}
 
\begin{proof}[Proof of Lemma~\ref{lem:rowidentity}]
Recall from~\eqref{eq:decomposem} and~\eqref{eq:decomposetau} that
\begin{equation*}
  \mathbb{E}\big[m_{\ell\ell'}(h)\mid Z(0)=a,Z(h)=b;Q\big]
  =\frac{q_{\ell\ell'}}{P_{ab}(h)}\int_{0}^{h} P_{a\ell}(u)\,P_{\ell'b}(h-u)\,du,
\end{equation*}
\begin{equation*}
  \mathbb{E}\big[\tau_{\ell}(h)\mid Z(0)=a,Z(h)=b;Q\big]
  =\frac{1}{P_{ab}(h)}\int_{0}^{h} P_{a\ell}(u)\,P_{\ell b}(h-u)\,du .
\end{equation*}
Multiplying each identity by $P_{ab}(h)$ removes the denominator. We then sum over the terminal state
$b$ and interchange the finite sum with the integral. Since the rows of a
stochastic matrix sum to one, $\sum_{b}P_{\ell'b}(h-u)=\sum_{b}P_{\ell b}(h-u)=1$
for every $u\in[0,h]$, and therefore
\begin{equation*}
  \sum_{b=1}^{k}P_{ab}(h)\,
  \mathbb{E}\big[m_{\ell\ell'}(h)\mid Z(0)=a,Z(h)=b;Q\big]
  =
  q_{\ell\ell'}\int_{0}^{h}P_{a\ell}(u)\,du ,
\end{equation*}
\begin{equation*}
  \sum_{b=1}^{k}P_{ab}(h)\,
  \mathbb{E}\big[\tau_{\ell}(h)\mid Z(0)=a,Z(h)=b;Q\big]
  =
  \int_{0}^{h}P_{a\ell}(u)\,du .
\end{equation*}
Therefore,
\begin{align*}
\sum_{b=1}^k P_{ab}(h)\bar m_{ab}(Q) &= \sum_{b=1}^k P_{ab}(h) \mathbb{E}\big[m_{\ell\ell'}(h)\mid Z(0)=a,Z(h)=b;Q\big]\\&\quad - \sum_{b=1}^kP_{ab}(h) q_{\ell\ell'}\mathbb{E}\big[\tau_{\ell}(h)\mid Z(0)=a,Z(h)=b;Q\big]\\& = 0.
\end{align*}
\end{proof}

\subsection{Bias between \texorpdfstring{$M_{\theta_i^\ell}(\Theta^\star)$}{the theta update} and \texorpdfstring{$\theta_{i}^{\ell\star}$}{the true theta}}
\label{appendix:biastheta}

The following lemma bounds the bias of the parameter update at the true parameter. Its noise term depends on how sensitive the smoothing probabilities are to the observations, which we bound first. In the population map, the smoothing probabilities at $\Theta^\star$ are computed under the surrogate model of Definition~\ref{def:surrogate_function} with the reconstructed covariates $\hat\Psi(t_{n'})$ and the initial distribution $Z(t_0)\sim\pi(Q^\star)$ used in the analysis. For $\ell\in[k]$ and $n\in[N]$, define
\[
w_{\ell,n}:=p(Z(t_n)=\ell\mid Y_0^N;\Theta^\star).
\]
This is the probability assigned by the surrogate model to state $\ell$ at time $t_n$, given all observations $Y_0,\ldots,Y_N$ and with the parameters fixed at $\Theta^\star$. Thus $w_{\ell,n}$ is a function of the observations, whereas $\ell$ and $n$ specify the state and sampling time of interest.
Under this model, the log-emission of state $\ell'\in[k]$ at time $n'\in[N]$ is, up to an additive constant,

\[
e_{\ell'}(n'):=-\frac{1}{4\sigma^{\star2}}\sum_{j=1}^p \hat u_{\ell',n',j}^2,
\qquad
\hat u_{\ell',n',j}:=Y_{n',j}-Y_{n'-1,j}-\hat\Psi(t_{n'})^\top\theta_j^{\ell'\star},
\]

and the initial distribution of $Z(t_0)$ does not depend on $Y_0^N$. The observations enter $e_{\ell'}(n')$ through the increments $Y_{n',j}-Y_{n'-1,j}$ and, through the estimated trajectory, through $\hat\Psi(t_{n'})$. A change of $\hat\Psi_i(t_{n'})$ moves the emission means of all rows $j$ with $\theta_{ji}^{\ell'\star}\neq0$ for some $\ell'\in[k]$, and we measure the size of this effect by the largest column $\ell_1$ norm of the true parameters $\theta_{ji}^{\ell'\star}$,
\[
C_{\mathrm{col}}:=\max_{\ell'\in[k],\,i\in[p]}\sum_{j=1}^p\norm{\theta_{ji}^{\ell'\star}}_1.
\]

\begin{lemma}[Sensitivity of the smoothing probabilities]
\label{lem:posterior-sensitivity}
Suppose that Assumption~\ref{assumption:local} holds, let each $g_{b}$ be continuously differentiable on $[-R_X,R_X]$, and let $D\geq\max_{b\in[m]}\sup_{|x|\leq R_X}|g_{b}'(x)|$. Suppose that, for every $i\in[p]$, the estimated trajectory $\hat X_i$ depends on $Y_0^N$ only through $(Y_{0,i},\ldots,Y_{N,i})$ and is a Lipschitz function of it uniformly in time, that is, $\sup_{u\in[0,T]}|\hat X_i(u;y)-\hat X_i(u;y')|\leq C_{\mathrm{Lip}}\norm{y-y'}_2$ for some $C_{\mathrm{Lip}}<\infty$ and all $y,y'\in\RR^{N+1}$, and that there is a constant $C_{\hat X}$ such that, for all $i\in[p]$ and $n\in\{0,\ldots,N\}$ and almost every $Y_0^N$,
\[
\int_0^T\abr{\frac{\partial\hat X_i(u)}{\partial Y_{n,i}}}\,\mathrm du\leq hC_{\hat X}.
\]
Fix $i\in[p]$, $\ell\in[k]$ and $n\in[N]$, and let
\[
\begin{aligned}
\bar U_{n,i}&:=\max\cbr{|\hat u_{\ell',n',i}|:\ell'\in[k],\ n'\in\{n-1,n,n+1\}\cap[N]},\\
\bar U_{n,i}^{\mathrm{ch}}&:=\max\cbr{|\hat u_{\ell',n',j}|:\ell'\in[k],\ j\in[p],\ \theta_{ji}^{\ell'\star}\neq0,\ n'\in[N],\ |n'-n|\leq r_{\mathrm{rec}}+1},
\end{aligned}
\]
where the maximum over the empty set is zero. Then $w_{\ell,n}=p(Z(t_n)=\ell\mid Y_0^N;\Theta^\star)$ is a locally Lipschitz function of $Y_0^N$, and, for almost every $Y_0^N$,
\[
\abr{\frac{\partial w_{\ell,n}}{\partial Y_{n,i}}-\frac{\partial w_{\ell,n}}{\partial Y_{n-1,i}}}
\leq\frac{4}{\sigma^{\star2}}\,w_{\ell,n}(1-w_{\ell,n})\cbr{\bar U_{n,i}+\frac{Dh\,C_{\hat X}C_{\mathrm{col}}}{2}\,\bar U_{n,i}^{\mathrm{ch}}}.
\]
\end{lemma}

\begin{proof}

We prove the stated bound by applying the chain rule to the dependence of $w_{\ell,n}$ on the observations through the log-emissions. We have,
\[
\frac{\partial w_{\ell,n}}{\partial Y_{n_0,i}}
=\sum_{n'=1}^N\sum_{\ell'=1}^k
\frac{\partial w_{\ell,n}}{\partial e_{\ell'}(n')}
\frac{\partial e_{\ell'}(n')}{\partial Y_{n_0,i}},
\qquad n_0\in\{n-1,n\}.
\]
We first justify differentiability and bound the derivatives with respect to the log-emissions. We then differentiate the log-emissions with respect to the observations, and combine these bounds to control $\partial w_{\ell,n}/\partial Y_{n,i}-\partial w_{\ell,n}/\partial Y_{n-1,i}$.

For a state sequence $z_0^N=(z_0,\ldots,z_N)\in[k]^{N+1}$, the surrogate model gives
\[
\begin{aligned}
& p\bigl(Z(t_0)=z_0,\ldots,Z(t_N)=z_N\mid Y_0^N;\Theta^\star\bigr)\\
&\quad=
\frac{
\pi_{z_0}(Q^\star)\,P_{z_0z_1}(h)\exp\{e_{z_1}(1)\}
\,\cdots\,
P_{z_{N-1}z_N}(h)\exp\{e_{z_N}(N)\}
}{
\displaystyle\sum_{a_0,\ldots,a_N\in[k]}
\pi_{a_0}(Q^\star)\prod_{n'=1}^N
\bigl[P_{a_{n'-1}a_{n'}}(h)\exp\{e_{a_{n'}}(n')\}\bigr]
}.
\end{aligned}
\]
Here $P(h)=\exp(Q^\star h)$. In the numerator, $\pi_{z_0}(Q^\star)$ is the initial-state probability, and each sampling interval contributes a transition factor $P_{z_{n'-1}z_{n'}}(h)$ and an emission factor $\exp\{e_{z_{n'}}(n')\}$. The denominator sums the same expression over all possible state sequences, so the displayed probabilities sum to one.

Therefore, $w_{\ell,n}$ is a smooth function of the log-emissions $e_{\ell'}(n')$, $\ell'\in[k]$, $n'\in[N]$, and each $e_{\ell'}(n')$ is a smooth function of the increments and of $\hat\Psi(t_{n'})=\int_{t_{n'-1}}^{t_{n'}}g(\hat X(u))\,\mathrm du$. Since $g$ is continuously differentiable on $[-R_X,R_X]$ and $\hat X$ is a Lipschitz function of $Y_0^N$ uniformly in time, so is $\hat\Psi(t_{n'})$. Hence $w_{\ell,n}$ is a locally Lipschitz function of $Y_0^N$, it is differentiable almost everywhere, and the chain rule holds almost everywhere.

Recall that
\[
w_{\ell,n}
=p(Z(t_n)=\ell\mid Y_0^N;\Theta^\star)
=\EE[\one\{Z(t_n)=\ell\}\mid Y_0^N;\Theta^\star].
\]
Fix $n'\in[N]$ and $\ell'\in[k]$.  Set
\[
A:=\{Z(t_n)=\ell\},\qquad B_{\ell'}:=\{Z(t_{n'})=\ell'\}.
\]
For the following calculation, all probabilities denoted by $P$ are conditional on $Y_0^N$ under the surrogate model at $\Theta^\star$, evaluated before the perturbation described below. In particular, $P(A)=w_{\ell,n}$.

Since $A$ specifies only the state at time $t_n$, we obtain its probability by summing over all state sequences with $z_n=\ell$:
\[
w_{\ell,n}
=\sum_{\substack{z_0,\ldots,z_N\in[k]\\z_n=\ell}}
p(z_0^N\mid Y_0^N;\Theta^\star).
\]
For an auxiliary scalar $x\in\RR$, replace the single log-emission $e_{\ell'}(n')$ by $e_{\ell'}(n')+x$, keeping all other log-emissions and the initial and transition probabilities fixed. The emission factor at time $t_{n'}$ becomes
\[
\exp\{e_{z_{n'}}(n')+x\one\{z_{n'}=\ell'\}\}
=\exp\{e_{z_{n'}}(n')\}\exp\{x\one\{z_{n'}=\ell'\}\}.
\]
Therefore, 
\[
w_{\ell,n}(x)
=\frac{
\displaystyle\sum_{\substack{z_0,\ldots,z_N\in[k]\\z_n=\ell}}
\exp\{x\one\{z_{n'}=\ell'\}\}\,
p(z_0^N\mid Y_0^N;\Theta^\star)
}{
\displaystyle\sum_{z_0,\ldots,z_N\in[k]}
\exp\{x\one\{z_{n'}=\ell'\}\}\,
p(z_0^N\mid Y_0^N;\Theta^\star)
}.
\]

Split the numerator into sequences with $z_{n'}=\ell'$ and those with $z_{n'}\neq\ell'$. The multiplier in the first group is $\exp(x)$ and in the second group is $1$. Therefore,
\[
\begin{aligned}
&\sum_{\substack{z_0,\ldots,z_N\in[k]\\z_n=\ell}}
\exp\{x\one\{z_{n'}=\ell'\}\}\,
p(z_0^N\mid Y_0^N;\Theta^\star)\\
&\quad=\exp(x)
\sum_{\substack{z_0,\ldots,z_N\in[k]\\z_n=\ell,\ z_{n'}=\ell'}}
p(z_0^N\mid Y_0^N;\Theta^\star)+
\sum_{\substack{z_0,\ldots,z_N\in[k]\\z_n=\ell,\ z_{n'}\neq\ell'}}
p(z_0^N\mid Y_0^N;\Theta^\star)\\
&\quad=\exp(x)P(A\cap B_{\ell'})+P(A\cap B_{\ell'}^c).
\end{aligned}
\]
The first sum on the right fixes both $Z(t_n)=\ell$ and $Z(t_{n'})=\ell'$, so it equals $P(A\cap B_{\ell'})$. The second fixes $Z(t_n)=\ell$ and excludes state $\ell'$ at time $t_{n'}$, so it equals $P(A\cap B_{\ell'}^c)$.

The denominator has no restriction on $z_n$. Splitting it according to the state at time $t_{n'}$ gives
\[
\begin{aligned}
&\sum_{z_0,\ldots,z_N\in[k]}
\exp\{x\one\{z_{n'}=\ell'\}\}\,
p(z_0^N\mid Y_0^N;\Theta^\star)\\
&\quad=\exp(x)
\sum_{\substack{z_0,\ldots,z_N\in[k]\\z_{n'}=\ell'}}
p(z_0^N\mid Y_0^N;\Theta^\star)
+\sum_{\substack{z_0,\ldots,z_N\in[k]\\z_{n'}\neq\ell'}}
p(z_0^N\mid Y_0^N;\Theta^\star)\\
&\quad=\exp(x)P(B_{\ell'})+P(B_{\ell'}^c).
\end{aligned}
\]
Substituting these two sums into the expression for $w_{\ell,n}(x)$ yields
\[
w_{\ell,n}(x)
=\frac{\exp(x)P(A\cap B_{\ell'})+P(A\cap B_{\ell'}^c)}
{\exp(x)P(B_{\ell'})+P(B_{\ell'}^c)}.
\]
At $x=0$, the numerator is $P(A)=w_{\ell,n}$ and the denominator is $1$. Their derivatives at $x=0$ are $P(A\cap B_{\ell'})$ and $P(B_{\ell'})$, respectively. Therefore, the quotient rule gives
\[
\begin{aligned}
\left.\frac{\mathrm d}{\mathrm dx}w_{\ell,n}(x)\right|_{x=0}
&=\frac{P(A\cap B_{\ell'})\cdot1-P(A)P(B_{\ell'})}{1^2}\\
&=P(A\cap B_{\ell'})-w_{\ell,n}P(B_{\ell'}).
\end{aligned}
\]
Since $x$ perturbs only $e_{\ell'}(n')$, this derivative is $\partial w_{\ell,n}/\partial e_{\ell'}(n')$. Restoring the conditioning explicitly, we obtain
\[
\begin{aligned}
C_{n n'}(\ell'):=\frac{\partial w_{\ell,n}}{\partial e_{\ell'}(n')}
&=P(Z(t_n)=\ell,Z(t_{n'})=\ell'\mid Y_0^N;\Theta^\star)\\
&\quad-w_{\ell,n}\,p(Z(t_{n'})=\ell'\mid Y_0^N;\Theta^\star).
\end{aligned}
\]
The second term comes from differentiating the normalizing denominator. 

To bound the sum over $\ell'$, suppose first that $0<w_{\ell,n}<1$. With the same conditional-probability convention, the decomposition
\[
P(B_{\ell'})=w_{\ell,n}P(B_{\ell'}\mid A)+(1-w_{\ell,n})P(B_{\ell'}\mid A^c)
\]
and $P(A\cap B_{\ell'})=w_{\ell,n}P(B_{\ell'}\mid A)$ give
\[
C_{n n'}(\ell')=w_{\ell,n}(1-w_{\ell,n})
\{P(B_{\ell'}\mid A)-P(B_{\ell'}\mid A^c)\}.
\]
The events $B_1,\ldots,B_k$ partition the state space at time $t_{n'}$, so both conditional distributions sum to one. Hence
\[
\begin{aligned}
\sum_{\ell'=1}^k|C_{n n'}(\ell')|
&=w_{\ell,n}(1-w_{\ell,n})
\sum_{\ell'=1}^k|P(B_{\ell'}\mid A)-P(B_{\ell'}\mid A^c)|\\
&\leq w_{\ell,n}(1-w_{\ell,n})
\sum_{\ell'=1}^k\{P(B_{\ell'}\mid A)+P(B_{\ell'}\mid A^c)\}\\
&=2w_{\ell,n}(1-w_{\ell,n}).
\end{aligned}
\]
When $w_{\ell,n}=0$, both $P(A\cap B_{\ell'})$ and $w_{\ell,n}P(B_{\ell'})$ are zero. When $w_{\ell,n}=1$, they both equal $P(B_{\ell'})$. Thus $C_{n n'}(\ell')=0$ at either endpoint, and the same bound holds.

The derivatives with respect to $Y_{n,i}$ and $Y_{n-1,i}$ have two parts. To distinguish them, we regard the residuals, log-emissions and smoothing probabilities as functions of the observed increments and reconstructed covariates, treated as separate arguments. For $n_0\in\{n-1,n\}$, the notation $\left.\partial/\partial Y_{n_0,i}\right|_{\hat\Psi\ \mathrm{fixed}}$ denotes the contribution obtained by differentiating the observed increments while holding the covariate argument fixed. The notation $\left.\partial/\partial Y_{n_0,i}\right|_{\mathrm{through}\ \hat\Psi}$ denotes the chain-rule contribution obtained by differentiating the reconstructed covariates with respect to $Y_{n_0,i}$ while holding the increment argument fixed. Therefore, 
\[
\frac{\partial w_{\ell,n}}{\partial Y_{n_0,i}}
=\left.\frac{\partial w_{\ell,n}}{\partial Y_{n_0,i}}\right|_{\hat\Psi\ \mathrm{fixed}}
+\left.\frac{\partial w_{\ell,n}}{\partial Y_{n_0,i}}\right|_{\mathrm{through}\ \hat\Psi}.
\]

For $n_0\in\{n-1,n\}$ and $n'\in[N]$, the definition of the residual gives
\[
\left.\frac{\partial\hat u_{\ell',n',j}}{\partial Y_{n_0,i}}\right|_{\hat\Psi\ \mathrm{fixed}}
=\one\{j=i\}\bigl(\one\{n'=n_0\}-\one\{n'=n_0+1\}\bigr).
\]
Indeed, $Y_{n_0,i}$ appears with coefficient $+1$ in $Y_{n_0,i}-Y_{n_0-1,i}$ and with coefficient $-1$ in $Y_{n_0+1,i}-Y_{n_0,i}$, whenever the corresponding increment is among the observations. Differentiating the squared residuals in $e_{\ell'}(n')$ therefore gives
\[
\begin{aligned}
\left.\frac{\partial e_{\ell'}(n')}{\partial Y_{n_0,i}}\right|_{\hat\Psi\ \mathrm{fixed}}
&=-\frac{1}{4\sigma^{\star2}}\sum_{j=1}^p
2\hat u_{\ell',n',j}\left.\frac{\partial\hat u_{\ell',n',j}}{\partial Y_{n_0,i}}\right|_{\hat\Psi\ \mathrm{fixed}}\\
&=-\frac{\hat u_{\ell',n',i}}{2\sigma^{\star2}}
\bigl(\one\{n'=n_0\}-\one\{n'=n_0+1\}\bigr).
\end{aligned}
\]
Substituting the emission derivative just calculated gives
\[
\begin{aligned}
\left.\frac{\partial w_{\ell,n}}{\partial Y_{n_0,i}}\right|_{\hat\Psi\ \mathrm{fixed}}
&=\sum_{n'=1}^N\sum_{\ell'=1}^k
\frac{\partial w_{\ell,n}}{\partial e_{\ell'}(n')}
\left.\frac{\partial e_{\ell'}(n')}{\partial Y_{n_0,i}}\right|_{\hat\Psi\ \mathrm{fixed}}\\
&=\sum_{n'=1}^N\sum_{\ell'=1}^k C_{n n'}(\ell')
\left.\frac{\partial e_{\ell'}(n')}{\partial Y_{n_0,i}}\right|_{\hat\Psi\ \mathrm{fixed}}\\
&=-\frac{1}{2\sigma^{\star2}}\sum_{n'=1}^N\sum_{\ell'=1}^k
C_{n n'}(\ell')\hat u_{\ell',n',i}
\bigl(\one\{n'=n_0\}-\one\{n'=n_0+1\}\bigr).
\end{aligned}
\]
The second line uses the definition of $C_{n n'}(\ell')$, the third replaces the emission derivative by its displayed formula.
Only $n'=n_0$ and $n'=n_0+1$ can contribute. Thus the two choices $n_0=n$ and $n_0=n-1$ give the following expressions, with terms indexed by $n+1$ omitted when $n=N$ and terms indexed by $n-1$ omitted when $n=1$:
\[
\begin{aligned}
\left.\frac{\partial w_{\ell,n}}{\partial Y_{n,i}}\right|_{\hat\Psi\ \mathrm{fixed}}
&=\frac{1}{2\sigma^{\star2}}\sum_{\ell'=1}^k
\bigl\{-C_{nn}(\ell')\hat u_{\ell',n,i}+C_{n,n+1}(\ell')\hat u_{\ell',n+1,i}\bigr\},\\
\left.\frac{\partial w_{\ell,n}}{\partial Y_{n-1,i}}\right|_{\hat\Psi\ \mathrm{fixed}}
&=\frac{1}{2\sigma^{\star2}}\sum_{\ell'=1}^k
\bigl\{-C_{n,n-1}(\ell')\hat u_{\ell',n-1,i}+C_{nn}(\ell')\hat u_{\ell',n,i}\bigr\}.
\end{aligned}
\]
Subtracting the second line from the first gives
\[
\begin{aligned}
&\left.\frac{\partial w_{\ell,n}}{\partial Y_{n,i}}\right|_{\hat\Psi\ \mathrm{fixed}}
-\left.\frac{\partial w_{\ell,n}}{\partial Y_{n-1,i}}\right|_{\hat\Psi\ \mathrm{fixed}}\\
&\qquad=\frac{1}{2\sigma^{\star2}}\sum_{\ell'=1}^k
\cbr{C_{n,n+1}(\ell')\hat u_{\ell',n+1,i}-2C_{nn}(\ell')\hat u_{\ell',n,i}+C_{n,n-1}(\ell')\hat u_{\ell',n-1,i}}.
\end{aligned}
\]
In particular, the coefficient $-2$ comes from subtracting the positive $C_{nn}(\ell')\hat u_{\ell',n,i}$ term in the second line from the negative one in the first.

Finally, apply the triangle inequality, use $|\hat u_{\ell',n',i}|\leq\bar U_{n,i}$ for the contributing times, and then use $\sum_{\ell'=1}^k|C_{n n'}(\ell')|\leq2w_{\ell,n}(1-w_{\ell,n})$. This gives
\[
\begin{aligned}
&\left|\left.\frac{\partial w_{\ell,n}}{\partial Y_{n,i}}\right|_{\hat\Psi\ \mathrm{fixed}}
-\left.\frac{\partial w_{\ell,n}}{\partial Y_{n-1,i}}\right|_{\hat\Psi\ \mathrm{fixed}}\right|\\
&\quad\leq\frac{\bar U_{n,i}}{2\sigma^{\star2}}
\sum_{n'\in\{n-1,n,n+1\}\cap[N]}
\bigl(1+\one\{n'=n\}\bigr)\sum_{\ell'=1}^k|C_{n n'}(\ell')|\\
&\quad\leq\frac{\bar U_{n,i}}{2\sigma^{\star2}}
(1+2+1)\,2w_{\ell,n}(1-w_{\ell,n})\\
&\quad=\frac{4}{\sigma^{\star2}}w_{\ell,n}(1-w_{\ell,n})\bar U_{n,i}.
\end{aligned}
\]
At $n=1$ or $n=N$, omitting the unavailable neighboring term only decreases the sum used in this bound.

We now calculate the contribution through the reconstructed covariates. Fix $n_0\in\{n-1,n\}$. Since each $\hat X_j$ depends only on the observations of coordinate $j$, among the covariate blocks only $\hat\Psi_i$ depends on $Y_{n_0,i}$. Changing this block affects the residual of every row $j$ for which $\theta_{ji}^{\ell'\star}\neq0$.

The term involving $\hat\Psi_i(t_{n'})$ in $\hat u_{\ell',n',j}$ is $-\hat\Psi_i(t_{n'})^\top\theta_{ji}^{\ell'\star}$. Thus, with gradients written as column vectors,
\[
\frac{\partial\hat u_{\ell',n',j}}{\partial\hat\Psi_i(t_{n'})}
=-\theta_{ji}^{\ell'\star},
\]
and differentiating the squared residuals gives
\[
\begin{aligned}
\frac{\partial e_{\ell'}(n')}{\partial\hat\Psi_i(t_{n'})}
&=-\frac{1}{4\sigma^{\star2}}\sum_{j=1}^p
2\hat u_{\ell',n',j}\,(-\theta_{ji}^{\ell'\star})\\
&=\frac{1}{2\sigma^{\star2}}\sum_{j=1}^p
\hat u_{\ell',n',j}\,\theta_{ji}^{\ell'\star}.
\end{aligned}
\]
The chain rule therefore gives the contribution through $\hat\Psi$ explicitly as
\[
\begin{aligned}
\left.\frac{\partial w_{\ell,n}}{\partial Y_{n_0,i}}\right|_{\mathrm{through}\ \hat\Psi}
&=\sum_{n'=1}^N\sum_{\ell'=1}^k C_{n n'}(\ell')
\left(\frac{\partial e_{\ell'}(n')}{\partial\hat\Psi_i(t_{n'})}\right)^\top
\frac{\partial\hat\Psi_i(t_{n'})}{\partial Y_{n_0,i}}\\
&=\frac{1}{2\sigma^{\star2}}\sum_{n'=1}^N\sum_{\ell'=1}^k C_{n n'}(\ell')
\left(\sum_{j=1}^p\hat u_{\ell',n',j}\theta_{ji}^{\ell'\star}\right)^\top
\frac{\partial\hat\Psi_i(t_{n'})}{\partial Y_{n_0,i}}.
\end{aligned}
\]

We bound the two vectors in each inner product separately. By Assumption~\ref{assumption:local},
\[
\frac{\partial\hat\Psi_i(t_{n'})}{\partial Y_{n_0,i}}=0
\qquad\text{if }|n'-n_0|>r_{\mathrm{rec}}.
\]
For every remaining time, $|n'-n|\leq|n'-n_0|+|n_0-n|\leq r_{\mathrm{rec}}+1$. Consequently, each residual with $\theta_{ji}^{\ell'\star}\neq0$ in the preceding inner product is bounded by $\bar U_{n,i}^{\mathrm{ch}}$. The triangle inequality and the definition of $C_{\mathrm{col}}$ give
\[
\begin{aligned}
\left\|\sum_{j=1}^p\hat u_{\ell',n',j}\theta_{ji}^{\ell'\star}\right\|_1
&\leq\sum_{j:\,\theta_{ji}^{\ell'\star}\neq0}
|\hat u_{\ell',n',j}|\,\|\theta_{ji}^{\ell'\star}\|_1\\
&\leq\bar U_{n,i}^{\mathrm{ch}}\sum_{j=1}^p\|\theta_{ji}^{\ell'\star}\|_1
\leq C_{\mathrm{col}}\bar U_{n,i}^{\mathrm{ch}}.
\end{aligned}
\]

For the other vector, the $b$th component of $\hat\Psi_i(t_{n'})$, for $b\in[m]$, is
\[
[\hat\Psi_i(t_{n'})]_b
=\int_{t_{n'-1}}^{t_{n'}}g_b(\hat X_i(u))\,\mathrm du.
\]
For each $u$, the Lipschitz condition gives differentiability of $\hat X_i(u)$ in the observations almost everywhere. By Fubini's theorem, for almost every observation vector the derivative exists for almost every $u$. The uniform Lipschitz condition and the bound $|g_b'|\leq D$ also bound the difference quotients of the integrand by $DC_{\mathrm{Lip}}$. Thus dominated convergence justifies differentiation under the integral for almost every observation vector, giving
\[
\frac{\partial[\hat\Psi_i(t_{n'})]_b}{\partial Y_{n_0,i}}
=\int_{t_{n'-1}}^{t_{n'}}
g_b'(\hat X_i(u))\frac{\partial\hat X_i(u)}{\partial Y_{n_0,i}}\,\mathrm du.
\]
Taking absolute values and then the maximum over $b$ yields
\[
\left\|\frac{\partial\hat\Psi_i(t_{n'})}{\partial Y_{n_0,i}}\right\|_\infty
\leq D\int_{t_{n'-1}}^{t_{n'}}
\left|\frac{\partial\hat X_i(u)}{\partial Y_{n_0,i}}\right|\,\mathrm du.
\]
The sampling intervals partition $[0,T]$, so summing over $n'$ and applying the assumed integrated sensitivity bound gives
\[
\begin{aligned}
\sum_{n'=1}^N\left\|\frac{\partial\hat\Psi_i(t_{n'})}{\partial Y_{n_0,i}}\right\|_\infty
&\leq D\sum_{n'=1}^N\int_{t_{n'-1}}^{t_{n'}}
\left|\frac{\partial\hat X_i(u)}{\partial Y_{n_0,i}}\right|\,\mathrm du\\
&=D\int_0^T\left|\frac{\partial\hat X_i(u)}{\partial Y_{n_0,i}}\right|\,\mathrm du
\leq DhC_{\hat X}.
\end{aligned}
\]

Returning to the chain-rule expression, apply $|v^\top z|\leq\|v\|_1\|z\|_\infty$ to each inner product. Terms with $|n'-n_0|>r_{\mathrm{rec}}$ vanish; for the remaining terms, the column norm bound above gives
\[
\begin{aligned}
\left|\left.\frac{\partial w_{\ell,n}}{\partial Y_{n_0,i}}\right|_{\mathrm{through}\ \hat\Psi}\right|
&\leq\frac{1}{2\sigma^{\star2}}\sum_{n'=1}^N\sum_{\ell'=1}^k|C_{n n'}(\ell')|
\left\|\sum_{j=1}^p\hat u_{\ell',n',j}\theta_{ji}^{\ell'\star}\right\|_1
\left\|\frac{\partial\hat\Psi_i(t_{n'})}{\partial Y_{n_0,i}}\right\|_\infty\\
&\leq\frac{C_{\mathrm{col}}\bar U_{n,i}^{\mathrm{ch}}}{2\sigma^{\star2}}
\sum_{n'=1}^N\left(\sum_{\ell'=1}^k|C_{n n'}(\ell')|\right)
\left\|\frac{\partial\hat\Psi_i(t_{n'})}{\partial Y_{n_0,i}}\right\|_\infty\\
&\leq\frac{C_{\mathrm{col}}\bar U_{n,i}^{\mathrm{ch}}}{\sigma^{\star2}}
w_{\ell,n}(1-w_{\ell,n})
\sum_{n'=1}^N\left\|\frac{\partial\hat\Psi_i(t_{n'})}{\partial Y_{n_0,i}}\right\|_\infty\\
&\leq\frac{DhC_{\hat X}C_{\mathrm{col}}}{\sigma^{\star2}}
w_{\ell,n}(1-w_{\ell,n})\bar U_{n,i}^{\mathrm{ch}},
\end{aligned}
\]
where the third line uses $\sum_{\ell'=1}^k|C_{n n'}(\ell')|\leq2w_{\ell,n}(1-w_{\ell,n})$.

The bound on this second contribution holds for both $n_0=n$ and $n_0=n-1$. The triangle inequality therefore bounds their difference by twice that bound. Adding this to the increment bound already proved yields
\[
\begin{aligned}
\left|\frac{\partial w_{\ell,n}}{\partial Y_{n,i}}
-\frac{\partial w_{\ell,n}}{\partial Y_{n-1,i}}\right|
&\leq\frac{4}{\sigma^{\star2}}w_{\ell,n}(1-w_{\ell,n})\bar U_{n,i}\\
&\quad+\frac{2DhC_{\hat X}C_{\mathrm{col}}}{\sigma^{\star2}}
w_{\ell,n}(1-w_{\ell,n})\bar U_{n,i}^{\mathrm{ch}}\\
&=\frac{4}{\sigma^{\star2}}w_{\ell,n}(1-w_{\ell,n})
\cbr{\bar U_{n,i}+\frac{DhC_{\hat X}C_{\mathrm{col}}}{2}\bar U_{n,i}^{\mathrm{ch}}},
\end{aligned}
\]
which is the claimed inequality.

\end{proof}

\begin{lemma}[Bias of the update of $\theta_i^\ell$]
\label{lem:bias-theta}
Fix $i\in[p]$ and $\ell\in[k]$. Suppose that
$\max_{n\in[N],\,j\in[p]}|\rho_{n,j}|\leq\rho$,
$\|\Psi(t_n)\|_\infty\leq Bh$ and $\|\hat\Psi(t_n)\|_\infty\leq Bh$ for every $n\in[N]$ almost surely, that
$\theta_j^{a\star}$ is $ms$-sparse for every $j\in[p]$ and $a\in[k]$, and that the conditions of Lemma~\ref{lem:posterior-sensitivity} hold. 
Recall the definitions of 
$w_{\ell,n}$, $\alpha_{\ell\ell'}$ and $\bar\rho_i$ defined in Lemma~\ref{lem:bias-sigma}. Let $K_\Psi=N^{-1}\sum_{n=1}^N\EE[\Psi(t_n)\Psi(t_n)^\top]$, and define
\[
\chi_\ell:=\frac{1}{N}\sum_{n=1}^N\EE\sbr{w_{\ell,n}(1-w_{\ell,n})},
\]
and let $\bar U^{\mathrm{ch}}_{n,i}$ be as in Lemma~\ref{lem:posterior-sensitivity}. If $K_\Psi^{w,\ell}(\Theta^\star)\succ0$, then
\[
\begin{aligned}
\left\|
M_{\theta_i^\ell}(\Theta^\star)-\theta_i^{\ell\star}
\right\|_2
\leq{}&
\frac{1}{\sqrt{\sigma_{\min}\rbr{K_\Psi^{w,\ell}(\Theta^\star)}}}
\Bigg[
\bar\rho_i+
Bh\sqrt{2ms}\,
\Bigl\{
\sum_{\ell'\neq\ell}\alpha_{\ell\ell'}
\norm{\theta_i^{\ell'\star}-\theta_i^{\ell\star}}_2^2
\Bigr\}^{1/2}
\\
&+\min\Bigg\{\sqrt2\sigma^\star,\;
4\sqrt{\frac{\sigma_{\max}(K_\Psi)}{\sigma_{\min}(K_\Psi^{w,\ell}(\Theta^\star))}}
\Bigl(\sigma^\star\chi_\ell^{1/4}\\
&\qquad+\frac{\chi_\ell^{1/2}}{2}\Bigl\{\rho+Bh\sqrt{2ms}\max_{a,b\in[k]}\norm{\theta_i^{a\star}-\theta_i^{b\star}}_2+2Bh\sqrt{ms}\max_{a\in[k]}\norm{\theta_i^{a\star}}_2\Bigr\}\\
&\qquad+\frac{Dh\,C_{\hat X}C_{\mathrm{col}}}{4\sqrt2}\,\chi_\ell^{1/4}\max_{n\in[N]}\bigl\{\EE[(\bar U^{\mathrm{ch}}_{n,i})^4]\bigr\}^{1/4}\Bigr)\Bigg\}
\Bigg].
\end{aligned}
\]
\end{lemma}
Note that $\chi_\ell\leq1/4$, and that $\chi_\ell=0$ when the smoothing probabilities at $\Theta^\star$ take values in $\{0,1\}$. If, in addition, $\alpha_{\ell\ell'}=0$ for every $\ell'\neq\ell$, the bound reduces to $\bar\rho_i/\sqrt{\sigma_{\min}(K_\Psi^{w,\ell}(\Theta^\star))}$. Since $\rho_{n,i}=0$ unless $Z$ jumps in $(t_{n-1},t_n]$, an event of probability at most $h\max_{a\in[k]}|q^\star_{aa}|$, we have $\bar\rho_i\leq\rho\,(h\max_{a\in[k]}|q^\star_{aa}|)^{1/2}$.

The last term in the minimum comes from the dependence of the posterior on the observations through the estimated trajectory, and it also vanishes when $\chi_\ell=0$. Its moment is finite. Let $d_i$ be the number of rows $j$ with $\theta_{ji}^{a\star}\neq0$ for some $a\in[k]$, and let $G^{\mathrm{ch}}_{n,i}$ be the maximum of $|\varepsilon_{n',j}-\varepsilon_{n'-1,j}|$ over these rows and over $n'\in[N]$ with $|n'-n|\leq r_{\mathrm{rec}}+1$. Since each $\varepsilon_{n',j}-\varepsilon_{n'-1,j}$ has fourth moment $12$, $\EE[(G^{\mathrm{ch}}_{n,i})^4]\leq12(2r_{\mathrm{rec}}+3)d_i$, and the decomposition of $\hat u_{a,n',j}$ in the proof below and Minkowski's inequality give
\[
\begin{aligned}
\bigl\{\EE[(\bar U^{\mathrm{ch}}_{n,i})^4]\bigr\}^{1/4}
\leq{}&\sigma^\star\{12(2r_{\mathrm{rec}}+3)d_i\}^{1/4}+\rho\\
&+\max_j\Bigl\{Bh\sqrt{2ms}\max_{a,b\in[k]}\norm{\theta_j^{a\star}-\theta_j^{b\star}}_2+2Bh\sqrt{ms}\max_{a\in[k]}\norm{\theta_j^{a\star}}_2\Bigr\},
\end{aligned}
\]
where the last maximum is over the same rows. Hence this term does not grow with $p$ when every node enters the drifts of a bounded number of nodes with bounded parameters, which also bounds $C_{\mathrm{col}}$.

\begin{proof}
Fix $i\in[p]$ and $\ell\in[k]$, and let
\[
\theta_i^{\ell\dagger}:=M_{\theta_i^\ell}(\Theta^\star),
\qquad
v:=\theta_i^{\ell\dagger}-\theta_i^{\ell\star},
\qquad
x:=\cbr{v^\top K_\Psi^{w,\ell}(\Theta^\star)v}^{1/2}.
\]
Since $x\geq\sigma_{\min}(K_\Psi^{w,\ell}(\Theta^\star))^{1/2}\norm{v}_2$, it suffices to bound $x$ by the expression in square brackets. For this, we split $x^2$ into a noise term, a switching remainder term and a mismatch term. We bound the last two by the Cauchy--Schwarz inequality. For the noise term, we use either the Cauchy--Schwarz inequality or Gaussian integration by parts, which exploits the fact that $w_{\ell,n}$ depends on the noise only through the posterior.

By the first-order condition defining $\theta_i^{\ell\dagger}$,
\[
\frac{1}{N}\sum_{n=1}^N
\EE\sbr{
w_{\ell,n}
\Psi(t_n)
\rbr{
Y_{n,i}-Y_{n-1,i}
-\Psi(t_n)^\top\theta_i^{\ell\dagger}
}
}
=0.
\]
Therefore,
\[
K_\Psi^{w,\ell}(\Theta^\star)v
=
\frac{1}{N}\sum_{n=1}^N
\EE\sbr{
w_{\ell,n}
\Psi(t_n)
\rbr{
Y_{n,i}-Y_{n-1,i}
-\Psi(t_n)^\top\theta_i^{\ell\star}
}
}.
\]
By~\eqref{eq:approxY} with $r_{n,i}=0$ under Assumption~\ref{assumption:edgeset}, the residual at $\theta_i^{\ell\star}$ is
\[
Y_{n,i}-Y_{n-1,i}
-\Psi(t_n)^\top\theta_i^{\ell\star}
=
\sigma^\star\rbr{\varepsilon_{n,i}-\varepsilon_{n-1,i}}
+\rho_{n,i}+
\Psi(t_n)^\top
\rbr{\theta_i^{Z(t_n)\star}-\theta_i^{\ell\star}}.
\]
Taking the inner product with $v$, we obtain $x^2=T^\varepsilon+T^\rho+T^\Delta$, where
\[
\begin{aligned}
T^\varepsilon&:=\frac{\sigma^\star}{N}\sum_{n=1}^N\EE\sbr{w_{\ell,n}\,\Psi(t_n)^\top v\,\rbr{\varepsilon_{n,i}-\varepsilon_{n-1,i}}},
\qquad
T^\rho:=\frac{1}{N}\sum_{n=1}^N\EE\sbr{w_{\ell,n}\,\Psi(t_n)^\top v\,\rho_{n,i}},\\
T^\Delta&:=\frac{1}{N}\sum_{n=1}^N\EE\sbr{w_{\ell,n}\,\Psi(t_n)^\top v\,\Psi(t_n)^\top\rbr{\theta_i^{Z(t_n)\star}-\theta_i^{\ell\star}}}.
\end{aligned}
\]

Since $N^{-1}\sum_{n}\EE[w_{\ell,n}(\Psi(t_n)^\top v)^2]=x^2$, the Cauchy--Schwarz inequality with weights $w_{\ell,n}$ bounds each of the three terms by $x$ times the square root of the corresponding weighted mean square. Since $0\leq w_{\ell,n}\leq1$, $\EE(\varepsilon_{n,i}-\varepsilon_{n-1,i})^2=2$ and $N^{-1}\sum_n\EE[\rho_{n,i}^2]=\bar\rho_i^2$, this gives
\[
|T^\varepsilon|\leq\sqrt2\sigma^\star x,
\qquad
|T^\rho|\leq\bar\rho_i x.
\]
For $T^\Delta$, recall that $\|\Psi(t_n)\|_\infty\leq Bh$. Furthermore, since each $\theta_i^{a\star}$ is $ms$-sparse, $\theta_i^{a\star}-\theta_i^{b\star}$ is at most $2ms$-sparse. Therefore,
\begin{equation}\label{eq:bias-theta-sparsity}
\left|
\Psi(t_n)^\top
\rbr{\theta_i^{a\star}-\theta_i^{b\star}}
\right|
\leq
\|\Psi(t_n)\|_\infty
\left\|
\theta_i^{a\star}-\theta_i^{b\star}
\right\|_1
\leq
Bh\sqrt{2ms}\,
\left\|
\theta_i^{a\star}-\theta_i^{b\star}
\right\|_2,
\qquad a,b\in[k].
\end{equation}
Splitting according to the value of $Z(t_n)$, the terms with $Z(t_n)=\ell$ vanish, and we obtain
\[
\begin{aligned}
|T^\Delta|
&\leq
x\left\{
\frac{1}{N}\sum_{n=1}^N
\EE\sbr{
w_{\ell,n}
\rbr{
\Psi(t_n)^\top
\rbr{\theta_i^{Z(t_n)\star}-\theta_i^{\ell\star}}
}^2
}
\right\}^{1/2}\\
&\leq
Bh\sqrt{2ms}\,x
\left\{
\sum_{\ell'\neq\ell}\alpha_{\ell\ell'}
\norm{\theta_i^{\ell'\star}-\theta_i^{\ell\star}}_2^2
\right\}^{1/2}.
\end{aligned}
\]

It remains to bound $|T^\varepsilon|$ by $x$ times the second argument of the minimum. Let $\mathcal F_{XZ}$ be the $\sigma$-field generated by $(X(t),Z(t))_{t\in[t_0,t_N]}$. The variables $\Psi(t_n)$ and $\rho_{n,i}$ are $\mathcal F_{XZ}$-measurable, $v$ is deterministic, and, by~\eqref{eq:definey}, $\varepsilon_0,\ldots,\varepsilon_N$ are i.i.d.\ $\Ncal(0,I_p)$ and independent of $\mathcal F_{XZ}$. Given $\mathcal F_{XZ}$, we have $Y_{n'}=X(t_{n'})+\sigma^\star\varepsilon_{n'}$, so that, by Lemma~\ref{lem:posterior-sensitivity}, $w_{\ell,n}$ is a locally Lipschitz function of the noise, with $\partial w_{\ell,n}/\partial\varepsilon_{n',i}=\sigma^\star\partial w_{\ell,n}/\partial Y_{n',i}$ almost everywhere. The proof of that lemma also shows that $|\partial w_{\ell,n}/\partial Y_{n',i}|\leq\{\bar U_{n,i}+Dh\,C_{\hat X}C_{\mathrm{col}}\bar U^{\mathrm{ch}}_{n,i}/2\}/\sigma^{\star2}$ for $n'\in\{n-1,n\}$, and both maxima have finite moments. Therefore, Gaussian integration by parts~\citep{stein1981estimation}, which holds for absolutely continuous functions with integrable derivatives, gives
\[
\EE\sbr{w_{\ell,n}\varepsilon_{n',i}\mid\mathcal F_{XZ}}
=\sigma^\star\,\EE\sbr{\frac{\partial w_{\ell,n}}{\partial Y_{n',i}}\,\Big|\,\mathcal F_{XZ}},
\qquad n'\in\{n-1,n\}.
\]
Write $\mathsf V_{n,i}:=\bar U_{n,i}+Dh\,C_{\hat X}C_{\mathrm{col}}\bar U^{\mathrm{ch}}_{n,i}/2$. Conditioning on $\mathcal F_{XZ}$ in $T^\varepsilon$, then applying Lemma~\ref{lem:posterior-sensitivity} and the Cauchy--Schwarz inequality, we obtain
\[
\begin{aligned}
|T^\varepsilon|
&=\abr{\frac{\sigma^{\star2}}{N}\sum_{n=1}^N\EE\sbr{\Psi(t_n)^\top v\rbr{\frac{\partial w_{\ell,n}}{\partial Y_{n,i}}-\frac{\partial w_{\ell,n}}{\partial Y_{n-1,i}}}}}
\leq\frac{4}{N}\sum_{n=1}^N\EE\sbr{\abr{\Psi(t_n)^\top v}\,w_{\ell,n}(1-w_{\ell,n})\,\mathsf V_{n,i}}\\
&\leq4\cbr{v^\top K_\Psi v}^{1/2}
\left\{\frac{1}{N}\sum_{n=1}^N\EE\sbr{w_{\ell,n}^2(1-w_{\ell,n})^2\,\mathsf V_{n,i}^2}\right\}^{1/2}.
\end{aligned}
\]
The first factor satisfies $\{v^\top K_\Psi v\}^{1/2}\leq\sigma_{\max}(K_\Psi)^{1/2}\norm{v}_2\leq\{\sigma_{\max}(K_\Psi)/\sigma_{\min}(K_\Psi^{w,\ell}(\Theta^\star))\}^{1/2}x$.
By Minkowski's inequality, the second factor is at most the sum of the same expression with $\bar U_{n,i}$ and with $Dh\,C_{\hat X}C_{\mathrm{col}}\bar U^{\mathrm{ch}}_{n,i}/2$ in place of $\mathsf V_{n,i}$. For the first, \eqref{eq:approxY} with $r_{n',i}=0$ gives
\[
\hat u_{a,n',i}=\sigma^\star(\varepsilon_{n',i}-\varepsilon_{n'-1,i})+\rho_{n',i}+\Psi(t_{n'})^\top(\theta_i^{Z(t_{n'})\star}-\theta_i^{a\star})-\{\hat\Psi(t_{n'})-\Psi(t_{n'})\}^\top\theta_i^{a\star}.
\]
Since $\|\hat\Psi(t_{n'})\|_\infty\leq Bh$ and $\|\Psi(t_{n'})\|_\infty\leq Bh$, we have $|\{\hat\Psi(t_{n'})-\Psi(t_{n'})\}^\top\theta_i^{a\star}|\leq2Bh\|\theta_i^{a\star}\|_1\leq2Bh\sqrt{ms}\|\theta_i^{a\star}\|_2$. This bound holds almost surely, whereas the reconstruction bound $\delta_1$ of Section~\ref{ssec:main:results} holds only on the event $\mathcal E_{\mathrm{rec}}$, so we do not use $\delta_1$ in this expectation. By~\eqref{eq:bias-theta-sparsity},
\[
\begin{aligned}
\bar U_{n,i}\leq{}&\sigma^\star G_{n,i}+\rho+Bh\sqrt{2ms}\max_{a,b\in[k]}\norm{\theta_i^{a\star}-\theta_i^{b\star}}_2\\
&+2Bh\sqrt{ms}\max_{a\in[k]}\norm{\theta_i^{a\star}}_2,
\end{aligned}
\]
where $G_{n,i}:=\max_{n'\in\{n-1,n,n+1\}\cap[N]}|\varepsilon_{n',i}-\varepsilon_{n'-1,i}|$.
Since $w_{\ell,n}(1-w_{\ell,n})\leq1/4$, we have $N^{-1}\sum_n\EE[w_{\ell,n}^2(1-w_{\ell,n})^2]\leq\chi_\ell/4$. Furthermore, each $\varepsilon_{n',i}-\varepsilon_{n'-1,i}$ is $\Ncal(0,2)$, with fourth moment $12$, so that $\EE[G_{n,i}^4]\leq36$. Using $w_{\ell,n}^4(1-w_{\ell,n})^4\leq w_{\ell,n}(1-w_{\ell,n})/64$, the Cauchy--Schwarz inequality and the concavity of the square root, we obtain
\[
\frac{1}{N}\sum_{n=1}^N\EE\sbr{w_{\ell,n}^2(1-w_{\ell,n})^2G_{n,i}^2}
\leq\frac{1}{N}\sum_{n=1}^N\frac{3}{4}\cbr{\EE\sbr{w_{\ell,n}(1-w_{\ell,n})}}^{1/2}
\leq\frac{3}{4}\chi_\ell^{1/2}.
\]
By Minkowski's inequality, the part with $\bar U_{n,i}$ is therefore at most $\sigma^\star\chi_\ell^{1/4}+\frac12\chi_\ell^{1/2}\{\rho+Bh\sqrt{2ms}\max_{a,b}\norm{\theta_i^{a\star}-\theta_i^{b\star}}_2+2Bh\sqrt{ms}\max_a\norm{\theta_i^{a\star}}_2\}$. For the part with $\bar U^{\mathrm{ch}}_{n,i}$, the same three steps give
\[
\left\{\frac{1}{N}\sum_{n=1}^N\EE\sbr{w_{\ell,n}^2(1-w_{\ell,n})^2(\bar U^{\mathrm{ch}}_{n,i})^2}\right\}^{1/2}
\leq\frac{\chi_\ell^{1/4}}{\sqrt8}\max_{n\in[N]}\bigl\{\EE[(\bar U^{\mathrm{ch}}_{n,i})^4]\bigr\}^{1/4},
\]
and multiplying by $Dh\,C_{\hat X}C_{\mathrm{col}}/2$ gives the last term of the statement. This gives the required bound on $|T^\varepsilon|$.

Combining the bounds on $T^\varepsilon$, $T^\rho$ and $T^\Delta$ with $x^2=T^\varepsilon+T^\rho+T^\Delta$, we find that $x^2$ is at most $x$ times the expression in square brackets in the statement. Dividing by $x$ when $x>0$ completes the proof.
\end{proof}

\subsection{Proof of Lemma~\ref{lemma:aggregate-surrogate-bias}}

The component bounds in Lemma~\ref{lemma:aggregate-surrogate-bias} have the following dependencies.
The bound $b_{\theta,i}^{\ell}$ depends on $\sigma_{\min}(K_\Psi^{w,\ell}(\Theta^\star))$ and $\sigma_{\max}(K_\Psi)$.
Both $b_{\theta,i}^{\ell}$ and $b_{\sigma^2}$ depend on the smoothing probabilities at $\Theta^\star$ through $\alpha_{\ell\ell'}$, and $b_{\theta,i}^{\ell}$ also through  $\chi_\ell$.
The bound $b_{q_{\ell\ell'}}$ depends on the expected joint smoothing probabilities, as shown in Lemma~\ref{lem:bias-q}.
Through its noise term, $b_{\theta,i}^{\ell}$ also depends on the estimated trajectory: the term from the estimated trajectory involves $C_{\hat X}$ and the residuals $\bar U^{\mathrm{ch}}_{n,i}$ of Lemma~\ref{lem:posterior-sensitivity}, multiplied by the bound $D$ on the derivatives of the basis functions and by the column norm $C_{\mathrm{col}}$ of the true parameters $\theta_{ij}^{\ell\star}$.
The switching remainder enters $b_{\theta,i}^{\ell}$ through $\rho$ and $\bar\rho_i$, and $b_{\sigma^2}$ through $\bar\rho^2=p^{-1}\sum_{i=1}^p\bar\rho_i^2$.
These bias bounds, including their aggregate $b_{\mathrm{surr}}$, are unrelated to the drift $b_\ell$ of Proposition~\ref{prop:general_example}.

\begin{proof}
By Definition~\ref{def:distance},
\[
\begin{aligned}
\dist\bigl(M(\Theta^\star),\Theta^\star\bigr)
={}&
\max_{\ell\in[k],\,i\in[p]}
\norm{M_{\theta_i^\ell}(\Theta^\star)-\theta_i^{\ell\star}}_2\\
&+
\sum_{\ell\neq\ell'}
\abr{M_{q_{\ell\ell'}}(\Theta^\star)-q_{\ell\ell'}^\star}\\
&+
\abr{M_{\sigma^2}(\Theta^\star)-\sigma^{\star2}}.
\end{aligned}
\]
The hypotheses of the lemma imply those of Lemmas~\ref{lem:bias-theta}, \ref{lem:bias-q}, and
\ref{lem:bias-sigma}. Applying these lemmas to the three terms, respectively, we obtain
\[
\dist\bigl(M(\Theta^\star),\Theta^\star\bigr)
\leq
\max_{\ell\in[k],\,i\in[p]} b_{\theta,i}^{\ell}
+
\sum_{\ell\neq\ell'}b_{q_{\ell\ell'}}
+
b_{\sigma^2}
=b_{\mathrm{surr}}(\Theta^\star).
\]
\end{proof}

\subsection{Dependence of the bias on \texorpdfstring{$p$}{p}}
\label{appendix:bias:p}

The block bounds $b^\ell_{\theta,i}$ of Lemma~\ref{lem:bias-theta} depend on $p$ only through $\sigma_{\min}(K_\Psi^{w,\ell}(\Theta^\star))$ and $\sigma_{\max}(K_\Psi)$, the extreme eigenvalues of $pm\times pm$ matrices, through $\alpha_{\ell\ell'}\leq1$ and $\chi_\ell\leq1/4$, and, in the term from the estimated trajectory, through $C_{\mathrm{col}}$ and the number $d_i$ of rows $j$ with $\theta_{ji}^{a\star}\neq0$ for some $a\in[k]$, as explained in the remark after Lemma~\ref{lem:bias-theta}. We bound the first two under the conditions of Theorem~\ref{theorem:graphrecovery:full}. Here and below, $r_+:=r+r_{\mathrm{rec}}$ denotes the radius of the observation window on which a truncated smoothing probability depends under Assumption~\ref{assumption:local}, as established in Lemma~\ref{lemma:local:blocks}. Let $\epsilon_r:=10k\delta_{\min}^{-1}\zeta^{-8}\pi_{\min}^{-2}\{1-(\zeta\pi_{\min})^2\}^{r-1}$ be the bound in the second claim of Lemma~\ref{lemma:truncated_smoooth}, so that $|w_{\ell,\Theta^\star}(t_n)-\hat w_{\ell,\Theta^\star}(t_n)|\leq\epsilon_r$ for every $n$, where $w_{\ell,\Theta^\star}(t_n)$ and $\hat w_{\ell,\Theta^\star}(t_n)$ are the full and truncated smoothing probabilities of Section~\ref{ssec:approximateEM}, both computed with $\hat\Psi$, and $w_{\ell,\Theta^\star}(t_n)=w_{\ell,n}$. Then
\begin{align*}
K_\Psi^{w,\ell}(\Theta^\star)
&\succeq\frac1N\sum_{n=r_++1}^{N-r_+}\EE\sbr{w_{\ell,\Theta^\star}(t_n)\Psi(t_n)\Psi(t_n)^\top}\\
&\succeq\frac1N\sum_{n=r_++1}^{N-r_+}\EE\sbr{\hat w_{\ell,\Theta^\star}(t_n)\Psi(t_n)\Psi(t_n)^\top}-\epsilon_rK_\Psi\\
&\succeq\frac{N-2r_+}{N}\,c_{\mathrm{RE}}I-\epsilon_rK_\Psi.
\end{align*}
The first step drops $2r_+$ positive semidefinite boundary terms, and the second uses $|w-\hat w|\leq\epsilon_r$ and then adds the boundary terms back to $K_\Psi$. The third step is the condition that defines $c_{\mathrm{RE}}$, applied at $\Theta=\Theta^\star\in B(r_0,\Theta^\star)$. Therefore,
\[
\sigma_{\min}\bigl(K_\Psi^{w,\ell}(\Theta^\star)\bigr)\geq\frac{N-2r_+}{N}\,c_{\mathrm{RE}}-\epsilon_r\,\sigma_{\max}(K_\Psi).
\]
The sample-size condition of Theorem~\ref{theorem:graphrecovery:full} gives $N\geq4r_+$, and the theorem assumes $c_{\mathrm{RE}}^{-1}\sigma_{\max}(K_\Psi)\leq c_\Psi^{-1}$. The right-hand side is therefore at least $c_{\mathrm{RE}}/4$ once $\epsilon_r\leq c_\Psi/4$, a condition that does not involve $p$. Furthermore, $\sigma_{\max}(K_\Psi)\leq c_{\mathrm{RE}}/c_\Psi$, so that $\sigma_{\max}(K_\Psi)/\sigma_{\min}(K_\Psi^{w,\ell}(\Theta^\star))\leq4/c_\Psi$. Lemma~\ref{lem:bias-theta} then gives
\[
\begin{aligned}
b^\ell_{\theta,i}\leq\frac{2}{\sqrt{c_{\mathrm{RE}}}}\Bigl[&\bar\rho_i+Bh\sqrt{2ms}\Bigl(\sum_{\ell'\neq\ell}\alpha_{\ell\ell'}\|\theta_i^{\ell'\star}-\theta_i^{\ell\star}\|_2^2\Bigr)^{1/2}\\
&+\min\Bigl\{\sqrt2\sigma^\star,\,\frac{8}{\sqrt{c_\Psi}}\Bigl(\sigma^\star\chi_\ell^{1/4}\\
&\qquad+\frac{\chi_\ell^{1/2}}{2}\Bigl\{\rho+Bh\sqrt{2ms}\max_{a,b\in[k]}\|\theta_i^{a\star}-\theta_i^{b\star}\|_2+2Bh\sqrt{ms}\max_{a\in[k]}\|\theta_i^{a\star}\|_2\Bigr\}\\
&\qquad+\frac{Dh\,C_{\hat X}C_{\mathrm{col}}}{4\sqrt2}\,\chi_\ell^{1/4}\max_{n\in[N]}\bigl\{\EE[(\bar U^{\mathrm{ch}}_{n,i})^4]\bigr\}^{1/4}\Bigr)\Bigr\}\Bigr].
\end{aligned}
\]
The other two terms of $b_{\mathrm{surr}}(\Theta^\star)$ do not grow with $p$ either: $b_{\sigma^2}$ averages over the rows, and the end of the proof of Lemma~\ref{lem:bias-q} bounds $b_{q_{\ell\ell'}}$ by a quantity that depends only on $Q^\star$ and $h$. Therefore, since $\alpha_{\ell\ell'}\leq1$ and $\chi_\ell\leq1/4$, $b_{\mathrm{surr}}(\Theta^\star)$ has no factor that grows with $p$ when $k$, $m$, $s$ and $h$ are fixed, the norms $\max_{i\in[p],\,a\in[k]}\|\theta_i^{a\star}\|_2$ and the constants $B$, $D$, $\rho$, $\sigma^\star$ and $C_{\hat X}$ are bounded, every node enters the drifts of a bounded number of nodes with bounded parameters, so that $C_{\mathrm{col}}$ and $\max_{i\in[p]}d_i$ are bounded, and $c_{\mathrm{RE}}$ and $c_\Psi$ are bounded below. A bounded number of children per node bounds neither the sparsity $s$ of a row nor its norm, so these conditions are separate.

The condition on $c_{\mathrm{RE}}$ and $c_\Psi$ restricts the features. Restricting the matrix in the definition of $c_{\mathrm{RE}}$ to the block of a node $j$ and using $\hat w_{\ell,\Theta}(t_n)\leq1$, we obtain $c_{\mathrm{RE}}\leq N\sigma_{\min}(K_{\Psi_j})/(N-2r_+)\leq2/c_\Psi$, so that the assumption $c_{\mathrm{RE}}^{-1}\sigma_{\max}(K_\Psi)\leq c_\Psi^{-1}$ gives $\sigma_{\max}(K_\Psi)\leq2/c_\Psi^2$. On the other hand, with $\bar\mu:=N^{-1}\sum_{n=1}^N\EE[\Psi(t_n)]$, we have $K_\Psi\succeq\bar\mu\bar\mu^\top$, and hence $\sigma_{\max}(K_\Psi)\geq\|\bar\mu\|_2^2$. For example, if $g_2(x)=x^2$, as in Section~\ref{sec:simulation}, then $0\leq\Psi_{j,2}(t_n)\leq hR_X^2$, and the second diagonal entry of $K_{\Psi_j}$ is at least $c_\Psi$, so that $\bar\mu_{j,2}\geq c_\Psi/(hR_X^2)$ for every $j\in[p]$. The conditions of Theorem~\ref{theorem:graphrecovery:full} then require $p\leq2h^2R_X^4/c_\Psi^4$, so that $c_\Psi$ must decrease as $p$ grows for such a basis.

\section{Truncated filtering and smoothing probabilities}

\begin{figure*}
    \centering
    \begin{subfigure}[t]{0.5\textwidth}
        \centering
        \resizebox{!}{9em}{%
        \begin{tikzpicture}[->, >=stealth', auto, semithick, node distance=2.5cm]

  \tikzstyle{state}=[circle, draw, fill=gray!30, minimum size=1.5cm]
  \tikzstyle{emission}=[circle, draw, minimum size=1.5cm]

  \node[state] (z1)                         {$Z(t_1)$};
  \node[state] (z2) [right of=z1]           {$Z(t_2)$};
  \node[state] (z3) [right of=z2]           {$Z(t_3)$};
  \node[state] (z4) [right of=z3]           {$Z(t_4)$};

  \node [emission] (y1) [below=2cm of z1] {$Y_1$};
  \node [emission](y2) [below=2cm of z2] {$Y_2$};
  \node [emission](y3) [below=2cm of z3] {$Y_3$};
  \node [emission](y4) [below=2cm of z4] {$Y_4$};

  \path (z1) edge [->] (z2)
        (z2) edge [->] (z3)
        (z3) edge [->] (z4);

  \path (z1) edge [->]  (y1)
        (z2) edge [->]  (y2)
        (z3) edge [->]  (y3)
        (z4) edge [->]  (y4);

\end{tikzpicture}}
        \caption{Hidden Markov Model discussed in~\citet{yang2017statistical}}
    \end{subfigure}%
    \hfill
    \begin{subfigure}[t]{0.5\textwidth}
        \centering
       \resizebox{!}{9em}{ \begin{tikzpicture}[->, >=stealth', auto, semithick, node distance=2.5cm]

  \tikzstyle{state}=[circle, draw, fill=gray!30, minimum size=1.5cm]
  \tikzstyle{emission}=[circle, draw, minimum size=1.5cm]

  \node[state] (z1)                         {$Z(t_1)$};
  \node[state] (z2) [right of=z1]           {$Z(t_2)$};
  \node[state] (z3) [right of=z2]           {$Z(t_3)$};
  \node[state] (z4) [right of=z3]           {$Z(t_4)$};

  \node [emission] (y1) [below=2cm of z1] {$Y_1$};
  \node [emission](y2) [below=2cm of z2] {$Y_2$};
  \node [emission](y3) [below=2cm of z3] {$Y_3$};
  \node [emission](y4) [below=2cm of z4] {$Y_4$};

  \path (z1) edge [->] (z2)
        (z2) edge [->] (z3)
        (z3) edge [->] (z4);

  \path (y1) edge [->] (y2)
        (y2) edge [->] (y3)
        (y3) edge [->] (y4);

  \path (z1) edge [->]  (y1)
        (z2) edge [->]  (y2)
        (z3) edge [->]  (y3)
        (z4) edge [->]  (y4);

\end{tikzpicture}}
        \caption{Proposed Model}\label{fig:HMMs:proposed}
    \end{subfigure}
    \caption{Two Hidden Markov Models}
    \label{fig:HMMs}
\end{figure*}

Throughout this appendix, we fix a parameter $\Theta$ of the surrogate model in Section~\ref{ssec:graphest}. We write $P(\cdot\mid\cdot;\Theta)$ for probabilities and densities computed under this model, with the conventions specified in Appendix~\ref{ssec:truncatedforward}. Even at $\Theta=\Theta^\star$, these are probabilities assigned by the surrogate model. The concentration bounds in Appendices~\ref{ssec:aux:onestep}, \ref{sec:concentration-deviation} and~\ref{sec:restricted-eigenvalue} concern these computed quantities as random functions of the observations; their expectations and probabilities of deviations are evaluated under the data-generating distribution.

The forward--backward algorithm for hidden Markov models, which computes the smoothing probabilities in the E-step of the Baum--Welch algorithm~\citep{baum1970maximization}, can be decomposed into a forward recursion for the filtering probabilities $P(Z(t_n)=\ell\mid Y_0,\ldots,Y_n;\Theta)$ and a backward recursion for the conditional densities $P(Y_{n+1},\ldots,Y_N\mid Z(t_n)=\ell,Y_n;\Theta)$ of the future observations; together, these give the smoothing probabilities. While this algorithm is a long-standing standard for inferring the latent states, the restricted eigenvalue and deviation bound conditions in Appendix~\ref{ssec:onestep:conditions} concern sums over $n$ whose summands are weighted by the smoothing probabilities, and each full smoothing probability depends on all observations, which makes concentration bounds for these sums difficult to establish. Instead, we work with truncated smoothing probabilities, whose calculations retain only the interval factors in a window of observations. Under Assumption~\ref{assumption:local}, these probabilities are local functions of the data; truncation alone does not remove the dependence between the summands.

In this section, we show that, under Assumption~\ref{assumption:mixing}, the filtering and smoothing probabilities are close to their truncated versions in total variation distance. In later sections, we use the truncated marginal and joint smoothing probabilities to establish concentration bounds. We view the sequences of filtering and smoothing probabilities as discrete-time stochastic processes taking values in $[0,1]$. Under Assumption~\ref{assumption:local}, the truncated filtering and smoothing probabilities computed with $\hat\Psi$ are functions of fixed-width windows of the interval-block process $\{W_n\}$, as established in Lemma~\ref{lemma:local:blocks}. Under Assumptions~\ref{assumption:stationary} and~\ref{assumption:mixingbeta}, which make $\{W_n\}$ strictly stationary and geometrically $\beta$-mixing, the summands at the interior indices with the same remainder modulo $d_{\mathrm{rec}}$ therefore form a stationary $\beta$-mixing sequence. Concentration bounds for sample means of mixing processes~\citep{yu1994rates, merlevede2011bernstein}, applied to these sequences, give concentration bounds for sums weighted by the truncated smoothing probabilities and establish the restricted eigenvalue condition for the corresponding Gram matrix, as shown in Lemmas~\ref{lemma:sc_w}, \ref{lemma:deviation} and~\ref{lemma:RE:condition}.

Our proof techniques build on~\citet{van2008hidden} and~\citet{yang2017statistical}, and we extend their results to the joint smoothing probabilities $P(Z(t_n), Z(t_{n+1})\mid Y_0,\ldots, Y_N;\Theta)$ for $n=0,\ldots, N-1$. The hidden Markov model discussed in~\citet{yang2017statistical} has a different structure from ours, as shown in Figure~\ref{fig:HMMs}, so establishing the mixing property requires a new analysis. In Appendix~\ref{ssec:truncatedforward}, we define the forward operator $F_n$ and the forward kernel, and in Appendix~\ref{ssec:truncatedbackward}, the backward kernel. Using these operators and kernels, we then analyze the truncated smoothing probabilities in Appendix~\ref{ssec:truncatedsmooth}.

\subsection{Truncated filtering probability}
\label{ssec:truncatedforward}

For shorthand, we write $Z_n=Z(t_n)$ for $n=0,\ldots,N$ and $Y_{n_1}^{n_2}=\{Y_{n_1},\ldots,Y_{n_2}\}$ for $n_1\leq n_2$. The filtering probability up to sample $n$ is $P(Z_n=j\mid Y_0^n;\Theta)$. The transition probability matrix is $P=\exp(Qh)$, with entries $P_{ij}$, where $Q\in\Omega_Q$ is the transition rate matrix in $\Theta$, and $\pi=\pi(Q)$ is its stationary distribution. Assumption~\ref{assumption:mixing} therefore applies to $P$. Unlike in the main text, in Appendices~\ref{ssec:truncatedforward}--\ref{ssec:truncatedsmooth} the letters $i$ and $j$ index states and $a$, $b$, $m$ and $\ell$ index sampling times; in the last step of the proof of Lemma~\ref{lemma:truncated_smoooth}, $\ell$ and $\ell'$ are states again.

For each realization of $Y_0^N$, we compute the reconstructed covariates and use the emission densities $P(Y_n\mid Z_n=j,Y_{n-1};\Theta)$ as fixed positive factors in the forward--backward identities below. The full and truncated calculations use the same reconstructed covariates and the same emission factors at shared sampling indices.

A window $Y_a^b$ retains the emission factors at sampling indices $a<n\leq b$. Each such factor uses $\hat\Psi(t_n)$, which is computed from the estimated trajectory before truncation; we do not repeat wavelet smoothing using only $Y_a^b$. Under Assumption~\ref{assumption:local}, $\hat\Psi(t_n)$ is a function of $Y_{(n-r_{\mathrm{rec}})\vee0}^{(n+r_{\mathrm{rec}})\wedge N}$. This observation window can extend past $a$ or $b$ when $n$ is near an endpoint, so the retained factors can depend on observations outside $Y_a^b$. In particular, a truncation window of radius $r$ gives an observation window of radius $r_+=r+r_{\mathrm{rec}}$, as established in Lemma~\ref{lemma:local:blocks}.

State probabilities are obtained by normalizing the corresponding sums of surrogate path products. The factors for future observations are computed from the same transitions and emissions. We suppress dependence on the fixed $\Theta$ and these evaluated emission factors in the operators, kernels and auxiliary quantities defined below.

The surrogate calculation on a window $Y_a^b$ starts from $Z_a\sim\pi(Q)$ and treats $Y_a$ as a given initial observation. Thus $P(Z_a\in\cdot\mid Y_a;\Theta)=\pi(Q)$. Following Chapter~5 of~\citet{van2008hidden}, we define the forward operator $F_n$ at sampling index $n\geq1$ by

\[
(F_n\nu)(j)=\frac{\sum_i P(Y_n\mid Z_n=j, Y_{n-1};\Theta)P_{ij}\nu(i)}{\sum_{i,j'} P(Y_n\mid Z_n=j', Y_{n-1};\Theta)P_{ij'}\nu(i)}.
\]
Then, we can express the filtering probability as
\[
P(Z_n=j\mid Y_0^n;\Theta)=F_n P(Z_{n-1}\mid Y_0^{n-1};\Theta)(j).
\]
Iterating for $n-1$ steps, we obtain
\[
P(Z_n=j\mid Y_0^n;\Theta)=F_n\cdots F_2 P(Z_1\mid Y_0^1;\Theta)(j).
\]

Furthermore, for $n\geq \ell\geq 1$, we define the forward kernel $K_{\ell\mid n}$ by
\begin{equation}\label{eq:transition_kernel:ij}
K_{\ell\mid n}(i,j)=P(Z_\ell=j\mid Z_{\ell-1}=i, Y_{\ell-1}^n;\Theta).
\end{equation}
By the surrogate factorization represented in Figure~\ref{fig:HMMs:proposed}, with the emission factors held fixed, we have, for $n'<\ell\leq n$,
\[
P(Z_\ell=j\mid Y_{n'}^n;\Theta)=\sum_i K_{\ell\mid n}(i,j)P(Z_{\ell-1}=i\mid Y_{n'}^n;\Theta).
\]
For a probability measure $\nu$ and $0\leq \ell'<n$, we define the vector $\nu_{\ell'\mid n}$ by
\[
\nu_{\ell'\mid n}(i)=\frac{P(Y_{\ell'+1}^n\mid Z_{\ell'}=i, Y_{\ell'};\Theta)\nu(i)}{\sum_j P(Y_{\ell'+1}^n\mid Z_{\ell'}=j, Y_{\ell'};\Theta)\nu(j)}.
\]
Define
\[
\nu_{\ell'\mid n}^\top K_{\ell'+1\mid n}=\sum_i \nu_{\ell'\mid n}(i)K_{\ell'+1\mid n}(i,\cdot).
\]
A direct calculation shows that
\[
F_n\cdots F_{\ell'+1}\nu=\nu_{\ell'\mid n}^\top K_{\ell'+1\mid n}\cdots K_{n\mid n}.
\]
 
The following lemma shows that the dependence of the filtering probability on the initial distribution decays geometrically. Its second statement bounds the difference between $P(Z_n=i\mid Y_0^n;\Theta)$ and the truncated filtering probability $P(Z_n=i\mid Y_{n-\ell}^n;\Theta)$.

\begin{lemma}[Adapted from Lemma~13 in~\citet{yang2017statistical}]
\label{lemma:truncatedforward}
Suppose that Assumption~\ref{assumption:mixing} holds. Then, for any probability distributions $\nu,\nu'$ and $n\geq \ell$, we have
\[
\|F_n\cdots F_{\ell+1}\nu - F_n\cdots F_{\ell+1}\nu'\|_\infty\leq \zeta^{-2}(1-\zeta\pi_{\min})^{n-\ell}\|\nu-\nu'\|_1.
\]
Furthermore, we have
\[
\max_{i=1,\ldots, k}|P(Z_n=i\mid Y_0^n;\Theta)-P(Z_n=i\mid Y_{n-\ell}^n;\Theta)|\leq 2\zeta^{-2}(1-\zeta\pi_{\min})^{\ell}.
\]
\end{lemma}

\begin{proof}[Proof of Lemma~\ref{lemma:truncatedforward}]
We follow the proof of Lemma~13 in~\citet{yang2017statistical}. The only difference is that here the emission density is $P(Y_n\mid Z_n, Y_{n-1};\Theta)$. Define
\[
L_{\ell,n}(j)=P(Y_\ell^n\mid Z_\ell=j, Y_{\ell-1};\Theta)=P(Y_{\ell+1}^n\mid Y_\ell, Z_\ell=j;\Theta)P(Y_\ell\mid Z_\ell=j, Y_{\ell-1};\Theta).
\]
Then
\[
K_{\ell\mid n}(i,j)=\frac{P_{ij}L_{\ell,n}(j)}{\sum_{j'}P_{ij'}L_{\ell,n}(j')}.
\]
By Assumption~\ref{assumption:mixing}, $P_{ij}\geq \zeta\pi_j\geq \zeta\pi_{\min}$, and therefore, since $P_{ij'}\leq1$ for every $j'$,
\[
K_{\ell\mid n}(i,j)\geq \zeta\pi_{\min}\frac{L_{\ell,n}(j)}{\sum_{j'}L_{\ell,n}(j')}=:\zeta\pi_{\min}P_{\ell\mid n}(j),
\]
where $P_{\ell\mid n}$ is a probability distribution on $[k]$. Therefore, by the Doeblin minorization condition~\citep{nummelin2004general},
\[
K_{\ell\mid n}=\zeta\pi_{\min}P_{\ell\mid n}+(1-\zeta\pi_{\min})Q_{\ell\mid n},
\]
where $Q_{\ell\mid n}$ is a transition kernel on $[k]$. It follows that, for any probability vectors $\mu,\mu'$,
\[
\|\mu K_{\ell\mid n}-\mu'K_{\ell\mid n}\|_1\leq (1-\zeta\pi_{\min})\|\mu-\mu'\|_1.
\]
Using
\[
F_n\cdots F_{\ell+1}\nu=\nu_{\ell\mid n}^\top K_{\ell+1\mid n}\cdots K_{n\mid n},
\]
we obtain
\[
\|F_n\cdots F_{\ell+1}\nu - F_n\cdots F_{\ell+1}\nu'\|_\infty\leq (1-\zeta\pi_{\min})^{n-\ell}\|\nu_{\ell\mid n}-\nu'_{\ell\mid n}\|_1.
\]
Furthermore,
\[
\|\nu_{\ell\mid n}-\nu'_{\ell\mid n}\|_1\leq \frac{\max_i P(Y_{\ell+1}^n\mid Z_\ell=i, Y_\ell;\Theta)}{\min_i P(Y_{\ell+1}^n\mid Z_\ell=i, Y_\ell;\Theta)}\|\nu-\nu'\|_1.
\]
Bounding the ratio in the last display by Lemma~\ref{lemma:mixing:forward:bd} yields the first statement.

For the second statement, take $\nu=P(Z_{n-\ell}\mid Y_0^{n-\ell};\Theta)$ and $\nu'=P(Z_{n-\ell}\mid Y_{n-\ell};\Theta)=\pi(Q)$. Then
\[
F_n\cdots F_{n-\ell+1}\nu=P(Z_n\mid Y_0^n;\Theta),
\]
and
\[
F_n\cdots F_{n-\ell+1}\nu'=P(Z_n\mid Y_{n-\ell}^n;\Theta).
\]
Since $\|\nu-\nu'\|_1\leq 2$, the second statement follows from the first.
\end{proof}

The following lemma shows that, under Assumption~\ref{assumption:mixing}, the conditional density $P(Y_{\ell+1}^n\mid Z_\ell=\cdot, Y_\ell;\Theta)$ depends on the current state only up to a bounded factor.

\begin{lemma}\label{lemma:mixing:forward:bd}
     Suppose that Assumption~\ref{assumption:mixing} holds and $n\geq\ell$. Then
\[
     \frac{\max_i P(Y_{\ell+1}^n\mid Z_\ell=i, Y_\ell;\Theta)}{\min_i P(Y_{\ell+1}^n\mid Z_\ell=i, Y_\ell;\Theta)}\leq \zeta^{-2}.
     \]
\end{lemma}

\begin{proof}[Proof of Lemma~\ref{lemma:mixing:forward:bd}]
Write
\begin{multline*}
P(Y_{\ell+1}^n\mid Z_\ell=i, Y_\ell;\Theta) \\
= \sum_{z_n,\ldots, z_{\ell+1}} P(Y_n\mid z_n, Y_{n-1};\Theta)P(z_n\mid z_{n-1};\Theta)\cdots
P(Y_{\ell+1}\mid z_{\ell+1}, Y_\ell;\Theta)P(z_{\ell+1}\mid Z_\ell=i;\Theta).
\end{multline*}
By Assumption~\ref{assumption:mixing}, we have
\[
P(z_{\ell+1}\mid Z_\ell=i;\Theta)\leq \zeta^{-1}\pi(z_{\ell+1}),
\]
and similarly
\[
P(z_{\ell+1}\mid Z_\ell=i;\Theta)\geq \zeta\pi(z_{\ell+1}).
\]
Therefore,
\begin{multline*}
P(Y_{\ell+1}^n\mid Z_\ell=i, Y_\ell;\Theta) \\
\leq \zeta^{-1}\sum_{z_n,\ldots, z_{\ell+1}} P(Y_n\mid z_n, Y_{n-1};\Theta)P(z_n\mid z_{n-1};\Theta)\cdots P(Y_{\ell+1}\mid z_{\ell+1}, Y_\ell;\Theta)\pi(z_{\ell+1}),
\end{multline*}
and
\begin{multline*}
P(Y_{\ell+1}^n\mid Z_\ell=i, Y_\ell;\Theta) \\
\geq \zeta\sum_{z_n,\ldots, z_{\ell+1}} P(Y_n\mid z_n, Y_{n-1};\Theta)P(z_n\mid z_{n-1};\Theta)\cdots P(Y_{\ell+1}\mid z_{\ell+1}, Y_\ell;\Theta)\pi(z_{\ell+1}).
\end{multline*}
Taking the maximum with respect to $i$ in the upper bound and the minimum with respect to $i$ in the lower bound completes the proof.
\end{proof}

\subsection{Dependence on future observations}
\label{ssec:truncatedbackward}

We now control how the smoothing probabilities depend on the right end of a window of observations $Y_a^{n'}$. In the surrogate factorization at $\Theta$, the emission factor of $Y_{s+1}$ depends on the states only through $Z_{s+1}$. Summing the products of the initial distribution, transition probabilities and fixed emission factors over the other states gives, for $0\leq a\leq s<n'\leq N$,
\begin{align*}
P(Z_s=i,Z_{s+1}=j,Y_a^{n'};\Theta)
&=P(Z_s=i,Y_a^s;\Theta)\,P_{ij}\\
&\quad\times P(Y_{s+1}\mid Z_{s+1}=j,Y_s;\Theta)\,P(Y_{s+2}^{n'}\mid Z_{s+1}=j,Y_{s+1};\Theta),
\end{align*}
where the last factor equals $1$ when $s+1=n'$. The last two factors do not depend on $i$, so the backward kernel
\begin{equation}\label{eq:backward:kernel}
\tilde{K}_{a\mid s}(j,i)
:=P(Z_s=i\mid Z_{s+1}=j,Y_a^{n'};\Theta)
=\frac{P_{ij}P(Z_s=i\mid Y_a^s;\Theta)}{\sum_{i'}P_{i'j}P(Z_s=i'\mid Y_a^s;\Theta)}
\end{equation}
does not depend on the right endpoint $n'$ when the reconstructed covariates are held fixed. Therefore, for $a\leq s<\ell\leq n'$,
\begin{align*}
P(Z_s\in\cdot\mid Y_a^{n'};\Theta)&=\sum_{j=1}^kP(Z_{s+1}=j\mid Y_a^{n'};\Theta)\tilde{K}_{a\mid s}(j,\cdot),\\
P(Z_s\in\cdot\mid Y_a^{\ell};\Theta)&=\sum_{j=1}^kP(Z_{s+1}=j\mid Y_a^{\ell};\Theta)\tilde{K}_{a\mid s}(j,\cdot),
\end{align*}
with the same kernels in both recursions.

The following lemma shows that the influence of distant future observations on the smoothing probabilities decays geometrically.

\begin{lemma}\label{lemma:truncatedbackward}
Suppose that Assumption~\ref{assumption:mixing} holds. Then, for $0\leq a\leq m\leq\ell\leq n'\leq N$ and every realization of $Y_0^N$,
\[
\sum_{i=1}^k\abr{P(Z_m=i\mid Y_a^{n'};\Theta)-P(Z_m=i\mid Y_a^{\ell};\Theta)}
\leq2(1-\zeta^2)^{\ell-m}.
\]
In particular, since $\zeta\pi_{\min}\leq\zeta$, the left-hand side is at most $2\{1-(\zeta\pi_{\min})^2\}^{\ell-m}$.
\end{lemma}

\begin{proof}[Proof of Lemma~\ref{lemma:truncatedbackward}]
Write $u_s:=P(Z_s\in\cdot\mid Y_a^{n'};\Theta)$ and $\bar{u}_s:=P(Z_s\in\cdot\mid Y_a^{\ell};\Theta)$ for $a\leq s\leq\ell$, as row vectors. The two recursions above give $u_s=u_{s+1}\tilde{K}_{a\mid s}$ and $\bar{u}_s=\bar{u}_{s+1}\tilde{K}_{a\mid s}$ for $m\leq s<\ell$. Assumption~\ref{assumption:mixing} gives $\zeta\pi_j\leq P_{ij}\leq\zeta^{-1}\pi_j$ for all $i,j\in[k]$. Bounding the numerator of~\eqref{eq:backward:kernel} below and its denominator above, we obtain
\[
\tilde{K}_{a\mid s}(j,i)
\geq\frac{\zeta\pi_jP(Z_s=i\mid Y_a^s;\Theta)}{\zeta^{-1}\pi_j}
=\zeta^2P(Z_s=i\mid Y_a^s;\Theta),\qquad i,j\in[k].
\]
Therefore, every row of $\tilde{K}_{a\mid s}$ dominates $\zeta^2$ times the same probability vector. As in the proof of Lemma~\ref{lemma:truncatedforward}, this Doeblin minorization gives $\|v\tilde{K}_{a\mid s}\|_1\leq(1-\zeta^2)\|v\|_1$ for every row vector $v$ with $\sum_jv(j)=0$. Applying this bound for $s=\ell-1,\ldots,m$ yields
\[
\|u_m-\bar{u}_m\|_1\leq(1-\zeta^2)^{\ell-m}\|u_\ell-\bar{u}_\ell\|_1\leq2(1-\zeta^2)^{\ell-m},
\]
since $u_\ell$ and $\bar{u}_\ell$ are probability vectors.
\end{proof}

\subsection{Truncated smoothing probability}\label{ssec:truncatedsmooth}

\begin{proof}[Proof of Lemma~\ref{lemma:truncated_smoooth}]
We factor the full and the truncated joint smoothing probabilities into a filtering probability of $Z_n$, a conditional probability of $Z_{n+1}$ given the observations from $Y_n$ onward, a common factor and a likelihood ratio, and then compare the two factorizations term by term. Fix $0\leq n\leq N-1$ and $i,j\in[k]$, and write
\[
a=(n-r)\vee 0,\qquad b=(n+r)\wedge N,
\]
so that $b\geq n+1$ because $r\geq1$.

Define
\[
\mathsf f_i=P(Z_n=i\mid Y_0^n;\Theta),\qquad \bar{\mathsf f}_i=P(Z_n=i\mid Y_a^n;\Theta),
\]
and
\[
u_j=P(Z_{n+1}=j\mid Y_n^N;\Theta),\qquad \bar{u}_j=P(Z_{n+1}=j\mid Y_n^b;\Theta).
\]

By the surrogate factorization with the same fixed emission factors,
\[
P(Z_n=i,Z_{n+1}=j,Y_a^b;\Theta)
=
P(Z_n=i,Y_a^n;\Theta)P_{ij}P(Y_{n+1}^b\mid Z_{n+1}=j,Y_n;\Theta).
\]
Therefore,
\[
P(Z_n=i,Z_{n+1}=j\mid Y_a^b;\Theta)
=
\bar{\mathsf f}_i P_{ij}\frac{P(Y_{n+1}^b\mid Z_{n+1}=j,Y_n;\Theta)}{P(Y_{n+1}^b\mid Y_a^n;\Theta)}.
\]
Furthermore, by Bayes' rule,
\[
\bar{u}_j
=
\frac{P(Y_{n+1}^b\mid Z_{n+1}=j,Y_n;\Theta)P(Z_{n+1}=j\mid Y_n;\Theta)}{P(Y_{n+1}^b\mid Y_n;\Theta)},
\]
so
\[
P(Z_n=i,Z_{n+1}=j\mid Y_a^b;\Theta)
=
\bar{\mathsf f}_i\bar{u}_j
\frac{P_{ij}}{P(Z_{n+1}=j\mid Y_n;\Theta)}
\frac{P(Y_{n+1}^b\mid Y_n;\Theta)}{P(Y_{n+1}^b\mid Y_a^n;\Theta)}.
\]
Similarly,
\[
P(Z_n=i,Z_{n+1}=j\mid Y_0^N;\Theta)
=
\mathsf f_i u_j
\frac{P_{ij}}{P(Z_{n+1}=j\mid Y_n;\Theta)}
\frac{P(Y_{n+1}^N\mid Y_n;\Theta)}{P(Y_{n+1}^N\mid Y_0^n;\Theta)}.
\]

Next, the denominator of the factor common to both expressions satisfies
\[
P(Z_{n+1}=j\mid Y_n;\Theta)=\sum_{i'=1}^k P_{i'j}P(Z_n=i'\mid Y_n;\Theta)\geq \min_{i'}P_{i'j}\geq \delta_{\min},
\]
and $P_{ij}\leq 1$, so
\[
\frac{P_{ij}}{P(Z_{n+1}=j\mid Y_n;\Theta)}\leq \delta_{\min}^{-1}.
\]

Define
\[
R_N=\frac{P(Y_{n+1}^N\mid Y_n;\Theta)}{P(Y_{n+1}^N\mid Y_0^n;\Theta)},
\qquad
R_b=\frac{P(Y_{n+1}^b\mid Y_n;\Theta)}{P(Y_{n+1}^b\mid Y_a^n;\Theta)},
\]
and
\[
B_{N,n}:=\left|\frac{R_N}{R_b}-1\right|.
\]
Then
\begin{multline}
\abr{
P\bigg(Z_n=i, Z_{n+1}=j\mid Y_{0}^{N};\Theta\bigg)
-
P\bigg(Z_n=i, Z_{n+1}=j\mid Y_{a}^{b};\Theta\bigg)
} \\
\leq
\delta_{\min}^{-1}\Big\{
R_N \abr{\mathsf f_i u_j-\bar{\mathsf f}_i\bar{u}_j}
+
B_{N,n}R_b\bar{\mathsf f}_i\bar{u}_j
\Big\}. \label{eq:bound:truncatesmooth_revised}
\end{multline}

Summing over $Z_n$ in the surrogate factorization, we can write
\[
R_N=\frac{\sum_{s=1}^k P(Y_{n+1}^N\mid Z_n=s,Y_n;\Theta)P(Z_n=s\mid Y_n;\Theta)}{\sum_{s=1}^k P(Y_{n+1}^N\mid Z_n=s,Y_n;\Theta)P(Z_n=s\mid Y_0^n;\Theta)}.
\]
This is a ratio of two weighted averages of the same numbers $P(Y_{n+1}^N\mid Z_n=s,Y_n;\Theta)$, $s\in[k]$, so Lemma~\ref{lemma:mixing:forward:bd} gives $R_N\leq\zeta^{-2}$. The same argument, with $N$ replaced by $b$ and $Y_0^n$ by $Y_a^n$, gives $R_b\leq\zeta^{-2}$.
Furthermore,
\[
\abr{\mathsf f_i u_j-\bar{\mathsf f}_i\bar{u}_j}
\leq
\abr{\mathsf f_i-\bar{\mathsf f}_i}
+
\abr{u_j-\bar{u}_j}.
\]

If $a=0$, then $\mathsf f_i=\bar{\mathsf f}_i$. Otherwise, the second statement of Lemma~\ref{lemma:truncatedforward}, applied with its index $\ell$ set to $r$, gives
\[
\abr{\mathsf f_i-\bar{\mathsf f}_i}
\leq 2\zeta^{-2}(1-\zeta\pi_{\min})^r.
\]

If $b=N$, then $u_j=\bar{u}_j$. Otherwise, Lemma~\ref{lemma:truncatedbackward}, applied with its indices $(a,m,\ell,n')$ set to $(n,n+1,b,N)$, gives
\[
\sum_{j=1}^k\abr{u_j-\bar{u}_j}
\leq2(1-\zeta^2)^{b-n-1}
\leq2\cbr{1-(\zeta\pi_{\min})^2}^{r-1}.
\]
In particular, since $2\leq4(\zeta\pi_{\min})^{-2}$,
\[
\abr{u_j-\bar{u}_j}
\leq
4(\zeta\pi_{\min})^{-2}\cbr{1-(\zeta\pi_{\min})^2}^{r-1}.
\]

It remains to bound $B_{N,n}$. By the chain rule,
\[
R_N=
\frac{P(Y_{n+1}^b\mid Y_n;\Theta)P(Y_{b+1}^N\mid Y_n^b;\Theta)}
{P(Y_{n+1}^b\mid Y_0^n;\Theta)P(Y_{b+1}^N\mid Y_0^b;\Theta)},
\]
with the convention that $P(Y_{b+1}^N\mid \cdot;\Theta)=1$ when $b=N$. Hence
\[
\frac{R_N}{R_b}
=
\frac{P(Y_{n+1}^b\mid Y_a^n;\Theta)}{P(Y_{n+1}^b\mid Y_0^n;\Theta)}
\cdot
\frac{P(Y_{b+1}^N\mid Y_n^b;\Theta)}{P(Y_{b+1}^N\mid Y_0^b;\Theta)}.
\]
Therefore, if we define
\[
T_{1,n}:=\left|\frac{P(Y_{n+1}^b\mid Y_a^n;\Theta)}{P(Y_{n+1}^b\mid Y_0^n;\Theta)}-1\right|,
\qquad
T_{2,n}:=\left|\frac{P(Y_{b+1}^N\mid Y_n^b;\Theta)}{P(Y_{b+1}^N\mid Y_0^b;\Theta)}-1\right|,
\]
then
\[
B_{N,n}\leq
\frac{P(Y_{b+1}^N\mid Y_n^b;\Theta)}{P(Y_{b+1}^N\mid Y_0^b;\Theta)}\,T_{1,n}+T_{2,n}
\leq \zeta^{-2}T_{1,n}+T_{2,n},
\]
where the last inequality follows as for $R_N$, by conditioning on $Z_b$ and applying Lemma~\ref{lemma:mixing:forward:bd}.

If $a>0$, write $L_s=P(Y_{n+1}^b\mid Z_n=s,Y_n;\Theta)$. Then
\[
T_{1,n}
=
\frac{\left|\sum_{s=1}^k L_s\{P(Z_n=s\mid Y_a^n;\Theta)-P(Z_n=s\mid Y_0^n;\Theta)\}\right|}
{\sum_{s=1}^k L_s P(Z_n=s\mid Y_0^n;\Theta)}.
\]
Therefore, by Lemmas~\ref{lemma:mixing:forward:bd} and~\ref{lemma:truncatedforward},
\[
T_{1,n}
\leq
\frac{\max_s L_s}{\min_s L_s}
\sum_{s=1}^k\abr{P(Z_n=s\mid Y_a^n;\Theta)-P(Z_n=s\mid Y_0^n;\Theta)}
\leq
2k\zeta^{-4}(1-\zeta\pi_{\min})^r.
\]
If $a=0$, then $T_{1,n}=0$.

Similarly, if $b<N$, write $M_s=P(Y_{b+1}^N\mid Z_b=s,Y_b;\Theta)$. Then
\[
T_{2,n}
=
\frac{\left|\sum_{s=1}^k M_s\{P(Z_b=s\mid Y_n^b;\Theta)-P(Z_b=s\mid Y_0^b;\Theta)\}\right|}
{\sum_{s=1}^k M_s P(Z_b=s\mid Y_0^b;\Theta)}.
\]
Applying Lemmas~\ref{lemma:mixing:forward:bd} and~\ref{lemma:truncatedforward} at time $b=n+r$ yields
\[
T_{2,n}
\leq
2k\zeta^{-4}(1-\zeta\pi_{\min})^r.
\]
If $b=N$, then $T_{2,n}=0$.

Therefore,
\[
B_{N,n}
\leq
4k\zeta^{-6}(1-\zeta\pi_{\min})^r
\leq
4k\zeta^{-6}\cbr{1-(\zeta\pi_{\min})^2}^{r-1}.
\]

Substituting these bounds into~\eqref{eq:bound:truncatesmooth_revised} and using
\[
(1-\zeta\pi_{\min})^r\leq \cbr{1-(\zeta\pi_{\min})^2}^{r-1},
\]
we obtain
\begin{multline*}
    \abr{P\bigg(Z_n=i, Z_{n+1}=j\mid Y_{a}^{b};\Theta\bigg)
        -
        P\bigg(Z_n=i, Z_{n+1}=j\mid Y_{0}^{N};\Theta\bigg)
        }\\
        \leq Ck\delta_{\min}^{-1}\zeta^{-8}\pi_{\min}^{-2}\cbr{1-(\zeta\pi_{\min})^2}^{r-1},
\end{multline*}
for an absolute constant $C$. Setting $i=\ell$ and $j=\ell'$ proves the first claim.

For the second claim, summing the differences of the joint smoothing probabilities over $j$ gives
\begin{align*}
& \abr{P(Z_n=i\mid Y_0^N;\Theta)-P(Z_n=i\mid Y_a^b;\Theta)} \\
&\qquad \qquad \qquad \leq
\sum_{j=1}^k
\abr{
P\bigg(Z_n=i,Z_{n+1}=j\mid Y_{0}^{N};\Theta\bigg)
-
P\bigg(Z_n=i,Z_{n+1}=j\mid Y_{a}^{b};\Theta\bigg)
}\\
&\qquad \qquad \qquad \leq
\delta_{\min}^{-1}\Big\{
R_N\sum_{j=1}^k \abr{\mathsf f_i u_j-\bar{\mathsf f}_i\bar{u}_j}
+
B_{N,n}R_b
\Big\}.
\end{align*}
Since $\sum_j u_j=\sum_j\bar{u}_j=1$,
\[
\sum_{j=1}^k \abr{\mathsf f_i u_j-\bar{\mathsf f}_i\bar{u}_j}
\leq
\abr{\mathsf f_i-\bar{\mathsf f}_i}
+
\sum_{j=1}^k \abr{u_j-\bar{u}_j}.
\]
Therefore,
\begin{align*}
&\abr{P(Z_n=i\mid Y_0^N;\Theta)-P(Z_n=i\mid Y_a^b;\Theta)}
\\ 
&
\quad 
\leq
\delta_{\min}^{-1}\Bigg[
\zeta^{-2}\cbr{
2\zeta^{-2}(1-\zeta\pi_{\min})^r
+
4k(\zeta\pi_{\min})^{-2}\cbr{1-(\zeta\pi_{\min})^2}^{r-1}
}
+
4k\zeta^{-8}(1-\zeta\pi_{\min})^r
\Bigg]\\
&\quad
\leq
10\,\delta_{\min}^{-1}k\zeta^{-8}\pi_{\min}^{-2}\cbr{1-(\zeta\pi_{\min})^2}^{r-1}.
\end{align*}
Setting $i=\ell$ proves the second claim for $0\leq n\leq N-1$.

For $n=N$, we have $b=N$, and the two probabilities in the second claim coincide if $a=0$. If $a=N-r>0$, the second statement of Lemma~\ref{lemma:truncatedforward}, applied with its indices $(n,\ell)$ set to $(N,r)$, bounds the left-hand side by $2\zeta^{-2}(1-\zeta\pi_{\min})^r$, which is at most the right-hand side because $\delta_{\min},\zeta,\pi_{\min}\leq1$ and $(1-\zeta\pi_{\min})^r\leq\cbr{1-(\zeta\pi_{\min})^2}^{r-1}$.
\end{proof}

\section{Implementation details of the E-step}\label{appendix:algorithm:details}

In this section, we describe how to compute $m_{\ell\ell'}(Y_0^N;\Theta)$ and $\tau_\ell(Y_0^N;\Theta)$ in Definition~\ref{def:surrogate_function}.
Under the surrogate model, $m_{\ell\ell'}$ is conditionally independent of $Y_0^N$ given the discrete state sequence $Z_0^N$, so that
\begin{align}
    m_{\ell\ell'}(Y_0^N;\Theta)&:= \EE[m_{\ell\ell'}(T)\mid Y_0^N;\Theta]\notag\\
    &= \sum_{n=1}^N\sum_{a,b=1}^k P(Z(t_{n-1})=a, Z(t_{n})=b\mid Y_0^N;\Theta)\times\notag\\
    &\quad\quad\quad\quad\quad\EE[m_{\ell\ell'}(t_{n}-t_{n-1})\mid Z(t_{n-1})=a, Z(t_{n})=b;\Theta]\notag\\
    & = \sum_{n=1}^N\sum_{a,b=1}^k P(Z(t_{n-1})=a, Z(t_{n})=b\mid Y_0^N;\Theta)\times\notag\\
    &\quad\quad\quad\quad\quad\EE[m_{\ell\ell'}(h)\mid Z(0)=a, Z(h)=b;\Theta],\label{eq:computemll}
\end{align}
where the last equality follows from the time-homogeneity of the latent Markov chain. Here, $\EE[m_{\ell\ell'}(h)\mid Z(0)=a, Z(h)=b;\Theta]$ is the expected number of transitions from state $\ell$ to state $\ell'$ in a time interval of length $h$, given that the chain starts in state $a$ and is in state $b$ at time $h$. Furthermore, \citet{hobolth2005statistical} show that
\begin{align}
\EE[m_{\ell\ell'}(h)\mid Z(0)=a, Z(h)=b;\Theta] = \frac{q_{\ell\ell'}}{P_{ab}(h)}\int_0^h P_{a\ell}(u)P_{\ell'b}(h-u)\mathrm{d}u,\label{eq:decomposem}
\end{align}
where $P_{ab}(t)=[\exp(Qt)]_{ab}$. As discussed by~\citet{liu2015efficient}, the integral on the right-hand side of~\eqref{eq:decomposem} can be computed in several ways. In our implementation, we use the method of~\citet{van1978computing}.

Similarly, we have
\begin{align}
    \tau_\ell(Y_0^N;\Theta) & := \EE[\tau_\ell\mid Y_0^N;\Theta] \notag\\
    & = 
    \sum_{n=1}^N\sum_{a,b=1}^k P(Z(t_{n-1})=a,Z(t_n)=b\mid Y_0^N;\Theta)\times\notag\\
    &\quad\quad\quad\quad\quad\EE[\tau_\ell(t_n-t_{n-1})\mid Z(t_{n-1})=a,  Z(t_n)=b;\Theta]\notag\\
    & = 
    \sum_{n=1}^N\sum_{a,b=1}^k P(Z(t_{n-1})=a,Z(t_n)=b\mid Y_0^N;\Theta)\times\notag\\
    &\quad\quad\quad\quad\quad\EE[\tau_\ell(h)\mid  Z(0)=a, Z(h)=b;\Theta],\label{eq:computetaul}
\end{align}
where, by~\citet{hobolth2005statistical},
\begin{equation}\label{eq:decomposetau}
\EE[\tau_\ell(h)\mid Z(0)=a, Z(h)=b;\Theta] = \frac{1}{P_{ab}(h)}\int_0^h P_{a\ell}(u)P_{\ell b}(h-u)\mathrm{d}u. 
\end{equation}

\section{One-step analysis of the truncated surrogate EM algorithm}\label{sec:proof:onestep}

We analyze one update of the truncated surrogate EM algorithm defined in Section~\ref{ssec:approximateEM}. For a current parameter $\Theta$, write $\hat\Theta=\tilde M_N(\Theta)$ and $\check\Theta=M(\Theta)$ for the empirical and population updates. The smoothing probabilities in these updates are those of the surrogate model computed with $\hat\Psi$ and started from $\pi(Q)$, as specified in Section~\ref{ssec:approximateEM}.

We first state the notation and conditions for the update of $\theta_i^\ell$ in Appendix~\ref{ssec:onestep:conditions}. Appendix~\ref{ssec:onestep_theta} establishes its deterministic one-step bound, and Appendices~\ref{ssec:onestep_sigma} and~\ref{ssec:onestep_Q} give the corresponding bounds for $\sigma^2$ and $Q$. The auxiliary lemmas in Appendix~\ref{ssec:aux:onestep} and the subsequent concentration bound for $\Delta_\varepsilon$ in Lemma~\ref{lemma:deviation} control the quantities entering the deviation bound condition, and Lemma~\ref{lemma:RE:condition} establishes the restricted eigenvalue condition at a fixed parameter. Lemma~\ref{lemma:oneupdate:theta:2} assembles these results at each fixed $\Theta$. The proof in Appendix~\ref{ssec:mainproof} then combines the blockwise conclusions with population contraction and surrogate bias to control successive iterations under the uniform concentration assumption of Theorem~\ref{theorem:graphrecovery}.

\subsection{Notation and conditions for the update of \texorpdfstring{$\theta_i^\ell$}{theta}}\label{ssec:onestep:conditions}

We fix $i\in[p]$ and $\ell\in[k]$ and, to reduce notational overhead, write $\theta=\theta_i^\ell$, $\hat\theta=\hat\theta_i^\ell$, $\check{\theta}=\check{\theta}_i^\ell$ and $\theta^\star=\theta_i^{\ell\star}$, all in $\RR^{pm}$, where $\hat\theta_i^\ell$ and $\check\theta_i^\ell$ are the corresponding blocks of $\hat\Theta$ and $\check\Theta$. Recall from Section~\ref{ssec:popEM} that $\hat{\Psi}(t_n)=(\hat{\Psi}_1^\top(t_n),\ldots,\hat{\Psi}_p^\top(t_n))^\top\in\RR^{pm}$ and
$K_{\Psi}=N^{-1}\sum_{n=1}^N\EE[\Psi(t_n)\Psi(t_n)^\top]$, and let $\Delta=\hat{\theta}-\theta^{\star}$. With $\hat{\Kcal}_\Psi:=\{\hat K_{\Psi_j}\}_{j=1}^p$, we also write $\norm{\theta}_{1,\hat{\Kcal}_\Psi}=\sum_{j=1}^p\norm{\theta_{ij}^{\ell}}_{\hat{K}_{\Psi_j}}$,
where $\norm{\theta_{ij}^{\ell}}_{\hat{K}_{\Psi_j}}=\{\theta_{ij}^{\ell\top}\hat{K}_{\Psi_j}\theta_{ij}^{\ell}\}^{1/2}$ and $\hat{K}_{\Psi_j}=N^{-1}\sum_{n=1}^N\hat{\Psi}_j(t_n)\hat{\Psi}_j(t_n)^\top$, and in this notation, the penalty in~\eqref{eq:M_ntheta} is $\lambda_i^\ell\norm{\tilde\theta_i^\ell}_{1,\hat{\Kcal}_\Psi}$. Given a positive definite matrix $K$, we write $\norm{\cdot}_{K^*}$ for the dual norm of $\norm{\cdot}_{K}$. The dual norm of $\norm{\cdot}_{1,\hat{\Kcal}_\Psi}$ is then $\norm{\cdot}_{\infty,\hat{\Kcal}_\Psi^*}$, the maximum over $j$ of the dual norms $\norm{\cdot}_{\hat{K}_{\Psi_j}^*}$ of the blocks.

We study guarantees for the group lasso estimator under two well-known conditions, the restricted eigenvalue condition and the deviation bound condition~\citep{Agawal2012fast, loh2012high}, and we adapt the analysis to the weighted group norm $\norm{\cdot}_{1,\hat{\Kcal}_{\Psi}}$ and its dual $\norm{\cdot}_{\infty,\hat{\Kcal}_{\Psi}^*}$. Recall that $s$ is the maximum sparsity level.

\begin{assumption}[Restricted eigenvalue condition]
\label{assumption:RE}
    There exist $a>0$ and $\nu>0$ with $a\geq50 \max_j\sigma_{\max}(\hat{K}_{\Psi_j})s\nu$ such that, for every $\Theta\in B(r_0,\Theta^\star)$, $\ell\in[k]$ and $x\in\mathbb{R}^{pm}$,
    \begin{align}
            \frac{1}{N}\sum_{n=1}^N\hat{w}_{\ell,\Theta}(t_n)\cbr{\hat{\Psi}(t_n)^\top x}^2\geq a\norm{x}_2^2-\nu\norm{x}_{1,\hat{\Kcal}_\Psi}^2.\label{eq:re}
    \end{align}
\end{assumption}
Intuitively, the parameter $a$ plays the role of the minimum eigenvalue of the weighted Gram matrix in~\eqref{eq:re}, and in Lemma~\ref{lemma:RE:condition} it is of the order of the minimum eigenvalue of the population counterpart of this matrix. Since the weighted Gram matrix is singular when $pm>N$, the term $\nu\norm{x}_{1,\hat{\Kcal}_\Psi}^2$ is needed for~\eqref{eq:re} to hold for every $x\in\mathbb{R}^{pm}$, and the condition on $a$ makes this term at most $a\norm{x}_2^2/50$ when at most $s$ of the $p$ blocks of $x$ are nonzero.

To state the deviation bound condition, we define ${Y}_n^{\Delta}=Y_{n,i}-{Y}_{n-1,i}$ and the following four quantities. The deviation term $\Delta_\varepsilon$ is the statistical term: the truncated score, computed with the true covariates, minus its expectation. The terms $\Delta_\Psi^\top\theta^\star$ and $\Delta_X$ account for reconstruction error in the Gram-matrix and response--covariate contributions, respectively, and $\Delta_w$ accounts for the difference between the expected scores computed with the truncated and full smoothing probabilities:
\begin{align}
\Delta_\varepsilon&=\Delta_{\varepsilon}(\Theta)=\frac{1}{N}\sum_{n=1}^N\hat{w}_{\ell,\Theta}(t_n)\cbr{Y_n^\Delta-\Psi(t_n)^\top \theta^{\star}}\Psi(t_n) \notag\\
    &\quad\quad\quad\quad\quad - \frac{1}{N}\sum_{n=1}^N\EE\sbr{\hat{w}_{\ell,\Theta}(t_n)\cbr{Y_n^\Delta-\Psi(t_n)^\top \theta^{\star}}\Psi(t_n)},\label{eq:define:deltaepsilon}\\
\Delta_\Psi&=\Delta_\Psi(\Theta)=\frac{1}{N}\sum_{n=1}^N\hat{w}_{\ell,\Theta}(t_n)\cbr{\Psi(t_n)\Psi(t_n)^\top-\hat{\Psi}(t_n)\hat{\Psi}(t_n)^\top},\label{eq:define:deltapsi}\\
\Delta_X&=\Delta_X(\Theta)=\frac{1}{N}\sum_{n=1}^N\hat{w}_{\ell,\Theta}(t_n)Y_n^\Delta\cbr{\hat{\Psi}(t_n)-\Psi(t_n)},\label{eq:define:deltax}\\
    \Delta_w&=\Delta_w(\Theta)=\frac{1}{N}\sum_{n=1}^N\EE\sbr{
    \cbr{\hat{w}_{\ell,\Theta}(t_n)-w_{\ell,\Theta}(t_n)}\cbr{Y_n^\Delta-\Psi(t_n)^\top\theta^{\star}}\Psi(t_n)
    }.\label{eq:define:deltaw}
\end{align}
Recall that $m$ is the number of basis functions and $r$ is the truncation radius. In addition, $\delta_1\geq0$ denotes a deterministic level for the reconstruction error, and Section~\ref{ssec:main:results} defines the event $\mathcal E_{\mathrm{rec}}$ on which $\max_{i\in[p],\,n\in[N]}\opnorm{\hat X_i-X_i}{[t_{n-1},t_n]}\leq\delta_1$.
\begin{assumption}[Deviation bound condition]
\label{assumption:DB}
There exists a deterministic function $\Delta_{\mathrm{DB}}$ of $(N,p,s,m,r,\delta_1)$ such that, for every $\Theta\in B(r_0,\Theta^\star)$, $i\in[p]$ and $\ell\in[k]$,
$
\max\cbr{\norm{ \Delta_{\varepsilon}}_{\infty,\hat{\Kcal}_{\Psi}^*},
\norm{ \Delta_{\Psi}^\top\theta^\star}_{\infty,\hat{\Kcal}_{\Psi}^*},
\norm{ \Delta_X}_{\infty,\hat{\Kcal}_{\Psi}^*},
\norm{ \Delta_{w}}_{\infty,\hat{\Kcal}_{\Psi}^*}
}\leq \Delta_{\mathrm{DB}}(N, p, s, m, r,\delta_1)
$.
\end{assumption}

\subsection{One-step update of \texorpdfstring{$\theta_i^\ell$}{theta}}\label{ssec:onestep_theta}

With the restricted eigenvalue and deviation bound conditions in place, the following lemma bounds the error of the truncated update $\tilde M_{N,\theta_i^\ell}(\Theta)$.

\begin{lemma}
\label{lemma:onestep:theta}
    Fix $i\in[p]$, $\ell\in[k]$ and $\Theta\in B(r_0,\Theta^\star)$. Suppose that $\hat{K}_{\Psi_j}\succ0$ for every $j\in[p]$, that~\eqref{eq:re} holds at $(\Theta,\ell)$ for every $x\in\RR^{pm}$, with constants $a>0$ and $\nu>0$ such that $a\geq50\max_j\sigma_{\max}(\hat{K}_{\Psi_j})s\nu$, that the inequality in Assumption~\ref{assumption:DB} holds at $(\Theta,i,\ell)$, and that
    \[
     \lambda_i^\ell\geq 4\max\cbr{4\Delta_{\mathrm{DB}}(N, p, s, m, r,\delta_1),
     \frac{1}{\sqrt{s}}
\frac{\sigma_{\max}(K_{\Psi})}{\max_j\sigma_{\max}^{1/2}(\hat{K}_{\Psi_j})}
\norm{\check{\theta}_i^\ell-\theta_i^{\ell\star}}_2
     }.
    \]
   Then we have
   $
   \norm{\hat{\theta}_i^\ell-\theta_{i}^{\ell\star}}_2\leq\frac{4}{a}\rbr{
5\lambda_i^\ell\max_j\sigma_{\max}^{1/2}(\hat{K}_{\Psi_j})\sqrt{s}+\sigma_{\max}(K_\Psi)\norm{\check{\theta}_i^\ell-\theta_i^{\ell\star} }_2}
   $
   and
   $
   \norm{\hat{\theta}_i^\ell-\theta_{i}^{\ell\star}}_{1,\hat{\Kcal}_\Psi}\leq5\max_j\sigma_{\max}^{1/2}(\hat{K}_{\Psi_j})\sqrt{s}\,\norm{\hat{\theta}_i^\ell-\theta_{i}^{\ell\star}}_2.
   $
\end{lemma}

Assumptions~\ref{assumption:RE} and~\ref{assumption:DB} give the restricted eigenvalue and deviation bound conditions of Lemma~\ref{lemma:onestep:theta} for every $\Theta\in B(r_0,\Theta^\star)$, $i\in[p]$ and $\ell\in[k]$. The first bound in Lemma~\ref{lemma:onestep:theta} shows that $\norm{\hat{\theta}_i^\ell-\theta_{i}^{\ell\star}}_2$ is bounded by a statistical error governed by $\Delta_{\mathrm{DB}}(N,p,s,m,r,\delta_1)$ and the corresponding population discrepancy $\norm{\check{\theta}_i^\ell-\theta_i^{\ell\star}}_2$. The second bound is the cone condition of the lasso, which shows that the error is approximately $ms$-sparse. We use it in the proof of Theorem~\ref{theorem:graphrecovery} to show that the iterates remain in $B(r_0,\Theta^\star)$. Note that the restricted eigenvalue and deviation bound conditions are conditions on the observed sample, so that Lemma~\ref{lemma:onestep:theta} is a deterministic statement on the event where they hold.

\begin{proof}

We use the notation of Appendix~\ref{ssec:onestep:conditions} and write $\lambda:=\lambda_i^\ell$. We now bound $\Delta=\hat\theta-\theta^\star$ after one update of the truncated surrogate EM algorithm.

By construction, it follows that
\begin{multline*}
\frac{1}{N}\sum_{n=1}^N \hat{w}_{\ell,\Theta}(t_n)\cbr{{Y}_n^{\Delta}-\hat{\Psi}(t_n)^\top\hat{\theta}}^2+\lambda\norm{\hat{\theta}}_{1,\hat{\Kcal}_{\Psi}}
\\\leq 
\frac{1}{N}\sum_{n=1}^N \hat{w}_{\ell,\Theta}(t_n)\cbr{{Y}_n^{\Delta}-\hat{\Psi}(t_n)^\top\theta^{\star}}^2+\lambda\norm{\theta^{\star}}_{1,\hat{\Kcal}_{\Psi}}.
\end{multline*}
Rearranging the terms, we obtain
\begin{multline}\label{eq:term1}
   \frac{1}{N}\sum_{n=1}^N \hat{w}_{\ell,\Theta}(t_n)\rbr{\hat{\Psi}(t_n)^\top \Delta}^2\\\leq
   \lambda\rbr{\norm{\theta^{\star}}_{1,\hat{\Kcal}_\Psi}-\norm{\hat{\theta}}_{1,\hat{\Kcal}_\Psi}}+{\frac{2}{N}\sum_{n=1}^N\hat{w}_{\ell,\Theta}(t_n)\cbr{{Y}_n^{\Delta}-\hat{\Psi}(t_n)^\top \theta^{\star}}\hat{\Psi}(t_n)^\top\Delta}.
\end{multline}
Recall that $\check{\theta}=\argmax_{\theta'} \Lcal(\theta'\mid \Theta)$. The first-order condition $\nabla_{\check{\theta}} \Lcal(\check{\theta}\mid \Theta)=0$ gives
\begin{equation}\label{eq:fact:tildetheta}
    \sum_{n=1}^N\EE\sbr{w_{\ell,\Theta}(t_n)\cbr{{Y}_n^{\Delta}-{\Psi}(t_n)^\top\check{\theta}}{\Psi}(t_n)}=0.
\end{equation}
Therefore, adding and subtracting terms, we can write
\begin{align*}
    \frac{1}{N}&\sum_{n=1}^N\hat{w}_{\ell,\Theta}(t_n)\cbr{{Y}_n^{\Delta}-\hat{\Psi}(t_n)^\top\theta^{\star}}\hat{\Psi}(t_n)\\
    &=
     {\frac{1}{N}\sum_{n=1}^N\hat{w}_{\ell,\Theta}(t_n)\cbr{{Y}_n^{\Delta}-\hat{\Psi}(t_n)^\top\theta^{\star}}\hat{\Psi}(t_n)}
     \\
     &\quad\quad\quad\quad\quad\quad\quad\quad\quad-
      {\frac{1}{N}\sum_{n=1}^N\hat{w}_{\ell,\Theta}(t_n)\cbr{{Y}_n^{\Delta}-{\Psi}(t_n)^\top\theta^{\star}}{\Psi}(t_n)}\\
      &\quad+{\frac{1}{N}\sum_{n=1}^N\hat{w}_{\ell,\Theta}(t_n)\cbr{{Y}_n^{\Delta}-{\Psi}(t_n)^\top\theta^{\star}}{\Psi}(t_n)}
      \\
     &\quad\quad\quad\quad\quad\quad\quad\quad\quad-{\frac{1}{N}\sum_{n=1}^N\EE\sbr{\hat{w}_{\ell,\Theta}(t_n)\cbr{{Y}_n^{\Delta}-{\Psi}(t_n)^\top \theta^{\star}}{\Psi}(t_n)}}\\
      &\quad +\frac{1}{N}\sum_{n=1}^N\EE\sbr{\hat{w}_{\ell,\Theta}(t_n)\cbr{{Y}_n^{\Delta}-{\Psi}(t_n)^\top\theta^{\star}}{\Psi}(t_n)}
      \\
     &\quad\quad\quad\quad\quad\quad\quad\quad\quad-
      \frac{1}{N}\sum_{n=1}^N\EE\sbr{{w}_{\ell,\Theta}(t_n)\cbr{{Y}_n^{\Delta}-{\Psi}(t_n)^\top \theta^{\star}}{\Psi}(t_n)}\\
      &\quad +
      \frac{1}{N}\sum_{n=1}^N\EE\sbr{w_{\ell,\Theta}(t_n)\cbr{{Y}_n^{\Delta}-{\Psi}(t_n)^\top \theta^{\star}}{\Psi}(t_n)}
      \\
     &\quad\quad\quad\quad\quad\quad\quad\quad\quad-
      \frac{1}{N}\sum_{n=1}^N\EE\sbr{w_{\ell,\Theta}(t_n)\cbr{{Y}_n^{\Delta}-{\Psi}(t_n)^\top\check{\theta}}{\Psi}(t_n)},
\end{align*}
where the last term is zero by~\eqref{eq:fact:tildetheta}. The first difference in the preceding display satisfies
\begin{align*}
&\frac{1}{N}\sum_{n=1}^N\hat{w}_{\ell,\Theta}(t_n)\Big[\cbr{Y_n^\Delta-\hat{\Psi}(t_n)^\top\theta^\star}\hat{\Psi}(t_n)
-\cbr{Y_n^\Delta-\Psi(t_n)^\top\theta^\star}\Psi(t_n)\Big]\\
&\qquad=\frac{1}{N}\sum_{n=1}^N\hat{w}_{\ell,\Theta}(t_n)Y_n^\Delta\cbr{\hat{\Psi}(t_n)-\Psi(t_n)}\\
&\qquad\qquad+\frac{1}{N}\sum_{n=1}^N\hat{w}_{\ell,\Theta}(t_n)\cbr{\Psi(t_n)\Psi(t_n)^\top-\hat{\Psi}(t_n)\hat{\Psi}(t_n)^\top}\theta^\star\\
&\qquad=\Delta_X+\Delta_\Psi^\top\theta^\star.
\end{align*}
Recall the weighted Gram matrix defined before Lemma~\ref{lemma:aggregate-surrogate-bias},
\[
K_\Psi^{w,\ell}=K_\Psi^{w,\ell}(\Theta)=\frac{1}{N}\sum_{n=1}^N\EE\sbr{w_{\ell,\Theta}(t_n)\Psi(t_n)\Psi(t_n)^\top}.
\]
With this notation, the decomposition above reads
\begin{align}\label{eq:theta:score:decomposition}
    \frac{1}{N}\sum_{n=1}^N\hat{w}_{\ell,\Theta}(t_n)\cbr{{Y}_n^{\Delta}-\hat{\Psi}(t_n)^\top\theta^{\star}}\hat{\Psi}(t_n)=
    \Delta_{\Psi}^\top \theta^{\star}
    +\Delta_X
    +\Delta_w
    +\Delta_\varepsilon+ K_\Psi^{w,\ell}(\check{\theta}-{\theta}^{\star}),
\end{align}
where $\Delta_\Psi$, $\Delta_X$, $\Delta_w$ and $\Delta_\varepsilon$ are defined in~\eqref{eq:define:deltapsi}, \eqref{eq:define:deltax}, \eqref{eq:define:deltaw} and~\eqref{eq:define:deltaepsilon}, respectively.
Substituting~\eqref{eq:theta:score:decomposition} into~\eqref{eq:term1}, and bounding the first four terms by duality of the norms and the last one by the Cauchy--Schwarz inequality, we obtain
\begin{align}
    \frac{1}{N}\sum_{n=1}^N \hat{w}_{\ell,\Theta}(t_n)\rbr{\hat{\Psi}(t_n)^\top \Delta}^2&\leq
   \lambda\rbr{\norm{\theta^{\star}}_{1,\hat{\Kcal}_\Psi}-\norm{\hat{\theta}}_{1,\hat{\Kcal}_\Psi}}+
   2\sigma_{\max}(K_\Psi^{w,\ell})\|\Delta\|_{2}\norm{\check{\theta}-\theta^{\star} }_2\notag\\
   &\quad+2\|\Delta\|_{1,\hat{\Kcal}_\Psi}\Bigl(\norm{\Delta_\varepsilon}_{\infty,\hat{\Kcal}_\Psi^*}+\norm{\Delta_w}_{\infty,\hat{\Kcal}_\Psi^*}\notag\\
   &\qquad\qquad\qquad+\norm{ \Delta_\Psi^\top\theta^\star}_{\infty,\hat{\Kcal}_\Psi^*}+\norm{\Delta_X}_{\infty,\hat{\Kcal}_\Psi^*}\Bigr),\label{eq:term2}
\end{align}
where 
$\norm{\Delta_\varepsilon}_{\infty,\hat{\Kcal}_\Psi^*}=\max_{j=1,\ldots,p}\norm{\Delta_{\varepsilon,j}}_{\hat{K}_{\Psi_j}^*}$ and $\norm{\cdot}_{\hat{K}_{\Psi_j}^*}$ is the dual norm of $\norm{\cdot}_{\hat{K}_{\Psi_j}}$.

Define $S:=\{j\in[p]:\theta_{ij}^{\ell\star}\neq0\}$, the set of nonzero blocks of $\theta^\star$, so that $|S|\leq s$, and let $S^c:=[p]\setminus S$.
By the inequality of Assumption~\ref{assumption:DB} at $(\Theta,i,\ell)$ and the penalty condition of Lemma~\ref{lemma:onestep:theta}, we have
\[
\lambda\geq 4\rbr{\norm{\Delta_\varepsilon}_{\infty,\hat{\Kcal}_\Psi^*}
+
\norm{\Delta_w}_{\infty,\hat{\Kcal}_\Psi^*}
+\norm{\Delta_\Psi^\top\theta^\star }_{\infty,\hat{\Kcal}_\Psi^*}
+\norm{\Delta_X}_{\infty,\hat{\Kcal}_\Psi^*}},
\]
and therefore
\begin{align}
   \lambda\rbr{\norm{\theta^{\star}}_{1,\hat{\Kcal}_\Psi}-\norm{\hat{\theta}}_{1,\hat{\Kcal}_\Psi}}&+2\|\Delta\|_{1,\hat{\Kcal}_\Psi}\rbr{\norm{\Delta_\varepsilon}_{\infty,\hat{\Kcal}_\Psi^*}
   +
   \norm{\Delta_w}_{\infty,\hat{\Kcal}_{\Psi}^*}
   +\norm{\Delta_\Psi^\top\theta^\star }_{\infty,\hat{\Kcal}_\Psi^*}
   +\norm{\Delta_X}_{\infty,\hat{\Kcal}_\Psi^*}}\notag\\
   &\leq 
\lambda\rbr{\norm{\theta^{\star}}_{1,\hat{\Kcal}_\Psi}-\norm{\theta^{\star} +\Delta}_{1,\hat{\Kcal}_\Psi}}
+
\frac{\lambda}{2}\norm{\Delta}_{1,\hat{\Kcal}_\Psi}\notag\\
&=\lambda\rbr{\norm{\theta^{\star}_S}_{1,\hat{\Kcal}_\Psi}-\norm{(\theta^{\star} +\Delta)_S}_{1,\hat{\Kcal}_\Psi}-\norm{\Delta_{S^c}}_{1,\hat{\Kcal}_\Psi}}\notag\\
&\quad +\frac{\lambda}{2}\rbr{\norm{\Delta_{S}}_{1,\hat{\Kcal}_\Psi}+\norm{\Delta_{S^c}}_{1,\hat{\Kcal}_\Psi}}\notag\\
&\leq \frac{3\lambda}{2}\norm{\Delta_S}_{1,\hat{\Kcal}_{\Psi}}
-
\frac{\lambda}{2}\norm{\Delta_{S^c}}_{1,\hat{\Kcal}_{\Psi}}.\label{eq:term3}
\end{align}
Combining~\eqref{eq:term2} and~\eqref{eq:term3}, and using that (i) the left-hand side of~\eqref{eq:term2} is bounded below by $0$ and (ii)
\[
\lambda\geq\frac{4}{\sqrt{s}}
\frac{\sigma_{\max}{(K_{\Psi})}}{\max_j\sigma_{\max}^{1/2}(\hat{K}_{\Psi_j})}
{\norm{\check{\theta}-\theta^{\star} }_2},
\]
we can conclude that
\begin{align}
\frac{\lambda}{2}\norm{\Delta_{S^c}}_{1,\hat{\Kcal}_{\Psi}}&\leq
\frac{3\lambda}{2}\norm{\Delta_{S}}_{1,\hat{\Kcal}_\Psi}+
2\sigma_{\max}(K_\Psi^{w,\ell})\|\Delta\|_{2}{\norm{\check{\theta}-\theta^{\star} }_2}.\notag\\
\intertext{Recall from Appendix~\ref{ssec:onestep:conditions} that $K_\Psi=N^{-1}\sum_{n=1}^N\EE[\Psi(t_n)\Psi(t_n)^\top]$. Since $0\leq w_{\ell,\Theta}(t_n)\leq1$, we have $K_\Psi^{w,\ell}\preceq K_\Psi$, and hence $\sigma_{\max}(K_\Psi^{w,\ell})\leq\sigma_{\max}(K_\Psi)$. Therefore, we can upper bound the above display as}
&\leq 
\frac{3\lambda}{2}\norm{\Delta_{S}}_{1,\hat{\Kcal}_\Psi}+
2\sigma_{\max}(K_\Psi)\|\Delta\|_{2}{\norm{\check{\theta}-\theta^{\star} }_2}\notag\\
&\leq \frac{3\lambda}{2}\max_{j}\sigma_{\max}^{1/2}(\hat{K}_{\Psi_j})\sqrt{s}\norm{\Delta}_2+\frac{\lambda}{2}\max_{j}\sigma_{\max}^{1/2}(\hat{K}_{\Psi_j})\sqrt{s}\norm{\Delta}_2 \notag\\
&=2\lambda\max_{j}\sigma_{\max}^{1/2}(\hat{K}_{\Psi_j})\sqrt{s}\norm{\Delta}_2.\notag
\end{align}
This implies that
\begin{align}
\norm{\Delta_S}_{1,\hat{\Kcal}_{\Psi}}+\norm{\Delta_{S^c}}_{1,\hat{\Kcal}_{\Psi}}=\norm{\Delta}_{1,\hat{\Kcal}_{\Psi}}\leq  5 \max_j\sigma_{\max}^{1/2}(\hat{K}_{\Psi_j})\sqrt{s}\norm{\Delta}_2.\label{eq:delta:1K}
\end{align}
This is the second claim of Lemma~\ref{lemma:onestep:theta}.

Combining~\eqref{eq:term2} with~\eqref{eq:term3} and~\eqref{eq:delta:1K}, we obtain
\begin{align}
 \frac{1}{N}\sum_{n=1}^N \hat{w}_{\ell,\Theta}(t_n)\rbr{\hat{\Psi}(t_n)^\top \Delta}^2&\leq 2\lambda\norm{\Delta}_{1,\hat{\Kcal}_\Psi}+{2\sigma_{\max}(K_\Psi^{w,\ell})\norm{\check{\theta}-\theta^{\star} }_2}\|\Delta\|_{2}\notag\\
 &\leq 2\|\Delta\|_{2}\big\{5\lambda\max_j\sigma_{\max}^{1/2}(\hat{K}_{\Psi_j})\sqrt{s}+\sigma_{\max}(K_\Psi)\norm{\check{\theta}-\theta^\star }_2\big\}.\label{eq:term4}
\end{align}
By~\eqref{eq:re} at $(\Theta,\ell)$, \eqref{eq:delta:1K} and $a\geq50\max_j\sigma_{\max}(\hat{K}_{\Psi_j})s\nu$, we have
\begin{align*}
        \frac{1}{N}\sum_{n=1}^N\hat{w}_{\ell,\Theta}(t_n)\cbr{\hat{\Psi}(t_n)^\top\Delta}^2&\geq a\norm{\Delta}_2^2-\nu\norm{\Delta}_{1,\hat{\Kcal}_\Psi}^2\\
        &\geq a\norm{\Delta}_2^2-25 \max_j\sigma_{\max}(\hat{K}_{\Psi_j})s\nu\norm{\Delta}_2^2\\
        &\geq\frac{a}{2}\norm{\Delta}_2^2. 
\end{align*}
Combining the above result with~\eqref{eq:term4}, we can conclude that
\[
\norm{\Delta}_2\leq\frac{4}{a}\rbr{
5\lambda\max_j\sigma_{\max}^{1/2}(\hat{K}_{\Psi_j})\sqrt{s}+\sigma_{\max}(K_\Psi)\norm{\check{\theta}-\theta^{\star} }_2
}.
\]
This is the first claim of Lemma~\ref{lemma:onestep:theta}, which completes the proof.
\end{proof}

\subsection{One-step update of \texorpdfstring{$\sigma^2$}{the variance}}\label{ssec:onestep_sigma}
We next control the variance block of $\hat\Theta=\tilde M_N(\Theta)$ relative to $\sigma^{\star2}$. Together with the bounds for $\theta_i^\ell$ and $Q$, this controls all blocks entering $\dist(\hat\Theta,\Theta^\star)$.

For $n\in[N]$, $i\in[p]$, and $\ell\in[k]$, write
$Y_{n,i}^\Delta:=Y_{n,i}-Y_{n-1,i}$ and
$\theta_i^\ell:=(\theta_{i1}^{\ell\top},\ldots,\theta_{ip}^{\ell\top})^\top\in\RR^{pm}$. Define
\begin{align*}
\hat{R}_{n,\ell} &:= \sum_{i=1}^p\left(Y_{n,i}^\Delta-\sum_{j=1}^p\hat{\Psi}_j(t_n)^\top \theta^\ell_{ij}\right)^2;\\
R_{n,\ell} &:= \sum_{i=1}^p\left(Y_{n,i}^\Delta-\sum_{j=1}^p{\Psi}_j(t_n)^\top \theta^\ell_{ij}\right)^2.
\end{align*}
The following lemma bounds the error of one update of the noise variance.

\begin{lemma}
\label{lemma:onestep:sigma}
Suppose that Assumption~\ref{assumption:mixing} holds, and let $\Theta\in\Omega$. Recall that $\zeta$, $\delta_{\min}$ and $\pi_{\min}$ are defined in Assumption~\ref{assumption:mixing}. Let $\delta_1\geq0$, let $\mathcal E_{\mathrm{rec}}$ be the event of Section~\ref{ssec:main:results} on which $\max_{i\in[p]}\max_{n\in[N]}\opnorm{X_i-\hat{X}_i}{[t_{n-1}, t_n]}\leq\delta_1$, and define
\begin{align*}
\delta_2= \abr{\frac{1}{2Np}\sum_{n=1}^N\sum_{\ell=1}^k \cbr{\hat{w}_{\ell,\Theta}(t_n)R_{n,\ell}-
    \EE\sbr{\hat{w}_{\ell,\Theta}(t_n)R_{n,\ell}}}}.
\end{align*}
    Further assume that $|X_i(t)|\leq R_X$ and $|\hat X_i(t)|\leq R_X$ for all $i\in[p]$ and $t\in[0,T]$, and that $\sup_{|x|\leq R_X}\max_{q\in[m]}|g_q'(x)|\leq D$. Write
    $\hat{\sigma}^2:=\tilde M_{N,\sigma^2}(\Theta)$ and
    $\check{\sigma}^2:=M_{\sigma^2}(\Theta)$.
    Let $\mathcal R_\Theta>0$ satisfy
    \[
    \max_{\ell\in[k]}\frac1{Np}\sum_{n=1}^N\EE[R_{n,\ell}]
    \leq\mathcal R_\Theta.
    \]
    On the event $\mathcal E_R\cap\mathcal E_{\mathrm{rec}}$, where
    \[
    \mathcal E_R:=\left\{
    \max_{\ell\in[k]}\frac1{Np}\sum_{n=1}^N
    \hat w_{\ell,\Theta}(t_n)R_{n,\ell}\leq2\mathcal R_\Theta
    \right\},
    \]
    we have
    \begin{align*}
    \abr{\hat{\sigma}^2-\sigma^{\star 2}}
    &\leq \sqrt{2}\,kD\sqrt h\,\delta_1
    \max_{\ell\in[k]}\left\{\frac1p\sum_{i=1}^p\norm{\theta_i^\ell}_1^2\right\}^{1/2}
    \sqrt{\mathcal R_\Theta} \\
    &\quad+\frac{k}{2}D^2h\delta_1^2
    \max_{\ell\in[k]}\left\{\frac1p\sum_{i=1}^p\norm{\theta_i^\ell}_1^2\right\}
    +\delta_2 \\
    &\quad+5k^2\delta_{\min}^{-1}\mathcal R_\Theta
    \zeta^{-8}\pi_{\min}^{-2}\cbr{1-(\zeta\pi_{\min})^2}^{r-1}
    +\abr{\check{\sigma}^2-\sigma^{\star2}}.
    \end{align*}
\end{lemma}
\begin{proof}[Proof of Lemma~\ref{lemma:onestep:sigma}.]

With $Q$ and $\theta^\ell$ fixed for $\ell=1,\ldots, k$, the update of $\sigma^2$ is
\begin{align*}
\hat{\sigma}^2 = \argmax_{{\sigma'}^2}& -\frac{pN}{2}\log 2{\sigma'}^2 \\&\quad- \frac{1}{4{\sigma'}^2}\sum_{n=1}^N\sum_{\ell=1}^k
\hat{w}_{\ell,\Theta}(t_n)
\sum_{i=1}^p\left(Y_{n,i}-Y_{n-1,i}-\sum_{j=1}^p\hat{\Psi}_j(t_n)^\top\theta^\ell_{ij}\right)^2.
\end{align*}

The first-order condition gives
\[
\hat{\sigma}^2 = \frac{1}{2Np}\sum_{n=1}^N\sum_{\ell=1}^k \hat{w}_{\ell,\Theta}(t_n)\sum_{i=1}^p\left(Y_{n,i}-Y_{n-1,i}-\sum_{j=1}^p\hat{\Psi}_j(t_n)^\top\theta^\ell_{ij}\right)^2.
\]

By the triangle inequality, we have
\begin{align*}
    \abr{\hat{\sigma}^2-{\sigma}^{\star2}} &\leq  \abr{\frac{1}{2Np}\sum_{n=1}^N\sum_{\ell=1}^k \hat{w}_{\ell,\Theta}(t_n)\rbr{\hat{R}_{n,\ell}-R_{n,\ell}}}\\
    &\quad 
    +\abr{\frac{1}{2Np}\sum_{n=1}^N\sum_{\ell=1}^k \cbr{\hat{w}_{\ell,\Theta}(t_n)R_{n,\ell}-
    \EE\sbr{\hat{w}_{\ell,\Theta}(t_n)R_{n,\ell}}}}\\
    &\quad+ 
    \abr{\frac{1}{2Np}\sum_{n=1}^N\sum_{\ell=1}^k \cbr{\EE\sbr{\hat{w}_{\ell,\Theta}(t_n)R_{n,\ell}}-
     \EE\sbr{w_{\ell,\Theta}(t_n)R_{n,\ell}}}}\\
     &\quad  +\abr{\frac{1}{2Np}\sum_{n=1}^N\sum_{\ell=1}^k
     \rbr{\EE\sbr{w_{\ell,\Theta}(t_n)R_{n,\ell}}}
     -
     \sigma^{\star 2}
     }\\
     &= T_1 + T_2+T_3 + T_4. 
\end{align*}
We bound the four terms in turn. For $n\in[N]$ and $\ell\in[k]$, define the vector of fitted-value differences
\[
d_{n,\ell}:=\left(\sum_{j=1}^p\rbr{\hat{\Psi}_j(t_n)-\Psi_j(t_n)}^\top\theta_{ij}^\ell\right)_{i=1}^p.
\]
Let $v_{n,\ell}:=\bigl(Y_{n,i}^\Delta-\sum_{j=1}^p\Psi_j(t_n)^\top\theta_{ij}^\ell\bigr)_{i=1}^p$ denote the residual vector.
We work on the event $\mathcal E_R\cap\mathcal E_{\mathrm{rec}}$. Since
$\hat R_{n,\ell}=\|v_{n,\ell}-d_{n,\ell}\|_2^2$ and
$R_{n,\ell}=\|v_{n,\ell}\|_2^2$, we have
\[
|\hat R_{n,\ell}-R_{n,\ell}|
\leq 2\sqrt{R_{n,\ell}}\,\|d_{n,\ell}\|_2+\|d_{n,\ell}\|_2^2.
\]

By H\"older's inequality, we have
\begin{align}
\bignorm{d_{n,\ell}}_2
&=
\sqrt{
\sum_{i'=1}^p\cbr{\sum_{j=1}^p\rbr{
\hat{\Psi}_{j}(t_n)
-
\Psi_{j}(t_n)
}^\top \theta_{i'j}^\ell}^2
}\notag\\
&\leq \sqrt{\sum_{i'}\norm{\theta_{i'}^\ell}_1^2}\norm{\Psi(t_n)-\hat{\Psi}(t_n)}_\infty,\label{eq:tmp100}
\end{align}
where $\theta_{i'}^\ell=(\theta_{i'1}^{\ell\top},\ldots,\theta_{i'p}^{\ell\top})^\top\in\RR^{pm}$.

The mean value theorem with the bound on $g_j'$, the definition of $\mathcal E_{\mathrm{rec}}$ and the Cauchy--Schwarz inequality give
\[
\norm{\Psi(t_n)-\hat{\Psi}(t_n)}_\infty
\leq D\sqrt{h}\,\delta_1.
\]
Therefore, \eqref{eq:tmp100}, the preceding reconstruction bound, the weighted Cauchy--Schwarz inequality and the
definition of $\mathcal E_R$ give
\begin{align*}
T_1
&\leq\sum_{\ell=1}^k
\left\{\frac1{Np}\sum_{n=1}^N
\hat w_{\ell,\Theta}(t_n)R_{n,\ell}\right\}^{1/2}
\left\{\frac1{Np}\sum_{n=1}^N
\hat w_{\ell,\Theta}(t_n)\|d_{n,\ell}\|_2^2\right\}^{1/2}\\
&\quad+\frac1{2Np}\sum_{n=1}^N\sum_{\ell=1}^k
\hat w_{\ell,\Theta}(t_n)\|d_{n,\ell}\|_2^2\\
&\leq \sqrt{2}\,kD\sqrt h\,\delta_1
\max_{\ell\in[k]}\left\{\frac1p\sum_{i=1}^p\norm{\theta_i^\ell}_1^2\right\}^{1/2}
\sqrt{\mathcal R_\Theta}\\
&\quad+\frac{k}{2}D^2h\delta_1^2
\max_{\ell\in[k]}\left\{\frac1p\sum_{i=1}^p\norm{\theta_i^\ell}_1^2\right\}.
\end{align*}

By the definition of $\delta_2$, we have $T_2=\delta_2$. Next, the second bound of Lemma~\ref{lemma:truncated_smoooth} gives
\begin{align*}
    T_3 &= \abr{\frac{1}{2Np}\sum_{n=1}^N\sum_{\ell=1}^k\EE\sbr{\cbr{\hat{w}_{\ell,\Theta}(t_n)-w_{\ell,\Theta}(t_n)} R_{n,\ell}}}\\
    &\leq5k^2\delta_{\min}^{-1}\mathcal R_\Theta
    \zeta^{-8}\pi_{\min}^{-2}\cbr{1-(\zeta\pi_{\min})^2}^{r-1}.
\end{align*}
Recall that $\check{\sigma}^2=M_{\sigma^2}(\Theta)$, so that
\[
\check{\sigma}^2=\frac{1}{2Np}\sum_{n=1}^N\sum_{\ell=1}^k\EE[w_{\ell,\Theta}(t_n)R_{n,\ell}].
\]
Therefore, $T_4 = \abr{\check{\sigma}^2-\sigma^{\star2}}$.
Combining the results from $T_1$ to $T_4$, we have
\begin{multline*}
\abr{\hat{\sigma}^2-\sigma^{\star 2}}
\leq \sqrt{2}\,kD\sqrt h\,\delta_1
\max_{\ell\in[k]}\left\{\frac1p\sum_{i=1}^p\norm{\theta_i^\ell}_1^2\right\}^{1/2}
\sqrt{\mathcal R_\Theta}\\
\quad+\frac{k}{2}D^2h\delta_1^2
\max_{\ell\in[k]}\left\{\frac1p\sum_{i=1}^p\norm{\theta_i^\ell}_1^2\right\}
+\delta_2\\
+5k^2\delta_{\min}^{-1}\mathcal R_\Theta
\zeta^{-8}\pi_{\min}^{-2}\cbr{1-(\zeta\pi_{\min})^2}^{r-1}
+\abr{\check{\sigma}^2-\sigma^{\star2}}.
\end{multline*}
\end{proof}
\subsection{One-step update of \texorpdfstring{$Q$}{Q}}\label{ssec:onestep_Q}

The transition-rate update uses the truncated posterior expected numbers of transitions and occupation times defined in Section~\ref{ssec:approximateEM}. The following lemma bounds its error by the corresponding population-update discrepancy and the statistical and truncation terms.

\begin{lemma}\label{lemma:onestep:q}
Suppose that Assumption~\ref{assumption:mixing} holds, and let $\Theta\in\Omega$.
Recall that $\zeta$, $\delta_{\min}$ and $\pi_{\min}$ are defined in
Assumption~\ref{assumption:mixing}. Let
\[
\check Q=M_Q(\Theta),
\qquad
\hat Q=\tilde M_{N,Q}(\Theta),
\]
let $\hat w_{ab,\Theta}(t_n)$ be the truncated joint smoothing probabilities of Section~\ref{ssec:approximateEM}, and define
\[
\delta_{3,1}:=\frac{1}{N}\sum_{a,b\in[k]}
\left|\sum_{n=1}^{N}\left\{
\hat{w}_{ab,\Theta}(t_n)-\EE[\hat{w}_{ab,\Theta}(t_n)]
\right\}\right|.
\]
Suppose that
\[
\underline\tau_\Theta:=\min_{\ell\in[k]}\frac{1}{N}
\EE\sbr{\tau_\ell(Y_0^N;\Theta)}>0.
\]
There exist constants $C_{\mathrm{tr}},C_Q<\infty$, where $C_{\mathrm{tr}}$ depends only on $h$, and $C_Q$ depends only on $h$, $\delta_{\min}$, $\underline\tau_\Theta$ and $\max_{a\neq b}q^\star_{ab}+r_0$, such that the following holds. If $r$ is large enough that
\[
C_{\mathrm{tr}}k^3\delta_{\min}^{-1}\zeta^{-8}\pi_{\min}^{-2}
\cbr{1-(\zeta\pi_{\min})^2}^{r-1}
\leq\frac{\underline\tau_\Theta}{2},
\]
then, on the event $\delta_{3,1}\leq\underline\tau_\Theta/(4h)$,
\[
\sum_{\ell\neq\ell'}
\left|\hat q_{\ell\ell'}-q_{\ell\ell'}^\star\right|
\leq
\sum_{\ell\neq\ell'}
\left|\check q_{\ell\ell'}-q_{\ell\ell'}^\star\right|
+C_Qk(k-1)\left[
\delta_{3,1}+
k^3\delta_{\min}^{-1}\zeta^{-8}\pi_{\min}^{-2}
\cbr{1-(\zeta\pi_{\min})^2}^{r-1}
\right].
\]
\end{lemma}
\begin{proof}[Proof of Lemma~\ref{lemma:onestep:q}]
For $a,b\in[k]$, let
\[
w_{ab,\Theta}(t_n)
:=P\bigl(Z(t_{n-1})=a,Z(t_n)=b\mid Y_0^N;\Theta\bigr)
\]
be the joint smoothing probabilities, whose truncated versions are $\hat w_{ab,\Theta}(t_n)$. By~\eqref{eq:computemll}, and following~\citet{bladt2005statistical, liu2015efficient}, the posterior expected number of transitions from state $\ell$ to state $\ell'$ is
\[
m_{\ell\ell'}(Y_0^N;\Theta)
=\sum_{n=1}^{N}\sum_{a,b}{w}_{ab,\Theta}(t_n)\EE\sbr{m_{\ell\ell'}(h)\mid Z(0)=a,Z(h)=b;Q},
\]
where $m_{\ell\ell'}(t)$ is the number of transitions from state $\ell$ to $\ell'$ in the time interval $[0,t]$. Replacing $w_{ab,\Theta}(t_n)$ by $\hat w_{ab,\Theta}(t_n)$ in this sum gives the truncated posterior expected number of transitions $\hat m_{\ell\ell'}(Y_0^N;\Theta)$ of Section~\ref{ssec:approximateEM}. The full and truncated posterior expected occupation times $\tau_\ell(Y_0^N;\Theta)$ and $\hat\tau_\ell(Y_0^N;\Theta)$ are given by the same two sums with $\tau_\ell(h)$ in place of $m_{\ell\ell'}(h)$, as in~\eqref{eq:computetaul}.

The truncated transition-rate update of Section~\ref{ssec:approximateEM} maximizes $\sum_{\ell\neq\ell'}\hat m_{\ell\ell'}(Y_0^N;\Theta)\log\tilde q_{\ell\ell'}+\sum_{\ell=1}^k\tilde q_\ell\hat\tau_\ell(Y_0^N;\Theta)$ over all transition rate matrices $\tilde Q$. Its maximizer $\hat{Q}$ is given by
\[
\hat{q}_{\ell\ell'}=\frac{\hat m_{\ell\ell'}(Y_0^N;\Theta)}{\hat\tau_\ell(Y_0^N;\Theta)}
\]
for $\ell\neq \ell'$ and $\hat{q}_{\ell\ell}=-\sum_{\ell'\neq \ell}\hat{q}_{\ell\ell'}$.
For brevity, we write $m_{\ell\ell'}:=m_{\ell\ell'}(Y_0^N;\Theta)$, $\tau_\ell:=\tau_\ell(Y_0^N;\Theta)$, $\hat m_{\ell\ell'}:=\hat m_{\ell\ell'}(Y_0^N;\Theta)$ and $\hat\tau_\ell:=\hat\tau_\ell(Y_0^N;\Theta)$.
By the triangle inequality, for each $\ell\neq \ell'$, we have
\begin{align*}
    \abr{\hat{q}_{\ell\ell'}-q_{\ell\ell'}^\star}
    &\leq \abr{\frac{\hat m_{\ell\ell'}}{\hat\tau_\ell}
    -\frac{\EE[\hat m_{\ell\ell'}]}{\EE[\hat\tau_\ell]}}\\
    &\quad +
    \abr{\frac{\EE[\hat m_{\ell\ell'}]}{\EE[\hat\tau_\ell]} -
    \frac{\EE[m_{\ell\ell'}]}{\EE[\tau_\ell]}}\\
    &\quad +
\abr{\frac{\EE[m_{\ell\ell'}]}{\EE[\tau_\ell]} -
q_{\ell\ell'}^\star}\\
    &= T_5 + T_6 + T_7.
\end{align*}

We first bound $T_6$. Set
\[
\eta_r:=\sup_{n\in[N],\,a,b\in[k]}
\EE\sbr{\left|\hat{w}_{ab,\Theta}(t_n)-w_{ab,\Theta}(t_n)\right|}.
\]
Lemma~\ref{lemma:truncated_smoooth} bounds the truncation error for the pair $(Z(t_n),Z(t_{n+1}))$ and the window $Y_{(n-r)\vee0}^{(n+r)\wedge N}$. Its proof applies verbatim to the pair $(Z(t_{n-1}),Z(t_n))$ and the same window, which extends $r-1$ sampling intervals before $t_{n-1}$ and $r$ after $t_n$, since every geometric factor in that proof still has base at most $1-(\zeta\pi_{\min})^2$ and exponent at least $r-1$. Therefore,
\[
\eta_r\leq C_{\mathrm{sm}}k\delta_{\min}^{-1}\zeta^{-8}\pi_{\min}^{-2}
\cbr{1-(\zeta\pi_{\min})^2}^{r-1}
\]
for an absolute constant $C_{\mathrm{sm}}$, which we take to be at least $10$. The bounds $\epsilon_r$, $\eta_r$ and $\bar\eta_r$ are related but distinct: $\epsilon_r$ bounds differences of the smoothing probabilities, $\eta_r$ bounds an expected difference of joint smoothing probabilities, and $\bar\eta_r$ is a realization-wise bound for both kinds of differences.  For the fixed transition
$\ell\to\ell'$, let
\[
\begin{aligned}
C_{m,\ell\ell'}&:=\sup_{Q\in\Omega_Q}\max_{a,b\in[k]}
\EE\sbr{m_{\ell\ell'}(h)\mid Z(0)=a,Z(h)=b;Q}.
\end{aligned}
\]
It is finite because $\Omega_Q$ is bounded, and we derive an explicit bound below. Since
$0\leq\tau_\ell(h)\leq h$, we take $C_{\mathrm{tr}}:=C_{\mathrm{sm}}h$.

By the definition of $\underline\tau_\Theta$, we have $\EE[\tau_\ell]\geq N\underline\tau_\Theta$.
For fixed $\ell\neq\ell'$, we now bound $C_{m,\ell\ell'}$. Let $q_{\max}:=\sup_{Q\in\Omega_Q}\max_{a\neq b}q_{ab}$, which is at most $\max_{a\neq b}q^\star_{ab}+r_0$ by the definition of $\Omega_Q$. By \eqref{eq:decomposem} and Assumption~\ref{assumption:mixing},
\[
\begin{aligned}
\EE\!\left[
m_{\ell\ell'}(h)
\mid Z(0)=a,Z(h)=b;Q
\right]
&=
\frac{q_{\ell\ell'}}{P_{ab}(h)}
\int_0^h
P_{a\ell}(u)P_{\ell'b}(h-u)\,\mathrm du\\
&\leq
\frac{q_{\max}}{\delta_{\min}}\int_0^h 1\,\mathrm du\\
&=
\frac{h q_{\max}}{\delta_{\min}}.
\end{aligned}
\]
Therefore,
\[
C_{m,\ell\ell'}
\leq\frac{h q_{\max}}{\delta_{\min}}<\infty.
\]
Since the joint smoothing probabilities sum to one, \eqref{eq:computemll} yields
\[
\EE\sbr{m_{\ell\ell'}}
\leq NC_{m,\ell\ell'}.
\]
Using \eqref{eq:computemll} again, we have
\[
\begin{aligned}
&\left|\EE\sbr{\hat m_{\ell\ell'}}
-\EE\sbr{m_{\ell\ell'}}\right|\\
&\quad\leq\sum_{n=1}^N\sum_{a,b=1}^k
\EE\sbr{\left|\hat{w}_{ab,\Theta}(t_n)-w_{ab,\Theta}(t_n)\right|}
\EE\sbr{m_{\ell\ell'}(h)\mid Z(0)=a,Z(h)=b;Q}\\
&\quad\leq Nk^2C_{m,\ell\ell'}\eta_r.
\end{aligned}
\]
Similarly, \eqref{eq:computetaul} gives
\[
\left|\EE\sbr{\hat\tau_\ell}
-\EE\sbr{\tau_\ell}\right|
\leq Nk^2h\eta_r.
\]
Combining these bounds with the truncation condition in the statement, we obtain
\[
\begin{aligned}
\EE\sbr{\hat\tau_\ell}
&=\EE\sbr{\tau_\ell}\\
&\quad+\sum_{n=1}^N\sum_{a,b=1}^k
\EE\sbr{\hat{w}_{ab,\Theta}(t_n)-w_{ab,\Theta}(t_n)}
\EE\sbr{\tau_\ell(h)\mid Z(0)=a,Z(h)=b;Q}\\
&\geq N\underline\tau_\Theta-Nk^2h\eta_r\\
&\geq N\underline\tau_\Theta
-NC_{\mathrm{tr}}k^3\delta_{\min}^{-1}\zeta^{-8}\pi_{\min}^{-2}
\cbr{1-(\zeta\pi_{\min})^2}^{r-1}\\
&\geq\frac{N\underline\tau_\Theta}{2}.
\end{aligned}
\]
Therefore,
\[
\begin{aligned}
T_6
&\leq
\frac{\left|\EE\sbr{\hat m_{\ell\ell'}}
-\EE\sbr{m_{\ell\ell'}}\right|}
{\EE\sbr{\hat\tau_\ell}}\\
&\quad+
\frac{\EE\sbr{m_{\ell\ell'}}
\left|\EE\sbr{\hat\tau_\ell}
-\EE\sbr{\tau_\ell}\right|}
{\EE\sbr{\hat\tau_\ell}
\EE\sbr{\tau_\ell}}\\
&\leq
k^2C_{m,\ell\ell'}\left(\frac{2}{\underline\tau_\Theta}
+\frac{2h}{\underline\tau_\Theta^2}\right)\eta_r\\
&\leq C_{Q,\ell\ell'}k^3\delta_{\min}^{-1}\zeta^{-8}\pi_{\min}^{-2}
\cbr{1-(\zeta\pi_{\min})^2}^{r-1},
\end{aligned}
\]
where $C_{Q,\ell\ell'}:=C_{\mathrm{sm}}C_{m,\ell\ell'}(2/\underline\tau_\Theta+2h/\underline\tau_\Theta^2)$.

We next bound $T_5$. By the representation of
$\hat m_{\ell\ell'}$ above,
\[
\begin{aligned}
&\left|
\hat m_{\ell\ell'}
-\EE\sbr{\hat m_{\ell\ell'}}
\right|\\
&\quad\leq
\sum_{a,b\in[k]}
\EE\sbr{m_{\ell\ell'}(h)\mid Z(0)=a,Z(h)=b;Q}
\left|
\sum_{n=1}^N\left\{
\hat w_{ab,\Theta}(t_n)-\EE[\hat w_{ab,\Theta}(t_n)]
\right\}
\right|\\
&\quad\leq NC_{m,\ell\ell'}\delta_{3,1}.
\end{aligned}
\]
Similarly, since $0\leq\tau_\ell(h)\leq h$,
\[
\begin{aligned}
&\left|
\hat\tau_\ell
-\EE\sbr{\hat\tau_\ell}
\right|\\
&\quad\leq
h\sum_{a,b\in[k]}
\left|
\sum_{n=1}^N\left\{
\hat w_{ab,\Theta}(t_n)-\EE[\hat w_{ab,\Theta}(t_n)]
\right\}
\right|\\
&\quad=Nh\delta_{3,1}.
\end{aligned}
\]
The lower bound established above for $T_6$ and the condition on
$\delta_{3,1}$ in the statement imply
\[
\begin{aligned}
\hat\tau_\ell
&\geq
\EE\sbr{\hat\tau_\ell}
-Nh\delta_{3,1}\\
&\geq
\frac{N\underline\tau_\Theta}{2}
-Nh\delta_{3,1}\\
&\geq
\frac{N\underline\tau_\Theta}{4}.
\end{aligned}
\]
Furthermore, since the truncated joint smoothing probabilities sum to one,
\[
\begin{aligned}
\EE\sbr{\hat m_{\ell\ell'}}
&\leq
\sum_{n=1}^N
\EE\sbr{\sum_{a,b\in[k]}
\hat w_{ab,\Theta}(t_n)C_{m,\ell\ell'}}\\
&=NC_{m,\ell\ell'}.
\end{aligned}
\]
Using the same ratio decomposition as above,
\[
\begin{aligned}
T_5
&\leq
\frac{\left|
\hat m_{\ell\ell'}
-\EE\sbr{\hat m_{\ell\ell'}}
\right|}
{\hat\tau_\ell}\\
&\quad+
\frac{\EE\sbr{\hat m_{\ell\ell'}}
\left|
\hat\tau_\ell
-\EE\sbr{\hat\tau_\ell}
\right|}
{\hat\tau_\ell
\EE\sbr{\hat\tau_\ell}}\\
&\leq
\frac{NC_{m,\ell\ell'}\delta_{3,1}}
{N\underline\tau_\Theta/4}
+
\frac{NC_{m,\ell\ell'}\cdot Nh\delta_{3,1}}
{(N\underline\tau_\Theta/4)(N\underline\tau_\Theta/2)}\\
&=
C_{m,\ell\ell'}\left(
\frac{4}{\underline\tau_\Theta}
+\frac{8h}{\underline\tau_\Theta^2}
\right)\delta_{3,1}.
\end{aligned}
\]

Finally, recall that $\check Q=M_Q(\Theta)$, so that
\[
\check{q}_{\ell\ell'} = \frac{\EE\sbr{m_{\ell\ell'}}}{\EE\sbr{\tau_\ell}}.
\]
Therefore, $T_7=\abr{\check{q}_{\ell\ell'}-{q}_{\ell\ell'}^\star}$. Combining the bounds on
$T_5$ to $T_7$ gives, for every $\ell\neq\ell'$,
\[
\left|\hat q_{\ell\ell'}-q_{\ell\ell'}^\star\right|
\leq
C_{m,\ell\ell'}\left(
\frac{4}{\underline\tau_\Theta}
+\frac{8h}{\underline\tau_\Theta^2}
\right)\delta_{3,1}
+
C_{Q,\ell\ell'}k^3\delta_{\min}^{-1}\zeta^{-8}\pi_{\min}^{-2}
\cbr{1-(\zeta\pi_{\min})^2}^{r-1}+
\left|\check q_{\ell\ell'}-q_{\ell\ell'}^\star\right|.
\]
Since $C_{m,\ell\ell'}\leq hq_{\max}/\delta_{\min}$, both $C_{Q,\ell\ell'}$ and $C_{m,\ell\ell'}(4/\underline\tau_\Theta+8h/\underline\tau_\Theta^2)$ are at most
\[
C_Q:=\frac{hq_{\max}}{\delta_{\min}}\max\{C_{\mathrm{sm}},1\}
\left(
\frac{4}{\underline\tau_\Theta}
+\frac{8h}{\underline\tau_\Theta^2}
\right),
\]
which depends only on $h$, $\delta_{\min}$, $\underline\tau_\Theta$ and $q_{\max}\leq\max_{a\neq b}q^\star_{ab}+r_0$. Summing the preceding bound over $\ell\neq\ell'$ proves the claim.
\end{proof}

\subsection{Auxiliary bounds}\label{ssec:aux:onestep}

We now control the quantities entering the one-step bounds. Lemma~\ref{lemma:deltaW} bounds the truncation term $\Delta_w$, and Lemmas~\ref{lemma:deltaPsi} and~\ref{lemma:deltaX} bound the reconstruction terms $\Delta_\Psi^\top\theta^\star$ and $\Delta_X$. The locality results below provide the observation-window representation needed for concentration. Lemmas~\ref{lemma:deviation} and~\ref{lemma:RE:condition} subsequently establish the concentration bound for $\Delta_\varepsilon$ and the restricted eigenvalue condition at a fixed parameter. These results are assembled in Lemma~\ref{lemma:oneupdate:theta:2} and used, under the uniform concentration assumption, in the proof of Theorem~\ref{theorem:graphrecovery}.

The following lemma bounds the term of the deviation bound condition that comes from truncating the smoothing probabilities.
\begin{lemma}\label{lemma:deltaW}
Fix $i\in[p]$ and $\ell\in[k]$, let $\Theta\in\Omega$, and define $Y_n^{\Delta}=Y_{n,i}-Y_{n-1,i}$ for $n\in[N]$. Suppose that Assumptions~\ref{assumption:edgeset} and~\ref{assumption:mixing} hold, that, almost surely, $|g_{q}(X_{i'}(t))|\leq B$ for all $i'\in[p]$, $q\in[m]$ and $t\in[0,T]$ and $\max_{n\in[N]}|\rho_{n,i}|\leq\rho$, and that $\hat K_{\Psi_j}\succ0$ for every $j\in[p]$. Define the coefficient
\[
C = \rbr{\max_j\sigma_{\min}^{-1/2}(\hat{K}_{\Psi_j})}Bh
\left\{\frac{2\sigma^\star}{\sqrt{\pi}}+\rho+
Bh\max_{a,b\in[k]}\norm{\theta_i^{a\star}-\theta_i^{b\star}}_1\right\}
\frac{10}{\delta_{\min}}\zeta^{-8}\pi_{\min}^{-2},
\]
where $\zeta$ is the transition-comparability constant defined in~\eqref{eq:mixing}, not a $\beta$-mixing coefficient, and $\delta_{\min}$ and $\pi_{\min}$ are defined in Assumption~\ref{assumption:mixing}; in $2\sigma^\star/\sqrt{\pi}$, $\pi=3.14159\ldots$.
    Then, we have
    \begin{multline*}
    \bignorm{\frac{1}{N}\sum_{n=1}^N\EE\sbr{
    \cbr{\hat{w}_{\ell,\Theta}(t_n)-w_{\ell,\Theta}(t_n)}\cbr{{Y}_n^{\Delta}-\Psi(t_n)^\top \theta^{\ell\star }_i}\Psi(t_n)}
    }_{\infty,\hat{\Kcal}_\Psi^*}\\\leq C\sqrt{m}k
        \cbr{1-(\zeta\pi_{\min})^2}^{r-1}. 
    \end{multline*}
\end{lemma}
\begin{proof}[Proof of Lemma~\ref{lemma:deltaW}.]
    Write
    \begin{align*}
\norm{\Delta_w}_{\infty,\hat{\Kcal}_\Psi^*}&=\bignorm{\frac{1}{N}\sum_{n=1}^N\EE\sbr{
    \cbr{\hat{w}_{\ell,\Theta}(t_n)-w_{\ell,\Theta}(t_n)}\cbr{{Y}_n^{\Delta}-\Psi(t_n)^\top\theta^{\ell\star }_i}\Psi(t_n)}
    }_{\infty,\hat{\Kcal}_\Psi^*}\\
    &\leq\max_{j\in[p]}\cbr{\sigma_{\min}^{-1/2}(\hat{K}_{\Psi_j})\norm{(\Delta_w)_j}_2},
    \end{align*}
    where $(\Delta_w)_j$ is the $j$th block of $\Delta_w$, and we used $\norm{v}_{\hat{K}_{\Psi_j}^*}\leq\sigma_{\min}^{-1/2}(\hat{K}_{\Psi_j})\norm{v}_2$ for $v\in\RR^m$, with $\hat{K}_{\Psi_j}$ at its realized value.
    By the triangle inequality, Jensen's inequality and $\norm{\Psi_j(t_n)}_2\leq \sqrt{m}Bh$, we obtain
    \[
    \norm{\Delta_w}_{\infty,\hat{\Kcal}_\Psi^*}\leq \rbr{\max_j\sigma_{\min}^{-1/2}(\hat{K}_{\Psi_j})}\frac{Bh\sqrt{m}}{N}\sum_{n=1}^N\EE\sbr{\abr{\hat{w}_{\ell,\Theta}(t_n)-w_{\ell,\Theta}(t_n)}\abr{{Y}_n^{\Delta}-\Psi(t_n)^\top\theta^{\ell\star }_i}}.
    \]
    By the observation model and the definition of $\rho_{n,i}$, we have
    \[
    Y_n^\Delta-\Psi(t_n)^\top\theta_i^{\ell\star}
    =\sigma^\star(\varepsilon_{n,i}-\varepsilon_{n-1,i})+\rho_{n,i}
    +\Psi(t_n)^\top\left\{\theta_i^{Z(t_n)\star}-\theta_i^{\ell\star}\right\},
    \]
    and therefore
    \[
    \EE\sbr{\abr{{Y}_n^{\Delta}-\Psi(t_n)^\top\theta^{\ell\star }_i}}
    \leq \frac{2\sigma^\star}{\sqrt{\pi}}+\rho+
    Bh\max_{a,b\in[k]}\norm{\theta_i^{a\star}-\theta_i^{b\star}}_1.
    \]
    Combining the preceding displays with the second claim of Lemma~\ref{lemma:truncated_smoooth}, we obtain
    \begin{align*}
        \norm{\Delta_w}_{\infty,\hat{\Kcal}_\Psi^*}
        &\leq \left\{\max_j\sigma_{\min}^{-1/2}(\hat{K}_{\Psi_j})\right\}Bh\sqrt{m}\\
        &\quad\times\left\{\frac{2\sigma^\star}{\sqrt{\pi}}+\rho+
        Bh\max_{a,b\in[k]}\norm{\theta_i^{a\star}-\theta_i^{b\star}}_1\right\}\\
        &\quad\times k\frac{10}{\delta_{\min}}\zeta^{-8}\pi_{\min}^{-2}
        \cbr{1-(\zeta\pi_{\min})^2}^{r-1}.
    \end{align*}
\end{proof}

The next two lemmas bound the two terms of the deviation bound condition that come from the reconstruction error.
\begin{lemma}\label{lemma:deltaPsi}
    Let $X(t)$ and $\hat{X}(t)$ be two stochastic processes on $[0,T]$ with values in $[-R_X,R_X]^p$. Suppose that $\max_{i\in[p]}\max_{n\in[N]}\opnorm{X_i-\hat{X}_i}{[t_{n-1},t_n]}\leq\delta_1$,
    \[
    \sup_{|x|\leq R_X}\max_{q\in[m]}|g_q(x)|\leq B
    \quad\text{and}\quad
    \sup_{|x|\leq R_X}\max_{q\in[m]}|g'_q(x)|\leq D.
    \]
    Assume that $\hat K_{\Psi_j}\succ0$ for every $j\in[p]$ and that $\theta^\star \in\RR^{pm}$ is $ms$-sparse with $s\leq p$.
    Then, for every $\Theta\in\Omega$ and $\ell\in[k]$,
    \begin{multline*}
    \bignorm{ \frac{1}{N}\sum_{n=1}^N\hat{w}_{\ell,\Theta}(t_n)\{{\Psi}(t_n){\Psi}(t_n)^\top-\hat{\Psi}(t_n)\hat{\Psi}(t_n)^\top\}\theta^\star}_{\infty,\hat{\Kcal}_\Psi^*}\\\leq2\max_j\sigma_{\min}^{-1/2}(\hat{K}_{\Psi_j})\norm{\theta^\star }_2BD\delta_1 mh^{3/2}s^{1/2}.
    \end{multline*}
\end{lemma}
\begin{proof}[Proof of Lemma~\ref{lemma:deltaPsi}.]
    Write
    \begin{align}
\bignorm{\frac{1}{N}\sum_{n=1}^N&\hat{w}_{\ell,\Theta}(t_n)\{{\Psi}(t_n){\Psi}(t_n)^\top-\hat{\Psi}(t_n)\hat{\Psi}(t_n)^\top\}\theta^\star }_{\infty,\hat{\Kcal}_\Psi^*}\notag\\
&=\max_j\bignorm{\frac{1}{N}
         \sum_{n=1}^N\hat{w}_{\ell,\Theta}(t_n)\{{\Psi}_j(t_n){\Psi}(t_n)^\top-\hat{\Psi}_j(t_n)\hat{\Psi}(t_n)^\top\}
       \theta^\star }_{\hat{K}_{\Psi_j}^*}\notag\\
        &\leq\max_j\sigma_{\min}^{-1/2}(\hat{K}_{\Psi_j})\bignorm{\frac{1}{N}
        \sum_{n=1}^N\hat{w}_{\ell,\Theta}(t_n)\{{\Psi}_j(t_n){\Psi}(t_n)^\top-\hat{\Psi}_j(t_n)\hat{\Psi}(t_n)^\top\}\theta^\star
        }_{2}\notag\\
        &\leq \max_j\sigma_{\min}^{-1/2}(\hat{K}_{\Psi_j})\underbrace{\bignorm{
        \frac{1}{N}\sum_{n=1}^N\hat{w}_{\ell,\Theta}(t_n)\{{\Psi}(t_n)-\hat{\Psi}(t_n)\}^\top
        \theta^\star{\Psi}_j(t_n) }_{2}}_{T_{1,j}}\notag\\
        &\quad 
        +\max_j\sigma_{\min}^{-1/2}(\hat{K}_{\Psi_j})\underbrace{\bignorm{
        \frac{1}{N}\sum_{n=1}^N\hat{w}_{\ell,\Theta}(t_n)\hat{\Psi}(t_n)^\top\theta^\star({\Psi}_j(t_n)-\hat{\Psi}_j(t_n))
         }_{2}}_{T_{2,j}}.\label{eq:tmp0}
    \end{align}
    By the triangle inequality, we have
    \begin{align}
     T_{1,j}&=\bignorm{\frac{1}{N}
         \sum_{n=1}^N\hat{w}_{\ell,\Theta}(t_n)\{{\Psi}(t_n)-\hat{\Psi}(t_n)\}^\top\theta^\star{\Psi}_j(t_n)
        }_{2}\notag\\&\leq
        \frac{1}{N}\sum_{n=1}^N\hat{w}_{\ell,\Theta}(t_n)
        \abr{\theta^{\star\top}\{{\Psi}(t_n)-\hat{\Psi}(t_n)\}}\norm{{\Psi}_j(t_n)}_{2}\notag\\
        &\leq \frac{Bh\sqrt{m}}{N}\sum_{n=1}^N\hat{w}_{\ell,\Theta}(t_n)
        \abr{\theta^{\star\top}\{{\Psi}(t_n)-\hat{\Psi}(t_n)\}}.\label{eq:tmp1}
    \end{align}
    Let $\Scal$ denote the support of $\theta^\star$, so that $|\Scal|\leq ms$. By the Cauchy--Schwarz inequality and $0\leq\hat{w}_{\ell,\Theta}(t_n)\leq1$, we have
   \begin{align}
       &\frac{Bh\sqrt{m}}{N}\sum_{n=1}^N\hat{w}_{\ell,\Theta}(t_n)
       \abr{\theta_\Scal^{\star\top}[{\Psi}(t_n)-\hat{\Psi}(t_n)]_{\Scal}}\notag\\
       &\quad\leq \frac{Bh\sqrt{m}}{N}\sum_{n=1}^N\hat{w}_{\ell,\Theta}(t_n)
       \norm{\theta_\Scal^\star}_2\norm{[{\Psi}(t_n)-\hat{\Psi}(t_n)]_{\Scal}}_2\notag\\
       &\quad\leq Bhm\sqrt{s}\norm{\theta^\star}_2
       \max_{n\in[N]}\norm{{\Psi}(t_n)-\hat{\Psi}(t_n)}_\infty.\label{eq:tmp2}
   \end{align}
For each $i\in[p]$, $q\in[m]$, and $n\in[N]$, we have
\begin{align}
\abr{{\Psi}_{iq}(t_n)-\hat{\Psi}_{iq}(t_n)}
&\leq\int_{t_{n-1}}^{t_n}\abr{g_q(X_i(u))-g_q(\hat{X}_i(u))}\mathrm{d}u\notag\\
&\leq D\sqrt{h}\,\opnorm{X_i-\hat{X}_i}{[t_{n-1},t_n]}
\leq D\sqrt{h}\delta_1.
\label{eq:tmp3}
\end{align}
Substituting~\eqref{eq:tmp3} into~\eqref{eq:tmp2}, and the result into~\eqref{eq:tmp1}, we obtain
\begin{equation}\label{eq:T1}
T_{1,j}\leq BD\delta_1 m\sqrt{s}\,h^{3/2}\norm{\theta^\star}_2.
\end{equation}
Similarly, since $\norm{\hat{\Psi}(t_n)}_\infty\leq Bh$, we have
\begin{align}
    T_{2,j}&=\bignorm{
        \frac{1}{N}\sum_{n=1}^N\hat{w}_{\ell,\Theta}(t_n)\cbr{\theta^{\star\top}\hat{\Psi}(t_n)}\cbr{{\Psi}_j(t_n)-\hat{\Psi}_j(t_n)}
        }_{2}\notag\\
       &\leq\frac{1}{N}\sum_{n=1}^N\hat{w}_{\ell,\Theta}(t_n)\abr{\theta^{\star\top}_\Scal[\hat{\Psi}(t_n)]_{\Scal}}\norm{{\Psi}_j(t_n)-\hat{\Psi}_j(t_n)}_2 \notag\\
       &\leq \norm{\theta^{\star}}_2Bh\sqrt{ms}\max_{n\in[N]}\norm{{\Psi}_j(t_n)-\hat{\Psi}_j(t_n)}_2\notag\\
       &\leq \norm{\theta^{\star}}_2Bhm\sqrt{s}\max_{n\in[N]}\norm{{\Psi}_j(t_n)-\hat{\Psi}_j(t_n)}_\infty.\notag
       \intertext{Substituting~\eqref{eq:tmp3} into the last line, we arrive at}
       &\leq \norm{\theta^{\star}}_2BmD\delta_1\sqrt{s}\,h^{3/2}.\label{eq:T2}
\end{align}
The bounds~\eqref{eq:T1} and~\eqref{eq:T2} hold for every $j\in[p]$. Substituting them into~\eqref{eq:tmp0}, we obtain
\begin{multline*}
\bignorm{ \frac{1}{N}\sum_{n=1}^N\hat{w}_{\ell,\Theta}(t_n)\{{\Psi}(t_n){\Psi}(t_n)^\top-\hat{\Psi}(t_n)\hat{\Psi}(t_n)^\top\}\theta^\star}_{\infty,\hat{\Kcal}_\Psi^*}\\\leq 2\max_j\sigma_{\min}^{-1/2}(\hat{K}_{\Psi_j})\norm{\theta^\star}_2BD\delta_1 m\sqrt{s}\,h^{3/2}.
\end{multline*}
\end{proof}

\begin{lemma}\label{lemma:deltaX}
Fix $i\in[p]$ and $\ell\in[k]$. Suppose that $|X_j(t)|\leq R_X$ and $|\hat X_j(t)|\leq R_X$ for all $j\in[p]$ and $t\in[0,T]$, and that $\delta_1\geq0$ and $D\geq0$ satisfy
\[
\max_{j\in[p]}\max_{n\in[N]}\opnorm{X_j-\hat X_j}{[t_{n-1},t_n]}\leq\delta_1,
\qquad
\sup_{|x|\leq R_X}\max_{q\in[m]}|g_q'(x)|\leq D,
\]
and assume that $\hat K_{\Psi_j}\succ0$ for every $j\in[p]$,
and that, for some constant $C_Y>0$,
\[
\frac{1}{N}\sum_{n=1}^N\rbr{Y_n^\Delta}^2\leq C_Y^2.
\]
Then, for every $\Theta\in\Omega$,
\[
\norm{\Delta_X}_{\infty,\hat{\Kcal}_\Psi^*}
\leq
\rbr{\max_{j\in[p]}\sigma_{\min}^{-1/2}(\hat K_{\Psi_j})}
C_YD\sqrt{mh}\,\delta_1.
\]
\end{lemma}
\begin{proof}[Proof of Lemma~\ref{lemma:deltaX}.]
For each $j\in[p]$, let $(\Delta_X)_j$ denote the $j$th block of $\Delta_X$. Since $0\leq\hat w_{\ell,\Theta}(t_n)\leq1$, the definition of the dual norm and the Cauchy--Schwarz inequality imply
\begin{align*}
\norm{(\Delta_X)_j}_{\hat K_{\Psi_j}^*}
&\leq \sigma_{\min}^{-1/2}(\hat K_{\Psi_j})
\left\|\frac{1}{N}\sum_{n=1}^N\hat w_{\ell,\Theta}(t_n)Y_n^\Delta
\cbr{\hat\Psi_j(t_n)-\Psi_j(t_n)}\right\|_2\\
&\leq \sigma_{\min}^{-1/2}(\hat K_{\Psi_j})
\left\{\frac{1}{N}\sum_{n=1}^N\rbr{Y_n^\Delta}^2\right\}^{1/2}
\left\{\frac{1}{N}\sum_{n=1}^N
\norm{\hat\Psi_j(t_n)-\Psi_j(t_n)}_2^2\right\}^{1/2}.
\end{align*}
For every $q\in[m]$,
\begin{align*}
\abr{\hat\Psi_{j,q}(t_n)-\Psi_{j,q}(t_n)}
&\leq \int_{t_{n-1}}^{t_n}
\abr{g_q\rbr{\hat X_j(u)}-g_q\rbr{X_j(u)}}\,\mathrm{d}u\\
&\leq D\sqrt{h}\,\opnorm{\hat X_j-X_j}{[t_{n-1},t_n]}
\leq D\sqrt{h}\,\delta_1.
\end{align*}
Therefore, $\norm{\hat\Psi_j(t_n)-\Psi_j(t_n)}_2\leq D\sqrt{mh}\,\delta_1$. Substituting this bound into the first display and taking the maximum over $j\in[p]$ proves the claim.
\end{proof}

In Assumption~\ref{assumption:local}, the integer $d_{\mathrm{rec}}$ allows the rule that computes $\hat\Psi(t_n)$ to repeat periodically along the sampling grid; for wavelet smoothing, $d_{\mathrm{rec}}=2^{J-j_0}$. When $|X_i(t)|\leq R_X$, the bound on $\hat X_i$ in that assumption can be enforced by projecting $\hat X_i(t)$ onto $[-R_X,R_X]$, which does not increase $\opnorm{\hat X_i-X_i}{[t_{n-1},t_n]}$.

The concentration argument has three parts. First, Lemma~\ref{lemma:wavelet:local} verifies Assumption~\ref{assumption:local} for deterministic wavelet thresholds. Second, Lemma~\ref{lemma:local:blocks} shows that the computed posterior statistics inherit mixing from the interval-block process, using the observation-window radius $r_+=r+r_{\mathrm{rec}}$. Third, we apply concentration to the truncated joint smoothing probabilities, squared-residual statistics, score and weighted Gram matrix in Lemmas~\ref{lemma:sc_w}, \ref{lemma:sc_T2}, \ref{lemma:deviation} and~\ref{lemma:RE:condition}. These are fixed-parameter results; their uniform use along the iteration is a separate hypothesis of Theorem~\ref{theorem:graphrecovery}.

All four concentration proofs use the same decomposition. The interior indices $\mathcal I_{r_+}=\{r_++1,\ldots,N-r_+\}$ have complete observation windows $(W_{n-r_+},\ldots,W_{n+r_+})$. We group these indices by their remainders modulo $d_{\mathrm{rec}}$ to obtain stationary sequences, apply a mixing inequality within each group and combine the bounds. The remaining $2r_+$ boundary indices are controlled separately. This is where the observation-window radius enters the concentration bounds. The concentration bounds use the observation-window radius $r_+=r+r_{\mathrm{rec}}$, whereas the bound on the difference between the full and truncated smoothing probabilities decays geometrically in the truncation radius $r$.

\begin{lemma}\label{lemma:wavelet:local}
Consider the wavelet smoothing of Section~\ref{ssec:waveletsmooth}, with $N=2^J$, followed by the projection of $\hat X_i(t)$ onto $[-R_X,R_X]$, and suppose that $|X_i(t)|\leq R_X$ for all $i\in[p]$ and $t\in[0,T]$, almost surely. Suppose that the number of decomposition levels $J-j_0$ does not depend on $N$, and that the thresholds $\lambda_{jl}$ do not depend on the data or on $l$. Suppose further that there is an integer $S\geq1$ such that, for every $j\in\{j_0,\ldots,J\}$, every function of level $j$ in the construction of Appendix~\ref{appendix:wavelet-reconstruction} is supported in an interval of length at most $S2^{-j}$, that the boundary-adjusted functions of level $j$ are supported within $[0, S2^{-j}]$ or $[1-S2^{-j},1]$, and that $\phi_{J,n-1}$ is supported in $[(n-1-S)2^{-J},(n-1+S)2^{-J}]$ for $n\in[N]$. Suppose also that $\phi$ and $\psi$ are supported in intervals of length at most $S$, and that every boundary-adjusted function of level $j$ has the form $2^{j/2}\phi^{\mathrm{bd}}(2^j\,\cdot)$ or $2^{j/2}\phi^{\mathrm{bd}}(2^j(\cdot-1))$ for some $\phi^{\mathrm{bd}}$ in a finite set $\mathcal B$ of bounded functions that depends on neither $j$ nor $N$. Finally, suppose the following translate condition: for $j\in\{j_0,J\}$, every translate $2^{j/2}\phi(2^j\cdot-l)$ whose support lies in $(0,1)$ is the scaling function of level $j$ and index $l$ in the construction, and, for $j\in\{j_0,\ldots,J-1\}$, every translate $2^{j/2}\psi(2^j\cdot-l)$ whose support lies in $(0,1)$ is the wavelet of level $j$ and index $l$. Then Assumption~\ref{assumption:local} holds with
\[
r_{\mathrm{rec}}=S(2^{J-j_0}+2),
\qquad
d_{\mathrm{rec}}=2^{J-j_0}.
\]
Furthermore, $\hat X_i$ depends on $Y_0^N$ only through $(Y_{1,i},\ldots,Y_{N,i})$, it is a Lipschitz function of it uniformly in time, as in Lemma~\ref{lem:posterior-sensitivity}, and the condition of Lemma~\ref{lem:posterior-sensitivity} holds with a constant $C_{\hat X}$ that depends only on the wavelets and on $J-j_0$.
\end{lemma}
For the Daubechies~$3$ wavelets of our experiments, $S=5$ for the functions away from the boundary.

\begin{proof}[Proof of Lemma~\ref{lemma:wavelet:local}]

Fix $i\in[p]$. We first verify the boundedness and locality required by Assumption~\ref{assumption:local}, and then show that the calculation at interior indices repeats every $d_{\mathrm{rec}}$ sampling intervals. We finish by proving the Lipschitz property and the integrated derivative bound in Lemma~\ref{lem:posterior-sensitivity}.

The projection gives $|\hat X_i(t)|\leq R_X$, and the same bound on $X_i$ is assumed. The coefficients in Appendix~\ref{appendix:wavelet-reconstruction} use only $(Y_{1,i},\ldots,Y_{N,i})$. Soft thresholding with deterministic thresholds and pointwise projection preserve this dependence, so $\hat X_i$ uses no observations from other coordinates and does not depend on $Y_{0,i}$.

To verify locality, fix an output interval $(t_{n'-1},t_{n'})$, with $n'\in[N]$. Write $D_j:=2^{J-j}$ for the number of sampling intervals in a level-$j$ grid spacing. Since $T=Nh=2^Jh$, a support length $S2^{-j}$ on $[0,1]$ corresponds to a physical-time length
\[
TS2^{-j}=S2^{J-j}h=SD_jh.
\]
The expansion of $\hat X_i$ uses scaling functions at level $j_0$ and wavelet functions at levels $j_0,\ldots,J-1$. Only functions that are nonzero on $(t_{n'-1},t_{n'})$ can contribute there. For any such function at level $j$, let $I$ be a physical-time interval containing its support, of length at most $SD_jh$. Because $I$ meets the output interval,
\[
\inf I>t_{n'-1}-SD_jh,
\qquad
\sup I<t_{n'}+SD_jh.
\]

We next identify the observations that can enter its coefficient. For example, substituting the definition of $\tilde X_i$ into a detail coefficient gives
\[
\begin{aligned}
\tilde\eta_{i,jl}
&=\frac1T\int_0^T\tilde X_i(t)\psi_{jl}(t/T)\,\mathrm dt\\
&=N^{-1/2}\sum_{n=1}^N Y_{n,i}
\int_0^1\phi_{J,n-1}(u)\psi_{jl}(u)\,\mathrm du.
\end{aligned}
\]
The same identity holds for $\hat\xi_{i,j_0l}$ with $\psi_{jl}$ replaced by $\phi_{j_0l}$. Thus, an observation $Y_{n,i}$ can enter the coefficient only if the support of $\phi_{J,n-1}(t/T)$ overlaps $I$ on a set of positive measure. By the assumed support of $\phi_{J,n-1}$, this requires
\[
(n-1-S)h<\sup I<t_{n'}+SD_jh,
\qquad
(n-1+S)h>\inf I>t_{n'-1}-SD_jh.
\]
Using $t_n=nh$, these inequalities give
\[
-S(D_j+1)<n-n'<S(D_j+1)+1.
\]
Since the indices and $S(D_j+1)$ are integers, every observation that can enter this coefficient satisfies
\[
|n-n'|\leq S(D_j+1)
\leq S(D_{j_0}+1)
\leq r_{\mathrm{rec}}.
\]
Soft thresholding acts on each coefficient separately, and the projection acts pointwise, so neither operation introduces dependence on additional observations. Consequently,
\[
\hat\Psi_i(t_{n'})
=\int_{t_{n'-1}}^{t_{n'}}g(\hat X_i(t))\,\mathrm dt
\]
is a measurable function of the observations within $r_{\mathrm{rec}}$ sampling intervals of $t_{n'}$. Applying this argument to each coordinate proves the locality assertion in Assumption~\ref{assumption:local}.

\emph{Claim and purpose.}
The locality argument shows that $\hat\Psi(t_{n'})$ is a measurable function of the observations around $n'$. We now show that the same function is used whenever interior indices have the same remainder modulo $d_{\mathrm{rec}}$. Specifically, for each fixed $N$, there are measurable functions $H_0,\ldots,H_{d_{\mathrm{rec}}-1}$ such that
\[
\hat\Psi(t_{n'})
=H_c\!\left(Y_{n'-r_{\mathrm{rec}}},\ldots,Y_{n'+r_{\mathrm{rec}}}\right),
\qquad c=n'\bmod d_{\mathrm{rec}},
\quad r_{\mathrm{rec}}<n'\leq N-r_{\mathrm{rec}}.
\]
This property is used in Lemma~\ref{lemma:local:blocks} to obtain stationary sequences, one for each remainder modulo $d_{\mathrm{rec}}$, from the stationary interval-block process. Locality identifies the window of observations that enters; the representation above ensures that the same function is applied to each such window at indices with the same remainder.

\emph{Proof of the claim.}
Fix an interior index $n'$, that is, $r_{\mathrm{rec}}<n'\leq N-r_{\mathrm{rec}}$, and write
\[
n'=a d_{\mathrm{rec}}+c,
\qquad 0\leq c<d_{\mathrm{rec}}.
\]
For this part of the proof, let
\[
y=(y_{-r_{\mathrm{rec}}},\ldots,y_{r_{\mathrm{rec}}}),
\qquad y_k:=Y_{n'+k}\in\RR^p,
\]
denote its window of observations, with $y_{k,i}$ its coordinate-$i$ observation. The letters $a$, $b$, $k$, $v$ and $z$ below are also local to this part of the proof.

We first exclude boundary-adjusted functions. These are the scaling functions and wavelets modified near $0$ or $1$ in Appendix~\ref{appendix:wavelet-reconstruction}. Since $D_j\leq D_{j_0}=d_{\mathrm{rec}}$, $r_{\mathrm{rec}}=S(d_{\mathrm{rec}}+2)$ and $n'$ is interior, we have, for every $j\in\{j_0,\ldots,J\}$,
\[
t_{n'-1}\geq r_{\mathrm{rec}}h>SD_jh,
\qquad
t_{n'}\leq T-r_{\mathrm{rec}}h<T-SD_jh.
\]
Thus, a boundary-adjusted level-$j$ function, supported within $[0,SD_jh]$ or $[T-SD_jh,T]$ in physical time, vanishes on the output interval. Moreover, for a support interval $I$ of a contributing function, the locality calculation gives
\[
\begin{aligned}
\inf I&>t_{n'-1}-SD_{j_0}h\geq2Sh,\\
\sup I&<t_{n'}+SD_{j_0}h\leq T-2Sh.
\end{aligned}
\]
The boundary-adjusted level-$J$ functions are supported within $[0,Sh]$ or $[T-Sh,T]$, so they cannot contribute to its coefficient. All functions entering this calculation are therefore ordinary translates from the fixed wavelet construction.

For $j=j_0,\ldots,J-1$, write
\[
m_j:=\frac{d_{\mathrm{rec}}}{D_j}=2^{j-j_0}\in\NN.
\]
Parameterize the output interval and relabel the translation index by
\[
t=(n'-1+v)h,\qquad 0<v<1,
\qquad l=am_j+b.
\]
Using $T=Nh$ and $D_j=N2^{-j}$, we obtain
\[
\begin{aligned}
\psi_{j,am_j+b}(t/T)
&=2^{j/2}\psi\!\left(\frac{n'-1+v}{D_j}-am_j-b\right)\\
&=2^{j/2}\psi\!\left(
\frac{a d_{\mathrm{rec}}+c-1+v}{D_j}
-a\frac{d_{\mathrm{rec}}}{D_j}-b\right)\\
&=2^{j/2}\psi\!\left(\frac{c-1+v}{D_j}-b\right).
\end{aligned}
\]
The coefficient $a$ cancels. The same calculation holds for the scaling functions, with $m_{j_0}=1$.

We next express the coefficients in the same coordinates. By locality,
\[
\tilde\eta_{i,j,am_j+b}
=N^{-1/2}\sum_{k=-r_{\mathrm{rec}}}^{r_{\mathrm{rec}}}y_{k,i}
\int_0^1\phi_{J,n'+k-1}(u)\psi_{j,am_j+b}(u)\,\mathrm du.
\]
Under the change of variable $u=(n'-1+z)/N$, the contributing ordinary translates become
\[
\begin{aligned}
\phi_{J,n'+k-1}(u)&=N^{1/2}\phi(z-k),\\
\psi_{j,am_j+b}(u)&=2^{j/2}\psi\!\left(\frac{c-1+z}{D_j}-b\right).
\end{aligned}
\]
 Consequently,
\[
\begin{aligned}
\tilde\eta_{i,j,am_j+b}
&=\frac{2^{j/2}}{N}
\sum_{k=-r_{\mathrm{rec}}}^{r_{\mathrm{rec}}}y_{k,i}
\int_{\RR}\phi(z-k)
\psi\!\left(\frac{c-1+z}{D_j}-b\right)\,\mathrm dz\\
&=: \tilde\eta_{i,jb}(y;c).
\end{aligned}
\]
Similarly, we have
\[
\begin{aligned}
\hat\xi_{i,j_0,a+b}
&=\frac{2^{j_0/2}}{N}
\sum_{k=-r_{\mathrm{rec}}}^{r_{\mathrm{rec}}}y_{k,i}
\int_{\RR}\phi(z-k)
\phi\!\left(\frac{c-1+z}{D_{j_0}}-b\right)\,\mathrm dz\\
&=: \hat\xi_{i,j_0b}(y;c).
\end{aligned}
\]
Applying the common level-$j$ threshold to $\tilde\eta_{i,jb}(y;c)$ gives $\hat\eta_{i,jb}(y;c)$. Since the threshold does not depend on the translation index, this operation also depends only on $y$, $c$ and the fixed wavelet construction.

Reindexing the contributing terms, the wavelet expansion before projection at $t=(n'-1+v)h$ is therefore
\[
\begin{aligned}
&\sum_b\hat\xi_{i,j_0b}(y;c)\,
2^{j_0/2}\phi\!\left(\frac{c-1+v}{D_{j_0}}-b\right)\\
&\quad+\sum_{j=j_0}^{J-1}\sum_b\hat\eta_{i,jb}(y;c)\,
2^{j/2}\psi\!\left(\frac{c-1+v}{D_j}-b\right).
\end{aligned}
\]
Every coefficient, basis function and summation index set in this expression is determined by $y$, $c$ and the fixed wavelet construction. Pointwise projection preserves this dependence.

For an arbitrary $y\in(\RR^p)^{2r_{\mathrm{rec}}+1}$, define the coordinate-$i$ block $[H_c(y)]_i$ by projecting the displayed expansion onto $[-R_X,R_X]$, applying $g$, and integrating over $v\in[0,1]$ with the factor $h$. This is a measurable function of $y$. For the window of observations of $n'$, $y_k=Y_{n'+k}$, the identities above and $\mathrm dt=h\,\mathrm dv$ give
\[
\begin{aligned}
\hat\Psi_i(t_{n'})
&=\int_{t_{n'-1}}^{t_{n'}}g(\hat X_i(t))\,\mathrm dt\\
&=h\int_0^1g\!\left(\hat X_i((n'-1+v)h)\right)\,\mathrm dv
=[H_c(y)]_i.
\end{aligned}
\]
Combining the coordinates proves the claimed representation, which is the repetition property in Assumption~\ref{assumption:local}. If no interior index has remainder $c$, the representation places no condition on $H_c$. The functions $H_c$ may depend on $N$ and the fixed wavelet construction; the formula shows that their dependence on $n'$ is only through $c=n'\bmod d_{\mathrm{rec}}$.

Finally,
\[
d_{\mathrm{rec}}=2^{J-j_0},
\qquad r_{\mathrm{rec}}=S(d_{\mathrm{rec}}+2),
\]
are independent of $N$ because $S$ and $J-j_0$ are fixed.

It remains to establish the Lipschitz property and the integrated derivative bound. All norms of basis functions in this part of the proof are taken on $[0,1]$. Let $A_1$ and $A_\infty$ be common bounds for the $L_1$ and supremum norms of the fixed scaling function, wavelet and boundary-adjusted functions before dilation. Thus, for every function used at level $j$,
\[
\begin{aligned}
\|\phi_{jl}\|_1,\ \|\psi_{jl}\|_1&\leq 2^{-j/2}A_1,\\
\|\phi_{jl}\|_\infty,\ \|\psi_{jl}\|_\infty&\leq 2^{j/2}A_\infty.
\end{aligned}
\]

First, take two observation vectors $y,y'\in\RR^{N+1}$ for coordinate $i$. Orthonormality of the level-$J$ scaling functions gives
\[
\begin{aligned}
\int_0^1|\tilde X_i(Tu;y)-\tilde X_i(Tu;y')|^2\,\mathrm du
&=N^{-1}\sum_{n=1}^N|y_n-y'_n|^2\\
&\leq N^{-1}\norm{y-y'}_2^2.
\end{aligned}
\]
Each coefficient difference before thresholding is an inner product of $\tilde X_i(T\,\cdot\,;y)-\tilde X_i(T\,\cdot\,;y')$ with a unit-norm basis function. The Cauchy--Schwarz inequality therefore gives
\[
\begin{aligned}
|\hat\xi_{i,j_0l}(y)-\hat\xi_{i,j_0l}(y')|
&\leq N^{-1/2}\norm{y-y'}_2,\\
|\tilde\eta_{i,jl}(y)-\tilde\eta_{i,jl}(y')|
&\leq N^{-1/2}\norm{y-y'}_2.
\end{aligned}
\]
Since the thresholds are deterministic, soft thresholding is the same $1$-Lipschitz map for both observation vectors, and hence
\[
|\hat\eta_{i,jl}(y)-\hat\eta_{i,jl}(y')|
\leq|\tilde\eta_{i,jl}(y)-\tilde\eta_{i,jl}(y')|.
\]
Pointwise projection onto $[-R_X,R_X]$ also cannot increase a difference. Applying the triangle inequality to the expansion before projection yields
\[
\begin{aligned}
&\sup_{t\in[0,T]}|\hat X_i(t;y)-\hat X_i(t;y')|\\
&\quad\leq N^{-1/2}\left\{
\sum_{l=0}^{2^{j_0}-1}\|\phi_{j_0l}\|_\infty
+\sum_{j=j_0}^{J-1}\sum_{l=0}^{2^j-1}\|\psi_{jl}\|_\infty
\right\}\norm{y-y'}_2\\
&\quad=:C_{\mathrm{Lip}}\norm{y-y'}_2.
\end{aligned}
\]
This choice of $C_{\mathrm{Lip}}$ is finite for every $N$ and is uniform in $t$ and $i$.  We next prove the separate bound with $C_{\hat X}$ independent of $N$.

To obtain the integrated derivative bound, fix $n\in[N]$. The coefficient identity above gives
\[
\frac{\partial\tilde\eta_{i,jl}}{\partial Y_{n,i}}
=N^{-1/2}\int_0^1\phi_{J,n-1}(u)\psi_{jl}(u)\,\mathrm du.
\]
By the norm bounds above,
\[
\begin{aligned}
\left|\frac{\partial\tilde\eta_{i,jl}}{\partial Y_{n,i}}\right|
&\leq N^{-1/2}
\left(\int_0^1|\phi_{J,n-1}(u)|\,\mathrm du\right)
\sup_{u\in[0,1]}|\psi_{jl}(u)|\\
&\leq N^{-1/2}2^{-J/2}2^{j/2}A_1A_\infty
=N^{-1}2^{j/2}A_1A_\infty.
\end{aligned}
\]
The corresponding function in the expansion satisfies
\[
\int_0^T|\psi_{jl}(t/T)|\,\mathrm dt
=T\int_0^1|\psi_{jl}(u)|\,\mathrm du
\leq T2^{-j/2}A_1.
\]
Multiplying these two bounds shows that the contribution of one detail coefficient to the integrated derivative is at most
\[
\left|\frac{\partial\tilde\eta_{i,jl}}{\partial Y_{n,i}}\right|
\int_0^T|\psi_{jl}(t/T)|\,\mathrm dt
\leq\frac{T}{N}A_1^2A_\infty
=hA_1^2A_\infty.
\]
The same calculation for a scaling coefficient gives
\[
\left|\frac{\partial\hat\xi_{i,j_0l}}{\partial Y_{n,i}}\right|
\int_0^T|\phi_{j_0l}(t/T)|\,\mathrm dt
\leq hA_1^2A_\infty.
\]

Combining everything, we have
\[
\begin{aligned}
\left|\frac{\partial\hat X_i(t)}{\partial Y_{n,i}}\right|
&\leq\sum_{l=0}^{2^{j_0}-1}
\left|\frac{\partial\hat\xi_{i,j_0l}}{\partial Y_{n,i}}\right|
|\phi_{j_0l}(t/T)|+\sum_{j=j_0}^{J-1}\sum_{l=0}^{2^j-1}
\left|\frac{\partial\tilde\eta_{i,jl}}{\partial Y_{n,i}}\right|
|\psi_{jl}(t/T)|.
\end{aligned}
\]

We now sum the contributions of the coefficients that depend on $Y_{n,i}$. Since
\[
\frac{\partial\tilde\eta_{i,jl}}{\partial Y_{n,i}}
=
N^{-1/2}\int_0^1
\phi_{J,n-1}(u)\psi_{jl}(u)\,\mathrm du,
\]
we have
\[
\operatorname{supp}(\phi_{J,n-1})\cap\operatorname{supp}(\psi_{jl})
\text{ has measure zero}
\quad\Longrightarrow\quad
\frac{\partial\tilde\eta_{i,jl}}{\partial Y_{n,i}}=0.
\]

At level $j$, the ordinary translates have support lengths at most $S2^{-j}$ and spacing $2^{-j}$. The support of $\phi_{J,n-1}$ lies in an interval of length at most $2S2^{-J}$. Therefore, the number of ordinary translates whose supports can overlap it is at most
\[
1+\frac{2S2^{-J}+S2^{-j}}{2^{-j}}
=1+2S2^{j-J}+S
\leq 3S+1.
\]
Adding the fixed number of boundary-adjusted functions gives a bound $A_S$ on the number of contributing coefficients in each family. This bound is independent of $N$, $n$, and $j$.

Each contributing coefficient contributes at most $hA_1^2A_\infty$ to the integrated derivative bound established above. There is one scaling family and there are $J-j_0$ detail families. Consequently, for almost every observation vector,
\[
\begin{aligned}
\int_0^T
\left|\frac{\partial\hat X_i(t)}{\partial Y_{n,i}}\right|
\,\mathrm dt
&\leq (J-j_0+1)A_S\,hA_1^2A_\infty\\
&=hC_{\hat X},
\end{aligned}
\]
where
\[
C_{\hat X}:=(J-j_0+1)A_SA_1^2A_\infty.
\]
Since the wavelet construction and $J-j_0$ are fixed, $C_{\hat X}$ is independent of $N$ and $T$. For $n=0$, the derivative is zero because the estimator does not use $Y_{0,i}$. Taking the finite union of the exceptional sets gives the bound simultaneously for every coordinate and observation index.

\end{proof}

The second lemma shows that, under Assumption~\ref{assumption:local}, the truncated marginal and joint smoothing probabilities are local functions of the interval blocks. For $n\in[N]$ and $c\in\{0,\ldots,d_{\mathrm{rec}}-1\}$, we write $n\equiv c$ if $c$ is the remainder of $n$ modulo $d_{\mathrm{rec}}$, and we write $\mathcal I_{r_+}:=\{r_++1,\ldots,N-r_+\}$ for the interior indices and $\mathcal I_{r_+,c}:=\{n\in\mathcal I_{r_+}:n\equiv c\}$ for the class of interior indices with remainder $c$.

\begin{lemma}\label{lemma:local:blocks}
Suppose that Assumption~\ref{assumption:stationary} and Assumption~\ref{assumption:local} hold. Fix an integer $r\geq1$ and $\Theta\in\Omega$, and recall that $r_+=r+r_{\mathrm{rec}}$.

\emph{Observation window.}
There are measurable functions $F_0,\ldots,F_{d_{\mathrm{rec}}-1}$ such that, for every $n\in\mathcal I_{r_+}$, the collection of truncated marginal and joint smoothing probabilities equals
\[
\left((\hat w_{\ell,\Theta}(t_n))_{\ell\in[k]},
(\hat w_{ab,\Theta}(t_n))_{a,b\in[k]}\right)
=F_c(W_{n-r_+},\ldots,W_{n+r_+}),
\]
where $c$ is the remainder of $n$ modulo $d_{\mathrm{rec}}$.

\emph{Inherited mixing.}
Let $G_0,\ldots,G_{d_{\mathrm{rec}}-1}$ be measurable functions into a Euclidean space, and set
$U_n=G_c(W_{n-r_+},\ldots,W_{n+r_+})$ for $n\in\mathcal I_{r_+}$ with remainder $c$. Within each class $\mathcal I_{r_+,c}$, $(U_n)$ is a finite segment of a strictly stationary sequence $(U^{(c)}_\tau)_{\tau\in\mathbb Z}$ satisfying
\[
\beta_{U^{(c)}}(\tau)
\leq
\begin{cases}
\beta_W(d_{\mathrm{rec}}\tau-2r_+),&d_{\mathrm{rec}}\tau>2r_+,\\
1,&d_{\mathrm{rec}}\tau\leq2r_+,
\end{cases}
\qquad \tau\geq1.
\]
Here $\beta_W$ and $\beta_{U^{(c)}}$ are the $\beta$-mixing coefficients of Definition~\ref{definition:discretebetamixing} for $\{W_n\}$ and $(U^{(c)}_\tau)$. Under Assumption~\ref{assumption:mixingbeta}, this sequence is geometrically $\beta$-mixing:$\beta_{U^{(c)}}(\tau)\leq2\exp(-c_U\tau)$ for $\tau\geq1$, where $c_U:=d_{\mathrm{rec}}\min\{c_W/2,(\log2)/(4r_+)\}$ depends only on $r_+$, $d_{\mathrm{rec}}$ and $c_W$.
\end{lemma}
The subtraction of $2r_+$ measures overlap between two observation windows: indices $d_{\mathrm{rec}}\tau$ sampling intervals apart have windows separated by $d_{\mathrm{rec}}\tau-2r_+$ sampling intervals. Thus a larger reconstruction radius delays the separation needed for the original mixing bound to be effective.

\begin{proof}[Proof of Lemma~\ref{lemma:local:blocks}]

Fix $n\in\mathcal I_{r_+}$ and let $c=n\bmod d_{\mathrm{rec}}$. We first identify the observations needed by each retained emission factor. Since $n\in\mathcal I_{r_+}$ and $r_+=r+r_{\mathrm{rec}}$, we have
\[
n-r>r_{\mathrm{rec}},
\qquad n+r\leq N-r_{\mathrm{rec}}.
\]
Thus the truncation window is $Y_{n-r}^{n+r}$ without clipping. For each index $m$, we have
\[
n-r<m\leq n+r
\quad\Longrightarrow\quad
r_{\mathrm{rec}}<m\leq N-r_{\mathrm{rec}}.
\]
Assumption~\ref{assumption:local} therefore gives
\[
\hat\Psi(t_m)
=H_{c'}\!\left(Y_{m-r_{\mathrm{rec}}},\ldots,Y_{m+r_{\mathrm{rec}}}\right),
\qquad c'=m\bmod d_{\mathrm{rec}}.
\]
Every input index in this expression lies between $n-r_+$ and $n+r_+$, because
\[
n-r_+\leq m-r_{\mathrm{rec}}
\leq m+r_{\mathrm{rec}}\leq n+r_+.
\]
The indices $m-1$ and $m$ of the observations in the emission factor also lie in this range. Here $\hat\Psi(t_m)$ is computed before truncation; wavelet smoothing is not repeated on $Y_{n-r}^{n+r}$.

The truncated marginal and joint smoothing probabilities are obtained by summing and normalizing products of the initial distribution $\pi(Q)$, transition probabilities and the retained emission factors. For fixed $\Theta$, these are measurable operations, with positive denominators. Hence all these probabilities are measurable functions of $Y_{n-r_+}^{n+r_+}$.

To identify which function is used, note that the offsets $m-n$ of the retained indices range over $\{-r+1,\ldots,r\}$ for every $n$, and that
\[
c'=m\bmod d_{\mathrm{rec}}=(c+m-n)\bmod d_{\mathrm{rec}}.
\]
Thus the relative input indices $m-n-r_{\mathrm{rec}},\ldots,m-n+r_{\mathrm{rec}}$ are fixed, and the choice of $H_{c'}$ depends on $n$ only through $c$. The initial distribution, transition probabilities and formulas for the emission factors are also the same for every $n\in\mathcal I_{r_+}$. Since $Y_j=X(t_j)+\sigma^\star\varepsilon_j$ is the same measurable function of $W_j$ at every $j$, composition gives a measurable function $F_c$ satisfying
\[
\left((\hat w_{\ell,\Theta}(t_n))_{\ell\in[k]},
(\hat w_{ab,\Theta}(t_n))_{a,b\in[k]}\right)
=F_c(W_{n-r_+},\ldots,W_{n+r_+}).
\]
This proves the first claim.

For the second claim, use a two-sided stationary extension of $\{W_n\}$ and fix $c$. By stationarity, each finite-dimensional distribution of the extension is one of $\{W_n\}_{n\in\mathbb N}$, and $\beta$-mixing coefficients are limits of coefficients between finitely many coordinates. Hence the extension satisfies Assumption~\ref{assumption:mixingbeta} with the same constant $c_W$. Define
\[
U^{(c)}_\tau
:=G_c(W_{c+d_{\mathrm{rec}}\tau-r_+},\ldots,
W_{c+d_{\mathrm{rec}}\tau+r_+}),
\qquad \tau\in\mathbb Z.
\]
The original $U_n$ with $n\equiv c$ and $n\in\mathcal I_{r_+}$ form a finite segment of this sequence. To check strict stationarity, take any finite set $\mathcal T\subset\mathbb Z$. Stationarity of $W$ implies
\[
\begin{aligned}
&\bigl(W_{c+d_{\mathrm{rec}}(\tau+1)-r_+},\ldots,W_{c+d_{\mathrm{rec}}(\tau+1)+r_+}\bigr)_{\tau\in\mathcal T}\\
&\qquad\stackrel{d.}{=}
\bigl(W_{c+d_{\mathrm{rec}}\tau-r_+},\ldots,W_{c+d_{\mathrm{rec}}\tau+r_+}\bigr)_{\tau\in\mathcal T}.
\end{aligned}
\]
Applying $G_c$ to each window gives
\[
(U^{(c)}_{\tau+1})_{\tau\in\mathcal T}
\stackrel{d.}{=}
(U^{(c)}_{\tau})_{\tau\in\mathcal T}.
\]
Since $\mathcal T$ is arbitrary, iterating this identity shows that every finite-dimensional distribution is invariant under all integer shifts, which proves strict stationarity.

Finally, for $\tau_0\in\mathbb Z$ and $\tau\geq1$, the past and future of this sequence satisfy
\[
\begin{aligned}
\sigma(U^{(c)}_{\tau'}:\tau'\leq\tau_0)
&\subseteq\sigma(W_j:j\leq c+d_{\mathrm{rec}}\tau_0+r_+),\\
\sigma(U^{(c)}_{\tau'}:\tau'\geq\tau_0+\tau)
&\subseteq\sigma(W_j:j\geq c+d_{\mathrm{rec}}(\tau_0+\tau)-r_+).
\end{aligned}
\]
The past and future index sets are therefore separated by
\[
\{c+d_{\mathrm{rec}}(\tau_0+\tau)-r_+\}
-\{c+d_{\mathrm{rec}}\tau_0+r_+\}
=d_{\mathrm{rec}}\tau-2r_+
\]
sampling intervals. When this separation is positive, monotonicity of $\beta$ under restriction of the two $\sigma$-fields, followed by the supremum over $\tau_0$, gives
\[
\beta_{U^{(c)}}(\tau)
\leq\beta_W(d_{\mathrm{rec}}\tau-2r_+).
\]
When $d_{\mathrm{rec}}\tau\leq2r_+$, the two index sets can overlap and we use $\beta_{U^{(c)}}(\tau)\leq1$. Under Assumption~\ref{assumption:mixingbeta}, both cases imply
\[
\beta_{U^{(c)}}(\tau)
\leq2\exp\{c_W(2r_+-d_{\mathrm{rec}}\tau)\}.
\]
It remains to bring this bound to the form of Definition~\ref{definition:discretebetamixing}. If $d_{\mathrm{rec}}\tau\geq4r_+$, then $2r_+-d_{\mathrm{rec}}\tau\leq-d_{\mathrm{rec}}\tau/2$, so the bound is at most $2\exp(-c_Wd_{\mathrm{rec}}\tau/2)\leq2\exp(-c_U\tau)$. If $d_{\mathrm{rec}}\tau<4r_+$, then $c_U\tau<\log2$, so $2\exp(-c_U\tau)>1\geq\beta_{U^{(c)}}(\tau)$. Hence
\[
\beta_{U^{(c)}}(\tau)\leq2\exp(-c_U\tau),
\qquad\tau\geq1,
\]
which is the asserted geometric mixing bound.

\end{proof}

The following lemma shows that the averages of the truncated joint smoothing probabilities concentrate around their expectations.
\begin{lemma}\label{lemma:sc_w}
Suppose that Assumptions~\ref{assumption:stationary}
and~\ref{assumption:mixingbeta} and Assumption~\ref{assumption:local} hold. Fix an integer $r\geq1$ that does
not depend on $N$, and fix $\Theta\in\Omega$. Define
\[
\delta_{3,1}:=
\frac{1}{N}\sum_{a,b\in[k]}
\left|
\sum_{n=1}^N\left\{
\hat w_{ab,\Theta}(t_n)
-\EE[\hat w_{ab,\Theta}(t_n)]
\right\}
\right|.
\]
Let
\[
u_{N,r_+,k}:=k^2\log 2+2\log(N-2r_+).
\]
Then there are constants $a_{r_+},b_{r_+}>0$, depending only on $r_+$, $d_{\mathrm{rec}}$ and the
geometric $\beta$-mixing rate, such that, if $N-2r_+>4$ and
\[
N-2r_+\geq b_{r_+} u_{N,r_+,k}^2,
\]
then
\[
P\left(
\delta_{3,1}>
\frac{N-2r_+}{N}a_{r_+}\sqrt{\frac{u_{N,r_+,k}}{N-2r_+}}+\frac{4r_+}{N}
\right)
\leq 2d_{\mathrm{rec}}(N-2r_+)^{-4}.
\]
\end{lemma}
For fixed $k$, the sample-size condition $N-2r_+\geq b_{r_+}u_{N,r_+,k}^2$ holds for all sufficiently large $N$, because $u_{N,r_+,k}$ grows logarithmically in $N-2r_+$. Since $u_{N,r_+,k}\geq k^2\log2$, this condition also requires $N-2r_+\geq b_{r_+}(\log2)^2k^4$, which is why Theorem~\ref{theorem:graphrecovery} assumes $N\gtrsim k^4$.

\begin{proof}[Proof of Lemma~\ref{lemma:sc_w}]
All expectations in this proof are taken under the original
data-generating distribution. We first bound the sum over the interior indices $\mathcal I_{r_+}=\{r_++1,\ldots,N-r_+\}$ and then add the contribution of the boundary indices. By Lemma~\ref{lemma:local:blocks}, for $n\in\mathcal I_{r_+}$, $\hat w_{ab,\Theta}(t_n)$ is a measurable function of $(W_{n-r_+},\ldots,W_{n+r_+})$ that depends on $n$ only through its remainder modulo $d_{\mathrm{rec}}$. We first assume that $d_{\mathrm{rec}}=1$. Then
\[
\left\{
\hat w_{ab,\Theta}(t_n):n\in\mathcal I_{r_+}
\right\}
\]
is a finite segment of a strictly stationary, geometrically
$\beta$-mixing process. Its mixing constants may depend on $r_+$, but not on $N$. We treat $d_{\mathrm{rec}}>1$ at the end of the proof.

Fix $(s_{ab})_{a,b\in[k]}\in\{-1,1\}^{k^2}$. Since the truncated joint
smoothing probabilities form a probability vector,
\[
\hat w_{ab,\Theta}(t_n)\geq0,
\qquad
\sum_{a,b\in[k]}\hat w_{ab,\Theta}(t_n)=1.
\]
Therefore,
\[
-1\leq\sum_{a,b\in[k]}s_{ab}\hat w_{ab,\Theta}(t_n)\leq1,
\]
and
\[
\left|
\sum_{a,b\in[k]}s_{ab}
\left\{
\hat w_{ab,\Theta}(t_n)-\EE[\hat w_{ab,\Theta}(t_n)]
\right\}
\right|
\leq2.
\]
The scalar sequence in the preceding display is centered, strictly
stationary, geometrically $\beta$-mixing, and sub-Weibull$(2)$.
Lemma~\ref{lemma:tailbound:mixing}, with $\gamma_1=1$,
$\gamma_2=2$, and $\gamma=2/3$, therefore gives constants
$C_1(r_+),C_2(r_+)>0$ such that, for every $t>(N-2r_+)^{-1}$,
\[
\begin{aligned}
&P\Bigg(
\bigg|
\frac1{N-2r_+}\sum_{n=r_++1}^{N-r_+}\sum_{a,b\in[k]}s_{ab}
\left\{
\hat w_{ab,\Theta}(t_n)-\EE[\hat w_{ab,\Theta}(t_n)]
\right\}
\bigg|>t
\Bigg)\\
&\qquad\leq
(N-2r_+)\exp\left\{-\frac{((N-2r_+)t)^{2/3}}{C_1(r_+)}\right\}\\
&\qquad\quad+\exp\left\{-\frac{(N-2r_+)t^2}{C_2(r_+)}\right\}.
\end{aligned}
\]
For any real numbers $(x_{ab})_{a,b\in[k]}$,
\[
\sum_{a,b\in[k]}|x_{ab}|
=
\max_{(s_{ab})\in\{-1,1\}^{k^2}}
\sum_{a,b\in[k]}s_{ab}x_{ab}.
\]
Applying this identity with
\[
x_{ab}:=
\frac{1}{N}\sum_{n=r_++1}^{N-r_+}\left\{
\hat w_{ab,\Theta}(t_n)-\EE[\hat w_{ab,\Theta}(t_n)]
\right\},
\]
and taking a union bound over the $2^{k^2}$ sign arrays, we obtain
\[
\begin{aligned}
&P\Bigg(
\frac1N\sum_{a,b\in[k]}
\left|
\sum_{n=r_++1}^{N-r_+}\left\{
\hat w_{ab,\Theta}(t_n)-\EE[\hat w_{ab,\Theta}(t_n)]
\right\}
\right|
>\frac{N-2r_+}{N}t
\Bigg)\\
&\qquad\leq
2^{k^2}\left[
(N-2r_+)\exp\left\{-\frac{((N-2r_+)t)^{2/3}}{C_1(r_+)}\right\}\right.\\
&\qquad\qquad\left.+
\exp\left\{-\frac{(N-2r_+)t^2}{C_2(r_+)}\right\}
\right].
\end{aligned}
\]
For each of the $2r_+$ indices outside
$\{r_++1,\ldots,N-r_+\}$,
\[
\begin{aligned}
\sum_{a,b\in[k]}
\left|
\hat w_{ab,\Theta}(t_n)-\EE[\hat w_{ab,\Theta}(t_n)]
\right|
&\leq
\sum_{a,b\in[k]}\hat w_{ab,\Theta}(t_n)
+\sum_{a,b\in[k]}\EE[\hat w_{ab,\Theta}(t_n)]\\
&=2.
\end{aligned}
\]
Therefore,
\[
\begin{aligned}
\delta_{3,1}
&\leq
\frac1N\sum_{a,b\in[k]}
\left|
\sum_{n=r_++1}^{N-r_+}\left\{
\hat w_{ab,\Theta}(t_n)-\EE[\hat w_{ab,\Theta}(t_n)]
\right\}
\right|+\frac{4r_+}{N}.
\end{aligned}
\]

Choose $a_{r_+}\geq1$ such that
\[
\frac{a_{r_+}^2}{C_2(r_+)}\geq2,
\]
and then choose $b_{r_+}>0$ such that
\[
\frac{a_{r_+}^{2/3}b_{r_+}^{1/3}}{C_1(r_+)}
\geq
\frac{a_{r_+}^2}{C_2(r_+)}+\frac12.
\]
Set
\[
t=a_{r_+}\sqrt{\frac{u_{N,r_+,k}}{N-2r_+}}.
\]
Since $N-2r_+>4$, $a_{r_+}\geq1$, and $u_{N,r_+,k}>1$, we have
\[
(N-2r_+)t=a_{r_+}\sqrt{(N-2r_+)u_{N,r_+,k}}>1.
\]
Therefore, $t>(N-2r_+)^{-1}$, as required to apply
Lemma~\ref{lemma:tailbound:mixing}.
Since $N-2r_+\geq b_{r_+} u_{N,r_+,k}^2$,
\[
\begin{aligned}
\frac{((N-2r_+)t)^{2/3}}{C_1(r_+)}
&=
\frac{a_{r_+}^{2/3}(N-2r_+)^{1/3}u_{N,r_+,k}^{1/3}}{C_1(r_+)}\\
&\geq
\frac{a_{r_+}^2u_{N,r_+,k}}{C_2(r_+)}+\log(N-2r_+)\\
&=
\frac{(N-2r_+)t^2}{C_2(r_+)}+\log(N-2r_+).
\end{aligned}
\]
Therefore,
\[
(N-2r_+)\exp\left\{-\frac{((N-2r_+)t)^{2/3}}{C_1(r_+)}\right\}
\leq
\exp\left\{-\frac{(N-2r_+)t^2}{C_2(r_+)}\right\}.
\]
Combining the preceding displays gives
\[
\begin{aligned}
P\left(
\delta_{3,1}>
\frac{N-2r_+}{N}a_{r_+}\sqrt{\frac{u_{N,r_+,k}}{N-2r_+}}+\frac{4r_+}{N}
\right)
&\leq
2^{k^2+1}
\exp\left\{-\frac{a_{r_+}^2u_{N,r_+,k}}{C_2(r_+)}\right\}\\
&\leq
2^{k^2+1}\exp\{-2u_{N,r_+,k}\}\\
&=
2\exp(-k^2\log2)(N-2r_+)^{-4}\\
&\leq2(N-2r_+)^{-4}.
\end{aligned}
\]
This proves the claim for $d_{\mathrm{rec}}=1$.

If $d_{\mathrm{rec}}>1$, we apply the argument to each class $\mathcal I_{r_+,c}:=\{n\in\mathcal I_{r_+}:n\equiv c\}$, $c=0,\ldots,d_{\mathrm{rec}}-1$. By Lemma~\ref{lemma:local:blocks}, each class gives a finite segment of a strictly stationary, geometrically $\beta$-mixing sequence, and it has $n_c\geq(N-2r_+)/(2d_{\mathrm{rec}})$ elements, since $N-2r_+\geq d_{\mathrm{rec}}$ when $b_{r_+}$ is large enough. In each class, we apply Lemma~\ref{lemma:tailbound:mixing} with $n_c$ in place of $N-2r_+$, keep $u_{N,r_+,k}$, and use the threshold $t_c:=a_{r_+}\{u_{N,r_+,k}/n_c\}^{1/2}$. After $b_{r_+}$ is multiplied by $2d_{\mathrm{rec}}$, so that $n_c\geq b_{r_+}u_{N,r_+,k}^2$, the preceding calculation shows that each class fails with probability at most $2(N-2r_+)^{-4}$. Outside the union of these events, the Cauchy--Schwarz inequality gives
\[
\frac1N\sum_{c=0}^{d_{\mathrm{rec}}-1}n_ct_c
=\frac{a_{r_+}}N\sum_{c=0}^{d_{\mathrm{rec}}-1}\sqrt{n_cu_{N,r_+,k}}
\leq\frac{N-2r_+}{N}\sqrt{d_{\mathrm{rec}}}\,a_{r_+}\sqrt{\frac{u_{N,r_+,k}}{N-2r_+}}.
\]
Therefore, the claim holds after $a_{r_+}$ is multiplied by $\sqrt{d_{\mathrm{rec}}}$, and the union bound over the $d_{\mathrm{rec}}$ classes gives the factor $d_{\mathrm{rec}}$ in the probability.
\end{proof}

The next lemma gives a similar concentration bound for the average of the weighted squared residuals.
\begin{lemma}\label{lemma:sc_T2}
Suppose that Assumptions~\ref{assumption:edgeset}, \ref{assumption:stationary}
and~\ref{assumption:mixingbeta} and Assumption~\ref{assumption:local} hold. Let $\Theta\in
B(r_0,\Theta^\star)$, let $r\geq1$ be fixed, and suppose that
$N-2r_+>4$. Let $R_{n,\ell}$ be as in Appendix~\ref{ssec:onestep_sigma}, evaluated at the blocks $\theta_i^\ell$ of $\Theta$. Fix $\ell\in[k]$, and suppose that, almost surely,
$\max_{n\in[N],\,i\in[p]}|\rho_{n,i}|\leq\rho$ and $\max_{n\in[N]}\norm{\Psi(t_n)}_\infty\leq Bh$. Recall that
$r_{n,i}=0$ under Assumption~\ref{assumption:edgeset}. Define
\[
\mathsf D:=
2\sigma^{\star2}+\rho^2
+2B^2h^2
\max_{\substack{\ell',\ell''\in[k]\\\ell'\neq\ell''}}\max_{i\in[p]}
\norm{\theta_i^{\ell'\star}-\theta_i^{\ell''\star}}_1^2
+2C_{\mathrm{cone}}^2msB^2h^2r_0^2,
\]
where $C_{\mathrm{cone}}=10c_\Psi^{-1}$ is the constant in the definition of $B(r_0,\Theta^\star)$ in Section~\ref{ssec:popEM}.
There exist constants $C_1(r_+),C_2(r_+)>0$, depending only on $r_+$, $d_{\mathrm{rec}}$ and the geometric $\beta$-mixing constants, such that, if
\[
N-2r_+\geq C_2(r_+)\{5\log(N-2r_+)\}^4,
\]
then
\begin{multline*}
P\Bigg(\frac{1}{Np}\left|\sum_{n=1}^N
\left\{\hat{w}_{\ell,\Theta}(t_n)R_{n,\ell}
-\EE\sbr{\hat{w}_{\ell,\Theta}(t_n)R_{n,\ell}}\right\}\right|\\*
>C_1(r_+)\mathsf D
\left\{\sqrt{\frac{5\log(N-2r_+)}{N-2r_+}}
+\frac{r_+(\log N)^{3/2}}{N}\right\}\Bigg)
\leq3d_{\mathrm{rec}}(N-2r_+)^{-4}.
\end{multline*}
\end{lemma}

\begin{proof}[Proof of Lemma~\ref{lemma:sc_T2}]
We first bound the sub-Weibull norm of each summand, and then we bound the sums over the interior and the boundary indices separately. Write $Z(t_n)=z$. Then
\begin{align*}
\hat{w}_{\ell,\Theta}(t_n)R_{n,\ell}
&=\hat{w}_{\ell,\Theta}(t_n)\sum_{i=1}^p
\left\{Y_{n,i}-Y_{n-1,i}
-\sum_{j=1}^p\Psi_j(t_n)^\top\theta_{ij}^\ell\right\}^2\\
&=\hat{w}_{\ell,\Theta}(t_n)\sum_{i=1}^p
\left\{X_i(t_n)+\sigma^\star\varepsilon_{n,i}-X_i(t_{n-1})
-\sigma^\star\varepsilon_{n-1,i}
-\sum_{j=1}^p\Psi_j(t_n)^\top\theta_{ij}^\ell\right\}^2.
\intertext{Recalling the definition of $\rho_{n,i}$ in~\eqref{eq:define:rho_n}, this is equal to}
&\hat{w}_{\ell,\Theta}(t_n)\sum_{i=1}^p
\left\{\sigma^\star(\varepsilon_{n,i}-\varepsilon_{n-1,i})+\rho_{n,i}
+\sum_{j=1}^p\Psi_j(t_n)^\top
\left(\theta_{ij}^{z\star}-\theta_{ij}^{\ell\star}
+\theta_{ij}^{\ell\star}-\theta_{ij}^\ell\right)\right\}^2\\
&\leq4\hat{w}_{\ell,\Theta}(t_n)\sum_{i=1}^p
\Bigg\{\sigma^{\star2}\varepsilon_{n,i}^2+\sigma^{\star2}\varepsilon_{n-1,i}^2+\rho_{n,i}^2\\
&\hspace{11em}+2B^2h^2\left(
\norm{\theta_i^{z\star}-\theta_i^{\ell\star}}_1^2
+\norm{\theta_i^{\ell\star}-\theta_i^\ell}_1^2\right)\Bigg\},
\end{align*}
where we used $\norm{\Psi(t_n)}_\infty\leq Bh$ and
$(a+b)^2\leq2a^2+2b^2$.
Since $\Theta\in B(r_0,\Theta^\star)$, the $\ell_1$ condition in the definition of $B(r_0,\Theta^\star)$ gives
\[
\frac{1}{p}\sum_{i=1}^p
\norm{\theta_i^{\ell\star}-\theta_i^\ell}_1^2
\leq
\max_{i\in[p]}
\norm{\theta_i^{\ell\star}-\theta_i^\ell}_1^2
\leq C_{\mathrm{cone}}^2ms\,r_0^2.
\]
Applying Lemma~\ref{lemma:subweibull:product} twice and using the preceding bound, we obtain
\begin{align*}
\norm{\hat{w}_{\ell,\Theta}(t_n)R_{n,\ell}}_{\psi_{2/3}}
&\leq C\norm{\hat{w}_{\ell,\Theta}(t_n)}_{\psi_2}
\sum_{i=1}^p\Bigg\{
\sigma^{\star2}\norm{\varepsilon_{n,i}}_{\psi_2}^2
+\sigma^{\star2}\norm{\varepsilon_{n-1,i}}_{\psi_2}^2
+\norm{\rho_{n,i}}_{\psi_2}^2\\
&\hspace{10em}+2B^2h^2\norm{\theta_i^{z\star}-\theta_i^{\ell\star}}_1^2
+2B^2h^2\norm{\theta_i^{\ell\star}-\theta_i^\ell}_1^2
\Bigg\}\\
&\leq Cp\Bigg\{
2\sigma^{\star2}+\rho^2
+2B^2h^2\max_{\substack{\ell',\ell''\in[k]\\\ell'\neq\ell''}}\max_{i\in[p]}
\norm{\theta_i^{\ell'\star}-\theta_i^{\ell''\star}}_1^2\\
&\hspace{5em}+2C_{\mathrm{cone}}^2msB^2h^2r_0^2
\Bigg\}.
\end{align*}
Applying Lemma~\ref{lemma:mean:subweibull} and the triangle inequality,
we obtain
\begin{align*}
\frac{1}{p}
\norm{\hat{w}_{\ell,\Theta}(t_n)R_{n,\ell}
-\EE\sbr{\hat{w}_{\ell,\Theta}(t_n)R_{n,\ell}}}_{\psi_{2/3}}
&\leq2C\Bigg\{2\sigma^{\star2}+\rho^2+2B^2h^2
\max_{\substack{\ell',\ell''\in[k]\\\ell'\neq\ell''}}\max_{i\in[p]}
\norm{\theta_i^{\ell'\star}-\theta_i^{\ell''\star}}_1^2\\
&\quad+2C_{\mathrm{cone}}^2msB^2h^2r_0^2\Bigg\}.
\end{align*}
Let
\[
T_{n,\ell}:=
\hat{w}_{\ell,\Theta}(t_n)R_{n,\ell}
-\EE\sbr{\hat{w}_{\ell,\Theta}(t_n)R_{n,\ell}},
\]
and let $\mathcal I_{r_+}:=\{r_++1,\ldots,N-r_+\}$ and
$\mathcal J_{r_+}:=[N]\setminus\mathcal I_{r_+}$. The preceding calculation
shows, uniformly in $n$, that
\[
\norm{T_{n,\ell}}_{\psi_{2/3}}
\leq Cp\mathsf D.
\]
By Lemma~\ref{lemma:local:blocks}, for $n\in\mathcal I_{r_+}$, the truncated smoothing probability is a measurable
function of $(W_{n-r_+},\ldots,W_{n+r_+})$ that depends on $n$ only through its remainder modulo $d_{\mathrm{rec}}$, while
$R_{n,\ell}$ is a function of $(W_{n-1},W_n)$, which lies in this block because $r_+\geq1$.
We first assume that $d_{\mathrm{rec}}=1$. Then $\{T_{n,\ell}:n\in\mathcal I_{r_+}\}$ is a finite segment of a
centered, strictly stationary, geometrically $\beta$-mixing process, and we treat $d_{\mathrm{rec}}>1$ at the end of the proof.
Write $n':=N-2r_+$, $u:=5\log n'$, and $K:=C\mathsf D$, where $C$ is
chosen so that $\|T_{n,\ell}/p\|_{\psi_{2/3}}\leq K$ for all $n$.
If $K=0$, the conclusion is immediate. Otherwise, we apply
Lemma~\ref{lemma:tailbound:mixing} to the normalized variables
$T_{n,\ell}/(pK)$, whose sub-Weibull norms are at most one.
With $1/\gamma=1+3/2=5/2$, this gives constants
$\widetilde C_1(r_+),\widetilde C_2(r_+)>0$ such that, for $t>1/n'$,
\[
P\left(\frac1{n'pK}\left|\sum_{n\in\mathcal I_{r_+}}T_{n,\ell}\right|>t\right)
\leq n'\exp\left\{-\frac{(tn')^{2/5}}{\widetilde C_1(r_+)}\right\}
+\exp\left\{-\frac{t^2n'}{\widetilde C_2(r_+)}\right\}.
\]
Choose $a\geq\max\{1,\sqrt{\widetilde C_2(r_+)}\}$ and set
$t:=a\sqrt{u/n'}$, which satisfies $t>1/n'$ since $n'>4$.
If $n'\geq \widetilde C_1(r_+)^5a^{-2}u^4$, then
\[
\frac{(tn')^{2/5}}{\widetilde C_1(r_+)}
=\frac{a^{2/5}(un')^{1/5}}{\widetilde C_1(r_+)}\geq u,
\qquad
\frac{t^2n'}{\widetilde C_2(r_+)}\geq u.
\]
Therefore,
\[
P\left(\frac1{n'p}\left|\sum_{n\in\mathcal I_{r_+}}T_{n,\ell}\right|
>aK\sqrt{\frac{u}{n'}}\right)
\leq n'e^{-u}+e^{-u}\leq2(n')^{-4}.
\]

For the boundary contribution, fix $v\geq1$. The bound
$\|T_{n,\ell}/p\|_{\psi_{2/3}}\leq K$ implies
$\EE[|T_{n,\ell}/p|^v]\leq(Kv^{3/2})^v$.
Applying Markov's inequality to $|T_{n,\ell}/p|^v$ gives
\[
P\left(\frac{|T_{n,\ell}|}{p}>eK v^{3/2}\right)
\leq
\frac{\EE[|T_{n,\ell}/p|^v]}{(eK v^{3/2})^v}
\leq e^{-v}.
\]
Taking $v:=6\log N$ and using $|\mathcal J_{r_+}|\leq2r_+\leq N$ yields
\[
P\left(\max_{n\in\mathcal J_{r_+}}\frac{|T_{n,\ell}|}{p}
>eK(6\log N)^{3/2}\right)\leq N^{-5}.
\]
Therefore, with probability at least $1-2(n')^{-4}-N^{-5}$,
\[
\frac1{Np}\left|\sum_{n=1}^NT_{n,\ell}\right|
\leq\frac{n'}{N}aK\sqrt{\frac{u}{n'}}
+\frac{2er_+K(6\log N)^{3/2}}{N}.
\]
Choose $C_2(r_+)\geq\widetilde C_1(r_+)^5/a^2$ and
$C_1(r_+)\geq C\max\{a,2e\,6^{3/2}\}$. Since $n'/N\leq1$ and
$N^{-5}\leq(n')^{-4}$, the stated bound follows. Because of the normalization,
these constants do not depend on $K$ or $\mathsf D$.

If $d_{\mathrm{rec}}>1$, we apply the argument for the interior indices to each class $\mathcal I_{r_+,c}:=\{n\in\mathcal I_{r_+}:n\equiv c\}$, $c=0,\ldots,d_{\mathrm{rec}}-1$, which gives a finite segment of a centered, strictly stationary, geometrically $\beta$-mixing sequence by Lemma~\ref{lemma:local:blocks} and has $n'_c\geq n'/(2d_{\mathrm{rec}})$ elements once $n'\geq d_{\mathrm{rec}}$. We use $n'_c$ in place of $n'$ in Lemma~\ref{lemma:tailbound:mixing} and in the threshold, and we keep $u=5\log n'$. After $C_2(r_+)$ is multiplied by $2d_{\mathrm{rec}}$, so that $n'_c\geq\widetilde C_1(r_+)^5a^{-2}u^4$, each class fails with probability at most $(n'_c+1)e^{-u}\leq2(n')^{-4}$. Outside the union of these events, the Cauchy--Schwarz inequality gives
\[
\frac1{n'p}\left|\sum_{n\in\mathcal I_{r_+}}T_{n,\ell}\right|
\leq\sum_{c=0}^{d_{\mathrm{rec}}-1}\frac{n'_c}{n'}\,aK\sqrt{\frac{u}{n'_c}}
\leq\sqrt{d_{\mathrm{rec}}}\,aK\sqrt{\frac{u}{n'}}.
\]
Therefore, the stated bound holds after $C_1(r_+)$ is multiplied by $\sqrt{d_{\mathrm{rec}}}$, with failure probability at most $2d_{\mathrm{rec}}(n')^{-4}+N^{-5}\leq3d_{\mathrm{rec}}(n')^{-4}$.
\end{proof}

\section{Proof of the concentration bound for \texorpdfstring{$\Delta_\varepsilon$}{Delta epsilon}}
\label{sec:concentration-deviation}

In this section, we bound the deviation term $\Delta_\varepsilon$ in Assumption~\ref{assumption:DB}. The proof of Lemma~\ref{lemma:deviation} uses Lemma~\ref{lemma:DB:bmixing}, which extends the concentration inequality of Lemma~\ref{lemma:tailbound:mixing} for $\beta$-mixing sequences to a block-weighted maximum norm.

\begin{lemma}
\label{lemma:deviation}
Suppose that Assumptions~\ref{assumption:edgeset}--\ref{assumption:mixingbeta} and Assumption~\ref{assumption:local} hold, and that $|g_j(X_i(t))| \leq B$ for all $t \in [0,T]$, $i = 1, \ldots, p$, and $j = 1, \ldots, m$.
Fix $\Theta\in\Omega$, an index $i\in[p]$, a state $\ell\in[k]$ and a truncation radius $r\in\NN$, and suppose that
\[
N-2r_+\geq C_0m^4\frac{\{\log(5(N-2r_+)p)\}^{5}}{\log(5p)}
\]
for a sufficiently large constant $C_0$ that depends on $r_+$, $d_{\mathrm{rec}}$ and the mixing coefficients.
By Assumption~\ref{assumption:edgeset}, the basis residual $r_{n,i}$ in~\eqref{eq:approxY} vanishes. Define the switching remainder:
\begin{align*}
\rho_{n,i} &:= \int_{t_{n-1}}^{t_n} \theta_i^{Z(u)\star\top} g(X(u))\,\mathrm{d}u - \theta_i^{Z(t_n)\star\top} \int_{t_{n-1}}^{t_n}g(X(u))\,\mathrm{d}u\\
&=\int_{t_{n-1}}^{t_n} \theta_i^{Z(u)\star\top} g(X(u))\,\mathrm{d}u - \theta_i^{Z(t_n)\star\top}\Psi(t_n),\quad n = 1, \ldots, N,
\end{align*}
and assume that $\max_n |\rho_{n,i}| \leq \rho$ almost surely.
Let $Y_n^\Delta := Y_{n,i} - Y_{n-1,i}$, and define the truncated smoothing probability
\begin{equation*}
\hat{w}_{\ell,\Theta}(t_n) = P\rbr{Z(t_n)=\ell\mid Y_{(n-r)\vee 0}^{(n+r)\wedge N};\Theta}
\end{equation*}
and the deviation term:
\begin{multline*}
\Delta_{\varepsilon} := \frac{1}{N} \sum_{n=1}^N \hat{w}_{\ell,\Theta}(t_n) \left( Y_n^\Delta -  \Psi(t_n)^\top\theta_i^{\ell\star} \right) \Psi(t_n)\\ - \frac{1}{N} \sum_{n=1}^N \mathbb{E} \left[ \hat{w}_{\ell,\Theta}(t_n) \left( Y_n^\Delta - \Psi(t_n)^\top \theta_i^{\ell\star} \right) \Psi(t_n) \right].
\end{multline*}
Then there is an event $\mathcal G$ with $P(\mathcal G)\geq1-10d_{\mathrm{rec}}\exp(-c_3m\log p)$ such that, on $\mathcal G\cap\{\hat K_{\Psi_j}\succ0\text{ for every }j\in[p]\}$, we have
\[
\left\| \Delta_{\varepsilon} \right\|_{\infty, \hat{\Kcal}_{\Psi}^*}
\leq C_{\mathrm{dev}}\max_{j\in[p]}\sigma_{\min}^{-1/2}(\hat{K}_{\Psi_j})
\sqrt{m} Bh
\left\{2\sigma^\star+\rho+
Bh\max_{\ell'\in[k]}
\norm{\theta_i^{\ell'\star}-\theta_i^{\ell\star}}_1\right\}
\sqrt{ \frac{m \log p}{N} },
\]
where the constants $C_{\mathrm{dev}}>0$ and $c_3\geq2$ depend on $r_+$, $d_{\mathrm{rec}}$ and the mixing coefficients.
\end{lemma}

\begin{proof}
We first bound the sum over the interior indices, whose truncation windows are not clipped, and then control the boundary terms. Recall that
\begin{align*}
Y_{n,i}= X_i(t_n) + \sigma^\star\varepsilon_{n,i} &= X_{i}(t_{n-1}) + \sum_j\int_{t_{n-1}}^{t_n}\theta_{ij}^{Z(u)\star\top}g(X_j(u))\mathrm{d}u + \sigma^\star\varepsilon_{n,i}\\
&=(Y_{n-1,i}-\sigma^\star\varepsilon_{n-1,i}) + \sum_j\int_{t_{n-1}}^{t_n}\theta_{ij}^{Z(u)\star\top}g(X_j(u))\mathrm{d}u + \sigma^\star\varepsilon_{n,i}
\end{align*}
Rearranging, we obtain
\begin{multline}\label{eq:define:rho_n}
Y_{n,i} - Y_{n-1,i} - \theta_i^{Z(t_n)\star\top} \Psi(t_n) \\= \sigma^\star(\varepsilon_{n,i} - \varepsilon_{n-1,i}) + \underbrace{\sum_j \int_{t_{n-1}}^{t_n} \theta_{ij}^{Z(u)\star\top} g(X_j(u)) \, \mathrm{d}u - \theta_i^{Z(t_n)\star\top} \Psi(t_n)}_{:=\rho_{n,i}}.
\end{multline}
Since $i$ is fixed, we write $\rho_n$ for $\rho_{n,i}$ in what follows.

For every $n$,
\[
Y_n^\Delta-\Psi(t_n)^\top\theta_i^{\ell\star}
=\sigma^\star(\varepsilon_{n,i}-\varepsilon_{n-1,i})+\rho_n
+\Psi(t_n)^\top\left\{\theta_i^{Z(t_n)\star}-\theta_i^{\ell\star}\right\}.
\]
Let $\mathcal I_{r_+}:=\{r_++1,\ldots,N-r_+\}$ be the set of interior indices, and define
\[
\mathsf K_N:=\max_{j\in[p]}\sigma_{\min}^{-1/2}(\hat K_{\Psi_j}).
\]
For every $v=(v_1^\top,\ldots,v_p^\top)^\top\in\RR^{pm}$, with $\|v\|_{\infty,2}:=\max_{j\in[p]}\|v_j\|_2$, we have
\[
\|v\|_{\infty,\hat{\Kcal}_\Psi^*}
\leq \mathsf K_N\|v\|_{\infty,2}.
\]
Accordingly, in each application of Lemma~\ref{lemma:DB:bmixing} below,
we take its deterministic block matrices to be $A_j=I_m$ and then apply
the last display, which holds for every realization with
$\hat K_{\Psi_j}\succ0$ for every $j\in[p]$. We define the interior statistic
\begin{align*}
\Delta_\varepsilon^{\mathrm{int}}
:={}&\frac{1}{N-2r_+}\sum_{n\in\mathcal I_{r_+}}
\hat w_{\ell,\Theta}(t_n)
\left\{Y_n^\Delta-\Psi(t_n)^\top\theta_i^{\ell\star}\right\}\Psi(t_n)\\
&-\frac{1}{N-2r_+}\sum_{n\in\mathcal I_{r_+}}\EE\left[
\hat w_{\ell,\Theta}(t_n)
\left\{Y_n^\Delta-\Psi(t_n)^\top\theta_i^{\ell\star}\right\}\Psi(t_n)
\right].
\end{align*}
Using~\eqref{eq:define:rho_n}, we decompose
\[
\Delta_\varepsilon^{\mathrm{int}}
=\Delta_{\varepsilon_1}-\Delta_{\varepsilon_2}
+\Delta_\rho+\Delta_Z,
\]
where
\begin{align*}
\Delta_{\varepsilon_1}
={}&\frac{1}{N-2r_+}\sum_{n\in\mathcal I_{r_+}}
\hat w_{\ell,\Theta}(t_n)\sigma^\star\varepsilon_{n,i}\Psi(t_n)
-\frac{1}{N-2r_+}\sum_{n\in\mathcal I_{r_+}}\EE\left[
\hat w_{\ell,\Theta}(t_n)\sigma^\star\varepsilon_{n,i}\Psi(t_n)\right],\\
\Delta_{\varepsilon_2}
={}&\frac{1}{N-2r_+}\sum_{n\in\mathcal I_{r_+}}
\hat w_{\ell,\Theta}(t_n)\sigma^\star\varepsilon_{n-1,i}\Psi(t_n)
-\frac{1}{N-2r_+}\sum_{n\in\mathcal I_{r_+}}\EE\left[
\hat w_{\ell,\Theta}(t_n)\sigma^\star\varepsilon_{n-1,i}\Psi(t_n)\right],\\
\Delta_\rho
={}&\frac{1}{N-2r_+}\sum_{n\in\mathcal I_{r_+}}
\hat w_{\ell,\Theta}(t_n)\rho_n\Psi(t_n)
-\frac{1}{N-2r_+}\sum_{n\in\mathcal I_{r_+}}\EE\left[
\hat w_{\ell,\Theta}(t_n)\rho_n\Psi(t_n)\right],\\
\Delta_Z
:={}&\frac{1}{N-2r_+}\sum_{n\in\mathcal I_{r_+}}
\hat w_{\ell,\Theta}(t_n)\Psi(t_n)\Psi(t_n)^\top
\left\{\theta_i^{Z(t_n)\star}-\theta_i^{\ell\star}\right\}\\
&-\frac{1}{N-2r_+}\sum_{n\in\mathcal I_{r_+}}\EE\left[
\hat w_{\ell,\Theta}(t_n)\Psi(t_n)\Psi(t_n)^\top
\left\{\theta_i^{Z(t_n)\star}-\theta_i^{\ell\star}\right\}
\right].
\end{align*}
By the triangle inequality,
\begin{equation}\label{eq:tmp4}
\norm{\Delta_\varepsilon^{\mathrm{int}}}_{\infty,\hat{\Kcal}_\Psi^*}
\leq \norm{\Delta_{\varepsilon_1}}_{\infty,\hat{\Kcal}_\Psi^*}
+\norm{\Delta_{\varepsilon_2}}_{\infty,\hat{\Kcal}_\Psi^*}
+\norm{\Delta_\rho}_{\infty,\hat{\Kcal}_\Psi^*}
+\norm{\Delta_Z}_{\infty,\hat{\Kcal}_\Psi^*}.
\end{equation}

\paragraph*{Step 1: Tail control for $\Delta_{\varepsilon_1}$.}
Let $\gamma_2 = 2/3$. For each $n\in\mathcal I_{r_+}$ and $j$, we have
\begin{multline*}
\left\| \hat{w}_{\ell,\Theta}(t_n) \sigma^\star\varepsilon_{n,i} \Psi_j(t_n) - \EE\left[ \hat{w}_{\ell,\Theta}(t_n) \sigma^\star\varepsilon_{n,i} \Psi_j(t_n) \right] \right\|_{\psi_{\gamma_2}} \\
\leq 2 \left\| \hat{w}_{\ell,\Theta}(t_n) \sigma^\star\varepsilon_{n,i} \Psi_j(t_n) \right\|_{\psi_{\gamma_2}}
\leq 6\sqrt{3} \norm{\hat{w}_{\ell,\Theta}(t_n)}_{\psi_2} \sigma^\star\norm{\varepsilon_{n,i}}_{\psi_2} \norm{\Psi_j(t_n)}_{\psi_2},
\end{multline*}
where the first inequality follows from Lemma~\ref{lemma:mean:subweibull} and the triangle inequality, and the second from Lemma~\ref{lemma:subweibull:product}, applied twice.
Since $\hat{w}_{\ell,\Theta}(t_n) \in [0,1]$ and $\varepsilon_{n,i} \sim \mathcal{N}(0, 1)$, we have $\|\hat{w}_{\ell,\Theta}(t_n)\|_{\psi_2} \leq 1$ and $\|\varepsilon_{n,i}\|_{\psi_2} \leq C$. Finally,
\begin{align*}
\norm{\Psi_j(t_n)}_{\psi_2}&=\sup_{u\in\mathbb{S}^{m-1}}\norm{u^\top\Psi_j(t_n)}_{\psi_2}\notag\\
    &=\sup_{u\in\mathbb{S}^{m-1}}\sup_{q\geq 1}q^{-1/2}\cbr{\EE|u^\top\Psi_j(t_n)|^q}^{1/q}\notag\\
    &\leq \sup_{u\in\mathbb{S}^{m-1}}\sup_{q\geq 1}q^{-1/2}\cbr{\EE\norm{u}_1^q\norm{\Psi_j(t_n)}_{\infty}^q}^{1/q}\leq Bh\sqrt{m}.
\end{align*}
Combining these bounds, we obtain
\[
\left\|\hat{w}_{\ell,\Theta}(t_n)\sigma^\star\varepsilon_{n,i}\Psi_j(t_n)^\top- \EE\sbr{\hat{w}_{\ell,\Theta}(t_n)\sigma^\star\varepsilon_{n,i}\Psi_j(t_n)^\top}\right\|_{\psi_{\gamma_2}}\leq C_1 \sigma^\star\sqrt{m} Bh.
\]
Therefore, for each $n, j$, $\hat{w}_{\ell,\Theta}(t_n)\sigma^\star\varepsilon_{n,i}\Psi_j(t_n)^\top- \EE\sbr{\hat{w}_{\ell,\Theta}(t_n)\sigma^\star\varepsilon_{n,i}\Psi_j(t_n)^\top}$ is a sub-Weibull($2/3$) random vector with sub-Weibull norm bounded by $C_1 \sigma^\star\sqrt{m} Bh$.

\paragraph*{Step 2: Statistical properties of
$\Delta_{\varepsilon_1}$.}
For $n\in\mathcal I_{r_+}$, the truncation window is not clipped, and Lemma~\ref{lemma:local:blocks} shows that $\hat w_{\ell,\Theta}(t_n)$ is a measurable function of $(W_{n-r_+},\ldots,W_{n+r_+})$ that depends on $n$ only through its remainder modulo $d_{\mathrm{rec}}$. Since $\varepsilon_{n,i}$ and $\Psi(t_n)$ are functions of $W_n$,
\[
U_n:=\hat w_{\ell,\Theta}(t_n)
\sigma^\star\varepsilon_{n,i}\Psi(t_n)^\top
\]
is of the form considered in that lemma. Therefore, for each $c$, the vectors $U_n$, $n\in\mathcal I_{r_+,c}$, form a finite segment of a strictly stationary sequence whose $\beta$-mixing coefficients satisfy
\[
\beta(\tau)
\leq 2\exp(-c_U\tau),
\qquad\tau\geq1,
\]
with $c_U$ as in Lemma~\ref{lemma:local:blocks}.
Since $r_+$ and $d_{\mathrm{rec}}$ are fixed, the rate $c_U>0$ does not depend on $N$, and each of these sequences is geometrically
$\beta$-mixing. The same argument applies to the summands involving
$\varepsilon_{n-1,i}$, $\rho_n$, and
$\Psi(t_n)^\top\left\{\theta_i^{Z(t_n)\star}-\theta_i^{\ell\star}\right\}$. Steps~3--6 below are stated for $d_{\mathrm{rec}}=1$, in which case $\{U_n:n\in\mathcal I_{r_+}\}$ is itself a finite segment of a strictly stationary, geometrically $\beta$-mixing sequence. Step~7 explains the changes for $d_{\mathrm{rec}}>1$.

\paragraph*{Step 3: Concentration of $\Delta_{\varepsilon_1}$ over the blocks $\theta_i^\ell$, at fixed $\Theta$.}
By Steps 1--2 and Lemma~\ref{lemma:DB:bmixing}, applied as described above
with $\gamma_1=1$ and $\gamma_2=2/3$, so that $1/\gamma=1/\gamma_1+1/\gamma_2=5/2$, we obtain:
\begin{align}\label{eq:tmp5}
P\left( \norm{\Delta_{\varepsilon_1}}_{\infty,\hat{\Kcal}_\Psi^*} \geq C_5\mathsf K_N\sqrt{m} Bh \sigma^\star \sqrt{ \frac{m \log p}{N-2r_+} } \right)
\leq 2 \exp(-C_6 m \log p).
\end{align}

\paragraph*{Step 4: Concentration of $\Delta_{\varepsilon_2}$ over the blocks $\theta_i^\ell$, at fixed $\Theta$.}
Since $\varepsilon_{n-1,i}$ has the same properties, the argument of Step 3 gives
\begin{align}\label{eq:tmp5_5}
P\left( \norm{\Delta_{\varepsilon_2}}_{\infty,\hat{\Kcal}_\Psi^*} \geq C_5\mathsf K_N\sqrt{m} Bh \sigma^\star \sqrt{ \frac{m \log p}{N-2r_+} } \right)
\leq 2 \exp(-C_6 m \log p).
\end{align}

\paragraph*{Step 5: Concentration of $\Delta_\rho$ over the blocks $\theta_i^\ell$, at fixed $\Theta$.}
As in Steps 1--3, we can verify that the sequence $\{\hat{w}_{\ell,\Theta}(t_n)\rho_n\Psi(t_n)^\top:n\in\mathcal I_{r_+}\}$ is sub-Weibull($2/3$) with norm bounded by $C_4 \sqrt{m} Bh \rho$. Lemma~\ref{lemma:DB:bmixing} then gives
\begin{align}\label{eq:tmp6}
P\left( \norm{\Delta_\rho}_{\infty,\hat{\Kcal}_\Psi^*} \geq C_7\mathsf K_N\sqrt{m} Bh \rho \sqrt{ \frac{m \log p}{N-2r_+} } \right)
\leq 2 \exp(-C_8 m \log p).
\end{align}

\paragraph*{Step 6: Concentration of $\Delta_Z$ over the blocks $\theta_i^\ell$, at fixed $\Theta$.}
Since $\|\Psi(t_n)\|_\infty\leq Bh$, we have
\[
\left|
\Psi(t_n)^\top
\left\{\theta_i^{Z(t_n)\star}-\theta_i^{\ell\star}\right\}
\right|
\leq Bh\max_{\ell'\in[k]}
\norm{\theta_i^{\ell'\star}-\theta_i^{\ell\star}}_1.
\]
Therefore, as in Steps 2--3, the summands defining $\Delta_Z$ are
geometrically $\beta$-mixing and sub-Weibull$(2/3)$, with norm bounded
by
\[
C\sqrt m\,B^2h^2
\max_{\ell'\in[k]}
\norm{\theta_i^{\ell'\star}-\theta_i^{\ell\star}}_1.
\]
Lemma~\ref{lemma:DB:bmixing} then gives
\begin{align}\label{eq:tmp6_5}
P\left(
\norm{\Delta_Z}_{\infty,\hat{\Kcal}_\Psi^*}
\geq
C\mathsf K_N
\sqrt m\,B^2h^2
\max_{\ell'\in[k]}
\norm{\theta_i^{\ell'\star}-\theta_i^{\ell\star}}_1
\sqrt{\frac{m\log p}{N-2r_+}}
\right)
\leq2\exp(-C_9m\log p).
\end{align}

\paragraph*{Boundary terms.}
Let
\[
\mathcal J_{r_+}:=\{1,\ldots,r_+\}\cup\{N-r_++1,\ldots,N\}.
\]
Then $|\mathcal J_{r_+}|\leq 2r_+$, and the original all-index statistic
satisfies
\[
\Delta_\varepsilon
=\frac{N-2r_+}{N}\Delta_\varepsilon^{\mathrm{int}}+\mathsf R_{N,r_+}^{\mathrm{bd}},
\]
where
\begin{align*}
\mathsf R_{N,r_+}^{\mathrm{bd}}:={}&\frac1N\sum_{n\in\mathcal J_{r_+}}\Bigg[
\hat w_{\ell,\Theta}(t_n)
\left\{Y_n^\Delta-\Psi(t_n)^\top\theta_i^{\ell\star}\right\}\Psi(t_n)\\
&\qquad\qquad-\EE\left[
\hat w_{\ell,\Theta}(t_n)
\left\{Y_n^\Delta-\Psi(t_n)^\top\theta_i^{\ell\star}\right\}\Psi(t_n)
\right]\Bigg].
\end{align*}
For $n\in\mathcal J_{r_+}$, define
\[
S_n:=
\hat w_{\ell,\Theta}(t_n)
\left\{Y_n^\Delta-\Psi(t_n)^\top\theta_i^{\ell\star}\right\}
\Psi(t_n).
\]
By~\eqref{eq:define:rho_n}, the bound in Step 6, $0\leq\hat w_{\ell,\Theta}(t_n)\leq1$,
$|\rho_n|\leq\rho$, and $\|\Psi_j(t_n)\|_2\leq\sqrt m\,Bh$, we have
\[
\begin{aligned}
\max_{j\in[p]}
\|S_{n,j}-\EE[S_{n,j}]\|_{\hat K_{\Psi_j}^*}
\leq{}& C\mathsf K_N\sqrt m\,Bh\\
&\times\left\{\sigma^\star\bigl(1+|\varepsilon_{n,i}-\varepsilon_{n-1,i}|\bigr)+\rho
+Bh\max_{\ell'\in[k]}\norm{\theta_i^{\ell'\star}-\theta_i^{\ell\star}}_1\right\}.
\end{aligned}
\]
The additive $1$ controls the contribution of the noise to the centering
term, and the last two deterministic terms control the centered
contributions of $\rho_n$ and of the state mismatch.
Since $\varepsilon_{n,i}-\varepsilon_{n-1,i}\sim\mathcal N(0,2)$,
a union bound over the at most $2r_+$ boundary indices gives
\[
\max_{n\in\mathcal J_{r_+}}\max_{j\in[p]}
\|S_{n,j}-\EE[S_{n,j}]\|_{\hat K_{\Psi_j}^*}
\leq
C\mathsf K_N\sqrt m\,Bh
\left\{2\sigma^\star+\rho
+Bh\max_{\ell'\in[k]}\norm{\theta_i^{\ell'\star}-\theta_i^{\ell\star}}_1\right\}
\sqrt{m\log p}
\]
with probability at least
\[
1-2\exp(-c_bm\log p),
\]
where $c_b\geq2$ once the constant $C$ is large enough, depending on $r_+$.
On this event,
\[
\norm{\mathsf R_{N,r_+}^{\mathrm{bd}}}_{\infty,\hat{\Kcal}_\Psi^*}
\leq C\mathsf K_N
\sqrt m\,Bh
\left\{2\sigma^\star+\rho
+Bh\max_{\ell'\in[k]}\norm{\theta_i^{\ell'\star}-\theta_i^{\ell\star}}_1\right\}
\frac{r_+\sqrt{m\log p}}{N}.
\]

\paragraph*{Step 7: Final bound via union bound.}
Let $\mathcal G$ be the event on which none of the four deviations in~\eqref{eq:tmp5}, \eqref{eq:tmp5_5},
\eqref{eq:tmp6}, and~\eqref{eq:tmp6_5} exceeds its threshold and the boundary bound holds. On $\mathcal G\cap\{\hat K_{\Psi_j}\succ0\text{ for every }j\in[p]\}$, the deterministic
inequality~\eqref{eq:tmp4} and the identity for the all-index statistic give
\[
\begin{aligned}
\norm{\Delta_\varepsilon}_{\infty,\hat{\Kcal}_\Psi^*}
\leq{}&
\frac{N-2r_+}{N}C\mathsf K_N
\sqrt m\,Bh
\left\{2\sigma^\star+\rho
+Bh\max_{\ell'\in[k]}\norm{\theta_i^{\ell'\star}-\theta_i^{\ell\star}}_1\right\}
\sqrt{\frac{m\log p}{N-2r_+}}\\
&+C\mathsf K_N\sqrt m\,Bh
\left\{2\sigma^\star+\rho
+Bh\max_{\ell'\in[k]}\norm{\theta_i^{\ell'\star}-\theta_i^{\ell\star}}_1\right\}
\frac{r_+\sqrt{m\log p}}{N}.
\end{aligned}
\]
By a union bound, $P(\mathcal G)\geq
1-10\exp(-c_3m\log p)$ with $c_3:=\min\{C_6,C_8,C_9,c_b\}$. Here $C_6$, $C_8$ and $C_9$ are the constant $C_3\geq2$ of Lemma~\ref{lemma:DB:bmixing}, so that $c_3\geq2$. If $d_{\mathrm{rec}}>1$, we apply Steps~3--6 to each class $\mathcal I_{r_+,c}$, $c=0,\ldots,d_{\mathrm{rec}}-1$. Each class has at least $(N-2r_+)/(2d_{\mathrm{rec}})$ elements once $N-2r_+\geq d_{\mathrm{rec}}$, so that the sample-size condition of Lemma~\ref{lemma:DB:bmixing} holds in each class after $C_0$ is multiplied by $2d_{\mathrm{rec}}$, and the thresholds of the class averages are at most $(2d_{\mathrm{rec}})^{1/2}$ times those in~\eqref{eq:tmp5}--\eqref{eq:tmp6_5}. Since each of the four interior averages is a convex combination of its class averages and the norm is convex, the four bounds hold with constants multiplied by $(2d_{\mathrm{rec}})^{1/2}$ outside the union of $4d_{\mathrm{rec}}$ events, and $P(\mathcal G)\geq1-(8d_{\mathrm{rec}}+2)\exp(-c_3m\log p)\geq1-10d_{\mathrm{rec}}\exp(-c_3m\log p)$.
For fixed $r_+$, the second term is smaller than the first by a factor of
order $r_+/\sqrt N$, and $N-2r_+\asymp N$. Therefore, after enlarging the constants, we obtain, on $\mathcal G\cap\{\hat K_{\Psi_j}\succ0\text{ for every }j\in[p]\}$,
\[
\norm{\Delta_{\varepsilon}}_{\infty,\hat{\Kcal}_\Psi^*}\leq
C_{\mathrm{dev}}\mathsf K_N\sqrt{m} Bh
\left\{2\sigma^\star+\rho
+Bh\max_{\ell'\in[k]}\norm{\theta_i^{\ell'\star}-\theta_i^{\ell\star}}_1\right\}
\sqrt{\frac{m\log p}{N}}.
\]
Since $P(\mathcal G)\geq1-10d_{\mathrm{rec}}\exp(-c_3m\log p)$, this completes the proof.
\end{proof}

\begin{lemma}
\label{lemma:DB:bmixing}
Let $A = \{A_1, \ldots, A_p\}$ be a collection of $m \times m$ positive definite matrices such that $\max_{i} \sigma_{\max}(A_i) < c_0$ for some constant $c_0 > 0$. For a vector ${\mathsf U} = ({\mathsf U_1}^\top, \ldots, \mathsf U_p^\top)^\top \in \RR^{pm}$ with ${\mathsf U_i}\in \RR^m$, define the block-weighted maximum norm
\[
\|{\mathsf U}\|_{\infty, A} := \max_{i=1, \ldots, p} \left({\mathsf U_i}^\top A_i {\mathsf U_i}\right)^{1/2}.
\]
Suppose that $({\mathsf U_n})_{n=1}^N$ is a strictly stationary, zero-mean sequence in $\RR^{pm}$ that is $\beta$-mixing with coefficients $\beta(n)\leq2\exp(-cn^{\gamma_1})$ for some constants $c>0$ and $\gamma_1>0$. Suppose further that, for each $i = 1, \ldots, p$, the block ${\mathsf U_{n,i}}$ is sub-Weibull$(\gamma_2)$ with sub-Weibull norm at most $K_1$. Define the effective tail parameter by
\[
\frac{1}{\gamma} := \frac{1}{\gamma_1} + \frac{1}{\gamma_2}.
\]
Assume that $\gamma<1$.
Then, there exist constants $C_1, C_2 > 0$ and $C_3\geq2$, depending only on
$\gamma_1$, $\gamma_2$ and $c$, such that if
\[
N \geq C_1m^{2/\gamma-1}
\frac{\{\log(5Np)\}^{2/\gamma}}{\log(5p)},
\]
then
\[
P \left( \left\| \frac{1}{N} \sum_{n=1}^N {\mathsf U_n} \right\|_{\infty, A} \geq C_2 K_1 c_0^{1/2} \sqrt{\frac{m \log p}{N}} \right) \leq 2 \exp(-C_3 m \log p).
\]
\end{lemma}

\begin{proof}[Proof of Lemma~\ref{lemma:DB:bmixing}.]

Define the empirical average ${S_{\mathsf U}} := \frac{1}{N} \sum_{n=1}^N {\mathsf U_n} \in \RR^{pm}$ and let ${S_{\mathsf U_i}} := \frac{1}{N} \sum_{n=1}^N {\mathsf U_{n,i}} \in \RR^m$ denote the $i$-th block. Then,
\[
\|{S_{\mathsf U}}\|_{\infty, A} = \max_{i=1,\ldots,p} \left( {S_{\mathsf U_i}}^\top A_i {S_{\mathsf U_i}} \right)^{1/2} \leq c_0^{1/2} \max_{i} \|{S_{\mathsf U_i}}\|_2 = c_0^{1/2} \|{S_{\mathsf U}}\|_{\infty, 2}.
\]
Let $\mathcal{N}(1/2)$ be a $1/2$-net of the unit sphere in $\RR^m$ (see Chapter 5 of~\citet{wainwright2019high}). Then, for each $i$,
\begin{equation*}
\|{S_{\mathsf U_i}}\|_2 = \sup_{\|u\|_2 \leq 1} u^\top {S_{\mathsf U_i}} \leq \sup_{u \in \mathcal{N}(1/2)} |u^\top {S_{\mathsf U_i}}| + \frac{1}{2} \|{S_{\mathsf U_i}}\|_2.
\end{equation*}
Rearranging gives
\[
\|{S_{\mathsf U_i}}\|_2 \leq 2 \sup_{u \in \mathcal{N}(1/2)} |u^\top {S_{\mathsf U_i}}|.
\]
Hence, for any $t > 0$,
\begin{equation}\label{eq:psw}
\mathbb{P} \left( \|{S_{\mathsf U_i}}\|_2 > 2t \right) \leq \mathbb{P} \left( \sup_{u \in \mathcal{N}(1/2)} |u^\top {S_{\mathsf U_i}}| > t \right).
\end{equation}
Fix $u \in \mathcal{N}(1/2)$. Since $\|u\|_2=1$, the definition of the sub-Weibull norm of a random vector gives $\|u^\top {\mathsf U_{n,i}}\|_{\psi_{\gamma_2}}\leq K_1$. Furthermore, $u^\top {\mathsf U_{n,i}}$ is a measurable function of ${\mathsf U_n}$, so the sequence $(u^\top {\mathsf U_{n,i}})_{n=1}^N$ is strictly stationary with mean zero, and its $\beta$-mixing coefficients are at most those of $({\mathsf U_n})_{n=1}^N$. Let $C_4$ and $C_5$ denote the constants $C_1$ and $C_2$ of Lemma~\ref{lemma:tailbound:mixing}. Enlarging $C_5$ if necessary, we assume that $C_5\geq1$. Lemma~\ref{lemma:tailbound:mixing}, applied to $u^\top {\mathsf U_{n,i}}/K_1$, then gives, for every $t>K_1/N$,
\[
\mathbb{P} \left( |u^\top {S_{\mathsf U_i}}| > t \right) \leq N \exp \left( - \frac{(tN)^\gamma}{K_1^\gamma C_4} \right) + \exp \left( - \frac{t^2 N}{K_1^2 C_5} \right).
\]
Since $|\mathcal{N}(1/2)| \leq 5^m$, a union bound over the net gives
\[
\mathbb{P} \left( \sup_{u \in \mathcal{N}(1/2)} |u^\top {S_{\mathsf U_i}}| > t \right) \leq 5^m N \exp \left( - \frac{(tN)^\gamma}{K_1^\gamma C_4} \right) + 5^m \exp \left( - \frac{t^2 N}{K_1^2 C_5} \right).
\]
Combining this bound with~\eqref{eq:psw} and taking a union bound over $i = 1, \ldots, p$, we obtain
\begin{align*}
\mathbb{P} \left( \|{S_{\mathsf U}}\|_{\infty, A} > 2c_0^{1/2} t \right)
&\leq \sum_{i=1}^p \mathbb{P} \left( \|{S_{\mathsf U_i}}\|_2 > 2t \right) \\
&\leq p \cdot 5^m N \exp \left( - \frac{(tN)^\gamma}{K_1^\gamma C_4} \right) + p \cdot 5^m \exp \left( - \frac{t^2 N}{K_1^2 C_5} \right).
\end{align*}
Now choose
\begin{equation}\label{eq:condition}
t = K_1 \max \left\{ \sqrt{ \frac{3C_5m \log(5p)}{N} },  \frac{1}{N} \left( 3C_4m \log(5Np) \right)^{1/\gamma} \right\}.
\end{equation}
Since $C_5\geq1$, we have $t>K_1/N$. Furthermore, $t^2N/(K_1^2C_5)\geq3m\log(5p)$ and $(tN)^\gamma/(K_1^\gamma C_4)\geq3m\log(5Np)$. Since $m\geq1$, we have $p\cdot5^m\leq(5p)^m$ and $p\cdot5^mN\leq(5Np)^m$, and therefore
\[
\mathbb{P} \left( \|{S_{\mathsf U}}\|_{\infty, A} > 2c_0^{1/2} t \right) \leq (5Np)^{-2m}+(5p)^{-2m}\leq 2 \exp(-2m \log p).
\]
The first entry of the maximum in~\eqref{eq:condition} is the larger one if and only if
\[
N \geq \frac{(3C_4)^{2/\gamma}}{3C_5}\,m^{2/\gamma-1}
\frac{\{\log(5Np)\}^{2/\gamma}}{\log(5p)}.
\]
Let $C_1:=\max\{3,(3C_4)^{2/\gamma}/(3C_5)\}$. The sample-size condition of the lemma then implies this inequality, and it also gives $N>4$, as Lemma~\ref{lemma:tailbound:mixing} requires, since $\{\log(5Np)\}^{2/\gamma}/\log(5p)\geq\log5$. For $p\geq2$, we have $\log(5p)\leq4\log p$, and hence
\[
2c_0^{1/2}t=2c_0^{1/2}K_1\sqrt{\frac{3C_5m\log(5p)}{N}}\leq4\sqrt{3C_5}\,c_0^{1/2}K_1\sqrt{\frac{m\log p}{N}}.
\]
Therefore, the lemma holds with $C_2=4\sqrt{3C_5}$ and $C_3=2$ when $p\geq2$, and its claim is trivial when $p=1$. These constants depend only on $\gamma_1$, $\gamma_2$ and $c$. This completes the proof.
\end{proof}

\section{Proof of the restricted eigenvalue condition}
\label{sec:restricted-eigenvalue}

In this section, we show that, for fixed $\Theta$ and $\ell$, the restricted eigenvalue condition~\eqref{eq:re} holds with high probability under mild conditions on the sample size and the underlying process. Assumption~\ref{assumption:RE} requires this inequality uniformly over $\Theta\in B(r_0,\Theta^\star)$. The following lemma gives a lower bound on the weighted empirical quadratic form, which yields~\eqref{eq:re} at the fixed $(\Theta,\ell)$.

\begin{lemma}\label{lemma:RE:condition}
Fix $\Theta\in\Omega$ and $\ell\in[k]$. Suppose that the conditions of
Lemmas~\ref{lemma:deviation} and~\ref{lemma:empirical:eigenvalue} hold, that $2s\leq p$, and that
\[
mh\left\{C_2B^2h\sqrt{\frac{m\log(9p)}{N}}
+2BD\sqrt{h}\,\delta_1+D^2\delta_1^2\right\}\leq\frac{c_\Psi}{2},
\]
where $C_2$ is the constant of Lemma~\ref{lemma:empirical:eigenvalue}. Define
\[
\sigma_{\mathrm{RE}} = \sigma_{\min}\rbr{\EE\sbr{\frac{1}{N-2r_+}
\sum_{n\in\mathcal I_{r_+}}\hat{w}_{\ell,\Theta}(t_n)
\Psi(t_n)\Psi(t_n)^\top}}>0.
\]
Let $t>4d_{\mathrm{rec}}(N-2r_+)^{-1}$ satisfy
\[
t\leq \frac{\sigma_{\mathrm{RE}}}{4\rbr{27+5400c_\Psi^{-2}}},
\qquad
\delta_1^2\leq\frac{\sigma_{\mathrm{RE}} c_\Psi^2}{1600smD^2h},
\]
and suppose that
\[
N-2r_+\geq \bar C_0\left\{
\frac{B^2h^2sm\{sm\log p+\log(N-2r_+)\}^{3/2}}{t}
\vee
\frac{B^4h^4s^3m^3\log p}{t^2}
\right\}.
\]
Then there is an event $\mathcal E_{\mathrm{RE}}$ with
\[
\PP(\mathcal E_{\mathrm{RE}}^c)\leq2d_{\mathrm{rec}}\exp\rbr{-c_1sm\log p}
+2\exp\rbr{-c_2m\log(9p)}
\]
such that, on $\mathcal E_{\mathrm{RE}}\cap\mathcal E_{\mathrm{rec}}$, where $\mathcal E_{\mathrm{rec}}$ is the event of Section~\ref{ssec:main:results}, we have, for every $x\in\RR^{pm}$,
\[
\frac{1}{N}\sum_{n=1}^N\hat{w}_{\ell,\Theta}(t_n)
\bigl(x^\top\hat{\Psi}(t_n)\bigr)^2
\geq a\|x\|_2^2-\nu\|x\|_{1,\hat{\Kcal}_\Psi}^2,
\]
where
\[
a=\frac{N-2r_+}{2N}(\sigma_{\mathrm{RE}}-27t)
\]
and
\[
\nu=\frac{N-2r_+}{N}\left\{
\frac{27t}{sc_\Psi}+\frac{2mD^2h\delta_1^2}{c_\Psi}
\right\}.
\]
Furthermore, on the same event,
\[
a\geq50s\max_{j\in[p]}\sigma_{\max}(\hat{K}_{\Psi_j})\nu.
\]
Here, $\bar C_0$ is a sufficiently large constant and $c_1,c_2\geq2$ are constants, all depending only on $r_+$, $d_{\mathrm{rec}}$ and the mixing coefficients.
\end{lemma}

\begin{proof}
We proceed in four steps.

\emph{Step 1: Uniform concentration over sparse vectors for the oracle design.}

Write
\[
\KK(2sm)=\left\{x\in\RR^{pm}:\|x\|_0\leq2sm,
\ \|x\|_2\leq1\right\},
\]
as in Lemma~\ref{lemma:sparse2all}. For $x\in\KK(2sm)$, define
\[
\mathsf H(t_n,x)=\sqrt{\hat{w}_{\ell,\Theta}(t_n)}\,
x^\top\Psi(t_n),\qquad n\in\mathcal I_{r_+}.
\]
By Lemma~\ref{lemma:local:blocks}, for $n\in\mathcal I_{r_+}$, $\hat{w}_{\ell,\Theta}(t_n)$ is a measurable function of $(W_{n-r_+},\ldots,W_{n+r_+})$ that depends on $n$ only through its remainder modulo $d_{\mathrm{rec}}$, and $\Psi(t_n)$ is measurable with respect to the $n$th interval block. We give Steps~1--4 for $d_{\mathrm{rec}}=1$, in which case $\{\mathsf H(t_n,x)\}$ is a finite segment of a strictly stationary, geometrically $\beta$-mixing sequence, with constants depending on the original mixing constants and $r_+$. The end of the proof treats $d_{\mathrm{rec}}>1$.

Since $0\leq\hat{w}_{\ell,\Theta}(t_n)\leq1$,
$\|\Psi(t_n)\|_\infty\leq Bh$, and $x\in\KK(2sm)$,
\[
|\mathsf H(t_n,x)|\leq Bh\|x\|_1\leq Bh\sqrt{2sm}.
\]
Hence,
\[
\norm{\{\mathsf H(t_n,x)\}^2}_{\psi_2}
\leq2B^2h^2sm.
\]
By Lemma~\ref{lemma:mean:subweibull} and the triangle inequality,
\[
\norm{\{\mathsf H(t_n,x)\}^2-
\EE[\{\mathsf H(t_n,x)\}^2]}_{\psi_2}
\leq4B^2h^2sm.
\]
Since the centered sequence is sub-Weibull$(2)$,
Lemma~\ref{lemma:tailbound:mixing} with $\gamma_1=1$, $\gamma_2=2$,
and $\gamma=2/3$ gives, for every $t>2(N-2r_+)^{-1}$,
\begin{align*}
&\PP\left(\left|\frac{1}{N-2r_+}\sum_{n\in\mathcal I_{r_+}}
\left[\{\mathsf H(t_n,x)\}^2-
\EE\{\mathsf H(t_n,x)\}^2\right]\right|\geq\frac{t}{2}\right)\\
&\qquad\leq(N-2r_+)\exp\left\{-c\left(
\frac{t(N-2r_+)}{B^2h^2sm}\right)^{2/3}\right\}
+\exp\left\{-c\frac{t^2(N-2r_+)}{B^4h^4s^2m^2}\right\},
\end{align*}
for a constant $c>0$.  Set
\[
\mathsf A=
\frac{1}{N-2r_+}\sum_{n\in\mathcal I_{r_+}}
\left[\hat w_{\ell,\Theta}(t_n)\Psi(t_n)\Psi(t_n)^\top
-\EE\{\hat w_{\ell,\Theta}(t_n)\Psi(t_n)\Psi(t_n)^\top\}\right].
\]

For each coordinate support $\Scal\subset\{1,\ldots,pm\}$ with
$|\Scal|=2sm$, let $\Ncal(1/4,\Scal)$ be a $1/4$-net of the unit sphere
of vectors supported on $\Scal$.  Since
\[
|\Ncal(1/4,\Scal)|\leq9^{2sm},
\qquad
\binom{pm}{2sm}9^{2sm}\leq\left(\frac{9ep}{2s}\right)^{2sm},
\]
a union bound gives
\begin{align*}
\PP\bigg(
&\max_{\substack{\Scal\subset[pm],\ |\Scal|=2sm\\
u\in\Ncal(1/4,\Scal)}}|u^\top\mathsf A u|
\geq\frac{t}{2}
\bigg)\\
\leq{}&(N-2r_+)\exp\left\{-c\left(
\frac{t(N-2r_+)}{B^2h^2sm}\right)^{2/3}+C_{\mathrm{ent}}sm\log p\right\}\\
&+\exp\left\{-c\frac{t^2(N-2r_+)}{B^4h^4s^2m^2}
+C_{\mathrm{ent}}sm\log p\right\},
\end{align*}
for a universal constant $C_{\mathrm{ent}}\geq2$.
Since $2s\leq p$, every $x\in\KK(2sm)$ is supported on some $\Scal$ with $|\Scal|=2sm$. The argument at the end of the proof of Lemma~\ref{lemma:oracle:gram:concentration}, applied to the principal submatrix of the symmetric matrix $\mathsf A$ indexed by each such $\Scal$, therefore gives
\[
\sup_{x\in\KK(2sm)}|x^\top\mathsf A x|
\leq2\max_{\substack{\Scal\subset[pm],\ |\Scal|=2sm\\
u\in\Ncal(1/4,\Scal)}}|u^\top\mathsf A u|.
\]
Outside the event in the union bound above, the last two displays give
\[
\sup_{x\in\KK(2sm)}|x^\top\mathsf A x|\leq t.
\]

\emph{Step 2: Extension from sparse vectors to all vectors.}

Outside the event in Step 1, Lemma~\ref{lemma:sparse2all} gives, for all
$x\in\RR^{pm}$,
\[
\left|\frac{1}{N-2r_+}\sum_{n\in\mathcal I_{r_+}}
\left[\{\mathsf H(t_n,x)\}^2-
\EE\{\mathsf H(t_n,x)\}^2\right]\right|
\leq27t\left(\|x\|_2^2+\frac{\|x\|_1^2}{sm}\right).
\]
Therefore,
\begin{align*}
\frac{1}{N-2r_+}\sum_{n\in\mathcal I_{r_+}}\hat{w}_{\ell,\Theta}(t_n)
\{x^\top\Psi(t_n)\}^2
\geq(\sigma_{\mathrm{RE}}-27t)\|x\|_2^2-\frac{27t}{sm}\|x\|_1^2.
\end{align*}
On $\mathcal E_{\mathrm G}\cap\mathcal E_{\mathrm{rec}}$, where $\mathcal E_{\mathrm G}$ is the event of Lemma~\ref{lemma:empirical:eigenvalue}, that lemma gives
\[
\min_{j\in[p]}\sigma_{\min}(\hat{K}_{\Psi_j})\geq\frac{c_\Psi}{2},
\]
and hence
\[
\|x\|_1^2\leq\frac{2m}{c_\Psi}\|x\|_{1,\hat{\Kcal}_\Psi}^2.
\]
It follows that
\begin{align*}
\frac{1}{N-2r_+}\sum_{n\in\mathcal I_{r_+}}\hat{w}_{\ell,\Theta}(t_n)
\{x^\top\Psi(t_n)\}^2
\geq(\sigma_{\mathrm{RE}}-27t)\|x\|_2^2-
\frac{54t}{sc_\Psi}\|x\|_{1,\hat{\Kcal}_\Psi}^2.
\end{align*}

\emph{Step 3: Transfer from $\Psi$ to $\hat{\Psi}$.}

As in the proof of Lemma~\ref{lemma:empirical:eigenvalue}, on $\mathcal E_{\mathrm{rec}}$, the bound on $g'$ gives
\[
\max_{n\in[N]}\|\hat{\Psi}(t_n)-\Psi(t_n)\|_\infty
\leq D\sqrt{h}\,\delta_1,
\]
so that $|x^\top\{\hat{\Psi}(t_n)-\Psi(t_n)\}|\leq D\sqrt{h}\,\delta_1\|x\|_1$. Since $(\alpha+\beta)^2\geq\alpha^2/2-\beta^2$ for all $\alpha,\beta\in\RR$, we obtain
\[
\{x^\top\hat{\Psi}(t_n)\}^2\geq
\frac12\{x^\top\Psi(t_n)\}^2-D^2h\delta_1^2\|x\|_1^2.
\]
Combining this inequality with Step 2, we obtain
\begin{align*}
\frac{1}{N-2r_+}\sum_{n\in\mathcal I_{r_+}}\hat{w}_{\ell,\Theta}(t_n)
\{x^\top\hat{\Psi}(t_n)\}^2
\geq{}&\frac{\sigma_{\mathrm{RE}}-27t}{2}\|x\|_2^2\\
&-\left\{\frac{27t}{sc_\Psi}+
\frac{2mD^2h\delta_1^2}{c_\Psi}\right\}
\|x\|_{1,\hat{\Kcal}_\Psi}^2.
\end{align*}
Since the boundary summands are nonnegative,
\begin{align*}
\frac1N\sum_{n=1}^N\hat{w}_{\ell,\Theta}(t_n)
\{x^\top\hat{\Psi}(t_n)\}^2
\geq{}&\frac{N-2r_+}{2N}(\sigma_{\mathrm{RE}}-27t)\|x\|_2^2\\
&-\frac{N-2r_+}{N}\left\{\frac{27t}{sc_\Psi}+
\frac{2mD^2h\delta_1^2}{c_\Psi}\right\}
\|x\|_{1,\hat{\Kcal}_\Psi}^2.
\end{align*}

\emph{Step 4: Parameter selection.}

First, by the first display of the statement and Lemma~\ref{lemma:empirical:eigenvalue}, on $\mathcal E_{\mathrm G}\cap\mathcal E_{\mathrm{rec}}$, we have
\[
\min_{j\in[p]}\sigma_{\min}(\hat{K}_{\Psi_j})
\geq\frac{c_\Psi}{2},
\qquad
\max_{j\in[p]}\sigma_{\max}(\hat{K}_{\Psi_j})\leq2c_\Psi^{-1}.
\]

The conditions on $t$ and $\delta_1$ in the statement give
\[
27t+\frac{5400t}{c_\Psi^2}+
\frac{400smD^2h\delta_1^2}{c_\Psi^2}\leq\sigma_{\mathrm{RE}},
\]
which yields
\[
a\geq50s\max_{j\in[p]}\sigma_{\max}(\hat{K}_{\Psi_j})\nu.
\]

Finally, choose $N$ so that
\[
N-2r_+\geq \bar C_0\left\{
\frac{B^2h^2sm\{sm\log p+\log(N-2r_+)\}^{3/2}}{t}
\vee
\frac{B^4h^4s^3m^3\log p}{t^2}
\right\}.
\]
If $\bar C_0\geq\max\{(2C_{\mathrm{ent}}/c)^{3/2},2C_{\mathrm{ent}}/c\}$, the first term in the Step 1 probability bound is then at most
$\exp\{-C_{\mathrm{ent}}sm\log p\}$ because
\[
c\left(\frac{t(N-2r_+)}{B^2h^2sm}\right)^{2/3}
\geq2C_{\mathrm{ent}}sm\log p+\log(N-2r_+),
\]
which absorbs the factor $N-2r_+$. Likewise, the second term is at most
$\exp\{-C_{\mathrm{ent}}sm\log p\}$ because
\[
c\frac{t^2(N-2r_+)}{B^4h^4s^2m^2}
\geq2C_{\mathrm{ent}}sm\log p.
\]
Let $\mathcal E_{\mathrm{RE}}$ be the intersection of the complement of the event in Step~1 with the event $\mathcal E_{\mathrm G}$ of
Lemma~\ref{lemma:empirical:eigenvalue}. Then $\PP(\mathcal E_{\mathrm{RE}}^c)$ satisfies the claimed bound, with $c_1=C_{\mathrm{ent}}\geq2$ and with $c_2$ the constant $C_3\geq2$ in the exponent of the probability bound of Lemma~\ref{lemma:empirical:eigenvalue}, and Steps~2--4 hold on $\mathcal E_{\mathrm{RE}}\cap\mathcal E_{\mathrm{rec}}$.

If $d_{\mathrm{rec}}>1$, we apply Step~1 to each class $\mathcal I_{r_+,c}$, $c=0,\ldots,d_{\mathrm{rec}}-1$, with the average $\mathsf A_c$ of the summands of $\mathsf A$ over the class in place of $\mathsf A$. By Lemma~\ref{lemma:local:blocks}, each class gives a finite segment of a strictly stationary, geometrically $\beta$-mixing sequence, and it has $n_c\geq(N-2r_+)/(2d_{\mathrm{rec}})$ elements once $N-2r_+\geq d_{\mathrm{rec}}$, so that $t>4d_{\mathrm{rec}}(N-2r_+)^{-1}\geq2n_c^{-1}$. The tail bound of Step~1 then holds in each class with $n_c$ in place of $N-2r_+$, and the choice of $N$ in Step~4, with $\bar C_0$ enlarged by a factor that depends only on $d_{\mathrm{rec}}$, makes each class fail with probability at most $2\exp\{-C_{\mathrm{ent}}sm\log p\}$, since $\log n_c\leq\log(N-2r_+)$. Since $\mathsf A$ is a convex combination of the matrices $\mathsf A_c$, the bound $\sup_{x\in\KK(2sm)}|x^\top\mathsf Ax|\leq t$ holds outside the union of these $d_{\mathrm{rec}}$ events, and Steps~2--4 apply unchanged. This gives the factor $d_{\mathrm{rec}}$ in the probability bound and completes the proof.
\end{proof}

\section{Proof of main results}

In this section, we prove Theorem~\ref{theorem:graphrecovery} in Section~\ref{ssec:mainproof}, Corollary~\ref{corollary:main} in Section~\ref{ssec:corollary:main} and Corollary~\ref{corollary:full} in Section~\ref{ssec:corollary:full}. The proof of Theorem~\ref{theorem:graphrecovery} extends the one-step bounds of Lemmas~\ref{lemma:onestep:theta}, \ref{lemma:onestep:sigma} and~\ref{lemma:onestep:q} to multiple iterations. Lemma~\ref{lemma:oneupdate:theta:2} shows that, at each fixed $\Theta\in B(r_0,\Theta^\star)$, $i\in[p]$ and $\ell\in[k]$, the restricted eigenvalue and deviation bound conditions that Lemma~\ref{lemma:onestep:theta} requires hold on the reconstruction event $\mathcal E_{\mathrm{rec}}$, outside an event of small probability, under the setting of Theorem~\ref{theorem:graphrecovery}, and Lemma~\ref{lemma:oneupdate:sigma:2} verifies the conditions of Lemma~\ref{lemma:onestep:sigma} in the same way. We state and prove these two lemmas in Sections~\ref{ssec:oneupdate:theta:2} and~\ref{ssec:oneupdate:sigma:2}.

\subsection{Proof of Theorem~\ref{theorem:graphrecovery}}
\label{ssec:mainproof}
We first state Theorem~\ref{theorem:graphrecovery} with its constants, conditions and penalty written out.

\begin{theorem}[Explicit form of Theorem~\ref{theorem:graphrecovery}]
\label{theorem:graphrecovery:full}
Suppose that Assumptions~\ref{assumption:edgeset}--\ref{assumption:population}, \ref{assumption:mixing},
\ref{assumption:eigenvalue} and~\ref{assumption:local} hold, and let $m,s\geq1$ and $2s\leq p$.
Let $R_X$ be as in Assumption~\ref{assumption:local}, and let $B,D>0$ satisfy $\max_{j\in[m]}\sup_{|x|\leq R_X}|g_j(x)|\leq B$ and $\max_{j\in[m]}\sup_{|x|\leq R_X}|g_j'(x)|\leq D$. Let
\[
\underline\tau:=\inf_{Q\in\Omega_Q}\min_{\ell,a,b\in[k]}\EE[\tau_\ell(h)\mid Z(0)=a,Z(h)=b;Q].
\]
 Fix a sufficiently large truncation radius $r\geq1$, independent of $N$, such that the truncation condition in Lemma~\ref{lemma:onestep:q} holds with $\underline\tau_\Theta$ replaced by $\underline\tau$. Suppose that the regularity conditions of
Lemma~\ref{lemma:aggregate-surrogate-bias} hold, and that
the deterministic regularity conditions of
Lemmas~\ref{lemma:oneupdate:theta:2}, \ref{lemma:oneupdate:sigma:2},
\ref{lemma:onestep:q}, and~\ref{lemma:sc_w} hold throughout
$B(r_0,\Theta^\star)$ with common constants.
Assume that the concentration bounds used in Lemmas~\ref{lemma:oneupdate:theta:2}, \ref{lemma:oneupdate:sigma:2} and~\ref{lemma:sc_w}, including the one used to derive the restricted eigenvalue condition in Lemma~\ref{lemma:oneupdate:theta:2}, hold simultaneously for all $\Theta\in B(r_0,\Theta^\star)$, on events that do not depend on $\Theta$, intersected with $\mathcal E_{\mathrm{rec}}$ defined in Section~\ref{ssec:main:results} as in those lemmas, with common constants and the probability bounds specified in the proof below.
Let $\mathcal C_\varepsilon$, $\mathcal C_w$, and
$\mathcal C_{\mathrm{rec}}$ be as defined in
Lemma~\ref{lemma:oneupdate:theta:2}, and let $c_{\mathrm{RE}}>0$ be the constant in that lemma, which satisfies
$\sigma_{\min}\bigl(\EE[(N-2r_+)^{-1}\sum_{n=r_++1}^{N-r_+}\hat w_{\ell,\Theta}(t_n)\Psi(t_n)\Psi(t_n)^\top]\bigr)\geq c_{\mathrm{RE}}$ for every $\Theta\in B(r_0,\Theta^\star)$ and $\ell\in[k]$. That is, $c_{\mathrm{RE}}$ is a lower bound on the curvature of the population Gram matrix weighted by the truncated smoothing probabilities $\hat w_{\ell,\Theta}(t_n)$, from which the restricted eigenvalue condition is derived. Assume, as in that lemma, that $c_{\mathrm{RE}}^{-1}\sigma_{\max}(K_\Psi)\leq c_\Psi^{-1}$.
Let $a_{r_+}$ and $C_Q$ be the constants of Lemmas~\ref{lemma:sc_w} and~\ref{lemma:onestep:q}, respectively, with $C_Q$ evaluated at $\underline\tau$. Assume that
$N\gtrsim\max\{r_+,\,m^4\{\log(Np)\}^5/\log p,\,
(sm)^3\log p,\,k^4\}$, where the implicit constant depends only on $B$, $h$, $c_\Psi$, $c_{\mathrm{RE}}$, $r_+$, $d_{\mathrm{rec}}$ and the $\beta$-mixing rate in Assumption~\ref{assumption:mixingbeta}, and that
\[
N\geq64\Bigl(\frac{ha_{r_+}}{\underline\tau}\Bigr)^2\bigl(k^2\log2+2\log N\bigr)
\qquad\text{and}\qquad
N\geq\frac{32r_+h}{\underline\tau}.
\]
Let $C_0,C_1,C_2>0$ be sufficiently large constants that depend only on
the $\beta$-mixing rate and on
$(h,\Theta^\star,m,s,r_0,B,\rho,D,r,r_+,d_{\mathrm{rec}},\delta_{\min},c_\Psi)$, and let $C_3\geq32$ and $C_4\geq160\sqrt2$ be universal constants.
Let $\delta_1\geq0$ and $\mathcal E_{\mathrm{rec}}$ be as in Section~\ref{ssec:main:results}, and suppose that $\delta_1$ satisfies
$mh\{2BD\sqrt h\,\delta_1+D^2\delta_1^2\}\leq c_\Psi/4$
and $\delta_1\leq c_\Psi\sqrt{c_{\mathrm{RE}}}/(40D\sqrt{smh})$.
For an initialization $\Theta^{(0)}\in B(r_0,\Theta^\star)$, let $d_0:=\dist(\Theta^{(0)},\Theta^\star)$, and define the minimal penalty level
\[
\lambda_N:=C_\lambda\sqrt m\max\cbr{
\mathcal C_w k\cbr{1-(\zeta\pi_{\min})^2}^{r-1},\;
\mathcal C_{\mathrm{rec}}\sqrt{ms}\delta_1,\;
\mathcal C_\varepsilon\sqrt{\frac{m\log p}{N}}},
\]
where $C_\lambda>0$ is a sufficiently large constant that depends only on $r_+$, $d_{\mathrm{rec}}$ and the $\beta$-mixing rate. Define $\operatorname{err}_{N,\theta}:=C_4c_{\mathrm{RE}}^{-1}c_\Psi^{-1/2}\sqrt{s}\lambda_N$ and
\[
\begin{aligned}
\operatorname{err}_{N,\sigma^2}
:={}&
C_0k\left\{\sqrt{mh}\,\delta_1
+mh\delta_1^2\right\}
+C_1k\sqrt{\frac{\log N}{N}}+
C_2k^2\zeta^{-8}\pi_{\min}^{-2}
\rbr{1-(\zeta\pi_{\min})^2}^{r-1},\\
\operatorname{err}_{N,Q}
:={}&
C_Qk(k-1)\Biggl[
\frac{N-2r_+}{N}a_{r_+}
\sqrt{
\frac{k^2\log 2+2\log(N-2r_+)}{N-2r_+}
}
+\frac{4r_+}{N}\\
&\hspace{8em}+
k^3\delta_{\min}^{-1}\zeta^{-8}\pi_{\min}^{-2}
\rbr{1-(\zeta\pi_{\min})^2}^{r-1}
\Biggr].
\end{aligned}
\]
Set $\operatorname{err}_N:=\operatorname{err}_{N,\theta}
+\operatorname{err}_{N,\sigma^2}+\operatorname{err}_{N,Q}$. Let
\[
C_\kappa:=\max\cbr{\frac{C_3\sigma_{\max}(K_\Psi)}{c_{\mathrm{RE}}},1}+\frac{C_4C_\lambda\sigma_{\max}(K_\Psi)}{c_{\mathrm{RE}}\,c_\Psi},
\qquad
\bar\kappa:=C_\kappa\kappa,
\]
$e_N:=\operatorname{err}_N+C_\kappa b_{\mathrm{surr}}(\Theta^\star)$ and $d_t:=\bar\kappa^td_0+(1-\bar\kappa^t)e_N/(1-\bar\kappa)$ for $t\geq0$. At iteration $t$, use the same penalty for every $i\in[p]$ and $\ell\in[k]$,
\[
\lambda_i^\ell=\lambda^{(t)}:=\max\cbr{\lambda_N,\;
C_\lambda\frac{\sigma_{\max}(K_\Psi)}{\sqrt{c_\Psi s}}
\sbr{\kappa d_t+b_{\mathrm{surr}}(\Theta^\star)}},
\]
and write $\Theta^{(t+1)}=\tilde M_N\bigl(\Theta^{(t)};\lambda^{(t)}\bigr)$ for one truncated update. Assume that $\bar\kappa<1$ and $e_N\leq(1-\bar\kappa)r_0$.
Let $c\geq2$ be the smallest of the four constants in the exponents of the probability bound of Lemma~\ref{lemma:oneupdate:theta:2}, and define
$\delta_e:=48d_{\mathrm{rec}}\max\{kp^{-csm},\,kp^{1-cm},\,kN^{-4}\}$. Then the following conclusions hold
with probability at least $1-3\delta_e-\PP(\mathcal E_{\mathrm{rec}}^c)$. First,
$\Theta^{(t)}\in B(r_0,\Theta^\star)$ for every $t\geq0$.
Second, for every $L\geq0$,
\[
\dist\bigl(\Theta^{(L)},\Theta^\star\bigr)
\leq d_L=\bar\kappa^Ld_0+\frac{1-\bar\kappa^L}{1-\bar\kappa}\,e_N.
\]

\end{theorem}

\emph{Reading the rate for $\theta_i^\ell$.}
The three entries in $\lambda_N$ give, respectively, the truncation, reconstruction and statistical errors. More explicitly,
\[
\begin{aligned}
\operatorname{err}_{N,\theta}
={}&C_4C_\lambda c_{\mathrm{RE}}^{-1}c_\Psi^{-1/2}
\max\Biggl\{
\mathcal C_w k\sqrt{ms}\{1-(\zeta\pi_{\min})^2\}^{r-1}, \mathcal C_{\mathrm{rec}}ms\delta_1,\quad
\mathcal C_\varepsilon m\sqrt{\frac{s\log p}{N}}
\Biggr\}.
\end{aligned}
\]
The constants $\mathcal C_w$, $\mathcal C_{\mathrm{rec}}$ and $\mathcal C_\varepsilon$ of Lemma~\ref{lemma:oneupdate:theta:2} depend on $(B,D,h,c_\Psi,\sigma^\star,\rho,\zeta,\pi_{\min},\delta_{\min},m,s)$ and on $\Theta^\star$ through row-wise maxima of the norms of its blocks $\theta_{ij}^{\ell\star}$. The constants $C_1$ and $C_2$ in $\operatorname{err}_{N,\sigma^2}$ grow at least linearly in $ms$.

\emph{The cost of the observation window.}
The truncation radius $r$ and the observation-window radius $r_+=r+r_{\mathrm{rec}}$ have different roles. The truncation error bound contains the factor $\{1-(\zeta\pi_{\min})^2\}^{r-1}$. Concentration uses the $N-2r_+$ interior indices and constants depending on $r_+$ and $d_{\mathrm{rec}}$; in $\operatorname{err}_{N,Q}$, these give the interior statistical term and the boundary term $4r_+/N$. For fixed $r$, $r_{\mathrm{rec}}$ and $d_{\mathrm{rec}}$, we have $N-2r_+\asymp N$, so the square-root order of the statistical term is unchanged. The quantity $N-2r_+$ counts interior observations, not independent ones; the dependence is handled by the mixing bounds. The constants $C_\lambda$, $a_{r_+}$ and $c$ depend on $r_+$, $d_{\mathrm{rec}}$ and the mixing rate, so the contraction condition $\bar\kappa<1$ also depends on the reconstruction radius $r_{\mathrm{rec}}$ and on $d_{\mathrm{rec}}$. The constant $C_Q$ depends only on $h$, $\delta_{\min}$, $\underline\tau$ and $\max_{a\neq b}q^\star_{ab}+r_0$.

\emph{Theorem~\ref{theorem:graphrecovery:full} implies Theorem~\ref{theorem:graphrecovery}.}
Write $\mathsf u_N:=\{\log(Np)/N\}^{1/2}$ and $\mathsf e_r:=\{1-(\zeta\pi_{\min})^2\}^{r-1}$ in this paragraph only, and note that $\log p\leq\log(Np)$, $\log N\leq\log(Np)$ and $N^{-1}\leq\mathsf u_N$ for $Np\geq3$. Bounding the maximum in $\lambda_N$ by the sum of its entries gives
\[
\operatorname{err}_{N,\theta}\leq C_4C_\lambda c_{\mathrm{RE}}^{-1}c_\Psi^{-1/2}\sqrt{ms}\,\bigl\{\mathcal C_wk\mathsf e_r+\mathcal C_{\mathrm{rec}}\sqrt{ms}\,\delta_1+\mathcal C_\varepsilon\sqrt m\,\mathsf u_N\bigr\}.
\]
The two conditions on $\delta_1$ in Theorem~\ref{theorem:graphrecovery:full} hold when $\delta_1\leq c_\delta$ for a constant $c_\delta>0$ that depends on $(m,s,h,B,D,c_\Psi,c_{\mathrm{RE}})$, and then $mh\delta_1^2\leq mhc_\delta\delta_1$. Hence $\operatorname{err}_{N,\sigma^2}\leq C_0k(\sqrt{mh}+mhc_\delta)\delta_1+C_1k\mathsf u_N+C_2k^2\zeta^{-8}\pi_{\min}^{-2}\mathsf e_r$. In $\operatorname{err}_{N,Q}$, we have $k^2\log2+2\log(N-2r_+)\leq(k^2+2)\log(Np)$ and $(N-2r_+)^{1/2}/N\leq N^{-1/2}$, so that the interior statistical term is at most $a_{r_+}(k^2+2)^{1/2}\mathsf u_N$, and $4r_+/N\leq4r_+\mathsf u_N$. Adding these bounds and $C_\kappa b_{\mathrm{surr}}(\Theta^\star)$ gives
\[
e_N\leq C_e\bigl\{\mathsf u_N+\delta_1+\mathsf e_r+b_{\mathrm{surr}}(\Theta^\star)\bigr\},
\]
where $C_e$ is the largest of the resulting coefficients of $\mathsf u_N$, $\delta_1$, $\mathsf e_r$ and $b_{\mathrm{surr}}(\Theta^\star)$, and does not depend on $N$ or $\delta_1$. The sample-size conditions of Theorem~\ref{theorem:graphrecovery:full} are implied by $N\geq C_{\mathrm{ss}}\max\{m^4\{\log(Np)\}^5/\log p,\,(sm)^3\log p,\,k^4,\,r_+\}$ for a sufficiently large $C_{\mathrm{ss}}$: the two explicit conditions require $N\gtrsim k^2+\log N$ and $N\gtrsim r_+$, with constants depending on $h$, $a_{r_+}$ and $\underline\tau$, and $m^4\{\log(Np)\}^5/\log p\geq\{\log(Np)\}^4\geq\log(Np)\geq\log N$. Finally, since $c\geq2$ and $m,s\geq1$, we have $3\delta_e\leq144d_{\mathrm{rec}}k\max\{p^{-1},N^{-4}\}\leq144d_{\mathrm{rec}}k(p^{-1}+N^{-4})$. Therefore, the conclusions of Theorem~\ref{theorem:graphrecovery:full} give those of Theorem~\ref{theorem:graphrecovery}, with $C_\kappa$ and $e_N$ as defined in Theorem~\ref{theorem:graphrecovery:full}.

We now prove Theorem~\ref{theorem:graphrecovery:full}.

\begin{proof}[Proof of Theorem~\ref{theorem:graphrecovery:full}]
We first fix the constants and verify the sample-size conditions. We then construct the event on which we work and prove both claims of the theorem by induction on the iteration count.

Define
\[
\mathsf D
:=
2\sigma^{\star2}+\rho^2
+2B^2h^2
\max_{\substack{\ell,\ell'\in[k]\\\ell\neq\ell'}}
\max_{i\in[p]}
\norm{\theta_i^{\ell'\star}-\theta_i^{\ell\star}}_1^2
+2C_{\mathrm{cone}}^2msB^2h^2r_0^2.
\]
Let $\tilde C_0,\tilde C_1,\tilde C_2>0$ denote the constants of
Lemma~\ref{lemma:oneupdate:sigma:2}. We absorb the factors involving $\mathsf D$
into the constants of the theorem by choosing positive $C_0,C_1,C_2$ such that
\[
C_0\geq\tilde C_0\max\{1,\sqrt{\mathsf D}\},\qquad
C_1\geq\tilde C_1\mathsf D,\qquad
C_2\geq\tilde C_2\mathsf D.
\]
By Lemma~\ref{lemma:oneupdate:sigma:2} and the definition of $\mathsf D$, these constants can be chosen to depend only on $(h,\Theta^\star,m,s,r_0,B,\rho,D,r,r_+,d_{\mathrm{rec}},\delta_{\min},c_\Psi)$ and the mixing constants. Since $\mathsf D\geq2C_{\mathrm{cone}}^2msB^2h^2r_0^2$, the smallest such $C_1$ and $C_2$ grow at least linearly in $ms$.

\emph{Sample-size verification.}
The proof uses Lemmas~\ref{lemma:aggregate-surrogate-bias},
\ref{lemma:onestep:theta}, \ref{lemma:oneupdate:sigma:2},
\ref{lemma:sc_w}, and~\ref{lemma:onestep:q}, and
Lemma~\ref{lemma:oneupdate:theta:2} verifies the conditions needed for the
update of $\theta_i^\ell$. It remains to check their sample-size
requirements. The concentration bound for $\Delta_\varepsilon$, the Gram-matrix bounds and the restricted eigenvalue condition
require
\[
\begin{aligned}
N-2r_+&\gtrsim
m^4\frac{\{\log(5(N-2r_+)p)\}^5}{\log(5p)}\vee(sm)^{5/2}(\log p)^{3/2}\\
&\qquad\vee sm\{\log(N-2r_+)\}^{3/2}\vee(sm)^3\log p,\\
N&\gtrsim m^4(\log p)^3\vee(sm)^{5/2},
\end{aligned}
\]
as follows from Lemmas~\ref{lemma:deviation},
\ref{lemma:oracle:gram:concentration},
\ref{lemma:empirical:eigenvalue}, \ref{lemma:RE:condition},
and~\ref{lemma:oneupdate:theta:2}. Here and in the rest of this verification, the implicit constants depend only on $B$, $h$, $c_\Psi$, $c_{\mathrm{RE}}$, $r_+$, $d_{\mathrm{rec}}$ and the mixing constants. Indeed, since its $\sigma_{\mathrm{RE}}$ is at least $c_{\mathrm{RE}}$, Lemma~\ref{lemma:RE:condition} applies with its threshold $t$ equal to $\epsilon_{\mathrm{RE}}:=c_{\mathrm{RE}}/\{4(27+5400c_\Psi^{-2})\}$ once $N-2r_+>4d_{\mathrm{rec}}/\epsilon_{\mathrm{RE}}$, and its sample-size condition then involves $B^2h^2/\epsilon_{\mathrm{RE}}$ and $B^4h^4/\epsilon_{\mathrm{RE}}^2$; since $\{sm\log p+\log(N-2r_+)\}^{3/2}\leq\sqrt2[(sm\log p)^{3/2}+\{\log(N-2r_+)\}^{3/2}]$, it is implied by the first relation of the preceding display. Furthermore, under the condition $mh\{2BD\sqrt h\,\delta_1+D^2\delta_1^2\}\leq c_\Psi/4$ of the theorem, the eigenvalue condition of Lemma~\ref{lemma:oneupdate:theta:2} holds once $N\gtrsim B^4h^4m^3\log(9p)/c_\Psi^2$. The $r_+$ term in the sample-size condition of the theorem gives $N\geq4r_+$, and hence
$N-2r_+\asymp N$ and $N-2r_+>4$. Furthermore, for $m,s\geq1$ and $p\geq2$, we have
\[
\begin{aligned}
m^4\frac{\{\log(Np)\}^5}{\log p}
&\gtrsim m^4(\log p)^3\vee(\log N)^2,\\
(sm)^3\log p\vee
m^4\frac{\{\log(Np)\}^5}{\log p}
&\gtrsim(sm)^{5/2}\vee(sm)^{5/2}(\log p)^{3/2}\vee sm(\log N)^{3/2}.
\end{aligned}
\]
For the last term, $sm(\log N)^{3/2}\leq(sm)^2\vee(\log N)^3$, where $(sm)^2\leq(sm)^3\log p/\log2$ and, since $N\geq4$, $(\log N)^3\leq\{\log(Np)\}^4\leq\{\log(Np)\}^5/\log p$.
Therefore, the first three terms in the sample-size condition of the theorem imply all
sample-size conditions of Lemma~\ref{lemma:oneupdate:theta:2}.

Lemma~\ref{lemma:sc_w} additionally requires
$N-2r_+\gtrsim\{k^2\log 2+2\log(N-2r_+)\}^2$, which is implied by the $k^4$ term
and the logarithmic term above.
For Lemma~\ref{lemma:oneupdate:sigma:2}, we also verify the condition
from Lemma~\ref{lemma:sc_T2}, namely
\[
N-2r_+\geq C_2(r_+)\{5\log(N-2r_+)\}^4,
\]
where $C_2(r_+)$ is the sample-size constant in that lemma. For $m\geq1$
and $p\geq2$,
\[
m^4\frac{\{\log(Np)\}^5}{\log p}
\geq\{\log(Np)\}^4\geq(\log N)^4.
\]
Since $N-2r_+\geq N/2$, the theorem's sample-size condition implies the
required inequality when its implicit constant is chosen sufficiently
large, depending on $r_+$, $d_{\mathrm{rec}}$ and the mixing constants.
Finally, Lemma~\ref{lemma:aggregate-surrogate-bias} imposes no additional
sample-size condition, so all required sample-size conditions hold.

\emph{Probability bound.}
We use the uniform-concentration hypothesis of Theorem~\ref{theorem:graphrecovery:full} in the following precise form. There exist measurable events $\mathcal E_\theta$, $\mathcal E_{\sigma^2}$ and $\mathcal E_Q$ that do not depend on $\Theta$ or the iteration index. On $\mathcal E_\theta\cap\mathcal E_{\mathrm{rec}}$, the restricted eigenvalue condition and the empirical block-Gram and deviation bounds in Step~1 of the proof of Lemma~\ref{lemma:oneupdate:theta:2} hold simultaneously for every $\Theta\in B(r_0,\Theta^\star)$, $i\in[p]$ and $\ell\in[k]$. On $\mathcal E_{\sigma^2}\cap\mathcal E_{\mathrm{rec}}$, the bound of Lemma~\ref{lemma:oneupdate:sigma:2} holds for every $\Theta\in B(r_0,\Theta^\star)$. On $\mathcal E_Q$, the bound on $\delta_{3,1}$ in Lemma~\ref{lemma:sc_w} holds for every $\Theta\in B(r_0,\Theta^\star)$. All bounds use constants common to the local region, with the dependencies stated in the corresponding lemmas. The assumed probability bound for $\mathcal E_\theta$ is
\[
\begin{aligned}
\PP(\mathcal E_\theta^c)
\leq{}&
k\left[2d_{\mathrm{rec}}\exp\{-c_1sm\log p\}
+2\exp\{-c_2m\log(9p)\}\right]\\
&+2p\exp\{-c_4N\}
+10pkd_{\mathrm{rec}}\exp\{-c_3m\log p\}.
\end{aligned}
\]
The same assumption gives $\PP(\mathcal E_{\sigma^2}^c)\leq3kd_{\mathrm{rec}}(N-2r_+)^{-4}$ and $\PP(\mathcal E_Q^c)\leq2d_{\mathrm{rec}}(N-2r_+)^{-4}$, together with the variance-update bound and the concentration bound for the truncated joint smoothing probabilities on the respective events, uniformly over $B(r_0,\Theta^\star)$. We show that the conclusion of Lemma~\ref{lemma:onestep:q} holds on $\mathcal E_Q$ for every $\Theta\in B(r_0,\Theta^\star)$, with $\underline\tau_\Theta$ replaced by the constant $\underline\tau$ of Theorem~\ref{theorem:graphrecovery:full}. First, $\underline\tau>0$. By~\eqref{eq:decomposetau} and $P_{ab}(h)\leq1$, we have $\EE[\tau_\ell(h)\mid Z(0)=a,Z(h)=b;Q]\geq\int_0^hP_{a\ell}(u)P_{\ell b}(h-u)\,\mathrm du$, which is positive because every $Q\in\Omega_Q$ is irreducible, and continuous in $Q$ on the compact set $\Omega_Q$. Second, since the joint smoothing probabilities sum to one, \eqref{eq:computetaul} gives $\underline\tau_\Theta\geq\underline\tau$ for every $\Theta\in B(r_0,\Theta^\star)$. Therefore, the truncation condition of Lemma~\ref{lemma:onestep:q} with $\underline\tau$ implies the one with $\underline\tau_\Theta$, and the constant $C_Q$ of that lemma evaluated at $\underline\tau$ is at least its value at $\underline\tau_\Theta$. Finally, since $N-2r_+\leq N$ and $u_{N,r_+,k}\leq k^2\log2+2\log N$, the two conditions on $N$ that involve $\underline\tau$ in Theorem~\ref{theorem:graphrecovery:full} give
\[
\frac{N-2r_+}{N}a_{r_+}\sqrt{\frac{u_{N,r_+,k}}{N-2r_+}}
\leq a_{r_+}\sqrt{\frac{u_{N,r_+,k}}{N}}
\leq\frac{\underline\tau}{8h}
\qquad\text{and}\qquad
\frac{4r_+}{N}\leq\frac{\underline\tau}{8h}.
\]
Therefore, $\delta_{3,1}\leq\underline\tau/(4h)\leq\underline\tau_\Theta/(4h)$ on $\mathcal E_Q$, which is the event required by Lemma~\ref{lemma:onestep:q}.
Set $c:=\min\{c_1,c_2,c_3,c_4\}\geq2$ and, as in
Theorem~\ref{theorem:graphrecovery:full}, define
\[
\delta_e:=48d_{\mathrm{rec}}\max\{kp^{-csm},\,kp^{1-cm},\,kN^{-4}\}.
\]
Set
\[
\mathcal E
:=
\mathcal E_\theta\cap\mathcal E_{\sigma^2}\cap\mathcal E_Q\cap\mathcal E_{\mathrm{rec}}.
\]
We now show that each of the first three events fails with probability
at most $\delta_e$.
Since $p,k\geq1$, the first, second, and fourth
terms in the bound on $\PP(\mathcal E_\theta^c)$ satisfy
\[
\begin{aligned}
2kd_{\mathrm{rec}}\exp\{-c_1sm\log p\}
&\leq2kd_{\mathrm{rec}}p^{-csm},\\
2k\exp\{-c_2m\log(9p)\}
&\leq2kp^{-cm}\leq2kd_{\mathrm{rec}}p^{1-cm},\\
10pkd_{\mathrm{rec}}\exp\{-c_3m\log p\}
&\leq10kd_{\mathrm{rec}}p^{1-cm}.
\end{aligned}
\]
For the measurement-noise term, $N\geq m\log p$ gives
\[
2p\exp\{-c_4N\}
\leq2pe^{-cN}
\leq2p^{1-cm}
\leq2kd_{\mathrm{rec}}p^{1-cm}.
\]
Therefore, the sum of the four terms for $\mathcal E_\theta$ is at most
$d_{\mathrm{rec}}(2kp^{-csm}+14kp^{1-cm})$. Furthermore, $N\geq4r_+$ implies
$N/2\leq N-2r_+\leq N$, and hence
\[
2d_{\mathrm{rec}}(N-2r_+)^{-4}\leq32d_{\mathrm{rec}}N^{-4},
\qquad
3kd_{\mathrm{rec}}(N-2r_+)^{-4}\leq48kd_{\mathrm{rec}}N^{-4}.
\]
Combining these bounds yields
\[
\begin{aligned}
&\max\{\PP(\mathcal E_\theta^c),\,
\PP(\mathcal E_{\sigma^2}^c),\,\PP(\mathcal E_Q^c)\}\\
&\quad\leq d_{\mathrm{rec}}\max\left\{
2kp^{-csm}+14kp^{1-cm},
32N^{-4},
48kN^{-4}
\right\}\\
&\quad\leq48d_{\mathrm{rec}}\max\left\{
kp^{-csm},
kp^{1-cm},
kN^{-4}
\right\}=\delta_e.
\end{aligned}
\]
The union bound therefore gives
\[
\PP(\mathcal E)\geq1-3\delta_e-\PP(\mathcal E_{\mathrm{rec}}^c).
\]
\emph{Induction.}
From now on, we work on $\mathcal E$. Write $\hat C:=\max\cbr{C_3c_{\mathrm{RE}}^{-1}\sigma_{\max}(K_\Psi),1}$, so that
$C_\kappa-\hat C=C_4C_\lambda\sigma_{\max}(K_\Psi)/(c_{\mathrm{RE}}c_\Psi)$.
For $t\geq0$, define
\[
\check\Theta^{(t)}:=M(\Theta^{(t)}),
\qquad
\Theta^{(t+1)}:=\tilde M_N\bigl(\Theta^{(t)};\lambda^{(t)}\bigr).
\]
We show by induction on $t$ that
\begin{equation}
\Theta^{(t)}\in B(r_0,\Theta^\star)
\qquad\text{and}\qquad
\dist(\Theta^{(t)},\Theta^\star)\leq d_t.
\label{eq:theorem-induction}
\end{equation}
For $t=0$, this holds by the choice of $\Theta^{(0)}$ and the definition of $d_0$.
Note that $d_t$ is a convex combination of $d_0\leq r_0$ and
$e_N/(1-\bar\kappa)\leq r_0$, so $d_t\leq r_0$ for every $t\geq0$.
Suppose that~\eqref{eq:theorem-induction} holds for some $t\geq0$.

\emph{Penalty condition.}
At iteration $t$, we can apply Lemma~\ref{lemma:onestep:theta} provided that
\begin{equation}
\lambda^{(t)}\geq
4\max\left\{
4\Delta_{\mathrm{DB}}(N,p,s,m,r,\delta_1),
\frac{\sigma_{\max}(K_\Psi)}
{\sqrt{s}\max_{j\in[p]}\sigma_{\max}^{1/2}(\hat K_{\Psi_j})}
\norm{\check\theta_i^{\ell(t)}-\theta_i^{\ell\star}}_2
\right\}.
\label{eq:theorem-penalty-at-t}
\end{equation}
By Proposition~\ref{prop:mvt}, Lemma~\ref{lemma:aggregate-surrogate-bias}
and~\eqref{eq:theorem-induction},
\[
\begin{aligned}
\norm{\check\theta_i^{\ell(t)}-\theta_i^{\ell\star}}_2
&\leq\dist\bigl(M(\Theta^{(t)}),\Theta^\star\bigr)\\
&\leq\kappa\dist(\Theta^{(t)},\Theta^\star)
+b_{\mathrm{surr}}(\Theta^\star)\\
&\leq\kappa d_t
+b_{\mathrm{surr}}(\Theta^\star).
\end{aligned}
\]
On $\mathcal E$,
\[
\max_{j\in[p]}\sigma_{\max}^{1/2}(\hat K_{\Psi_j})
\geq\sqrt{c_\Psi/2},
\]
and the uniform deviation bounds of Step 1 in the proof of Lemma~\ref{lemma:oneupdate:theta:2} give
\[
\Delta_{\mathrm{DB}}(N,p,s,m,r,\delta_1)
\leq C_{\mathrm{DB}}\sqrt m\max\left\{
\mathcal C_w k(1-\zeta^2\pi_{\min}^2)^{r-1},
\mathcal C_{\mathrm{rec}}\sqrt{ms}\,\delta_1,
\mathcal C_\varepsilon\sqrt{\frac{m\log p}{N}}
\right\}.
\]
Therefore, since $C_{\mathrm{DB}}\geq\sqrt2$, the right-hand side of~\eqref{eq:theorem-penalty-at-t}
is at most
\[
\begin{aligned}
16C_{\mathrm{DB}}\sqrt m\max\Biggl\{&
\mathcal C_w k(1-\zeta^2\pi_{\min}^2)^{r-1},
\mathcal C_{\mathrm{rec}}\sqrt{ms}\,\delta_1,
\mathcal C_\varepsilon\sqrt{\frac{m\log p}{N}},\\
&\frac{c_\Psi^{-1/2}\sigma_{\max}(K_\Psi)}{\sqrt{ms}}
\left[\kappa d_t
+b_{\mathrm{surr}}(\Theta^\star)\right]
\Biggr\}
\leq\lambda^{(t)},
\end{aligned}
\]
where the last inequality follows from the definition of $\lambda^{(t)}$
and $C_\lambda\geq16C_{\mathrm{DB}}$.

\emph{Step for $\theta_i^\ell$.}
On $\mathcal E$, the restricted eigenvalue argument in the proof of
Lemma~\ref{lemma:oneupdate:theta:2} gives
$a\geq c_{\mathrm{RE}}/8$, and Lemma~\ref{lemma:empirical:eigenvalue} gives
$c_\Psi/2\leq\sigma_{\min}(\hat K_{\Psi_j})
\leq\sigma_{\max}(\hat K_{\Psi_j})\leq2c_\Psi^{-1}$ for every $j\in[p]$.
We apply Lemma~\ref{lemma:onestep:theta} with the penalty $\lambda^{(t)}$.
Since
$\lambda^{(t)}\leq\lambda_N+C_\lambda\sigma_{\max}(K_\Psi)(c_\Psi s)^{-1/2}
\{\kappa d_t+b_{\mathrm{surr}}(\Theta^\star)\}$, its first claim gives
\[
\begin{aligned}
\norm{\theta_i^{\ell(t+1)}-\theta_i^{\ell\star}}_2
&\leq160\sqrt2\,c_{\mathrm{RE}}^{-1}c_\Psi^{-1/2}
\sqrt s\,\lambda^{(t)}+32c_{\mathrm{RE}}^{-1}\sigma_{\max}(K_\Psi)
\norm{\check\theta_i^{\ell(t)}-\theta_i^{\ell\star}}_2\\
&\leq\operatorname{err}_{N,\theta}
+(C_\kappa-\hat C)\bigl\{\kappa d_t+b_{\mathrm{surr}}(\Theta^\star)\bigr\}
+C_3c_{\mathrm{RE}}^{-1}\sigma_{\max}(K_\Psi)
\norm{\check\theta_i^{\ell(t)}-\theta_i^{\ell\star}}_2.
\end{aligned}
\]
Here we used $C_3\geq32$ and $C_4\geq160\sqrt2$.
Write $\Delta:=\theta_i^{\ell(t+1)}-\theta_i^{\ell\star}$, with blocks
$\Delta_j\in\RR^m$. The second claim of Lemma~\ref{lemma:onestep:theta} and
the eigenvalue bounds above give
\[
\norm{\Delta}_1
\leq\sqrt m\sum_{j=1}^p\norm{\Delta_j}_2
\leq\sqrt{\frac{2m}{c_\Psi}}\norm{\Delta}_{1,\hat{\Kcal}_\Psi}
\leq\sqrt{\frac{2m}{c_\Psi}}\cdot5\sqrt{\frac{2}{c_\Psi}}\sqrt s\norm{\Delta}_2
=C_{\mathrm{cone}}\sqrt{ms}\,\norm{\Delta}_2.
\]

\emph{Full step.}
Since $\Theta^{(t)}\in B(r_0,\Theta^\star)$, the variance and
transition-rate bounds in Lemmas~\ref{lemma:oneupdate:sigma:2}
and~\ref{lemma:onestep:q} also apply. Taking the maximum over $i$ and $\ell$
in the bound for $\theta_i^\ell$, we obtain
\[
\begin{aligned}
\max_{\ell\in[k],\,i\in[p]}
\norm{\theta_i^{\ell(t+1)}-\theta_i^{\ell\star}}_2
&\leq C_3c_{\mathrm{RE}}^{-1}\sigma_{\max}(K_\Psi)
\max_{\ell\in[k],\,i\in[p]}
\norm{\check\theta_i^{\ell(t)}-\theta_i^{\ell\star}}_2
+\operatorname{err}_{N,\theta}\\
&\quad+(C_\kappa-\hat C)\bigl\{\kappa d_t+b_{\mathrm{surr}}(\Theta^\star)\bigr\},\\
\left|\sigma^{2(t+1)}-\sigma^{\star2}\right|
&\leq\left|\check\sigma^{2(t)}-\sigma^{\star2}\right|
+\operatorname{err}_{N,\sigma^2},\\
\sum_{\ell\neq\ell'}
\left|q_{\ell\ell'}^{(t+1)}-q_{\ell\ell'}^\star\right|
&\leq\sum_{\ell\neq\ell'}
\left|\check q_{\ell\ell'}^{(t)}-q_{\ell\ell'}^\star\right|
+\operatorname{err}_{N,Q}.
\end{aligned}
\]
Adding these bounds and using Definition~\ref{def:distance},
Proposition~\ref{prop:mvt}, Lemma~\ref{lemma:aggregate-surrogate-bias}
and~\eqref{eq:theorem-induction}, we obtain
\begin{equation}
\begin{aligned}
\dist\bigl(\Theta^{(t+1)},\Theta^\star\bigr)
&\leq\hat C\dist\bigl(\check\Theta^{(t)},\Theta^\star\bigr)
+\operatorname{err}_N
+(C_\kappa-\hat C)\bigl\{\kappa d_t+b_{\mathrm{surr}}(\Theta^\star)\bigr\}\\
&\leq C_\kappa\bigl\{\kappa d_t+b_{\mathrm{surr}}(\Theta^\star)\bigr\}+\operatorname{err}_N\\
&=\bar\kappa d_t+e_N=d_{t+1}\leq r_0.
\end{aligned}
\label{eq:theorem-one-step-recursion}
\end{equation}
Together with the cone bound, \eqref{eq:theorem-one-step-recursion} gives
$\max_{\ell,i}\norm{\theta_i^{\ell(t+1)}-\theta_i^{\ell\star}}_1
\leq C_{\mathrm{cone}}\sqrt{ms}\,d_{t+1}\leq C_{\mathrm{cone}}\sqrt{ms}\,r_0$.
Furthermore, the $Q$-update returns a valid transition rate matrix, since its off-diagonal entries, which are the ratios $\hat m_{\ell\ell'}(Y_0^N;\Theta^{(t)})/\hat\tau_\ell(Y_0^N;\Theta^{(t)})$ of Section~\ref{ssec:approximateEM}, are nonnegative and its rows sum to zero. Since~\eqref{eq:theorem-one-step-recursion} also gives $\|Q^{(t+1)}_\backslash-Q^\star_\backslash\|_1\leq r_0$ and $|\sigma^{2(t+1)}-\sigma^{\star2}|\leq r_0<\sigma^{\star2}$, we have $Q^{(t+1)}\in\Omega_Q$ and $\sigma^{2(t+1)}>0$, so $\Theta^{(t+1)}\in\Omega$.
Therefore, $\Theta^{(t+1)}\in B(r_0,\Theta^\star)$, and~\eqref{eq:theorem-induction}
holds at $t+1$. The same event $\mathcal E$ serves every iteration, because it
holds simultaneously for every point in $B(r_0,\Theta^\star)$.
By induction, \eqref{eq:theorem-induction} holds for every $t\geq0$. This
proves both claims of the theorem.
\end{proof}

\subsection{Proof of Corollary~\ref{corollary:main}}
\label{ssec:corollary:main}
\begin{proof}
On the event of Theorem~\ref{theorem:graphrecovery:full}, which has probability at least $1-3\delta_e-\PP(\mathcal E_{\mathrm{rec}}^c)\geq1-144d_{\mathrm{rec}}k(p^{-1}+N^{-4})-\PP(\mathcal E_{\mathrm{rec}}^c)$ as shown in Appendix~\ref{ssec:mainproof}, the condition of the corollary and the $L$-step bound of that theorem give
\[
\dist\bigl(\Theta^{(L)},\Theta^\star\bigr)
<\frac{\epsilon_{\mathrm{edge}}}{2}.
\]
Therefore, for every $i,j\in[p]$ and $\ell\in[k]$,
\[
\norm{\theta_{ij}^{\ell(L)}-\theta_{ij}^{\ell\star}}_2
\leq
\norm{\theta_i^{\ell(L)}-\theta_i^{\ell\star}}_2
\leq
\dist\bigl(\Theta^{(L)},\Theta^\star\bigr)
<\frac{\epsilon_{\mathrm{edge}}}{2}.
\]

If $(j,i)\in E^\ell$, then, by the definition of $\epsilon_{\mathrm{edge}}$,
\[
\norm{\theta_{ij}^{\ell\star}}_2
\geq\frac{3}{2}\epsilon_{\mathrm{edge}}.
\]
Hence,
\[
\norm{\theta_{ij}^{\ell(L)}}_2
\geq
\norm{\theta_{ij}^{\ell\star}}_2
-\norm{\theta_{ij}^{\ell(L)}-\theta_{ij}^{\ell\star}}_2
>\epsilon_{\mathrm{edge}}.
\]
Therefore, $(j,i)\in\hat E^\ell$. If instead $(j,i)\notin E^\ell$, then $\theta_{ij}^{\ell\star}=0$, and
\[
\norm{\theta_{ij}^{\ell(L)}}_2
<\frac{\epsilon_{\mathrm{edge}}}{2}
<\epsilon_{\mathrm{edge}},
\]
so $(j,i)\notin\hat E^\ell$. Therefore, $\hat E^\ell=E^\ell$ for every $\ell\in[k]$. Under Assumption~\ref{assumption:edgeset}, the same conclusion extends to $\tilde E^\ell$.
\end{proof}

\subsection{Proof of Corollary~\ref{corollary:full}}
\label{ssec:corollary:full}
Corollary~\ref{corollary:full} uses Theorem~\ref{theorem:graphrecovery:full} with $C_3\geq64$ and $C_4\geq320\sqrt2$, with $\mathcal C_w$ replaced by the constant $\mathcal C_w^{\mathrm{full}}\geq\mathcal C_w$ defined in the proof below, with $C_Q$ replaced by $5C_Q$, and with $C_2$ and $C_\lambda$ enlarged as specified in the proof.
\begin{proof}
The proof follows the proof of Theorem~\ref{theorem:graphrecovery}, and we only indicate the changes. Throughout this proof, $M_N(\Theta)$ denotes the full update of Corollary~\ref{corollary:full}. We compare it with the auxiliary truncated update $\tilde M_N(\Theta)$, introduced only for the analysis; the implemented algorithm does not truncate the observations. In this theoretical comparison, the full posterior starts from $\pi(Q)$ at $t_0$, whereas each truncated posterior starts from $\pi(Q)$ at the left endpoint of its window. Here $Q$ is the transition-rate block of the current parameter $\Theta$.

We proceed in five steps. For $\Theta\in B(r_0,\Theta^\star)$, we first bound the differences between the full and the truncated marginal and joint smoothing probabilities, and we then compare the variance and transition-rate blocks and the blocks $\theta_i^\ell$ of $M_N(\Theta)$ with those of the truncated update $\tilde M_N(\Theta)$. Finally, we repeat the induction of Appendix~\ref{ssec:mainproof}.

We work on the event $\mathcal E=\mathcal E_\theta\cap\mathcal E_{\sigma^2}\cap\mathcal E_Q\cap\mathcal E_{\mathrm{rec}}$ of the proof of Theorem~\ref{theorem:graphrecovery}, which has probability at least $1-3\delta_e-\PP(\mathcal E_{\mathrm{rec}}^c)\geq1-144d_{\mathrm{rec}}k(p^{-1}+N^{-4})-\PP(\mathcal E_{\mathrm{rec}}^c)$. The event $\mathcal E_\theta$ includes, for every $i\in[p]$, the event $\mathcal E_X$ of Step~1 in the proof of Lemma~\ref{lemma:oneupdate:theta:2}, whose failure probabilities give the term $2p\exp\{-c_4N\}$ in the bound on $\PP(\mathcal E_\theta^c)$. Therefore, on $\mathcal E_\theta$,
\begin{equation}\label{eq:full:CY}
\frac1N\sum_{n=1}^N\rbr{Y_{n,i}-Y_{n-1,i}}^2\leq C_Y^2\qquad\text{for every }i\in[p],
\end{equation}
where $C_Y$ is the constant in the proof of Lemma~\ref{lemma:oneupdate:theta:2}, with $\max_{\ell\in[k]}\norm{\theta_i^{\ell\star}}_1$ replaced by $\max_{i\in[p],\,\ell\in[k]}\norm{\theta_i^{\ell\star}}_1$, so that it does not depend on $i$. The five steps introduce no further event, so the probability bound is that of Theorem~\ref{theorem:graphrecovery}. Furthermore, since $|\hat X_j(t)|\leq R_X$ and $\max_{q\in[m]}\sup_{|x|\leq R_X}|g_q(x)|\leq B$, we have $\norm{\hat\Psi(t_n)}_\infty\leq Bh$ for every $n\in[N]$. Let
\[
\mathcal C_{\mathrm{full}}:=C_Y+Bh\max_{i\in[p],\,\ell\in[k]}\norm{\theta_i^{\ell\star}}_1.
\]

\emph{Step 1: Smoothing probabilities.}
Let $\Theta\in\Omega$, and let $w_{ab,\Theta}(t_n)=p(Z(t_{n-1})=a,Z(t_n)=b\mid Y_0^N;\Theta)$, for $a,b\in[k]$, denote the full joint smoothing probabilities. The proof of Lemma~\ref{lemma:truncated_smoooth} uses Assumption~\ref{assumption:mixing} for $P=\exp(Qh)$, and it treats the emission densities $P(Y_n\mid Z_n=j,Y_{n-1})$ as given positive functions of the data. Therefore, its bounds hold for the marginal and joint smoothing probabilities computed with $\hat\Psi$, whatever the dependence of $\hat\Psi$ on the data. The second bound of Lemma~\ref{lemma:truncated_smoooth}, $C_{\mathrm{sm}}\geq10$, and the argument in the proof of Lemma~\ref{lemma:onestep:q} for the pair $(Z(t_{n-1}),Z(t_n))$ give, for every realization of $Y_0^N$, every $n\in[N]$ and all $\ell,a,b\in[k]$,
\begin{equation}\label{eq:full:weights}
\abr{w_{\ell,\Theta}(t_n)-\hat w_{\ell,\Theta}(t_n)}\leq\bar\eta_r
\qquad\text{and}\qquad
\abr{w_{ab,\Theta}(t_n)-\hat w_{ab,\Theta}(t_n)}\leq\bar\eta_r.
\end{equation}

\emph{Step 2: Variance update.}
Let $\Theta\in B(r_0,\Theta^\star)$. By the definition of $M_N$ in Corollary~\ref{corollary:full} and the proof of Lemma~\ref{lemma:onestep:sigma}, $M_{N,\sigma^2}(\Theta)$ and $\tilde M_{N,\sigma^2}(\Theta)$ are the same linear function of the full and the truncated smoothing probabilities, respectively, with the residuals $\hat R_{n,\ell}$ of Appendix~\ref{ssec:onestep_sigma}. Therefore, \eqref{eq:full:weights} gives
\[
\abr{M_{N,\sigma^2}(\Theta)-\tilde M_{N,\sigma^2}(\Theta)}
\leq\frac{1}{2Np}\sum_{n=1}^N\sum_{\ell=1}^k\abr{w_{\ell,\Theta}(t_n)-\hat w_{\ell,\Theta}(t_n)}\hat R_{n,\ell}
\leq\frac{k\bar\eta_r}{2}\max_{\ell\in[k]}\frac{1}{Np}\sum_{n=1}^N\hat R_{n,\ell}.
\]
Since $\norm{\hat\Psi(t_n)}_\infty\leq Bh$, H\"older's inequality gives $|\hat\Psi(t_n)^\top\theta_i^\ell|\leq Bh\norm{\theta_i^\ell}_1$. Together with~\eqref{eq:full:CY} and the inequality $(x+y)^2\leq2x^2+2y^2$, this gives
\[
\frac1N\sum_{n=1}^N\cbr{Y_{n,i}-Y_{n-1,i}-\hat\Psi(t_n)^\top\theta_i^\ell}^2\leq2C_Y^2+2B^2h^2\norm{\theta_i^\ell}_1^2
\]
for every $i\in[p]$ and $\ell\in[k]$. Averaging over $i$ and using the bound $\max_{\ell\in[k]}p^{-1}\sum_{i=1}^p\norm{\theta_i^\ell}_1^2\leq C_\Theta ms$ in the proof of Lemma~\ref{lemma:oneupdate:sigma:2}, we obtain
\begin{equation}\label{eq:full:sigma}
\begin{aligned}
\abr{M_{N,\sigma^2}(\Theta)-\tilde M_{N,\sigma^2}(\Theta)}
&\leq k\bar\eta_r\rbr{C_Y^2+B^2h^2C_\Theta ms}\\
&=\frac{C_{\mathrm{sm}}\rbr{C_Y^2+B^2h^2C_\Theta ms}}{\delta_{\min}}\,k^2\zeta^{-8}\pi_{\min}^{-2}\cbr{1-(\zeta\pi_{\min})^2}^{r-1}.
\end{aligned}
\end{equation}

\emph{Step 3: Transition rates.}
Let $\Theta\in B(r_0,\Theta^\star)$ and $\ell\neq\ell'$. As in the proof of Lemma~\ref{lemma:onestep:q}, write $m_{\ell\ell'}$, $\tau_\ell$, $\hat m_{\ell\ell'}$ and $\hat\tau_\ell$ for the full and the truncated posterior expected numbers of transitions and occupation times, so that $M_{N,q_{\ell\ell'}}(\Theta)=m_{\ell\ell'}/\tau_\ell$ and $\tilde M_{N,q_{\ell\ell'}}(\Theta)=\hat m_{\ell\ell'}/\hat\tau_\ell$, and let $C_{m,\ell\ell'}\leq hq_{\max}/\delta_{\min}$ be the constant of that proof. By~\eqref{eq:computemll}, \eqref{eq:computetaul}, \eqref{eq:full:weights} and $0\leq\tau_\ell(h)\leq h$,
\[
\abr{m_{\ell\ell'}-\hat m_{\ell\ell'}}\leq Nk^2C_{m,\ell\ell'}\bar\eta_r
\qquad\text{and}\qquad
\abr{\tau_\ell-\hat\tau_\ell}\leq Nk^2h\bar\eta_r.
\]
On $\mathcal E_Q$, the proof of Theorem~\ref{theorem:graphrecovery} shows that $\delta_{3,1}\leq\underline\tau_\Theta/(4h)$, and the proof of Lemma~\ref{lemma:onestep:q} then gives $\hat\tau_\ell\geq N\underline\tau_\Theta/4\geq N\underline\tau/4$. Since $hk^2\bar\eta_r\leq\underline\tau/8$ by the condition of Corollary~\ref{corollary:full}, we also have $\tau_\ell\geq\hat\tau_\ell-Nk^2h\bar\eta_r\geq N\underline\tau/8$. Furthermore, $\hat m_{\ell\ell'}\leq NC_{m,\ell\ell'}$, because the truncated joint smoothing probabilities sum to one. Therefore,
\[
\abr{M_{N,q_{\ell\ell'}}(\Theta)-\tilde M_{N,q_{\ell\ell'}}(\Theta)}
\leq\frac{\abr{m_{\ell\ell'}-\hat m_{\ell\ell'}}}{\tau_\ell}
+\frac{\hat m_{\ell\ell'}\abr{\tau_\ell-\hat\tau_\ell}}{\tau_\ell\hat\tau_\ell}
\leq k^2C_{m,\ell\ell'}\rbr{\frac{8}{\underline\tau}+\frac{32h}{\underline\tau^2}}\bar\eta_r.
\]
Since $8/\underline\tau+32h/\underline\tau^2\leq4(4/\underline\tau+8h/\underline\tau^2)$, the definition of $C_Q$ at the end of the proof of Lemma~\ref{lemma:onestep:q}, evaluated at $\underline\tau$, gives $C_{\mathrm{sm}}C_{m,\ell\ell'}(8/\underline\tau+32h/\underline\tau^2)\leq4C_Q$. Summing over the $k(k-1)$ pairs, we obtain
\begin{equation}\label{eq:full:q}
\sum_{\ell\neq\ell'}\abr{M_{N,q_{\ell\ell'}}(\Theta)-\tilde M_{N,q_{\ell\ell'}}(\Theta)}
\leq4C_Qk(k-1)k^3\delta_{\min}^{-1}\zeta^{-8}\pi_{\min}^{-2}\cbr{1-(\zeta\pi_{\min})^2}^{r-1}.
\end{equation}

\emph{Step 4: Update of $\theta_i^\ell$.}
Fix $i\in[p]$, $\ell\in[k]$ and $\Theta\in B(r_0,\Theta^\star)$, use the notation of Appendix~\ref{ssec:onestep_theta}, and let $\Delta:=M_{N,\theta_i^\ell}(\Theta)-\theta^\star$ be the error of the update with the full smoothing probabilities. Suppose that~\eqref{eq:re} holds at $(\Theta,\ell)$ with constants $a,\nu>0$ such that $a\geq50\max_j\sigma_{\max}(\hat K_{\Psi_j})s\nu$, that the inequality of Assumption~\ref{assumption:DB} holds at $(\Theta,i,\ell)$, that $\hat K_{\Psi_j}\succ0$ for every $j\in[p]$, that~\eqref{eq:full:CY} holds, and that
\begin{align}
\lambda&\geq4\max\cbr{4\Delta_{\mathrm{DB}}(N,p,s,m,r,\delta_1)+\mathcal C_{\mathrm{full}}\bar\eta_r,\
\frac{1}{\sqrt s}\frac{\sigma_{\max}(K_\Psi)}{\max_j\sigma_{\max}^{1/2}(\hat K_{\Psi_j})}\norm{\check\theta-\theta^\star}_2},\label{eq:full:lambda}\\
\bar\eta_r&\leq\frac{a}{100\max_j\sigma_{\max}(\hat K_{\Psi_j})s}.\label{eq:full:eta}
\end{align}
We show that $\Delta$ satisfies the cone bound~\eqref{eq:delta:1K} and that
\begin{equation}\label{eq:full:theta}
\norm{M_{N,\theta_i^\ell}(\Theta)-\theta_i^{\ell\star}}_2
\leq\frac{8}{a}\rbr{5\lambda\max_j\sigma_{\max}^{1/2}(\hat K_{\Psi_j})\sqrt s+\sigma_{\max}(K_\Psi)\norm{\check\theta_i^\ell-\theta_i^{\ell\star}}_2},
\end{equation}
which is the first bound of Lemma~\ref{lemma:onestep:theta} with $4/a$ replaced by $8/a$. Let
\[
\Delta_{\mathrm{full}}:=\frac1N\sum_{n=1}^N\cbr{w_{\ell,\Theta}(t_n)-\hat w_{\ell,\Theta}(t_n)}\cbr{Y_n^\Delta-\hat\Psi(t_n)^\top\theta^\star}\hat\Psi(t_n)\in\RR^{pm}.
\]
The basic inequality~\eqref{eq:term1} holds with $w_{\ell,\Theta}$ in place of $\hat w_{\ell,\Theta}$, and adding $\Delta_{\mathrm{full}}$ to both sides of~\eqref{eq:theta:score:decomposition} gives
\[
\frac1N\sum_{n=1}^Nw_{\ell,\Theta}(t_n)\cbr{Y_n^\Delta-\hat\Psi(t_n)^\top\theta^\star}\hat\Psi(t_n)
=\Delta_\Psi^\top\theta^\star+\Delta_X+\Delta_w+\Delta_\varepsilon+K_\Psi^{w,\ell}(\check\theta-\theta^\star)+\Delta_{\mathrm{full}}.
\]
Therefore, \eqref{eq:term2} holds with $w_{\ell,\Theta}$ in place of $\hat w_{\ell,\Theta}$ on its left-hand side and with $\norm{\Delta_{\mathrm{full}}}_{\infty,\hat{\Kcal}_\Psi^*}$ added to the four dual norms. To bound this term, fix $j\in[p]$ and $u\in\RR^m$, and let $(\Delta_{\mathrm{full}})_j\in\RR^m$ denote the $j$th block of $\Delta_{\mathrm{full}}$. By~\eqref{eq:full:weights} and the Cauchy--Schwarz inequality,
\[
(\Delta_{\mathrm{full}})_j^\top u
\leq\bar\eta_r\cbr{\frac1N\sum_{n=1}^N\rbr{Y_n^\Delta-\hat\Psi(t_n)^\top\theta^\star}^2}^{1/2}\norm{u}_{\hat K_{\Psi_j}}.
\]
By~\eqref{eq:full:CY}, Minkowski's inequality and $|\hat\Psi(t_n)^\top\theta^\star|\leq Bh\norm{\theta^\star}_1$, the square root is at most $\mathcal C_{\mathrm{full}}$. Hence $\norm{\Delta_{\mathrm{full}}}_{\infty,\hat{\Kcal}_\Psi^*}\leq\mathcal C_{\mathrm{full}}\bar\eta_r$, and~\eqref{eq:full:lambda} and the inequality of Assumption~\ref{assumption:DB} at $(\Theta,i,\ell)$ imply that $\lambda$ is at least four times the sum of the five dual norms. The steps from~\eqref{eq:term2} to~\eqref{eq:term4} then apply verbatim, since they use only the basic inequality, the nonnegativity of its left-hand side, the two conditions on $\lambda$ and $\sigma_{\max}(K_\Psi^{w,\ell})\leq\sigma_{\max}(K_\Psi)$. They give the cone bound~\eqref{eq:delta:1K} and
\begin{equation}\label{eq:full:term4}
\frac1N\sum_{n=1}^Nw_{\ell,\Theta}(t_n)\cbr{\hat\Psi(t_n)^\top\Delta}^2
\leq2\norm{\Delta}_2\cbr{5\lambda\max_j\sigma_{\max}^{1/2}(\hat K_{\Psi_j})\sqrt s+\sigma_{\max}(K_\Psi)\norm{\check\theta-\theta^\star}_2}.
\end{equation}
For the lower bound, write $\Delta_j\in\RR^m$ for the blocks of $\Delta$. Minkowski's inequality gives $\{N^{-1}\sum_{n=1}^N(\hat\Psi(t_n)^\top\Delta)^2\}^{1/2}\leq\sum_{j=1}^p\norm{\Delta_j}_{\hat K_{\Psi_j}}=\norm{\Delta}_{1,\hat{\Kcal}_\Psi}$. Therefore, \eqref{eq:full:weights}, \eqref{eq:re} at $(\Theta,\ell)$ and the cone bound give
\[
\begin{aligned}
\frac1N\sum_{n=1}^Nw_{\ell,\Theta}(t_n)\cbr{\hat\Psi(t_n)^\top\Delta}^2
&\geq\frac1N\sum_{n=1}^N\hat w_{\ell,\Theta}(t_n)\cbr{\hat\Psi(t_n)^\top\Delta}^2-\frac{\bar\eta_r}{N}\sum_{n=1}^N\cbr{\hat\Psi(t_n)^\top\Delta}^2\\
&\geq a\norm{\Delta}_2^2-(\nu+\bar\eta_r)\norm{\Delta}_{1,\hat{\Kcal}_\Psi}^2\\
&\geq\cbr{a-25\max_j\sigma_{\max}(\hat K_{\Psi_j})s(\nu+\bar\eta_r)}\norm{\Delta}_2^2
\geq\frac a4\norm{\Delta}_2^2,
\end{aligned}
\]
where the last inequality follows from $a\geq50\max_j\sigma_{\max}(\hat K_{\Psi_j})s\nu$ and~\eqref{eq:full:eta}. Combining this bound with~\eqref{eq:full:term4} proves~\eqref{eq:full:theta}.

\emph{Step 5: Induction.}
We now repeat the induction in the proof of Theorem~\ref{theorem:graphrecovery} for $\Theta^{(t+1)}=M_N(\Theta^{(t)};\lambda^{(t)})$, with $\check\Theta^{(t)}:=M(\Theta^{(t)})$. Suppose that~\eqref{eq:theorem-induction} holds for some $t\geq0$. Since $\Theta^{(t)}\in B(r_0,\Theta^\star)$ and the bounds that define $\mathcal E$ hold uniformly over $B(r_0,\Theta^\star)$, Steps~2--4 apply at $\Theta=\Theta^{(t)}$.

For $\theta_i^\ell$, Step~4 replaces Lemma~\ref{lemma:onestep:theta}. On $\mathcal E$, \eqref{eq:re} and the inequality of Assumption~\ref{assumption:DB} hold at $\Theta^{(t)}$, for every $i\in[p]$ and $\ell\in[k]$, with $a\geq c_{\mathrm{RE}}/8$, and~\eqref{eq:Kpsi:ub} and~\eqref{eq:full:CY} hold, so~\eqref{eq:full:eta} follows from the condition $\bar\eta_r\leq c_\Psi c_{\mathrm{RE}}/(1600s)$ of Corollary~\ref{corollary:full}. For~\eqref{eq:full:lambda}, let
\[
\mathcal C_w^{\mathrm{full}}:=\mathcal C_w+C_{\mathrm{sm}}\delta_{\min}^{-1}\zeta^{-8}\pi_{\min}^{-2}\mathcal C_{\mathrm{full}},
\]
so that $\mathcal C_{\mathrm{full}}\bar\eta_r=(\mathcal C_w^{\mathrm{full}}-\mathcal C_w)k\{1-(\zeta\pi_{\min})^2\}^{r-1}$. The bound on $\Delta_{\mathrm{DB}}$ in the proof of Theorem~\ref{theorem:graphrecovery} and $m\geq1$ give
\[
\begin{aligned}
&4\cbr{4\Delta_{\mathrm{DB}}(N,p,s,m,r,\delta_1)+\mathcal C_{\mathrm{full}}\bar\eta_r}\\
&\quad\leq(16C_{\mathrm{DB}}+4)\sqrt m\max\cbr{
\mathcal C_w^{\mathrm{full}}k\cbr{1-(\zeta\pi_{\min})^2}^{r-1},\
\mathcal C_{\mathrm{rec}}\sqrt{ms}\,\delta_1,\
\mathcal C_\varepsilon\sqrt{\frac{m\log p}{N}}}.
\end{aligned}
\]
Since $\lambda_N$ is defined with $\mathcal C_w^{\mathrm{full}}$ in place of $\mathcal C_w$, the penalty condition in the proof of Theorem~\ref{theorem:graphrecovery} gives~\eqref{eq:full:lambda} with $\lambda=\lambda^{(t)}$ once $C_\lambda\geq16C_{\mathrm{DB}}+4$. Therefore, \eqref{eq:full:theta}, $a\geq c_{\mathrm{RE}}/8$, \eqref{eq:Kpsi:ub} and $\lambda^{(t)}\leq\lambda_N+C_\lambda\sigma_{\max}(K_\Psi)(c_\Psi s)^{-1/2}\{\kappa d_t+b_{\mathrm{surr}}(\Theta^\star)\}$ give
\[
\begin{aligned}
\norm{\theta_i^{\ell(t+1)}-\theta_i^{\ell\star}}_2
&\leq320\sqrt2\,c_{\mathrm{RE}}^{-1}c_\Psi^{-1/2}\sqrt s\,\lambda^{(t)}+64c_{\mathrm{RE}}^{-1}\sigma_{\max}(K_\Psi)\norm{\check\theta_i^{\ell(t)}-\theta_i^{\ell\star}}_2\\
&\leq\operatorname{err}_{N,\theta}+(C_\kappa-\hat C)\bigl\{\kappa d_t+b_{\mathrm{surr}}(\Theta^\star)\bigr\}+C_3c_{\mathrm{RE}}^{-1}\sigma_{\max}(K_\Psi)\norm{\check\theta_i^{\ell(t)}-\theta_i^{\ell\star}}_2,
\end{aligned}
\]
where we used $C_3\geq64$ and $C_4\geq320\sqrt2$. This is the bound for $\theta_i^\ell$ in the proof of Theorem~\ref{theorem:graphrecovery}, and the cone bound~\eqref{eq:delta:1K} gives the $\ell_1$ bound on $\theta_i^{\ell(t+1)}-\theta_i^{\ell\star}$ there with the same constant $C_{\mathrm{cone}}$.

For the variance, Lemma~\ref{lemma:oneupdate:sigma:2} and~\eqref{eq:full:sigma} give
\[
\abr{\sigma^{2(t+1)}-\sigma^{\star2}}\leq\abr{\check\sigma^{2(t)}-\sigma^{\star2}}+\operatorname{err}_{N,\sigma^2}
\]
once $C_2\geq\tilde C_2\mathsf D+C_{\mathrm{sm}}(C_Y^2+B^2h^2C_\Theta ms)/\delta_{\min}$, with $\tilde C_2$ and $\mathsf D$ as in the proof of Theorem~\ref{theorem:graphrecovery}. This lower bound depends only on the quantities listed for $C_2$ in Theorem~\ref{theorem:graphrecovery:full}. For the transition rates, Lemma~\ref{lemma:onestep:q} and~\eqref{eq:full:q} give
\[
\sum_{\ell\neq\ell'}\abr{q_{\ell\ell'}^{(t+1)}-q_{\ell\ell'}^\star}
\leq\sum_{\ell\neq\ell'}\abr{\check q_{\ell\ell'}^{(t)}-q_{\ell\ell'}^\star}+\operatorname{err}_{N,Q},
\]
where $\operatorname{err}_{N,Q}$ is defined with $5C_Q$ in place of $C_Q$.

These are the three bounds of the full step in the proof of Theorem~\ref{theorem:graphrecovery}, and the rest of that proof applies verbatim. In particular, $M_{N,Q}(\Theta^{(t)})$ is a valid transition rate matrix, since $\tau_\ell>0$ by Step~3, and~\eqref{eq:theorem-one-step-recursion} gives $\dist(\Theta^{(t+1)},\Theta^\star)\leq d_{t+1}$ and $\Theta^{(t+1)}\in B(r_0,\Theta^\star)$. By induction, the conclusions of Theorem~\ref{theorem:graphrecovery} hold for the iterates of Corollary~\ref{corollary:full}. The proof of Corollary~\ref{corollary:main} uses only the bound on $\dist(\Theta^{(L)},\Theta^\star)$, so its conclusion holds as well. This completes the proof of Corollary~\ref{corollary:full}.
\end{proof}

\subsection{One-step bound for the update of \texorpdfstring{$\theta_i^\ell$}{theta}}
\label{ssec:oneupdate:theta:2}

The following lemma bounds the error of one update of $\theta_i^\ell$ at a fixed parameter in the local region.
\begin{lemma}\label{lemma:oneupdate:theta:2}

Assume that Assumptions~\ref{assumption:edgeset}--\ref{assumption:population}, \ref{assumption:mixing}, \ref{assumption:eigenvalue} and~\ref{assumption:local} hold, as well as the conditions of Lemmas~\ref{lemma:oracle:gram:concentration} and~\ref{lemma:RE:condition}.
Fix $i\in[p]$, $\ell\in[k]$, and
$\Theta\in B(r_0,\Theta^\star)$. Suppose that $Z(0)$ is sampled from
the stationary distribution $\pi^\star=\pi(Q^\star)$. Define
$\check{\Theta}=M(\Theta)$
and set $\check{\theta}_{i}^\ell =
(\check{\theta}_{i1}^{\ell\top},\ldots,\check{\theta}_{ip}^{\ell\top})^\top$
for this fixed $(i,\ell)$. Assume
\[
\sup_{|x|\leq R_X}\max_{j\in[m]}\abr{g_j(x)}\leq B,\quad
\sup_{|x|\leq R_X}\max_{j\in[m]}\abr{g_j'(x)}\leq D,
\]
for some constants $B,D>0$, where $R_X$ is the bound on $X$ and $\hat X$ of Assumption~\ref{assumption:local}. Assume also that $\max_{n\in[N],\,i\in[p]}|\rho_{n,i}|\leq\rho$ almost surely. Define
\[
\begin{aligned}
\mathcal{C}_w
&=c_\Psi^{-1/2}Bh\left\{\frac{2\sigma^\star}{\sqrt{\pi}}+\rho+
Bh\max_{i'\in[p]}\max_{\ell',\ell''\in[k]}
\norm{\theta_{i'}^{\ell'\star}-\theta_{i'}^{\ell''\star}}_1\right\}
\frac{10\zeta^{-8}\pi_{\min}^{-2}}{\delta_{\min}},\\
\mathcal{C}_\varepsilon
&=c_\Psi^{-1/2}Bh\left\{2\sigma^\star+\rho+
Bh\max_{i'\in[p]}\max_{\ell',\ell''\in[k]}
\norm{\theta_{i'}^{\ell'\star}-\theta_{i'}^{\ell''\star}}_1\right\},\\
\mathcal{C}_{\mathrm{rec}}
&=1+\frac{4BDh^{3/2}}{c_\Psi}
\max_{i'\in[p],\,\ell'\in[k]}\norm{\theta_{i'}^{\ell'\star}}_2+\sqrt{\frac{2}{c_\Psi}}D\left\{
Bh^{3/2}\max_{i'\in[p],\,\ell'\in[k]}\norm{\theta_{i'}^{\ell'\star}}_2
+(\rho+\sigma^\star)\sqrt{\frac{h}{ms}}\right\}.
\end{aligned}
\]
Assume further that the weighted oracle Gram matrix for
the fixed $(\Theta,\ell)$ is nondegenerate:
\[
\sigma_{\min}\left(\EE\left[\frac{1}{N-2r_+}
\sum_{n=r_++1}^{N-r_+}\hat w_{\ell,\Theta}(t_n)
\Psi(t_n)\Psi(t_n)^\top\right]\right)
\geq c_{\mathrm{RE}}>0.
\]
Assume in addition that
\[
c_{\mathrm{RE}}^{-1}\sigma_{\max}(K_\Psi)\leq c_\Psi^{-1}.
\]
Suppose further that $N \gtrsim \cbr{m^4(\log p)^3 \vee m^{5/2}s^{5/2}}$,
$N-2r_+\geq C_0m^4
\frac{\{\log(5(N-2r_+)p)\}^{5}}{\log(5p)},\quad
N\geq C_1[m\log(9p)]^2,\quad N\geq4r_+$, and
\[
mh\left\{C_2B^2h\sqrt{\frac{m\log(9p)}{N}}
+2BD\sqrt{h}\,\delta_1+D^2\delta_1^2\right\}\leq\frac{c_\Psi}{2},
\]
where $C_0$ is the constant of Lemma~\ref{lemma:deviation} and $C_1,C_2$ are the constants of Lemma~\ref{lemma:empirical:eigenvalue}. Set
\[
\begin{aligned}
\lambda_i^\ell
:={}& C_\lambda\sqrt{m}\max\Biggl\{
\mathcal{C}_w k\rbr{1-\zeta^2\pi_{\min}^2}^{r-1},\mathcal{C}_{\mathrm{rec}}\sqrt{ms}\delta_1,\mathcal{C}_\varepsilon\sqrt{\frac{m\log p}{N}},\\
&\hspace{55pt}
\frac{c_\Psi^{-1/2}\sigma_{\max}(K_{\Psi})}{\sqrt{ms}}
\norm{\check{\theta}_{i}^\ell-\theta_{i}^{\ell\star}}_2
\Biggr\}.
\end{aligned}
\]
Define
\[
\begin{aligned}
\delta_{e_0}:={}&2d_{\mathrm{rec}}\exp\{-c_1sm\log p\}
+2\exp\{-c_2m\log(9p)\}\\
&+10d_{\mathrm{rec}}\exp\{-c_3m\log p\}+2\exp\{-c_4N\},
\end{aligned}
\]
where $c_1$ and $c_2$ are the constants in the probability bound of Lemma~\ref{lemma:RE:condition}, $c_3$ is the constant in the probability bound of Lemma~\ref{lemma:deviation}, and $c_4=2$ is the constant of the $\chi^2$ tail bound in Step~1 of the proof. The constants $c_1$, $c_2$ and $c_3$ depend only on $r_+$, $d_{\mathrm{rec}}$ and the mixing coefficients, and $c_1,\ldots,c_4\geq2$. Here,
$C_\lambda>0$ is a sufficiently large constant that depends only on $r_+$, $d_{\mathrm{rec}}$ and the mixing coefficients.
Then there exist constants $C_{\theta,1},C_{\theta,2}>0$, depending only on $C_\lambda$, and an event of probability at least $1-\delta_{e_0}$ on whose intersection with the event $\mathcal E_{\mathrm{rec}}$ of Section~\ref{ssec:main:results} we have
\begin{multline*}
\norm{\hat{\theta}^\ell_{i} - \theta_{i}^{\ell\star}}_2
\leq C_{\theta,1} c_{\Psi}^{-2}
\norm{\check{\theta}^\ell_{i} - \theta_{i}^{\ell\star}}_2
\\
+ C_{\theta,2} c_{\mathrm{RE}}^{-1}c_\Psi^{-1}\cbr{
    \mathcal{C}_\varepsilon\frac{m\sqrt{s\log p}}{\sqrt{N}}
    + \mathcal{C}_{\mathrm{rec}}ms\delta_1
    + \mathcal{C}_w k\sqrt{ms}\rbr{1 - \zeta^2 \pi_{\min}^2}^{r-1}
}.
\end{multline*}
In particular, this bound holds with probability at least $1-\delta_{e_0}-\PP(\mathcal E_{\mathrm{rec}}^c)$.
\end{lemma}

\begin{proof}

We proceed in three steps.

\emph{Step 1: Deviation bounds and restricted eigenvalue condition.}

To apply Lemma~\ref{lemma:onestep:theta}, we first establish its restricted eigenvalue condition,~\eqref{eq:re} at $(\Theta,\ell)$.
Let $\mathcal E_{\mathrm{RE}}$ be the event of
Lemma~\ref{lemma:RE:condition}, the intersection of its oracle concentration event and the event $\mathcal E_{\mathrm G}$ of Lemma~\ref{lemma:empirical:eigenvalue}. By Lemma~\ref{lemma:RE:condition},
\[
\PP(\mathcal E_{\mathrm{RE}}^c)
\leq2d_{\mathrm{rec}}\exp\{-c_1sm\log p\}
+2\exp\{-c_2m\log(9p)\}.
\]
On $\mathcal E_{\mathrm{RE}}\cap\mathcal E_{\mathrm{rec}}$,
\begin{equation}\label{eq:Kpsi:ub}
\min_j \sigma_{\min}(\hat{K}_{\Psi_j}) \geq \frac{c_\Psi}{2}, 
\quad \max_j \sigma_{\max}(\hat{K}_{\Psi_j}) \leq 2 c_\Psi^{-1}.
\end{equation}
In particular, $\hat K_{\Psi_j}\succ0$ for every $j\in[p]$ on this event, as Lemmas~\ref{lemma:onestep:theta}, \ref{lemma:deltaW}, \ref{lemma:deltaPsi}, \ref{lemma:deltaX} and~\ref{lemma:deviation} require.
Lemma~\ref{lemma:RE:condition} also yields~\eqref{eq:re} at $(\Theta,\ell)$ on $\mathcal E_{\mathrm{RE}}\cap\mathcal E_{\mathrm{rec}}$, with constants $a,\nu$ that satisfy $a\geq50\max_j\sigma_{\max}(\hat K_{\Psi_j})s\nu$, and we let
$a_{\mathrm{RE}}$ denote the constant $a$ that it supplies. Since
$t\leq\sigma_{\mathrm{RE}}/[4\{27+5400c_\Psi^{-2}\}]$ in Lemma~\ref{lemma:RE:condition}, we have
$\sigma_{\mathrm{RE}}-27t\geq3\sigma_{\mathrm{RE}}/4$. Therefore, the nondegeneracy condition in the statement, $\sigma_{\mathrm{RE}}\geq c_{\mathrm{RE}}$, and
$N\geq4r_+$ give
\[
a_{\mathrm{RE}}
=\frac{N-2r_+}{2N}(\sigma_{\mathrm{RE}}-27t)
\geq\frac{3}{16}c_{\mathrm{RE}}.
\]

Next, we establish the inequality of Assumption~\ref{assumption:DB} at $(\Theta,i,\ell)$, that is, we bound $\norm{\Delta_w}_{\infty,\hat{\Kcal}_\Psi^*}$, $\norm{\Delta_\Psi^\top\theta^\star}_{\infty,\hat{\Kcal}_\Psi^*}$, $\norm{\Delta_X}_{\infty,\hat{\Kcal}_\Psi^*}$, and $\norm{\Delta_\varepsilon}_{\infty,\hat{\Kcal}_\Psi^*}$, in this order. The bounds on $\Delta_\Psi$ and $\Delta_X$ use the hypothesis of Lemmas~\ref{lemma:deltaPsi} and~\ref{lemma:deltaX} on $\hat X-X$, which holds on $\mathcal E_{\mathrm{rec}}$.
By Lemma~\ref{lemma:deltaW} and~\eqref{eq:Kpsi:ub}, we have
\[
\norm{\Delta_w}_{\infty,\hat{\Kcal}_\Psi^*}
\leq \sqrt{2}\mathcal{C}_w\sqrt{m}k\rbr{1-\zeta^2\pi_{\min}^2}^{r-1}.
\]
By Lemma~\ref{lemma:deltaPsi} and~\eqref{eq:Kpsi:ub}, we have
\[
\norm{\Delta_\Psi^\top\theta^\star}_{\infty,\hat{\Kcal}_\Psi^*} \leq \frac{4}{c_\Psi}\norm{\theta^{\star}}_2 BDm \delta_1 s^{1/2}h^{3/2}.
\]
Since $\theta^\star=\theta_i^{\ell\star}$ for the
current row and state, the definition of $\mathcal C_{\mathrm{rec}}$
gives
\[
\norm{\Delta_\Psi^\top\theta^\star}_{\infty,\hat{\Kcal}_\Psi^*}
\leq
\mathcal C_{\mathrm{rec}}\sqrt m\sqrt{ms}\,\delta_1.
\]
Furthermore, we verify the moment condition required by
Lemma~\ref{lemma:deltaX}. By~\eqref{eq:approxY} with the basis residual
$r_{n,i}=0$ under Assumption~\ref{assumption:edgeset},
\[
Y_n^\Delta=\Psi(t_n)^\top\theta_{i}^{Z(t_n)\star}+\rho_{n,i}
+\sigma^\star\rbr{\varepsilon_{n,i}-\varepsilon_{n-1,i}}.
\]
Since $\norm{\Psi(t_n)}_\infty\leq Bh$, H\"older's inequality gives
$\abr{\Psi(t_n)^\top\theta_{i}^{Z(t_n)\star}}\leq Bh\max_{\ell\in[k]}\norm{\theta_i^{\ell\star}}_1$,
and $\max_{n\in[N],\,i\in[p]}\abr{\rho_{n,i}}\leq\rho$ by assumption. Because the
$\varepsilon_n\sim\Ncal(0,I)$ are i.i.d., a $\chi^2$ tail bound shows that
$N^{-1}\sum_{n=1}^N\rbr{\varepsilon_{n,i}-\varepsilon_{n-1,i}}^2$ is at most a universal constant on an event $\mathcal E_X$ with
$\PP(\mathcal E_X^c)\leq2\exp\{-c_4N\}$, where $c_4=2$. Indeed, $(\varepsilon_{n,i}-\varepsilon_{n-1,i})^2\leq2\varepsilon_{n,i}^2+2\varepsilon_{n-1,i}^2$, and since $\EE\exp(\varepsilon_{n,i}^2/4)=\sqrt2$, Markov's inequality gives $\PP(\sum_{n=0}^N\varepsilon_{n,i}^2\geq12N)\leq2^{(N+1)/2}e^{-3N}\leq e^{-2N}$, so that the universal constant $48$ suffices. On $\mathcal E_X$, the
hypothesis of Lemma~\ref{lemma:deltaX} holds with
\[
C_Y\asymp Bh\max_{\ell\in[k]}\norm{\theta_i^{\ell\star}}_1+\rho+\sigma^\star,
\]
and Lemma~\ref{lemma:deltaX} together with~\eqref{eq:Kpsi:ub} gives
\[
\norm{\Delta_X}_{\infty,\hat{\Kcal}_\Psi^*}
\leq \sqrt{\frac{2}{c_\Psi}}C_YD\sqrt{mh}\,\delta_1.
\]
By the choice of $C_Y$, $\|\theta_i^{\ell\star}\|_1\leq\sqrt{ms}\,
\|\theta_i^{\ell\star}\|_2$ and the definition of
$\mathcal C_{\mathrm{rec}}$, the right-hand side is at most $C\mathcal C_{\mathrm{rec}}\sqrt m\sqrt{ms}\,\delta_1$ for a universal constant $C$.
Finally, Lemma~\ref{lemma:deviation} together with~\eqref{eq:Kpsi:ub} gives
\[
\norm{\Delta_\varepsilon}_{\infty,\hat{\Kcal}_\Psi^*}
\leq C\mathcal{C}_\varepsilon\sqrt{m}\sqrt{\frac{m \log p}{N}},
\]
on $\mathcal E_\varepsilon\cap\mathcal E_{\mathrm{RE}}\cap\mathcal E_{\mathrm{rec}}$, where $\mathcal E_\varepsilon$ is the event $\mathcal G$ of Lemma~\ref{lemma:deviation}, so that
$\PP(\mathcal E_\varepsilon^c)\leq10d_{\mathrm{rec}}\exp\{-c_3m\log p\}$, and $C$ depends only on $r_+$, $d_{\mathrm{rec}}$ and the mixing coefficients. Therefore, we can
set
\[
\Delta_{\mathrm{DB}}(N, p, s, m, r, \delta_1) 
= C_{\mathrm{DB}} \sqrt{m}\max\cbr{
    \mathcal{C}_w k\rbr{1 - \zeta^2\pi_{\min}^2}^{r-1},
    \mathcal{C}_{\mathrm{rec}}\sqrt{ms}\delta_1,
    \mathcal{C}_\varepsilon\sqrt{\frac{m \log p}{N}}
},
\]
where $C_{\mathrm{DB}}\geq\sqrt2$ is the largest of the constant factors in the four bounds above. In particular, $C_{\mathrm{DB}}$ depends only on $r_+$, $d_{\mathrm{rec}}$ and the mixing coefficients.
Let
$\mathcal E=\mathcal E_{\mathrm{RE}}\cap\mathcal E_X
\cap\mathcal E_\varepsilon$. By the union bound and the definition of
$\delta_{e_0}$,
\[
\PP(\mathcal E)\geq1-\delta_{e_0}.
\]
The four bounds above hold on $\mathcal E\cap\mathcal E_{\mathrm{rec}}$, and in the rest of the proof, we work on this event.
Note that $\Delta_{\mathrm{DB}}$ is small when (i) $r$ is large, (ii) $\delta_1$ is small, that is, the estimated trajectories are accurate, and (iii) $N$ is large, so $\Delta_{\mathrm{DB}}$ is well controlled.

\emph{Step 2: Selecting $\lambda_i^\ell$.}

Lemma~\ref{lemma:onestep:theta} requires
\[
\lambda_i^\ell\geq4\max\cbr{
4\Delta_{\mathrm{DB}}(N, p, s, m, r,\delta_1),
\frac{1}{\sqrt{s}}\frac{\sigma_{\max}(K_{\Psi})}{\max_j\sigma_{\max}^{1/2}(\hat{K}_{\Psi_j})}
\norm{\check{\theta}_{i}^\ell - \theta_{i}^{\ell\star}}_2
}.
\]
For the second entry, $\sigma_{\max}(\hat{K}_{\Psi_j})\geq\sigma_{\min}(\hat{K}_{\Psi_j})$
for every $j$, so~\eqref{eq:Kpsi:ub} gives
$\max_j\sigma_{\max}^{1/2}(\hat{K}_{\Psi_j})\geq\min_j\sigma_{\min}^{1/2}(\hat{K}_{\Psi_j})\geq\sqrt{c_\Psi/2}$.
Substituting the expression for $\Delta_{\mathrm{DB}}$ obtained in Step 1 and using $4\sqrt2\leq16C_{\mathrm{DB}}$, we see that the right-hand side of the last display is at most
\[
\begin{aligned}
16C_{\mathrm{DB}}\sqrt{m}\max\Biggl\{&
\mathcal{C}_w k\rbr{1-\zeta^2\pi_{\min}^2}^{r-1},\mathcal{C}_{\mathrm{rec}}\sqrt{ms}\delta_1,\mathcal{C}_\varepsilon
\sqrt{\frac{m\log p}{N}},\\
&\frac{c_\Psi^{-1/2}\sigma_{\max}(K_{\Psi})}{\sqrt{ms}}
\norm{\check{\theta}_{i}^\ell-\theta_{i}^{\ell\star}}_2
\Biggr\}.
\end{aligned}
\]
Therefore, $\lambda_i^\ell$ satisfies the condition of Lemma~\ref{lemma:onestep:theta} whenever $C_\lambda\geq16C_{\mathrm{DB}}$.
Applying Lemma~\ref{lemma:onestep:theta} with $a=a_{\mathrm{RE}}\geq3c_{\mathrm{RE}}/16$ and using~\eqref{eq:Kpsi:ub}, we obtain
\begin{align}\label{eq:theta1}
\norm{\hat{\theta}_{i}^\ell - \theta_{i}^{\ell\star}}_2 
\leq Cc_{\mathrm{RE}}^{-1}\rbr{
    c_\Psi^{-1/2}10\lambda_i^\ell \sqrt{s}
    +\sigma_{\max}(K_{\Psi})\norm{\check{\theta}_{i}^\ell - \theta_{i}^{\ell\star}}_2
},
\end{align}
on $\mathcal E\cap\mathcal E_{\mathrm{rec}}$, where $C:=64/3$.

\emph{Step 3: Combining results.} 

By the definition of $\lambda_i^\ell$, we have
\[
\begin{aligned}
10\lambda_i^\ell\sqrt{s}
&=10C_\lambda\max\Biggl\{
\mathcal{C}_w\sqrt{ms}k
\rbr{1-\zeta^2\pi_{\min}^2}^{r-1},\mathcal{C}_{\mathrm{rec}}ms\delta_1,\mathcal{C}_\varepsilon m
\sqrt{\frac{s\log p}{N}},\\
&\hspace{60pt}
c_\Psi^{-1/2}\sigma_{\max}(K_{\Psi})
\norm{\check{\theta}_{i}^\ell-\theta_{i}^{\ell\star}}_2
\Biggr\}.
\end{aligned}
\]
Plugging this into \eqref{eq:theta1} and bounding the maximum by the sum, we obtain
\begin{multline*}
\norm{\hat{\theta}_{i}^\ell - \theta_{i}^{\ell\star}}_2 
\leq Cc_{\mathrm{RE}}^{-1}\Bigg\{
    10C_\lambda c_\Psi^{-1/2}\rbr{
        \mathcal{C}_w k\sqrt{ms}\rbr{1 - \zeta^2\pi_{\min}^2}^{r-1}
        + \mathcal{C}_{\mathrm{rec}}ms\delta_1
        + \mathcal{C}_\varepsilon m \sqrt{\frac{s \log p}{N}}
    }
    \\
    + \rbr{10C_\lambda c_\Psi^{-1}\sigma_{\max}(K_{\Psi})
    +\sigma_{\max}(K_{\Psi})}
    \norm{\check{\theta}_{i}^\ell - \theta_{i}^{\ell\star}}_2
\Bigg\},
\end{multline*}
on $\mathcal E\cap\mathcal E_{\mathrm{rec}}$. Since Assumption~\ref{assumption:eigenvalue} implies
$c_\Psi\leq1$, the additional condition
$c_{\mathrm{RE}}^{-1}\sigma_{\max}(K_\Psi)\leq c_\Psi^{-1}$ yields
\begin{multline*}
\norm{\hat{\theta}^\ell_{i} - \theta_{i}^{\ell\star}}_2
\leq C_{\theta,1} c_{\Psi}^{-2}
\norm{\check{\theta}^\ell_{i} - \theta_{i}^{\ell\star}}_2
\\
+ C_{\theta,2} c_{\mathrm{RE}}^{-1}c_\Psi^{-1}\cbr{
    \mathcal{C}_\varepsilon\frac{m \sqrt{s \log p}}{\sqrt{N}}
    + \mathcal{C}_{\mathrm{rec}}ms\delta_1
    + \mathcal{C}_w k\sqrt{ms}\rbr{1 - \zeta^2 \pi_{\min}^2}^{r-1}
},
\end{multline*}
on $\mathcal E\cap\mathcal E_{\mathrm{rec}}$, where $C_{\theta,1}:=C(10C_\lambda+1)$ and $C_{\theta,2}:=10CC_\lambda$. Since $\PP(\mathcal E)\geq1-\delta_{e_0}$, this completes the proof.
\end{proof}

\subsection{One-step bound for the variance update}
\label{ssec:oneupdate:sigma:2}

\begin{lemma}\label{lemma:oneupdate:sigma:2}
    Under Assumptions~\ref{assumption:edgeset}, \ref{assumption:stationary}, \ref{assumption:mixingbeta}, \ref{assumption:population}, \ref{assumption:mixing}, and \ref{assumption:eigenvalue}, and under Assumption~\ref{assumption:local},
    suppose that $Z(0)$ is sampled from the stationary distribution $\pi^\star=\pi(Q^\star)$.
    Suppose that $\Theta\in B(r_0,\Theta^\star)$, let $\delta_1\geq0$, and let $\mathcal E_{\mathrm{rec}}$ be the event of Section~\ref{ssec:main:results} on which
    \[
    \max_{i\in[p],\,n\in[N]}\opnorm{\hat{X}_i - X_i}{[t_{n-1}, t_{n}]}
    \leq \delta_1.
    \]
    Fix $r\geq1$, and suppose that $N\geq4r_+$ and that
    the sample-size conditions of Lemma~\ref{lemma:sc_T2}, including
    $N-2r_+>4$, hold. Suppose further that, almost surely,
    \[
    \max_{n\in[N],\,i\in[p]}|\rho_{n,i}|\leq\rho,
    \qquad
    \sup_{t\in[0,T]}\max_{i\in[p],\,j\in[m]}|g_j(X_i(t))|\leq B,
    \]
    \[
    \sup_{|x|\leq R_X}\max_{j\in[m]}|g_j'(x)|\leq D.
    \]
    Let $\hat{\sigma}^2=\tilde M_{N,\sigma^2}(\Theta)$ and
    $\check{\sigma}^2=M_{\sigma^2}(\Theta)$.
    Define
    \[
    \mathsf D:=2\sigma^{\star2}+\rho^2
    +2B^2h^2\max_{\substack{\ell,\ell'\in[k]\\\ell\neq\ell'}}
    \max_{i\in[p]}\norm{\theta_i^{\ell'\star}-\theta_i^{\ell\star}}_1^2
    +2C_{\mathrm{cone}}^2msB^2h^2r_0^2,
    \]
    where $C_{\mathrm{cone}}=10c_\Psi^{-1}$ is the constant in the definition of $B(r_0,\Theta^\star)$ in Section~\ref{ssec:popEM}.
    Then there exist positive constants $C_0$, $C_1$ and $C_2$, which depend on
    the geometric $\beta$-mixing constants and otherwise only on
    $(\Theta^\star,r_0,c_\Psi,D,r_+,d_{\mathrm{rec}},s)$, $(r_+,d_{\mathrm{rec}})$ and $(r_+,d_{\mathrm{rec}},\delta_{\min})$, respectively, and an event of probability at least $1-3kd_{\mathrm{rec}}(N-2r_+)^{-4}$ on whose intersection with $\mathcal E_{\mathrm{rec}}$ we have
    \begin{align*}
    \abr{\hat{\sigma}^2 - \sigma^{\star2}}
    &\leq C_0k\left\{\sqrt{mh\mathsf D}\,\delta_1
    +mh\delta_1^2\right\}
    +C_1k\mathsf D\sqrt{\frac{\log N}{N}}\\
    &\quad+C_2k^2\mathsf D\zeta^{-8}\pi_{\min}^{-2}
    \cbr{1-(\zeta\pi_{\min})^2}^{r-1}
    +\abr{\check{\sigma}^2-\sigma^{\star2}}.
    \end{align*}
    In particular, this bound holds with probability at least $1-3kd_{\mathrm{rec}}(N-2r_+)^{-4}-\PP(\mathcal E_{\mathrm{rec}}^c)$.
\end{lemma}

\begin{proof}[Proof of Lemma~\ref{lemma:oneupdate:sigma:2}]

We apply Lemma~\ref{lemma:onestep:sigma} on an event on which the
weighted residuals concentrate, and then bound the term $\delta_2$ of that
lemma. First, for a fixed $\ell$, the residual expansion in the proof of
Lemma~\ref{lemma:sc_T2}, with the factor $\hat w_{\ell,\Theta}(t_n)$ removed, gives
\[
\frac1{Np}\sum_{n=1}^N\EE[R_{n,\ell}]
\leq4\mathsf D.
\]
Next, since $\norm{\Psi(t_n)}_\infty\leq Bh$ by the bound on $g$, Lemma~\ref{lemma:sc_T2}
implies, under the stated sample-size conditions, that, for each $\ell\in[k]$, with probability at least
$1-3d_{\mathrm{rec}}(N-2r_+)^{-4}$,
\begin{multline*}
\frac{1}{Np}\left|\sum_{n=1}^N
\left\{\hat{w}_{\ell,\Theta}(t_n)R_{n,\ell}
-\EE\sbr{\hat{w}_{\ell,\Theta}(t_n)R_{n,\ell}}\right\}\right|\\*
\leq \bar C_1\mathsf D
\left\{\sqrt{\frac{5\log(N-2r_+)}{N-2r_+}}
+\frac{r_+(\log N)^{3/2}}{N}\right\},
\end{multline*}
where $\bar C_1$ denotes the constant $C_1(r_+)$ of that lemma. Let $\mathcal E_\ell$ denote this event. By a union bound,
\[
\PP\left(\bigcap_{\ell=1}^k\mathcal E_\ell\right)
\geq1-3kd_{\mathrm{rec}}(N-2r_+)^{-4}.
\]
On $\mathcal E:=\bigcap_{\ell=1}^k\mathcal E_\ell$, the population bound
above and the inequality $\hat w_{\ell,\Theta}\leq1$ give
\[
\max_{\ell\in[k]}\frac1{Np}\sum_{n=1}^N
\hat w_{\ell,\Theta}(t_n)R_{n,\ell}
\leq C_{\mathcal R}\mathsf D,
\qquad
C_{\mathcal R}:=4+2r_+\bar C_1,
\]
since $\sup_{x>1}(\log x)/x=1/e$ and $\sup_{x\geq1}(\log x)^{3/2}/x<1/2$ imply that the term in braces in the definition of $\mathcal E_\ell$ is at most $\sqrt{5/e}+r_+/2\leq2r_+$.
Set $\mathcal R_\Theta:=C_{\mathcal R}\mathsf D$. Since $C_{\mathcal R}\geq4$, the population bound above shows that $\mathcal R_\Theta$ satisfies the condition of Lemma~\ref{lemma:onestep:sigma}, and the preceding display gives
$\mathcal E\subseteq\mathcal E_R$. Furthermore, since $\Theta\in B(r_0,\Theta^\star)$, the inequality $\norm{\theta_i^\ell}_1^2\leq2\norm{\theta_i^{\ell\star}}_1^2+2\norm{\theta_i^\ell-\theta_i^{\ell\star}}_1^2$ and the $\ell_1$ condition in the definition of $B(r_0,\Theta^\star)$ give
\[
\max_{\ell\in[k]}\frac1p\sum_{i=1}^p\norm{\theta_i^\ell}_1^2
\leq \max_{\ell\in[k]}
\left\{\frac{2}{p}\sum_{i=1}^p\norm{\theta_i^{\ell\star}}_1^2
+2C_{\mathrm{cone}}^2ms\,r_0^2\right\}
\leq C_\Theta ms,
\]
where $C_\Theta:=2\max_{\ell\in[k],\,i\in[p]}\norm{\theta_i^{\ell\star}}_2^2+2C_{\mathrm{cone}}^2r_0^2$ depends only on $(\Theta^\star,r_0,c_\Psi)$. Here we used $\norm{\theta_i^{\ell\star}}_1^2\leq ms\norm{\theta_i^{\ell\star}}_2^2$, which holds because $\theta_i^{\ell\star}$ has at most $s$ nonzero blocks of length $m$. Therefore, on
$\mathcal E\cap\mathcal E_{\mathrm{rec}}\subseteq\mathcal E_R\cap\mathcal E_{\mathrm{rec}}$, Lemma~\ref{lemma:onestep:sigma} gives
\begin{align*}
\abr{\hat{\sigma}^2-\sigma^{\star2}}
&\leq C_0k\left\{\sqrt{mh\mathsf D}\,\delta_1
+mh\delta_1^2\right\}+\delta_2\\
&\quad+C_2k^2\mathsf D\zeta^{-8}\pi_{\min}^{-2}
\cbr{1-(\zeta\pi_{\min})^2}^{r-1}
+\abr{\check{\sigma}^2-\sigma^{\star2}},
\end{align*}
where $C_0:=s\max\{D\sqrt{2C_\Theta C_{\mathcal R}},\,D^2C_\Theta/2\}$ and $C_2:=5C_{\mathcal R}/\delta_{\min}$. For the remaining term, the triangle inequality and the
definition of $\delta_2$ yield
\begin{align*}
\delta_2
&\leq\sum_{\ell=1}^k\frac{1}{2Np}\left|\sum_{n=1}^N
\left\{\hat{w}_{\ell,\Theta}(t_n)R_{n,\ell}
-\EE\sbr{\hat{w}_{\ell,\Theta}(t_n)R_{n,\ell}}\right\}\right|\\
&\leq\frac{\bar C_1k}{2}\mathsf D
\left\{\sqrt{\frac{5\log(N-2r_+)}{N-2r_+}}
+\frac{r_+(\log N)^{3/2}}{N}\right\}.
\end{align*}
Since $N\geq4r_+$, the first term in braces is at most
$\sqrt{10\log N/N}$. Furthermore,
\[
\frac{r_+(\log N)^{3/2}}{N}
=\frac{r_+\log N}{\sqrt N}\sqrt{\frac{\log N}{N}}
\leq\frac{2r_+}{e}\sqrt{\frac{\log N}{N}},
\]
where the inequality follows from $\sup_{x\geq1}(\log x)/\sqrt x=2/e$.
Therefore,
$\delta_2\leq C_1k\mathsf D\sqrt{\log N/N}$ with $C_1:=(\sqrt{10}+2r_+/e)\bar C_1/2$. Substituting this bound completes the proof.
\end{proof}

\subsection{Concentration of the Gram matrices}
\label{ssec:gram:concentration}

The following lemma, which we use in the proofs of Lemmas~\ref{lemma:RE:condition} and~\ref{lemma:oneupdate:theta:2} and of Theorem~\ref{theorem:graphrecovery}, bounds the extreme singular values of the empirical
Gram matrices uniformly over the nodes.

\begin{lemma}[Empirical eigenvalue bound]
\label{lemma:empirical:eigenvalue}
Suppose that Assumption~\ref{assumption:eigenvalue} and the conditions
of Lemma~\ref{lemma:oracle:gram:concentration} hold, with $B$ denoting
the bound in that lemma. Let $\delta_1\geq0$, and let $\mathcal E_{\mathrm{rec}}$ be the event of Section~\ref{ssec:main:results} on which
\[
\max_{i\in[p]}\max_{n\in[N]}
\opnorm{\hat{X}_i-X_i}{[t_{n-1},t_n]}\leq\delta_1.
\]
Assume further that $|X_i(t)|\leq R_X$ and $|\hat X_i(t)|\leq R_X$ for all $i\in[p]$ and $t\in[0,T]$, and that
\[
\sup_{|x|\leq R_X}\max_{q\in[m]}|g_q'(x)|\leq D.
\]
Then there exist positive constants $C_1,C_2$ and $C_3\geq2$, depending only on
the constant $c_W$ in Assumption~\ref{assumption:mixingbeta}, such that, if $N\geq C_1[m\log(9p)]^2$, then the event $\mathcal E_{\mathrm G}$ on which the bound of Lemma~\ref{lemma:oracle:gram:concentration} holds has probability at least
$1-2\exp\{-C_3m\log(9p)\}$, and, on $\mathcal E_{\mathrm G}\cap\mathcal E_{\mathrm{rec}}$, simultaneously for every $i\in[p]$,
\begin{align*}
\sigma_{\min}(\hat{K}_{\Psi_i})
&\geq c_\Psi-mh\left\{C_2B^2h\sqrt{\frac{m\log(9p)}{N}}
+2BD\sqrt{h}\,\delta_1+D^2\delta_1^2\right\},\\
\sigma_{\max}(\hat{K}_{\Psi_i})
&\leq c_\Psi^{-1}+mh\left\{C_2B^2h\sqrt{\frac{m\log(9p)}{N}}
+2BD\sqrt{h}\,\delta_1+D^2\delta_1^2\right\}.
\end{align*}
In particular, if
\[
mh\left\{C_2B^2h\sqrt{\frac{m\log(9p)}{N}}
+2BD\sqrt{h}\,\delta_1+D^2\delta_1^2\right\}\leq\frac{c_\Psi}{2},
\]
then
\[
\min_{i\in[p]}\sigma_{\min}(\hat{K}_{\Psi_i})\geq\frac{c_\Psi}{2},
\qquad
\max_{i\in[p]}\sigma_{\max}(\hat{K}_{\Psi_i})\leq2c_\Psi^{-1}.
\]
\end{lemma}

Note that the preceding sufficient condition holds when $N$ is
sufficiently large and $\delta_1$ is sufficiently small. For
example, it is enough that
\[
C_2B^2h\sqrt{\frac{m\log(9p)}{N}}\leq\frac{c_\Psi}{6mh},
\]
and
\[
2BD\sqrt{h}\,\delta_1+D^2\delta_1^2\leq\frac{c_\Psi}{3mh}.
\]

\begin{proof}
On $\mathcal E_{\mathrm{rec}}$, for every $i\in[p]$, $n\in[N]$, and $j\in[m]$, the mean-value theorem
and the Cauchy--Schwarz inequality imply
\begin{align*}
\left|\hat{\Psi}_{i,j}(t_n)-\Psi_{i,j}(t_n)\right|
&=\left|\int_{t_{n-1}}^{t_n}
\left\{g_j(\hat{X}_i(u))-g_j(X_i(u))\right\}\,\mathrm{d}u\right|\\
&\leq D\sqrt{h}\opnorm{\hat{X}_i-X_i}{[t_{n-1},t_n]}\\
&\leq D\sqrt{h}\,\delta_1.
\end{align*}
Therefore,
\[
\norm{\hat{\Psi}_i(t_n)-\Psi_i(t_n)}_2
\leq\sqrt{m}D\sqrt{h}\,\delta_1.
\]
Furthermore,
\[
\norm{\Psi_i(t_n)}_2\leq\sqrt{m}\norm{\Psi_i(t_n)}_\infty
\leq\sqrt{m}Bh.
\]
As in Lemma~\ref{lemma:oracle:gram:concentration}, let
\[
G_{\Psi_i}=\frac{1}{N}\sum_{n=1}^N
\Psi_i(t_n)\Psi_i(t_n)^\top.
\]
Using
\begin{align*}
\hat{\Psi}_i(t_n)\hat{\Psi}_i(t_n)^\top
-\Psi_i(t_n)\Psi_i(t_n)^\top
&=\Psi_i(t_n)\{\hat{\Psi}_i(t_n)-\Psi_i(t_n)\}^\top\\
&\quad+\{\hat{\Psi}_i(t_n)-\Psi_i(t_n)\}\Psi_i(t_n)^\top\\
&\quad+\{\hat{\Psi}_i(t_n)-\Psi_i(t_n)\}
\{\hat{\Psi}_i(t_n)-\Psi_i(t_n)\}^\top,
\end{align*}
we obtain
\begin{align*}
\norm{\hat{K}_{\Psi_i}-G_{\Psi_i}}_2
&\leq\frac{1}{N}\sum_{n=1}^N\left[2\norm{\Psi_i(t_n)}_2
\norm{\hat{\Psi}_i(t_n)-\Psi_i(t_n)}_2\right.\\
&\hspace{10em}\left.+\norm{\hat{\Psi}_i(t_n)-\Psi_i(t_n)}_2^2\right]\\
&\leq mh\left\{2BD\sqrt{h}\,\delta_1+D^2\delta_1^2\right\}.
\end{align*}
By Lemma~\ref{lemma:oracle:gram:concentration}, $\PP(\mathcal E_{\mathrm G})\geq1-2\exp\{-C_3m\log(9p)\}$, and on $\mathcal E_{\mathrm G}$,
\[
\max_{i\in[p]}\norm{G_{\Psi_i}-K_{\Psi_i}}_2
\leq C_2mB^2h^2\sqrt{\frac{m\log(9p)}{N}}.
\]
Therefore, on $\mathcal E_{\mathrm G}\cap\mathcal E_{\mathrm{rec}}$,
\begin{align*}
\norm{\hat{K}_{\Psi_i}-K_{\Psi_i}}_2
&\leq\norm{\hat{K}_{\Psi_i}-G_{\Psi_i}}_2
+\norm{G_{\Psi_i}-K_{\Psi_i}}_2\\
&\leq mh\left\{C_2B^2h\sqrt{\frac{m\log(9p)}{N}}
+2BD\sqrt{h}\,\delta_1+D^2\delta_1^2\right\}.
\end{align*}
By Weyl's inequality,
\begin{align*}
\sigma_{\min}(\hat{K}_{\Psi_i})
&\geq\sigma_{\min}(K_{\Psi_i})
-\norm{\hat{K}_{\Psi_i}-K_{\Psi_i}}_2,\\
\sigma_{\max}(\hat{K}_{\Psi_i})
&\leq\sigma_{\max}(K_{\Psi_i})
+\norm{\hat{K}_{\Psi_i}-K_{\Psi_i}}_2.
\end{align*}
Substituting the preceding bound and using
Assumption~\ref{assumption:eigenvalue}, that is,
\[
\sigma_{\min}(K_{\Psi_i})\geq c_\Psi,
\qquad
\sigma_{\max}(K_{\Psi_i})\leq c_\Psi^{-1},
\]
proves the claimed inequalities. Finally, the displayed sufficient
condition gives
\[
\sigma_{\min}(\hat{K}_{\Psi_i})\geq\frac{c_\Psi}{2}.
\]
Since $c_\Psi\leq1$,
\[
\sigma_{\max}(\hat{K}_{\Psi_i})
\leq c_\Psi^{-1}+\frac{c_\Psi}{2}\leq2c_\Psi^{-1}.
\]
This completes the proof.
\end{proof}

Next, we prove the concentration bound for the oracle Gram matrices
that is used in the preceding proof.

\begin{lemma}[Concentration of the oracle Gram matrices]
\label{lemma:oracle:gram:concentration}
Suppose that Assumptions~\ref{assumption:stationary} and~\ref{assumption:mixingbeta} hold, and that, almost surely,
\[
\max_{i\in[p]}\max_{j\in[m]}\sup_{t\in[0,T]}
|g_j(X_i(t))|\leq B.
\]
For each $i\in[p]$, define
\[
G_{\Psi_i}:=\frac{1}{N}\sum_{n=1}^N
\Psi_i(t_n)\Psi_i(t_n)^\top,\qquad
K_{\Psi_i}:=\frac{1}{N}\sum_{n=1}^N
\EE\sbr{\Psi_i(t_n)\Psi_i(t_n)^\top}.
\]
Then there exist positive constants $C_1,C_2$ and $C_3\geq2$, depending only on
the constant $c_W$ in Assumption~\ref{assumption:mixingbeta}, such that, if $N\geq C_1[m\log(9p)]^2$, then, with probability at least
$1-2\exp\{-C_3m\log(9p)\}$,
\[
\max_{i\in[p]}\norm{G_{\Psi_i}-K_{\Psi_i}}_2
\leq C_2mB^2h^2\sqrt{\frac{m\log(9p)}{N}}.
\]
\end{lemma}

\begin{proof}
We first bound the quadratic forms of the centered oracle Gram matrices
on a net of the unit sphere, and then pass from the net to the operator
norm. Let $\mathcal{N}$ be a $1/4$-net of the Euclidean unit sphere
$\mathbb{S}^{m-1}$. The standard volumetric covering bound yields
\[
|\mathcal{N}|\leq \left(1+\frac{2}{1/4}\right)^m=9^m.
\]
Fix $i\in[p]$ and $u\in\mathcal{N}$, and define
\[
V_{n,i}(u):=\left\{u^\top\Psi_i(t_n)\right\}^2-
\EE\sbr{\left\{u^\top\Psi_i(t_n)\right\}^2}.
\]
Since $\Psi_i(t_n)$ is a measurable function of $W_n$, as noted in Section~\ref{sec:theory},
Assumption~\ref{assumption:stationary} implies that
$\{\Psi_i(t_n)\}_{n\in\mathbb{N}}$ is strictly stationary. Furthermore,
the $\sigma$-fields generated by $\{\Psi_i(t_{n'})\}_{n'\leq n}$ and
$\{\Psi_i(t_{n'})\}_{n'\geq n+j}$ are contained in $\Fcal_{-\infty}^n$
and $\Fcal_{n+j}^\infty$, respectively, so, by
Assumption~\ref{assumption:mixingbeta}, the $\beta$-mixing coefficients
of $\{\Psi_i(t_n)\}_{n\in\mathbb{N}}$ satisfy $\beta(j)\leq2\exp(-cj)$
for $j\in\NN$. Since the expectation in $V_{n,i}(u)$ does not depend
on $n$, the centered sequence $\{V_{n,i}(u)\}_{n\in\mathbb{N}}$ has the
same two properties.

For every $n$,
\[
\norm{\Psi_i(t_n)}_2\leq\sqrt{m}\norm{\Psi_i(t_n)}_\infty
\leq\sqrt{m}Bh.
\]
Since $\norm{u}_2=1$,
\[
|u^\top\Psi_i(t_n)|\leq\sqrt{m}Bh,
\]
and therefore
\[
0\leq\left\{u^\top\Psi_i(t_n)\right\}^2\leq mB^2h^2.
\]
Its expectation lies in the same interval, so
\[
|V_{n,i}(u)|\leq mB^2h^2.
\]
For every $q\geq1$,
\[
\cbr{\EE|V_{n,i}(u)|^q}^{1/q}\leq mB^2h^2.
\]
Hence, by Definition~\ref{def:subweibull},
\[
\norm{V_{n,i}(u)}_{\psi_2}
=\sup_{q\geq1}q^{-1/2}\cbr{\EE|V_{n,i}(u)|^q}^{1/q}
\leq mB^2h^2.
\]
In particular, $V_{n,i}(u)/(mB^2h^2)$ is sub-Weibull$(2)$ with sub-Weibull norm at most one.
Since $u^\top(G_{\Psi_i}-K_{\Psi_i})u=N^{-1}\sum_{n=1}^NV_{n,i}(u)$, we apply
Lemma~\ref{lemma:tailbound:mixing} to these normalized variables with
\[
\gamma_1=1,\qquad \gamma_2=2,\qquad \gamma=\frac{2}{3}.
\]
This gives, for every $t>mB^2h^2/N$ and a constant $C>0$ that depends only on $c$,
\begin{align*}
\PP\left(\left|u^\top(G_{\Psi_i}-K_{\Psi_i})u\right|>t\right)
&\leq N\exp\left\{-\frac{(tN)^{2/3}}
{C(mB^2h^2)^{2/3}}\right\}\\
&\quad+\exp\left\{-\frac{t^2N}{Cm^2B^4h^4}\right\}.
\end{align*}
We choose
\[
t=C_4mB^2h^2\sqrt{\frac{m\log(9p)}{N}},
\]
where $C_4\geq1$ is a constant chosen below. Since $m\log(9p)\geq\log9>2$, this
threshold exceeds $mB^2h^2/N$, and $N\geq C_1[m\log(9p)]^2>4$ when $C_1\geq1$, as Lemma~\ref{lemma:tailbound:mixing} requires. For the second
exponential term,
\[
\frac{t^2N}{Cm^2B^4h^4}=\frac{C_4^2}{C}m\log(9p).
\]
For the first exponential term,
\[
\frac{(tN)^{2/3}}{C(mB^2h^2)^{2/3}}
=\frac{C_4^{2/3}}{C}\cbr{Nm\log(9p)}^{1/3}.
\]
Since $\sup_{x\geq1}x^{-1/3}\log x=3/e$ and $m\log(9p)>1$, we have
$\log N\leq(3/e)\{Nm\log(9p)\}^{1/3}$. Hence, if $C_4^{2/3}/C\geq6/e$, the
first term on the right-hand side of the tail bound is at most
\[
\exp\left[-\frac{C_4^{2/3}}{2C}\cbr{Nm\log(9p)}^{1/3}\right]
\leq\exp\left\{-\frac{C_4^{2/3}C_1^{1/3}}{2C}m\log(9p)\right\},
\]
where the last inequality uses $N\geq C_1[m\log(9p)]^2$. Therefore, choosing
$C_4$ and then $C_1$ large enough, depending only on $C$, that
$C_5:=\min\{C_4^2/C,\,C_4^{2/3}C_1^{1/3}/(2C)\}\geq3$, we obtain
\[
\PP\left(\left|u^\top(G_{\Psi_i}-K_{\Psi_i})u\right|>t\right)
\leq2\exp\{-C_5m\log(9p)\}.
\]
Taking a union bound over $i\in[p]$ and $u\in\mathcal{N}$, and using
\[
p|\mathcal{N}|\leq p9^m\leq\exp\{m\log(9p)\},
\]
gives, with probability at least
$1-2\exp\{-C_3m\log(9p)\}$, where $C_3:=C_5-1\geq2$,
\[
\max_{i\in[p]}\sup_{u\in\mathcal{N}}
\left|u^\top(G_{\Psi_i}-K_{\Psi_i})u\right|\leq t.
\]
Finally, let $A$ be any symmetric $m\times m$ matrix and choose
$v\in\mathbb{S}^{m-1}$ such that
\[
\norm{A}_2=|v^\top Av|.
\]
Choose $u\in\mathcal{N}$ satisfying $\norm{u-v}_2\leq1/4$. Then
\begin{align*}
|v^\top Av-u^\top Au|
&\leq |(v-u)^\top Av|+|u^\top A(v-u)|\\
&\leq \frac{1}{2}\norm{A}_2.
\end{align*}
Therefore,
\[
\norm{A}_2\leq2\sup_{u\in\mathcal{N}}|u^\top Au|.
\]
Applying this inequality with $A=G_{\Psi_i}-K_{\Psi_i}$ gives
\[
\max_{i\in[p]}\norm{G_{\Psi_i}-K_{\Psi_i}}_2
\leq2t\leq C_2mB^2h^2\sqrt{\frac{m\log(9p)}{N}}.
\]
This completes the proof.
\end{proof}

\section{Auxiliary lemmas}

This section collects auxiliary results used in the proofs, with their
sources when they are taken from the literature.

\begin{definition}\label{def:subweibull}
For a random variable $X$ and $q\geq1$, write $\norm{X}_q:=(\EE|X|^q)^{1/q}$. For every $\gamma > 0$, we define the sub-Weibull norm of $X$ as
\[
\norm{X}_{\psi_{\gamma}} := \sup_{q \geq 1} q^{-1/\gamma}\norm{X}_q.
\]
For a random vector $V\in\RR^d$, we define its sub-Weibull norm by
\[
\norm{V}_{\psi_{\gamma}}:=
\sup_{\|u\|_2=1}\norm{u^\top V}_{\psi_{\gamma}}.
\]
A random variable or vector is sub-Weibull($\gamma$) if its sub-Weibull norm is finite.
\end{definition}

\begin{lemma}[Adapted from Lemma~13 in~\citet{wong2020lasso} and Theorem~1 in~\citet{merlevede2011bernstein}]
\label{lemma:tailbound:mixing}

Let $\{X_n\}_{n\in\NN}$ be a strictly stationary sequence of zero-mean random variables that are sub-Weibull($\gamma_2$) with sub-Weibull norm at most $K$. Suppose that $\{X_n\}_{n\in\NN}$ is $\beta$-mixing with coefficients satisfying $\beta(n) \leq 2\exp(-c n^{\gamma_1})$.

Let $S_N = \sum_{n=1}^N X_n$ and define $1/\gamma = 1/\gamma_1 + 1/\gamma_2$. Assume that $\gamma < 1$. Then, for $N > 4$ and any $t > 1/N$, we have
\[
P\rbr{ \abr{N^{-1} S_N} > t } \leq N \exp\rbr{ -\frac{(tN)^\gamma}{K^\gamma C_1} } + \exp\rbr{ -\frac{t^2 N}{K^2 C_2} },
\]
where the constants $C_1$ and $C_2$ depend only on $\gamma_1$, $\gamma_2$, and $c$.
\end{lemma}

\begin{lemma}\label{lemma:subweibull:product}
Let $X$ be a sub-Weibull($\gamma_1$) random variable and $Y$ a sub-Weibull($\gamma_2$) random variable. Let $1/\gamma = 1/\gamma_1 + 1/\gamma_2$. Then the product $XY$ is a sub-Weibull($\gamma$) random variable with sub-Weibull norm bounded by
\[
\norm{XY}_{\psi_\gamma} \leq \rbr{\frac{\gamma_1+\gamma_2}{\gamma_2}}^{\frac{1}{\gamma_1}} \rbr{\frac{\gamma_1+\gamma_2}{\gamma_1}}^{\frac{1}{\gamma_2}} \norm{X}_{\psi_{\gamma_1}} \norm{Y}_{\psi_{\gamma_2}}.
\]
\end{lemma}

\begin{proof}
Fix $q\geq1$. By H\"older's inequality with the conjugate exponents $\gamma_1/\gamma$ and $\gamma_2/\gamma$, we have
\[
\EE|XY|^q \leq \rbr{\EE|X|^{q\frac{\gamma_1}{\gamma}}}^{\frac{\gamma}{\gamma_1}} \rbr{\EE|Y|^{q\frac{\gamma_2}{\gamma}}}^{\frac{\gamma}{\gamma_2}}.
\]
Taking the $q$-th root of both sides yields
\[
\norm{XY}_q \leq \norm{X}_{q\frac{\gamma_1}{\gamma}} \norm{Y}_{q\frac{\gamma_2}{\gamma}}.
\]
Since $q\gamma_1/\gamma\geq1$ and $q\gamma_2/\gamma\geq1$, Definition~\ref{def:subweibull} bounds each factor:
\begin{align*}
\norm{XY}_q &\leq \norm{X}_{\psi_{\gamma_1}} \rbr{\frac{\gamma_1}{\gamma}}^{\frac{1}{\gamma_1}} q^{\frac{1}{\gamma_1}} \norm{Y}_{\psi_{\gamma_2}} \rbr{\frac{\gamma_2}{\gamma}}^{\frac{1}{\gamma_2}} q^{\frac{1}{\gamma_2}} \\
&= q^{\frac{1}{\gamma}} \rbr{\frac{\gamma_1+\gamma_2}{\gamma_2}}^{\frac{1}{\gamma_1}} \rbr{\frac{\gamma_1+\gamma_2}{\gamma_1}}^{\frac{1}{\gamma_2}} \norm{X}_{\psi_{\gamma_1}} \norm{Y}_{\psi_{\gamma_2}}.
\end{align*}
Hence,
\[
\norm{XY}_{\psi_\gamma} = \sup_{q \geq 1} q^{-\frac{1}{\gamma}} \norm{XY}_q \leq \rbr{\frac{\gamma_1+\gamma_2}{\gamma_2}}^{\frac{1}{\gamma_1}} \rbr{\frac{\gamma_1+\gamma_2}{\gamma_1}}^{\frac{1}{\gamma_2}} \norm{X}_{\psi_{\gamma_1}} \norm{Y}_{\psi_{\gamma_2}}.
\]
\end{proof}

\begin{lemma}\label{lemma:mean:subweibull}
Let $X$ be a sub-Weibull($\gamma$) random variable. Then
\[
\norm{\EE X}_{\psi_\gamma} \leq \norm{X}_{\psi_\gamma}.
\]
\end{lemma}

\begin{proof}[Proof of Lemma~\ref{lemma:mean:subweibull}]
We have $\norm{\EE X}_{\psi_\gamma} \leq \abr{\EE X} \leq \EE\abr{X} = \norm{X}_1 \leq \norm{X}_{\psi_\gamma}$.
\end{proof}

\begin{lemma}[Theorem~5.1 in~\citet{bradley2005basic}]\label{lemma:betamix:upperbound}
Let $\{\mathsf A_n\}_{n\in\mathbb Z}$ and
$\{\mathsf B_n\}_{n\in\mathbb Z}$ be independent processes, and let
$\mathsf C_n:=(\mathsf A_n,\mathsf B_n)$. Then, for every $j\in\NN$,
\[
\beta_{\mathsf C}(j)\leq\beta_{\mathsf A}(j)+\beta_{\mathsf B}(j).
\]
\end{lemma}

\begin{lemma}[Strictly stationary Markov process]\label{lemma:strictstationary}
Suppose that $(X_t \in \Xcal)_{t \geq 0}$ is a time-homogeneous Markov process with stationary distribution $\Pi$, and let $x_0 \sim \Pi$ be the initial state. Then the process $(X_t)_{t \geq 0}$ with $X_0 = x_0$ is strictly stationary.

Furthermore, let $f:\Xcal \rightarrow \RR^d$ be a bounded Borel measurable function, and define the process $(\Phi_n)_{n \in \NN}$ by
\[
\Phi_n = \int_{t_{n-1}}^{t_n} f(X_u) \, \mathrm{d}u,
\]
where $t_n - t_{n-1} = h$ for some fixed $h > 0$ and $n \in \NN$. Then $(\Phi_n)_{n \in \NN}$ is strictly stationary.
\end{lemma}

\begin{proof}[Proof of Lemma~\ref{lemma:strictstationary}]
We first show that $(X_t)$ is strictly stationary. Let $J\in\NN$, $u_1 \leq \cdots \leq u_J$ and $\tau \geq 0$, and consider the joint distribution
\begin{align*}
    P(X_{u_1}, \ldots, X_{u_J})
    &= P(X_{u_J} \mid X_{u_{J-1}}) \cdots P(X_{u_2} \mid X_{u_1}) P(X_{u_1}) \\
    &= P(X_{u_J+\tau} \mid X_{u_{J-1}+\tau}) \cdots P(X_{u_2+\tau} \mid X_{u_1+\tau}) P(X_{u_1}),
\end{align*}
where the second line uses the time-homogeneity of the transition kernel. Since $X_0 \sim \Pi$, we have $P(X_{u_1}) = P(X_{u_1+\tau}) = \Pi$, and hence
\begin{align*}
    P(X_{u_1}, \ldots, X_{u_J})
    &= P(X_{u_J+\tau} \mid X_{u_{J-1}+\tau}) \cdots P(X_{u_2+\tau} \mid X_{u_1+\tau}) P(X_{u_1+\tau}) \\
    &= P(X_{u_1+\tau}, \ldots, X_{u_J+\tau}).
\end{align*}
Therefore, $(X_t)$ is strictly stationary.

Now consider $\Phi_n = \int_{t_{n-1}}^{t_n} f(X_u) \, \mathrm{d}u$. For fixed $x \in \Xcal$ and $j \in \NN$, consider the conditional distribution
\begin{align*}
    P(\Phi_1, \ldots, \Phi_J \mid X_{t_0} = x)
    &= P\rbr{ \int_{t_0}^{t_1} f(X_u) \, \mathrm{d}u, \ldots, \int_{t_{J-1}}^{t_J} f(X_u) \, \mathrm{d}u  \Big|  X_{t_0} = x } \\
    &= P\rbr{ \int_{t_j}^{t_{j+1}} f(X_u) \, \mathrm{d}u, \ldots, \int_{t_{j+J-1}}^{t_{j+J}} f(X_u) \, \mathrm{d}u  \Big|  X_{t_j} = x } \\
    &= P(\Phi_{j+1}, \ldots, \Phi_{j+J} \mid X_{t_j} = x),
\end{align*}
where the second equality uses the time-homogeneity of the Markov process.

Since $X_{t_0} \sim \Pi$ and $X_{t_j} \sim \Pi$, we have $P(X_{t_0}) = P(X_{t_j})$. Therefore,
\begin{multline*}
P(\Phi_1, \ldots, \Phi_J) = \int P(\Phi_1, \ldots, \Phi_J \mid X_{t_0} = x) \, \Pi(\mathrm{d}x)
\\
= \int P(\Phi_{j+1}, \ldots, \Phi_{j+J} \mid X_{t_j} = x) \, \Pi(\mathrm{d}x)
= P(\Phi_{j+1}, \ldots, \Phi_{j+J}).
\end{multline*}
Hence, $(\Phi_n)$ is strictly stationary.
\end{proof}

\begin{lemma}[Lemma~12 in~\citet{loh2012high}]\label{lemma:sparse2all}
Let $s\geq1$ be an integer, let $\KK(s) = \{v \in \RR^p : \norm{v}_0 \leq s, \norm{v}_2 \leq 1\}$, let $A \in \RR^{p \times p}$ be a fixed matrix, and let $\delta > 0$ be a tolerance parameter. Suppose that
\[
\abr{v^\top A v} \leq \delta \quad \forall v \in \KK(2s).
\]
Then, for all $v \in \RR^p$,
\[
\abr{v^\top A v} \leq 27\delta \rbr{\norm{v}_2^2 + s^{-1}\norm{v}_1^2}.
\]
\end{lemma}

Since $v^\top Av$ is homogeneous of degree two in $v$, it suffices to verify the hypothesis of Lemma~\ref{lemma:sparse2all} for $v\in\KK(2s)$ with $\norm{v}_2=1$.

\section{Details of the simulations}
\label{appendix:simulation-details}

In this section, we give the details of the simulations in Section~\ref{sec:simulation} and report additional results.
\subsection{Experimental setup}\label{sses:expsetup}
In this subsection, we describe the general setup of the simulations and the implementation choices that give improved and stable performance.

To generate the data, we first simulate the continuous trajectories from the true ODE parameters, the true transition rate matrix and the initial conditions. Next, we sample the trajectories at equally spaced time points and corrupt the samples with i.i.d.\ centered Gaussian noise.
To sample the latent Markov chain $Z(t)$, we draw its initial state from the stationary distribution $\pi^\star=\pi(Q^\star)$, which satisfies $\pi^\star Q^\star=0$. Given $Z(t)$ and the true parameters $\Theta^\star$, we generate $X(t)$ and $\{Y_n\}$. We give the details of the parameters in the next subsection.

\subsubsection{Implementation details of wavelet smoothing}
We first apply the wavelet smoothing of Section~\ref{ssec:waveletsmooth}, using Daubechies $3$ wavelets~\citep{daubechies1992ten} in all experiments, and selecting the threshold as in Section~\ref{ssec:parameterselect}. Next, we approximate the integral $\hat{\Psi}_i(t_n)=\int_{t_{n-1}}^{t_n}g(\hat{X}_i(u))\mathrm{d}u$ with the trapezoidal rule $h\{g(\hat{X}_i(t_n))+g(\hat{X}_i(t_{n-1}))\}/2$. In the following experiments, we use the polynomial basis functions $g_i(x)=x^i$ for $i=1,\ldots,m$. We initialize $Q$ and the parameters $\{\theta_{ij}^\ell\}$ randomly, with $q_{\ell\ell'}\sim\text{Unif}(0,1)$ for $\ell\neq\ell'$, the diagonal entries of $Q$ set so that its rows sum to zero, and $\theta_{ij}^\ell\sim\Ncal(0,I)$. Our implementation parameterizes the noise by the variance $2\sigma^2$ of an increment under the surrogate model of Section~\ref{ssec:graphest}. It initializes this variance at $0.01$, so that $\sigma^2$ starts at $0.005=\sigma^{\star2}/2$, and $\sigma^2$ is updated at each M-step. The theoretical analysis permits coordinate- and state-specific penalty parameters $\lambda_i^\ell$, but we use a single shared penalty parameter $\lambda_i^\ell\equiv\lambda$ in all numerical experiments.

We compute the estimated trajectories with the function \texttt{denoise\_wavelet} of scikit-image~\citep{van2014scikit}, applied separately to the observed series of each node, with Daubechies~$3$ wavelets and soft thresholding. The implementation differs slightly from the proposed method and the theoretical analysis. First, we use BayesShrink to soft-threshold the level-$j$ detail coefficients of node $i$ at $\hat\sigma_{\mathrm{wav}}^2/\hat\sigma_{X,ij}$, where $\hat\sigma_{X,ij}$ estimates the standard deviation of the noiseless coefficients at that level, rather than using the universal threshold of Appendix~\ref{appendix:wavelet-reconstruction}. Second, at the boundary, the transform uses the symmetric signal extension of PyWavelets, and the estimated trajectory is cropped to the observed samples; this also handles sample sizes that are not powers of two. The noise level $\hat\sigma_{\mathrm{wav}}$ is computed by the function \texttt{estimate\_sigma} of scikit-image, applied to the $p\times n$ matrix of observations as a two-dimensional image: it is the median absolute value of the finest diagonal coefficients of a two-dimensional Daubechies~$2$ wavelet transform, divided by $0.6745$. The analysis of Section~\ref{sec:experiment} applies the same functions to the concatenated time-series described in Appendix~\ref{appendix:adhd-data}.

\subsubsection{Implementation details of Algorithm~\ref{alg:update:ode}}

Theorem~\ref{theorem:graphrecovery} requires an initialization in $B(r_0,\Theta^\star)$ and does not cover the random initialization above. Instead, we sweep over all candidate values of $\lambda$ with warm starts. This works well in most runs for $N\geq160$, but some runs converge to poor solutions, as reported in Appendix~\ref{ssec:simulation:withX}. The candidate set for $\lambda$ consists of the exponentials of $100$ equally spaced points in $[-7,-1]$. We first run Algorithm~\ref{alg:update:ode} with the largest value, $\lambda_{\mathrm{code}}=\exp(-1)$, starting from the random initial point $\Theta^{(0)}$ described above. When the loop of Algorithm~\ref{alg:update:ode} stops, we obtain the estimate $\hat{\Theta}_{\mathrm o}$, and we use its estimates of $\{\theta_{ij}^\ell\}$ and $Q$ to initialize Algorithm~\ref{alg:update:ode} for the next smaller value of $\lambda$, resetting $2\sigma^2$ to $0.01$ for each value of $\lambda$. Furthermore, at each M-step of Algorithm~\ref{alg:update:ode}, updating $\theta_i^\ell$ for $i=1,\ldots,p$ and $\ell=1,\ldots,k$ amounts to solving a weighted least-squares problem with a group lasso penalty. We solve it in the form~\eqref{eq:M_ntheta}, and we orthonormalize each group of columns so that the group lasso penalty is proportional to the penalty in~\eqref{eq:M_ntheta}. The solver, the Python package \emph{group\_lasso}, minimizes one half of the loss in~\eqref{eq:M_ntheta} plus $\lambda_{\mathrm{code}}\sqrt m$ times the sum of the block norms, so that $\lambda_{\mathrm{code}}=\lambda/(2\sqrt m)$; the candidate values above and all values of $\lambda$ reported for the proposed method in the simulations are values of $\lambda_{\mathrm{code}}$. We find that warm-starting each of these problems at the estimate from the previous M-step also improves the performance of the surrogate EM algorithm, compared with not using warm starts.

\paragraph*{Fitting with the intercept.}In our implementation, each of these group lasso problems is fitted with an additional unpenalized intercept: the response and design are multiplied by the square roots of the smoothing probabilities and then centered by their sample means before the group lasso step, and the fitted intercept is then discarded, so that it does not enter the E-step or the update of $\sigma^2$. The simulated processes have no constant term, and the fitted intercepts are small, but we found that this centering stabilizes the iterations of the surrogate EM algorithm: without it, the posterior frequently concentrates on a single state along the $\lambda$ path. The theoretical analysis in Section~\ref{sec:theory} covers the update~\eqref{eq:M_ntheta} without this centering.

\paragraph*{Stopping rule.}The stopping rule of Algorithm~\ref{alg:update:ode} monitors $\Lcal^{*}_{0,N}(\Theta)$, the log-likelihood of $Y_0^N$ under the surrogate model with parameter $\Theta$. This score contains no sparsity penalty. The forward recursion computes it with the initial distribution in use at that iteration and includes the term $\log\Ncal(Y_0;\hat X(t_0),2\sigma^2I_p)$ for the first observation. The superscript $(t)$ indexes algorithmic iterations and is unrelated to the physical time $t$. We use the tolerance $\epsilon_{\mathrm{tol}}=0.01$. The loop stops at the first $t\geq2$ for which $\Lcal^{*}_{0,N}(\Theta^{(t+1)})-\Lcal^{*}_{0,N}(\Theta^{(t)})<\epsilon_{\mathrm{tol}}$, including any decrease, or after $t_{\max}$ iterations, and returns the last iterate $\hat\Theta_{\mathrm o}$. The iteration limit is $t_{\max}=100$ for the proposed and oracle methods, except for the oracle fits of process~$1$ at $N\in\{80,120\}$ in Table~\ref{tab:roc_aux} and at $N=120$ in Table~\ref{tab:roc_auxa_v2}, which use $t_{\max}=10$ as the oracle method converges fast. In the forward--backward algorithm, each emission density is capped to prevent numerical underflow, at $e^{-100}$ when the algorithm is implemented in log space and at $10^{-100}$ when it is implemented with rescaled probabilities; our experiments use both forms.

\paragraph*{Edge selection criterion.}In all experiments with the proposed method, we set $\epsilon_{\mathrm{edge}}=0$, so that $\hat E^\ell$ contains the pairs whose estimated block is nonzero; the group lasso sets blocks exactly to zero. This is a practical choice rather than the threshold required by Corollary~\ref{corollary:main}. The MS-VAR baseline of Appendix~\ref{sec:baseline} instead counts a coefficient as an edge when its absolute value exceeds $10^{-3}$.

\subsubsection{Model selection and estimation error}\label{ssec:modelselection:experiment}

In this subsection, we illustrate the model selection procedure of Section~\ref{ssec:parameterselect} and assess the quality of the estimates. First, we report the number of basis functions and the penalty parameter selected at the chosen number of states. Second, we evaluate the estimates by their distances to the true parameters.

We evaluate the model selection procedure for the sample sizes $N\in\{80,120,160,200,320\}$, with $10$ runs for each, which, as in Section~\ref{sec:simulation}, are independent measurement-noise replicates conditional on one fixed latent path. We compute the surrogate BIC of Section~\ref{ssec:parameterselect} over the grid $\{1,2,3,4\}$ for the number of states, $\{1,2,3,4,5\}$ for the number of basis functions, and the exponentials of $50$ equally spaced points in $[-7,-1]$ for $\lambda$. We choose $k=2$ by visual inspection of the first elbow point of the minimal surrogate BIC as $k$ increases over the grid. For each run, we then select $m$ and $\lambda$ by minimizing the surrogate BIC at $k=2$, breaking ties toward the largest $\lambda$ and then the smallest $m$.
Figure~\ref{fig:parameter_selection} shows that, when $N\geq160$, the selected number of basis functions equals the true one in most runs. For sample sizes greater than $80$, the selected $\lambda$ appears to decrease as the sample size increases.

To assess the estimates compared to the true parameters, we do not use the selected $k$ and $m$. Instead, we use fits of the proposed method with $\hat{X}(t)$ at the true values $k^\star=2$ and $m^\star$, where $m^\star=3$ for process~$1$ and $m^\star=1$ for process~$2$, over the $100$ values of $\lambda$ of Section~\ref{sec:simulation}, and we select only $\lambda$, as described below. These fits come from repeating that experiment on the same data, with the same data scripts and values of $\lambda$, and saving the estimates at every $\lambda$; they differ from the fits of Table~\ref{tab:roc_aux} in their random initial values and in corrections to the implementation made after those fits, and Appendix~\ref{sses:expsetup} describes the implementation used here. The estimated and true parameters $\theta_{ij}^\ell$ then have the same dimension, and the estimated and true states can be matched by a permutation. We then compute the distances of the estimates to the true parameters. For this, we need to match the estimates of each state $\ell$ with the true parameters of the same state $\ell$. Since the order of the states is unknown, we find the permutations of the states that minimize the following two objective functions. Let $\Pcal(k)$ be the collection of all permutations of $1,\ldots,k$ and define
\begin{align}
\Xi_1 = {\argmin_{\tilde{\Xi}\in \Pcal(k)}\sum_{\ell=1}^k\sum_{i,j=1}^p\left\|{\hat{\theta}_{ij}^{\tilde{\Xi}(\ell)} - {\theta}_{ij}^{\ell\star}}\right\|_2}, \quad \Xi_2 = \argmin_{\tilde{\Xi}\in \Pcal(k)} \sum_{\ell\neq\ell'}\abr{\hat{q}_{\tilde{\Xi}(\ell)\tilde{\Xi}(\ell')}-q_{\ell\ell'}^\star}.\label{eq:permutation}
\end{align}
Each $\argmin$ in~\eqref{eq:permutation} is taken as the first minimizer when the permutations in $\Pcal(k)$ are listed in lexicographic order, starting with the identity; in particular, $\Xi_1$ is the identity when all $\hat{\theta}_{ij}^{\ell}$ are zero. We set $\Xi=\Xi_1$ if $\Xi_1$ is identical to $\Xi_2$. If $\Xi_1$ differs from $\Xi_2$, we exclude the estimate for that value of $\lambda$ from the evaluation. Either permutation would define a comparison with the true parameters, so this agreement rule is an evaluation exclusion rule rather than a necessity: it keeps only the values of $\lambda$ at which the matchings by $\{\theta_{ij}^\ell\}$ and by the transition rates agree. We report the numbers of valid values of $\lambda$ below and in Appendix~\ref{ssec:roc:computation}. After matching the states, we compute the distances of the matched estimates to $Q^\star$, $\{\theta_{ij}^{\ell\star}\}$ and $\sigma^{\star2}$, defined, respectively, as
\[
\norm{\hat{Q}(\Xi)-Q^\star}_{\mathrm F},\quad\sum_{\ell=1}^k\sum_{i,j=1}^p\norm{\hat{\theta}_{ij}^{\Xi(\ell)}-\theta_{ij}^{\ell\star}}_2,\quad \abr{\hat{\sigma}^{2}-\sigma^{\star2}},
\]
where $\hat{Q}(\Xi)=[\hat{q}_{\Xi(\ell)\Xi(\ell')}]_{\ell,\ell'}$ and $\norm{\cdot}_{\mathrm F}$ is the Frobenius norm. Since $k=2$ and the rows of $\hat{Q}(\Xi)$ and $Q^\star$ sum to zero, the matrix $\hat{Q}(\Xi)-Q^\star$ has rank at most one, so its Frobenius norm equals its spectral norm. 

Figure~\ref{fig:parameter_estimation} shows these distances for $N\leq200$, where, for each run, $\lambda$ minimizes the score among the values of $\lambda$ with $\Xi_1=\Xi_2$. The number of values of $\lambda$ with $\Xi_1=\Xi_2$ per run, out of $100$, has median (range) $13.5$ ($2$--$30$), $21.5$ ($16$--$57$), $56.5$ ($22$--$65$) and $53.5$ ($23$--$68$) for process~$1$ at $N=80$, $120$, $160$ and $200$, and $63$ ($40$--$83$), $98.5$ ($57$--$100$), $99$ ($83$--$100$) and $98$ ($81$--$100$) for process~$2$, so every run retains at least one value of $\lambda$. In one run of process~$2$ at $N=80$ and one at $N=120$, the surrogate EM algorithm returned non-finite estimates from $\lambda\approx0.0056$ and $\lambda\approx0.012$ on, after one state absorbed all observations; the path stops there, and the smaller values of $\lambda$ of these two runs are not fitted. For process~$1$ at $N=80$ and $N=120$, the median distance of $\{\theta_{ij}^\ell\}$ equals $\sum_{\ell,i,j}\|\theta_{ij}^{\ell\star}\|_2\approx12.9$, the distance of the zero estimate. The distances drop at $N=160$, sharply for $\sigma^2$ but only modestly for $\{\theta_{ij}^\ell\}$, whose median distance stays about three quarters of that of the zero estimate, and they do not decrease further at $N=200$. For process~$2$, the distances of $\{\theta_{ij}^\ell\}$ and $\sigma^2$ decrease with the sample size, and the distance of $Q$ is large at $N=80$.

\subsection{ROC curves and AUC summaries}\label{ssec:roc:computation}

For each run and each value of $\lambda$, let $\hat{E}^{\Xi(\ell)}$ denote the estimated edge set, where $\Xi$ is the permutation of $[k]$, defined in~\eqref{eq:permutation}, that matches the estimated states to the true states. The true positive rate (TPR) and the false positive rate (FPR) of state $\ell$ are
\begin{align*}
\text{TPR}(\ell) &= \frac{|\hat{E}^{\Xi(\ell)}\cap E^{\ell}|}{|E^\ell|},\qquad
\text{FPR}(\ell) = \frac{|\hat{E}^{\Xi(\ell)}\setminus E^{\ell}|}{p^2-|E^\ell|}.
\end{align*}
ROC curves are computed over all $p^2$ entries per state, including diagonal entries.
Finally, we keep the values of $\lambda$ for which the matching by $\{\theta_{ij}^\ell\}$ and the matching by the transition rates, $\Xi_1$ and $\Xi_2$ in~\eqref{eq:permutation}, coincide, with ties resolved as stated after~\eqref{eq:permutation}, and form the ROC curve of each run from their TPR and FPR.

The oracle method uses the true states, so its states need no matching: we evaluate its fits with the identity permutation at every $\lambda$, without the agreement rule. For each run, we keep the largest TPR among points with the same FPR, add the end points $(0,0)$ and $(1,1)$, and compute the AUC. A run for which $\Xi_1$ and $\Xi_2$ do not agree at any $\lambda$ over the grid is discarded.

In Tables~\ref{tab:roc_aux} and~\ref{tab:roc_auxa_v2}, we report a dash when any of the $10$ runs has no valid ROC curve; otherwise, each numerical cell gives the mean and standard deviation of the AUC over all $10$ runs. The standard deviation is computed with divisor $10$. The affected entries occur only for the proposed method on process~$1$ at $N=40$: $5$ of the $10$ runs with $\hat{X}(t)$ and $2$ of the $10$ runs with $X(t)$ given have no valid ROC curve. For completeness, over the $5$ valid runs with $\hat{X}(t)$, the mean AUCs of states~$1$ and~$2$ are $0.56$ and $0.55$; over the $8$ valid runs with $X(t)$ given, they are $0.45$ and $0.52$. These summaries are conditional on the runs being valid. Every other run in these two tables has at least one valid $\lambda$. For the proposed method with $\hat{X}(t)$, the number of valid values of $\lambda$ per run, out of $100$, has median (range) $0.5$ ($0$--$48$), $15$ ($4$--$37$), $21.5$ ($2$--$61$), $37$ ($34$--$63$) and $51.5$ ($31$--$64$) for process~$1$ at $N=40$, $80$, $120$, $160$ and $200$, and $64.5$ ($34$--$97$), $72$ ($59$--$88$), $92$ ($78$--$98$), $86.5$ ($78$--$99$) and $77.5$ ($77$--$99$) for process~$2$.

Figures~\ref{fig:roc_plot_hatgl}, \ref{fig:roc_plot_gl} and~\ref{fig:roc_msar} show pooled ROC curves. For each method, process, state and sample size, we pool the ROC points of the $10$ runs, that is, the points $(\text{FPR},\text{TPR})$ at the valid values of $\lambda$, keeping the largest TPR among points of the same run with the same FPR. We group the pooled points into FPR bins of width $0.1$, average the TPR within each nonempty bin, plot that average at the left endpoint of the bin, and join these averages by line segments from $(0,0)$ to $(1,1)$. This pooled display weights ROC points rather than runs equally, so a pooled curve need not be nondecreasing, and its area need not equal the corresponding mean AUC in Tables~\ref{tab:roc_aux}, \ref{tab:roc_auxa_v2} and~\ref{tab:roc_msar}. We omit a pooled curve if any of its $10$ runs has no valid ROC curve; this occurs only for the proposed method on process~$1$ at $N=40$ in Figures~\ref{fig:roc_plot_hatgl} and~\ref{fig:roc_plot_gl}.

\subsection{Parameters of data-generating process 1}\label{ssec:details:dgp1}
We now list the parameters of data-generating process 1 in~\eqref{eq:dgp1:1}--\eqref{eq:dgp1:2}. In state $\ell=1$, we have
\begin{align*}
    \theta_{11}^{1\star} &= (1.2, 0.3, -0.6)^\top , \theta_{12}^{1\star} = (0.1,0.2,0.2)^\top;\\
    \theta_{21}^{1\star} &= (-2.0, 0.0, 0.4)^\top , \theta_{22}^{1\star} = (0.5,0.2,-0.3)^\top;\\ 
    \theta_{33}^{1\star} &= (0.0, 0.0, 0.0)^\top,\theta_{34}^{1\star} = (-0.3,0.4,0.1)^\top;\\
    \theta_{43}^{1\star} & = (0.2,-0.1,-0.2)^\top,\theta_{44}^{1\star}= (0.0, 0.0, 0.0)^\top;\\
    \theta_{55}^{1\star} &= (0.0, 0.0, 0.0)^\top,\theta_{56}^{1\star} = (0.1,0.0,-0.8)^\top;\\
    \theta_{65}^{1\star} & = (0.0,0.0,0.5)^\top,\theta_{66}^{1\star} = (0.0, 0.0, 0.0)^\top. 
\end{align*}
The remaining parameters are $\theta_{ij}^{1\star}=(0.0,0.0,0.0)^\top$, and Figure~\ref{fig:dgp1:1} shows the graph of state $\ell=1$. In state $\ell=2$, we have
\begin{align*}
    \theta_{55}^{2\star} &= (1.2, 0.3, -0.6)^\top , \theta_{56}^{2\star} = (0.1,0.2,0.2)^\top;\\
    \theta_{65}^{2\star} &= (-2.0, 0.0, 0.4)^\top , \theta_{66}^{2\star} = (0.5,0.2,-0.3)^\top;\\ 
    \theta_{77}^{2\star} &= (0.0, 0.0, 0.0)^\top,\theta_{78}^{2\star} = (-0.3,0.4,0.1)^\top;\\
    \theta_{87}^{2\star} & = (0.2,-0.1,-0.2)^\top,\theta_{88}^{2\star}= (0.0, 0.0, 0.0)^\top;\\
    \theta_{99}^{2\star} &= (0.0, 0.0, 0.0)^\top,\theta_{9,10}^{2\star} = (0.1,0.0,-0.8)^\top;\\
    \theta_{10,9}^{2\star} & = (0.0,0.0,0.5)^\top,\theta_{10,10}^{2\star} = (0.0, 0.0, 0.0)^\top.
\end{align*}
The remaining parameters are $\theta_{ij}^{2\star}=(0.0,0.0,0.0)^\top$, and Figure~\ref{fig:dgp1:2} shows the graph of state $\ell=2$. Figure~\ref{fig:dgp1_xy} shows the simulated trajectories for the initial condition $X(0)$ and the matrix $Q^\star$ given in Section~\ref{sec:simulation}.

\begin{figure}
    \centering
    \begin{tikzpicture}[>=stealth, node distance=5cm]

\foreach \x in {1,2,...,10} {
    \ifodd\x
        \node[circle, draw] (X\x) at ({(\x-1)},1) {$X_{\x}$};
    \else
        \node[circle, draw] (X\x) at ({(\x-2)},-.5) {$X_{\x}$};
    \fi
}

\foreach \x in {1,2}
    \draw[->, line width=1.5pt] (X\x) edge [loop left] ();

\draw[<->, line width=1.5pt] (X1) to  (X2);
\draw[<->, line width=1.5pt] (X3) to  (X4);
\draw[<->, line width=1.5pt] (X5) to  (X6);
\end{tikzpicture}

    \caption{The graph associated with state $\ell=1$ has self-loops on nodes $X_1$ and $X_2$, and bidirectional edges between nodes $X_1$ and $X_2$, $X_3$ and $X_4$, and $X_5$ and $X_6$.}
    \label{fig:dgp1:1}
\end{figure}

\begin{figure}
    \centering
    \begin{tikzpicture}[>=stealth, node distance=5cm]

\foreach \x in {1,2,...,10} {
    \ifodd\x
        \node[circle, draw] (X\x) at ({(\x-1)},1) {$X_{\x}$};
    \else
        \node[circle, draw] (X\x) at ({(\x-2)},-.5) {$X_{\x}$};
    \fi
}

\foreach \x in {5,6}
    \draw[->, line width=1.5pt] (X\x) edge [loop left] ();

\draw[<->, line width=1.5pt] (X5) to  (X6);
\draw[<->, line width=1.5pt] (X7) to  (X8);
\draw[<->, line width=1.5pt] (X9) to  (X10);
\end{tikzpicture}
    \caption{The graph associated with state $\ell=2$ has self-loops on nodes $X_5$ and $X_6$, and bidirectional edges between nodes $X_5$ and $X_6$, $X_7$ and $X_8$, and $X_9$ and $X_{10}$.}
    \label{fig:dgp1:2}
\end{figure}

\begin{figure}
    \centering
    \includegraphics[width=1.\textwidth]{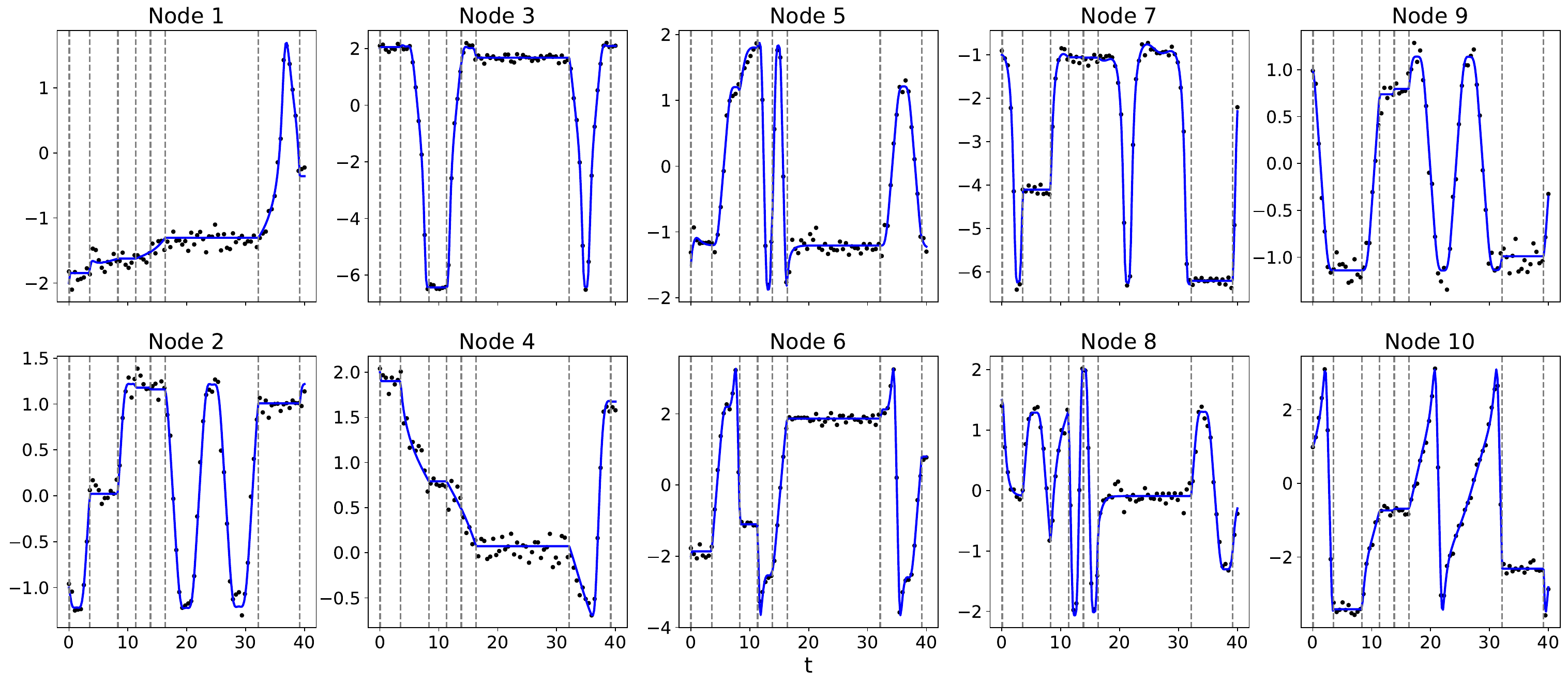}
    \caption{The simulated trajectories of data-generating process 1, defined in~\eqref{eq:dgp1:1}--\eqref{eq:dgp1:2}, with the parameters of Section~\ref{ssec:details:dgp1}. The blue solid lines denote $X_i(t)$ for $i=1,\ldots, 10$, the black dots denote the noisy observations $Y_{n,i}$ for $i=1,\ldots,10$ and $n=0,\ldots,N-1$, and the grey dashed lines mark the times at which the latent Markov chain $Z(t)$ switches states. }
    \label{fig:dgp1_xy}
\end{figure}

\begin{figure}
    \centering
        \begin{tikzpicture}[>=stealth, node distance=10cm]
        \def\nodesize{1cm}
        \foreach \c in {1,6}{
        \node[minimum size=\nodesize,circle, draw] (center) at (\c,0) {$X_{\c}$};
        \foreach \i in {1,...,4} {
            \pgfmathparse{int(\c + \i)}
            \edef\sum{\pgfmathresult}
            \node[minimum size=\nodesize,circle, draw] (X\i) at ({1.8*cos(\i*90)+\c},{1.8*sin(\i*90)}) {$X_{\sum}$};
        }

        \foreach \i in {1,...,4} {
            \draw[<->, line width=1.5pt] (center) -- (X\i);
        }
        }

        \foreach \c in {11,16}{
        \node[minimum size=\nodesize,circle, draw] (center) at (\c-10,-5.5) {$X_{\c}$};
        \foreach \i in {1,...,4} {
            \pgfmathparse{int(\c + \i)}
            \edef\sum{\pgfmathresult}
            \node[minimum size=\nodesize,circle, draw] (X\sum) at ({1.8*cos(\i*90)+\c-10},{1.8*sin(\i*90)-5.5}) {$X_{\sum}$};
        }

        \foreach \i in {1,...,4} {
            \pgfmathparse{int(\c + \i)}
            \edef\sum{\pgfmathresult}
            \draw[<->, line width=1.5pt] (center) -- (X\sum);
        }
        }

    \end{tikzpicture}
    \caption{The graph of state $1$ of the data-generating process $2$.}
    \label{fig:dgp2:1}
\end{figure}

\begin{figure}
    \centering
        \begin{tikzpicture}[>=stealth, node distance=5cm]
        \def\nodesize{1cm}

        \foreach \x in {1,...,10} {
            \node[minimum size=\nodesize,circle, draw] (X\x) at (\x*1.5,1) {$X_{\x}$};
        }

        \foreach \x [count=\i from 11] in {11,...,20} {
            \node[minimum size=\nodesize,circle, draw] (X\x) at (21*1.5-\i*1.5,-0.5) {$X_{\x}$};
        }
    
        \foreach \x [evaluate=\x as \nextx using int(\x+1)] in {1,...,19} {
            \draw[<->, line width=1.5pt] (X\x) -- (X\nextx);
        }

        \draw[<->, line width=1.5pt] (X1) to (X20);
        
    \end{tikzpicture}
    \caption{The graph of state $2$ of the data-generating process $2$.}
    \label{fig:dgp2:2}
\end{figure}

\begin{figure}
    \centering
    \includegraphics[width=1.\textwidth]{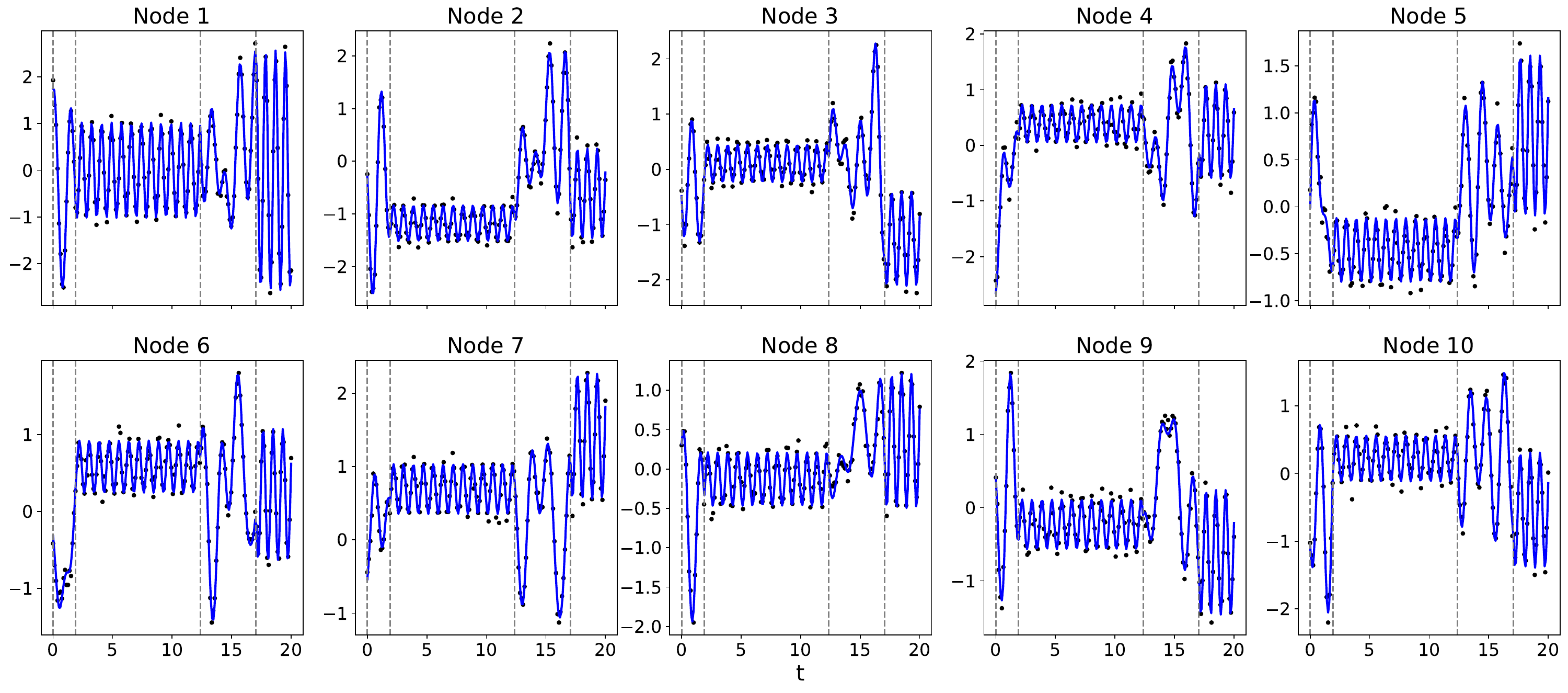}
    \caption{The simulated trajectories of data-generating process 2, defined in~\eqref{eq:dgp2:1}--\eqref{eq:dgp2:2}, with the parameters of Section~\ref{sec:simulation}. We plot the first half $[0,20]$ of the time interval $[0,40]$ and the first $10$ of the $p=20$ nodes.
    The blue solid lines denote $X_i(t)$, the black dots denote the noisy observations $Y_{n,i}$ for $i=1,\ldots,10$ at the observation times in $[0,20]$, and the grey dashed lines mark the times at which the latent Markov chain $Z(t)$ switches states. }
    \label{fig:dgp2_xy}
\end{figure}

\begin{figure}[!t]
    \centering
\includegraphics[width=\linewidth]{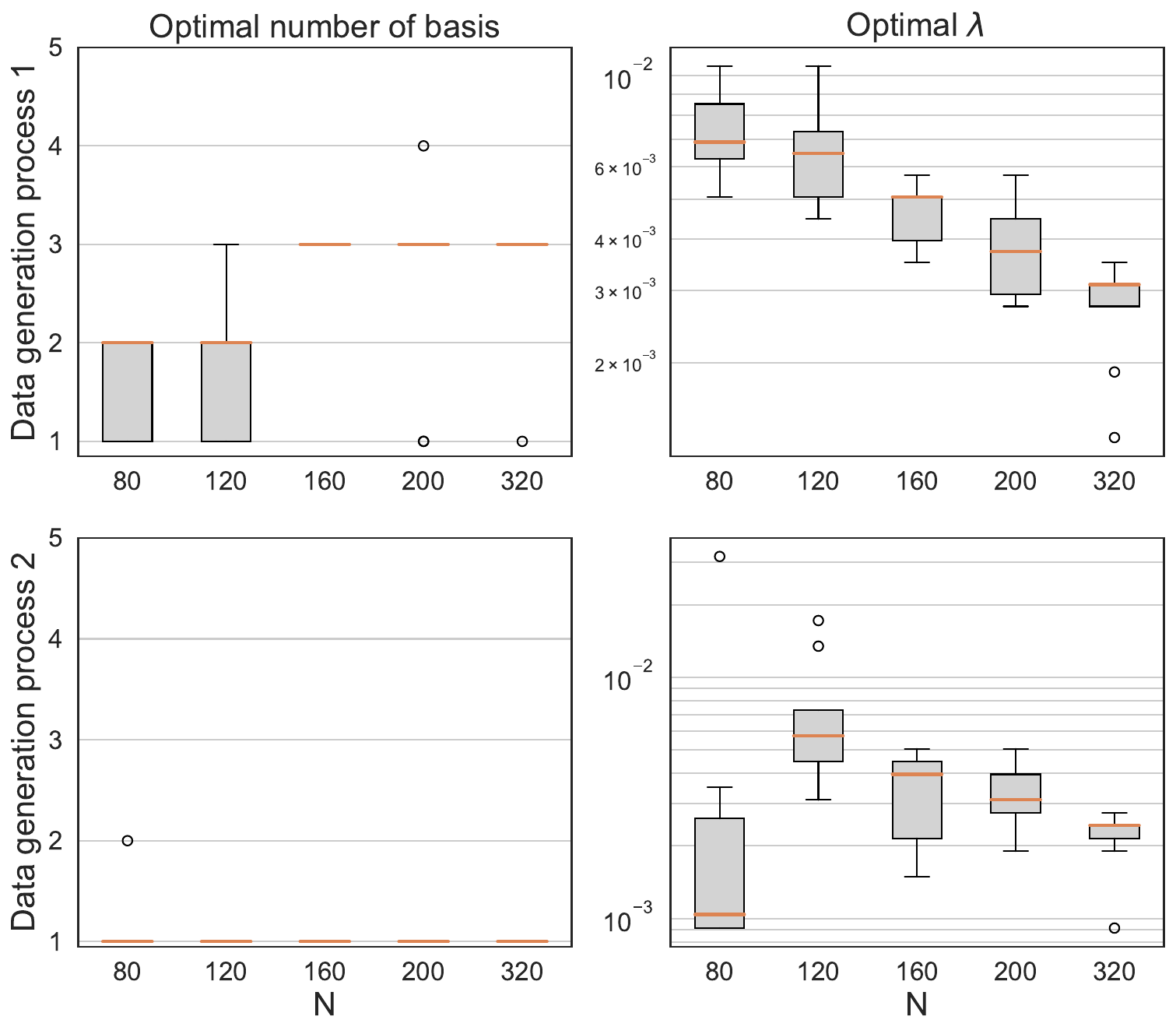}
    \caption{Box plots of the number of basis functions $m$ and the penalty parameter $\lambda$ selected at $k=2$ by the procedure described in Appendix~\ref{ssec:modelselection:experiment}. The orange line denotes the median, the upper line of the box denotes the third quartile (Q$3$), and the lower line of the box denotes the first quartile (Q$1$). The whiskers extend from the box toward Q$3+1.5$IQR and Q$1-1.5$IQR, where IQR is the inter-quartile range. The dots denote the outliers that lie outside the ends of the whiskers. The true number of states is $2$ for both cases, and the true number of basis functions is $3$ for process~$1$ and $1$ for process~$2$.
    When the sample size is sufficiently large, the selected number of basis functions equals the true one in most runs ($m=3$ for process~$1$ and $m=1$ for process~$2$). Furthermore, for process~$1$ the selected $\lambda$ decreases as the sample size increases, and for process~$2$ it tends to decrease for $N\geq120$.}
    \label{fig:parameter_selection}
\end{figure}
\begin{figure}[!t]
    \centering
    \includegraphics[width=\linewidth]{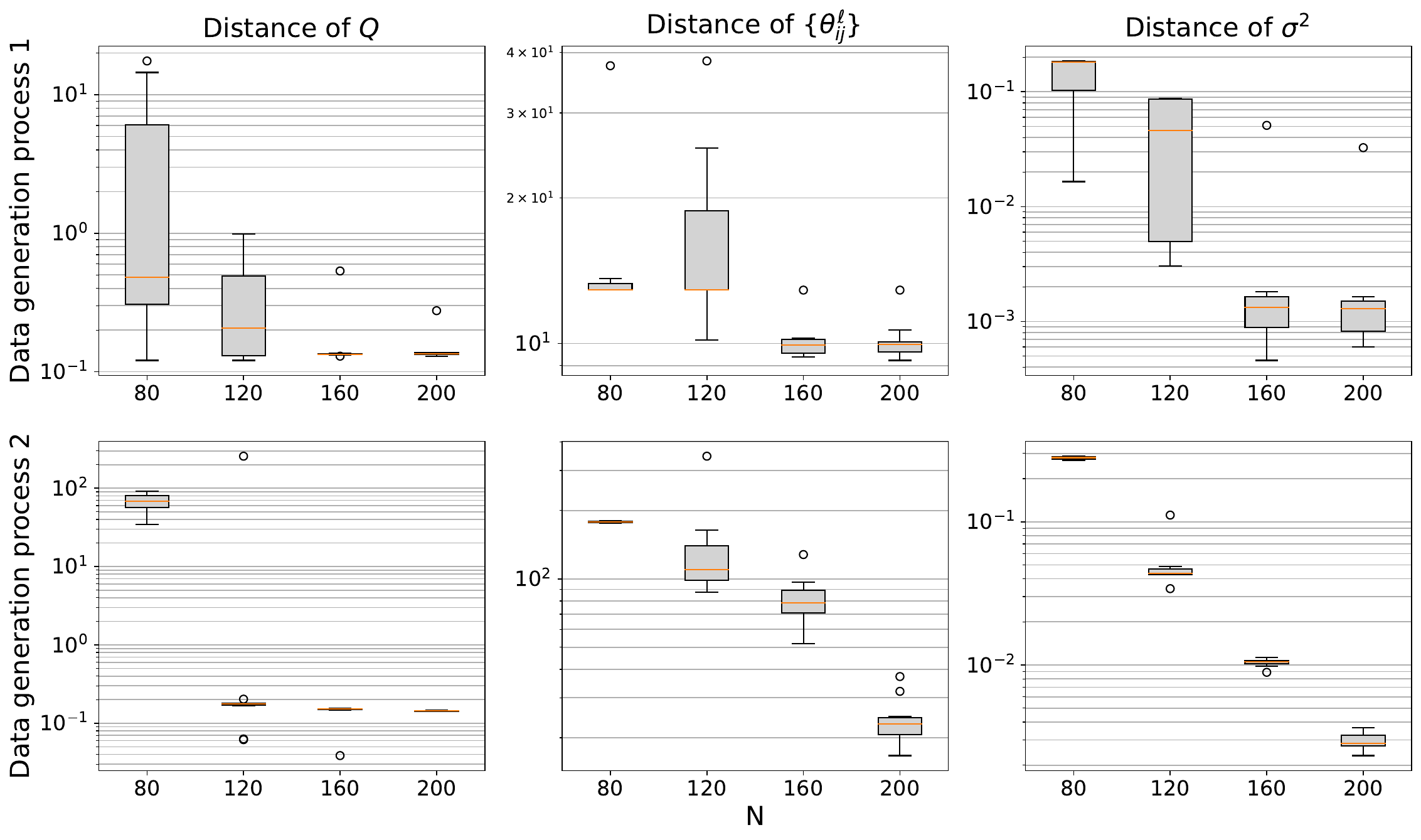}
    \caption{The distances of the estimated parameters to the true parameters for $N\leq200$, at the true $k^\star=2$ and $m^\star$ and with only $\lambda$ selected by the procedure described in Appendix~\ref{ssec:modelselection:experiment}. The top row shows process~$1$, where the distances drop between $N=120$ and $N=160$, sharply for $\sigma^2$ and modestly for $\{\theta_{ij}^\ell\}$; at $N\leq120$, the median distance of $\{\theta_{ij}^\ell\}$ is that of the zero estimate. The bottom row shows process~$2$, where the distances of $\{\theta_{ij}^\ell\}$ and $\sigma^2$ decrease as the sample size increases.}
    \label{fig:parameter_estimation}
\end{figure}

\subsection{Estimation given \texorpdfstring{$X(t)$}{X(t)}}\label{ssec:simulation:withX}

We repeat the simulations of Figure~\ref{fig:roc_plot_hatgl} with $X(t)$ given, so that we do not need the wavelet smoothing to obtain $\hat{X}(t)$. Using the true trajectory removes the reconstruction error, and the comparison assesses the effect of this error on graph recovery empirically.

Figure~\ref{fig:roc_plot_gl} and Table~\ref{tab:roc_auxa_v2} report the results of Algorithm~\ref{alg:update:ode} with the true $X(t)$, and we compare them with the results for $\hat{X}(t)$ in Table~\ref{tab:roc_aux}. For the proposed method, knowing $X(t)$ helps most at moderate sample sizes: at $N=80$ and $N=120$, the AUC is higher for all four graphs, by $0.03$ to $0.40$. At $N=40$, the proposed-method summaries for process~$1$ are omitted because the state-matching rule leaves some runs with no valid ROC curve, both with estimated trajectories and with the true trajectories, as reported in Appendix~\ref{ssec:roc:computation}. At $N=160$ and $N=200$, the two AUCs differ by at most $0.06$, except for state~$1$ of the second process at $N=160$, where the AUC is $0.77$ with $X(t)$ and $0.91$ with $\hat{X}(t)$; with $X(t)$, five of the ten runs converge to a poor solution for this state (AUC between $0.60$ and $0.62$), while the other five reach $0.91$ to $0.96$. For the oracle method, knowing $X(t)$ changes the AUC on the first process by at most $0.04$ when $N\geq 80$; on the second process it helps most at small sample sizes, raising the AUCs at $N=40$ from $0.77$ and $0.50$ to $0.86$ and $0.90$, and by $0.08$ to $0.21$ at $N=80$. Finally, among the reported comparisons with $X(t)$ given, the oracle method has an AUC at least as high as the proposed method in all but one cell, where the proposed method is higher by less than $0.01$. This comparison reflects the additional task of estimating the smoothing probabilities of the latent states, which the oracle method does not need.

\begin{figure}[!t]
    \centering
\includegraphics[width=\textwidth]{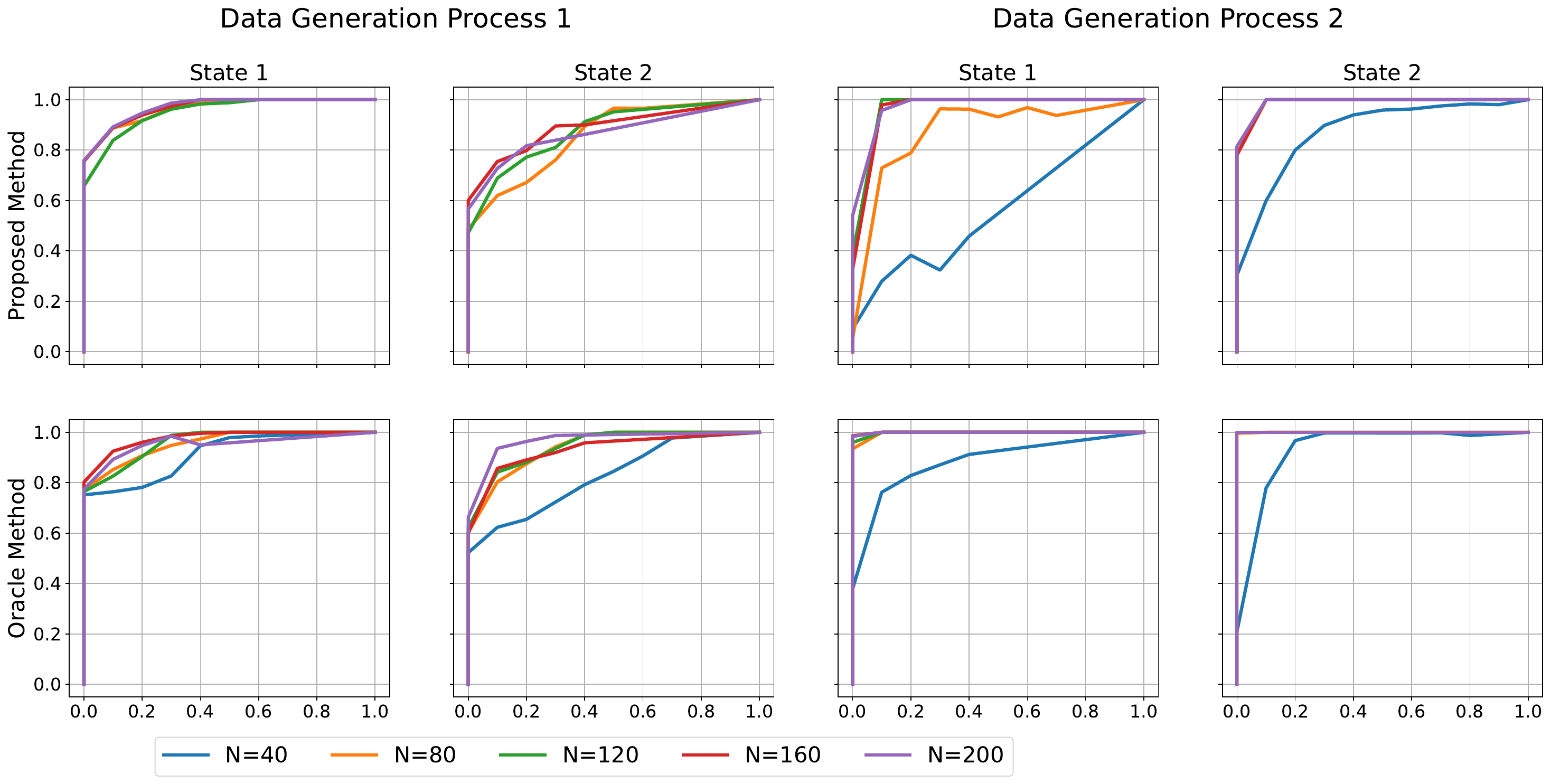}
    \caption{
    Pooled ROC curves of the two data-generating processes over $10$ runs (independent measurement-noise replicates conditional on one fixed latent path), with the true $X(t)$ given and FPR bins of width $0.1$, using the pooling procedure described in Appendix~\ref{ssec:roc:computation}. The area under a pooled curve can differ substantially from the mean AUC in Table~\ref{tab:roc_auxa_v2}. \textbf{Top row.} Results of the proposed method.
    \textbf{Bottom row.} Results of the oracle method, where the latent Markov chain $Z(t)$ is assumed to be known; its fits are evaluated with the identity permutation at every $\lambda$. The variation across runs and the comparison with estimated trajectories are discussed in the text.
    }
    \label{fig:roc_plot_gl}
\end{figure}

\begin{table}[t]
    \centering
    \caption{The AUCs of the two data-generating processes in Figure~\ref{fig:roc_plot_gl}, with the true $X(t)$ given. Each numerical cell gives the mean AUC and, in parentheses, its standard deviation over all $10$ runs; the runs are independent measurement-noise replicates conditional on one fixed latent path: they share $Z(t)$, $X(0)$ and $X(t)$ and differ only in the measurement noise. A dash indicates that some runs have no valid ROC curve under the state-matching rule: this occurs in $2$ of the $10$ runs for the proposed method on process~$1$ at $N=40$, as reported in Appendix~\ref{ssec:roc:computation}. The oracle method additionally assumes that the latent Markov chain is known and is evaluated with the identity permutation at every $\lambda$.}
    \label{tab:roc_auxa_v2}
    \resizebox{\textwidth}{!}{
    \begin{tabular}{lllllllll}\toprule{} & 
    \multicolumn{4}{l}{Data-Generating Process 1}
    &
    \multicolumn{4}{l}{Data-Generating Process 2}\\
    &\multicolumn{2}{l}{Proposed Method} & \multicolumn{2}{l}{Oracle Method} & 
 \multicolumn{2}{l}{Proposed Method} & \multicolumn{2}{l}{Oracle Method} \\
    State &           1 &           2 &           1 &           2 &           1 &           2 &           1 &           2 \\
    N   &             &             &             &             &             &             &             &             \\\midrule
40 & \textemdash & \textemdash & 0.90(0.016) & 0.80(0.027) & 0.53(0.049) & 0.84(0.010) & 0.86(0.006) & 0.90(0.005) \\
80 & 0.93(0.056) & 0.83(0.029) & 0.94(0.026) & 0.92(0.012) & 0.81(0.059) & 0.98(0.002) & 1.00(0.001) & 1.00(0.000) \\
120 & 0.92(0.046) & 0.85(0.040) & 0.95(0.017) & 0.93(0.014) & 0.90(0.135) & 0.98(0.000) & 1.00(0.000) & 1.00(0.000) \\
160 & 0.95(0.020) & 0.88(0.034) & 0.97(0.020) & 0.92(0.023) & 0.77(0.167) & 0.98(0.001) & 1.00(0.000) & 1.00(0.000) \\
200 & 0.96(0.023) & 0.86(0.044) & 0.96(0.029) & 0.95(0.024) & 0.93(0.046) & 0.98(0.001) & 1.00(0.000) & 1.00(0.000) \\
\bottomrule
\end{tabular}}
\end{table}

\section{MS-VAR baseline}
\label{sec:baseline}
In this section, we present the simulation results of a Markov-switching vector autoregressive (MS-VAR) model of order one, which serves as a baseline; we call it the MS-VAR baseline. This model can be viewed as a discrete-time approximation of Markov-switching linear differential equations, which are a special case of our model. The MS-VAR baseline is
\[
Y_n = A^{\tilde{Z}(t_n)}Y_{n-1}+\varepsilon^{\mathrm{AR}}_n,\qquad \varepsilon^{\mathrm{AR}}_n\sim\mathcal{N}(0,\sigma^2I),
\]
where $\tilde{Z}(t_n)$ is a discrete-time Markov chain with transition probability matrix $P$, and $A^\ell\in\RR^{p\times p}$ for $\ell=1,\ldots,k$. In our model, the sampled chain $Z(t_n)$ is such a chain, with $P=\exp(Q^\star h)$.
We fit the model by an EM algorithm and add a Lasso penalty with penalty parameter $\lambda$ to the M-step for $A^\ell$. As in the implementation of the proposed method described in Appendix~\ref{sses:expsetup}, each Lasso problem is fitted with an additional unpenalized intercept: the response and design are multiplied by the square roots of the smoothing probabilities and then centered before the Lasso step, and the fitted intercept is discarded, so that the fitted model has no intercept.

The estimated edge set of state $\ell$ is constructed as 
\[
\{(i,j):|\hat{A}^\ell_{ji}|>10^{-3}\}.
\] 
Note that we regress $Y_n$ on $Y_{n-1}$, so that, for a linear system with $B^\ell_{ij}=\theta^{\ell\star}_{ij}$, $\hat A^\ell$ estimates approximately $I+hB^\ell$, and the estimated edge set includes the diagonal. As for the proposed method, we score all $p^2$ entries of each state, against the edge set $E^\ell$ of the ODE, as described in Appendix~\ref{ssec:roc:computation}. Since the exact discretization of a linear system is $\exp(hB^\ell)=I+hB^\ell+O(h^2)$, its diagonal entries and the indirect entries from the $O(h^2)$ terms count as false positives wherever $E^\ell$ has no edge, which raises the FPR of the MS-VAR baseline.

As in Section~\ref{sses:expsetup}, we compute the ROC curve by decreasing $\lambda$ from $\exp(-1)$ to $\exp(-7)$. We initialize the entries of $A^\ell$ independently from the standard normal distribution. For the transition probability matrix $P$, we draw the entries from a uniform distribution and then normalize each row to sum to $1$.

The MS-VAR baseline estimates $P$ and $A^\ell$ rather than $Q$ and $\theta^\ell$, so we match its states in its own parameterization. The permutation $\Xi_2$ minimizes $\|\hat P^{\Xi}-\exp(Q^\star h)\|_{\mathrm F}$, where $\hat P^{\Xi}_{\ell\ell'}=\hat P_{\Xi(\ell)\Xi(\ell')}$, and $\Xi_1$ minimizes $\sum_{\ell=1}^k\|\hat A^{\Xi(\ell)}-A^{\ell\star}\|_{\mathrm F}$. Here $A^{\ell\star}=hB^\ell$ for data-generating process~$2$; for data-generating process~$1$, which is nonlinear and has no linear autoregressive target, $A^{\ell\star}$ is the matrix of scaled block norms $h\|\theta_{ij}^{\ell\star}\|_2$. As in Appendix~\ref{ssec:roc:computation}, the estimator associated with a $\lambda$ is evaluated only if $\Xi_1=\Xi_2$; this holds for at least $21$ of the $100$ values of $\lambda$ in every run. The oracle MS-VAR baseline replaces the estimated smoothing probabilities in the E-step by the true sample path of $Z(t_n)$, so that only the M-step is executed. Its states are therefore the true ones, and we evaluate it with the identity permutation at every $\lambda$.

Figure~\ref{fig:roc_msar} displays the results for data-generating processes $1$ and $2$, and Table~\ref{tab:roc_msar} reports the corresponding AUCs. Compared with the proposed method in Table~\ref{tab:roc_aux}, the MS-VAR baseline is better for every state at $N=80$, by $0.03$ to $0.30$, better for two of the four states at $N=120$, and worse for every state at $N\geq160$. The oracle MS-VAR baseline performs worse than the oracle method of Table~\ref{tab:roc_aux} on data-generating process~$1$, with AUCs of $0.83$--$0.84$ for state~$1$ and $0.64$--$0.72$ for state~$2$, against $0.90$--$0.94$ and $0.90$--$0.96$. On data-generating process~$2$, it performs comparably to this oracle method at $N\leq120$ and slightly worse at $N\geq160$.
This is because, by the construction in Section~\ref{sec:simulation}, the first data-generating process is governed by nonlinear ODEs, while the second is governed by linear ODEs.

In these simulations, the MS-VAR baseline's posterior concentrates on a single state at more values of $\lambda$ as $N$ increases. We have not determined the cause of this behavior.

\begin{figure}[!t]
    \centering
\includegraphics[width=\textwidth]{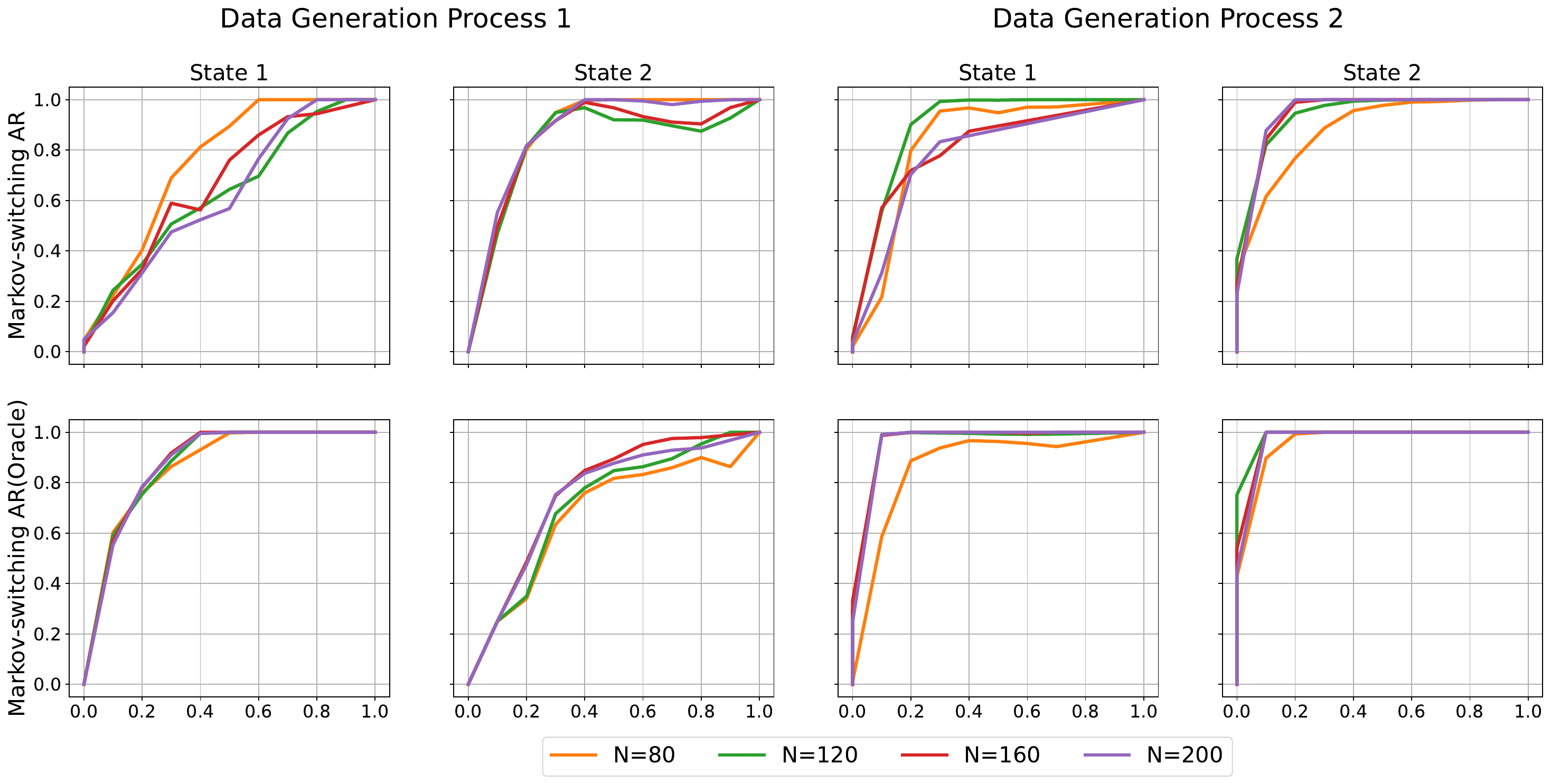}
    \caption{
    Pooled ROC curves of the two data-generating processes over $10$ runs (independent measurement-noise replicates conditional on one fixed latent path), with FPR bins of width $0.1$, using the pooling procedure described in Appendix~\ref{ssec:roc:computation}. The area under a pooled curve can differ substantially from the mean AUC in Table~\ref{tab:roc_msar}. The orange, green, red and purple curves correspond to $N=80$, $120$, $160$ and $200$, respectively, as in Figure~\ref{fig:roc_plot_hatgl}. \textbf{Top row.} Results of the MS-VAR baseline. Its posterior concentrates on one state as $N$ grows, as discussed in Appendix~\ref{sec:baseline}, and state~$1$ of data-generating process~$1$ is not recovered at any sample size.
    \textbf{Bottom row.} Results of the oracle MS-VAR baseline, where the sampled chain $\{Z(t_n)\}$ is assumed to be known; its fits are evaluated with the identity permutation at every $\lambda$.}
    \label{fig:roc_msar}
\end{figure}

\begin{table}
    \centering
    \caption{The AUCs of the two data-generating processes in Figure~\ref{fig:roc_msar}. Each numerical cell gives the mean AUC and, in parentheses, its standard deviation over all $10$ runs. The runs are independent measurement-noise replicates conditional on one fixed latent path: they share $Z(t)$, $X(0)$ and $X(t)$ and differ only in the measurement noise. The states of the MS-VAR baseline are matched as described in Appendix~\ref{sec:baseline}. The oracle MS-VAR baseline assumes that the latent Markov chain is known and is evaluated with the identity permutation.}
    \label{tab:roc_msar}
    \resizebox{\textwidth}{!}{
\begin{tabular}{lllllllll}\toprule{} & \multicolumn{4}{l}{Data-Generating Process 1}
    &
    \multicolumn{4}{l}{Data-Generating Process 2}\\
    &\multicolumn{2}{l}{MS-VAR Baseline} & \multicolumn{2}{l}{Oracle MS-VAR Baseline} &
 \multicolumn{2}{l}{MS-VAR Baseline} & \multicolumn{2}{l}{Oracle MS-VAR Baseline}  \\State & 1 & 2 & 1 & 2 & 1 & 2 & 1 & 2 \\ N &  &  &  &  &  &  &  &  \\
 
 80 & 0.65(0.070) & 0.82(0.020) & 0.83(0.007) & 0.64(0.030) & 0.78(0.048) & 0.86(0.014) & 0.82(0.009) & 0.94(0.002) \\
120 & 0.58(0.030) & 0.79(0.025) & 0.84(0.003) & 0.67(0.012) & 0.85(0.041) & 0.91(0.003) & 0.93(0.006) & 0.98(0.001) \\
160 & 0.62(0.047) & 0.81(0.025) & 0.84(0.004) & 0.72(0.014) & 0.72(0.077) & 0.91(0.002) & 0.94(0.003) & 0.96(0.001) \\
200 & 0.56(0.071) & 0.83(0.005) & 0.84(0.005) & 0.70(0.014) & 0.62(0.115) & 0.91(0.001) & 0.93(0.001) & 0.95(0.000) \\
\bottomrule\\\end{tabular}}
\end{table}

\section{Details of the NYU ADHD analysis}
\label{appendix:adhd-analysis}

In Section~\ref{sec:experiment}, we apply the proposed method to characterize Attention-Deficit/Hyperactivity Disorder (ADHD), a neurodevelopmental disorder that develops in early childhood. We focus on group differences in brain network dynamics between Typically Developing Children (TDC) and the ADHD-combined (ADHD-C) subtype, a common subtype of ADHD that presents both inattentiveness and hyperactivity/impulsivity.
\citet{shappell2021children} found that children with ADHD and TDC differ in the time they spend in resting-state network states and in how quickly they leave them. They modeled resting-state fMRI by a hidden semi-Markov model with multivariate Gaussian emissions, where both ADHD and TDC groups share the same network states. In particular, they found that children with ADHD spend more time in a hyperconnected state and less time in anticorrelated states than TDC. Furthermore, \citet{park2021state} found state-dependent differences in effective connectivity between children with ADHD-C and TDC.
Motivated by these findings, we estimate parameters $\{\theta_{ij}^\ell\}$ shared by the two groups, together with group-specific transition rate matrices $Q_{\text{ADHD-C}}$ and $Q_{\text{TDC}}$ and group-specific noise variances.

\subsection{Data and preprocessing}
\label{appendix:adhd-data}

We analyze resting-state fMRI from the NYU ADHD dataset~\citep{castellanos2008cingulate}, released through the ADHD-200 Initiative~\citep{bellec2017neuro}. We use the Athena-preprocessed ``filtfix'' release~\citep{bellec2017neuro}, whose regional time series were extracted by AFNI \texttt{3dROIstats} with the mask \texttt{aal\_mask\_pad} of the $116$ regions of interest $(p=116)$ of the Automated Anatomical Labeling (AAL) atlas~\citep{tzourio2002automated}. A subject is eligible if the resting-state scan passed quality control (\texttt{QC\_Rest\_1}$\,=1$), the age is between $7$ and $10$ years, which mitigates age effects that contribute to the development of ADHD, following~\citet{park2021state}, the diagnosis is TDC (\texttt{DX}$\,=0$) or ADHD-C (\texttt{DX}$\,=1$), and none of ADHD Measure, ADHD Index, Inattentive, Hyper/Impulsive, IQ Measure, Verbal IQ and Performance IQ is missing (coded $-999$). We use the regional time series without centering, scaling or detrending.

Each session has $172$ scans, with a sampling period of $2.0$ seconds. Therefore, each subject provides only $172$ time points to estimate $p^2mk$ parameters, which would result in unreliable estimates. Motivated by~\citet{shappell2021children}, we pool the subjects of each group. Of the $22$ TDC and $21$ ADHD-C subjects that meet these criteria, we use the first $15$ of each in the order of the phenotype file. We join the $15$ TDC series and then the $15$ ADHD-C series end to end into one series of $30\times172=5160$ scans, which we place at equally spaced times on $[0,10800]$, so that $N=5159$, and we use a joint estimator. We treat the concatenated series as one trajectory and do not correct the increments or the wavelet smoothing at the knots between subjects. The increment across a knot, from the last scan of one subject to the first scan of the next, is therefore artificial, and it is included: the $14$ knots within a group are ordinary sampling intervals of that group's E-step as discussed in Section~\ref{appendix:adhd-fitting}, and the knot between the two groups is excluded from both E-steps but enters the update of $\{\theta_{ij}^{\ell}\}$.

\subsection{Model fitting}
\label{appendix:adhd-fitting}

In the estimation step, we modify Algorithm~\ref{alg:update:ode} as follows. The parameters $\{\theta_{ij}^\ell\}$ are shared, and each group has its own transition rate matrix $Q_{\text{group}}$ and noise variance $\sigma^2_{\text{group}}$. The E-step is computed separately for each group, over the concatenated series of its $15$ subjects, with $\theta$, $Q_{\text{group}}$, $\sigma^2_{\text{group}}$ and an initial distribution of the group; we write $\Theta_{\text{group}}$ for these quantities and $P(\cdot\mid Y_0^N;\Theta_{\text{group}})$ for the posterior computed in this E-step. The initial distribution starts at $\pi(Q^{(0)}_{\text{group}})$, for the initial value $Q^{(0)}_{\text{group}}$ described below, and is replaced at each M-step by the posterior distribution of the state at the group's first scan. The M-step updates $Q_{\text{group}}$ from the posterior of that group, as in Section~\ref{ssec:graphest}, and $\sigma^2_{\text{group}}$ by the analogue of $M_{N,\sigma^2}$ over the increments within the group, except that the squared residual of the increment from $t_{n-1}$ to $t_n$ is weighted by $P(Z(t_{n-1})=\ell\mid Y_0^N;\Theta_{\text{group}})$ rather than $P(Z(t_n)=\ell\mid Y_0^N;\Theta_{\text{group}})$, and the sum is divided by $p$ times the $2580$ scans of the group rather than its $2579$ increments. It updates $\theta$ by~\eqref{eq:M_ntheta} over all $N=5159$ increments of the concatenated series of the $30$ subjects:
\begin{multline*}
\theta_i^\ell\leftarrow\argmax_{\tilde\theta_i^\ell}-\frac{1}{N}\sum_{n=1}^{N}w_{\ell,\Theta_{\text{group}}}(t_n)\rbr{Y_{n,i}-Y_{n-1,i}-\sum_{j=1}^p\hat\Psi_j(t_n)^\top\tilde\theta_{ij}^\ell}^2\\
-\lambda\sum_{j=1}^p\cbr{N^{-1}\sum_{n=1}^N(\hat\Psi_j(t_n)^\top\tilde\theta_{ij}^\ell)^2}^{1/2},
\end{multline*}
where $w_{\ell,\Theta_{\text{group}}}(t_n)=P(Z(t_n)=\ell\mid Y_0^N;\Theta_{\text{group}})$ for the group that contains $t_n$. 

For parsimony and ease of interpretation, we use the linear basis $g(x)=x$, that is, $m=1$. We use $k=3$ states and $\lambda_{\mathrm{code}}=e^{-8}\approx3.355\times10^{-4}$, a value on the grid of the exponentials of $10$ equally spaced points in $[-9,-6]$. We reach this value along the grid using the procedure for the simulations described in Appendix~\ref{sses:expsetup}: we run Algorithm~\ref{alg:update:ode} for the grid values from $\lambda_{\mathrm{code}}=e^{-6}$ downward, with at most $t_{\max}=10$ iterations for each value, initialize $\theta$ and $Q_{\text{group}}$ at the estimates for the previous value, and reset $\sigma^2_{\text{group}}$ to its initial value. The fit at $e^{-8}$ is the seventh on this path, so $Q^{(0)}_{\text{group}}$ above is the estimate at $e^{-23/3}$. At $e^{-6}$, the initial values are drawn independently: in each group, $q_{\ell\ell'}\sim\mathrm{Unif}(0,0.1)$ for $\ell\neq\ell'$, with the diagonal entries set so that the rows sum to zero, and $\theta_{ij}^\ell\sim\Ncal(0,0.01)$ for $i\neq j$ and $\theta_{ii}^\ell\sim\Ncal(0.1,0.01)$. We report $k$, $m$ and $\lambda$ as the values used for this fit, not as values selected by the surrogate BIC of Section~\ref{ssec:parameterselect}.

\subsection{Occupation-time computation}
\label{appendix:adhd-occupation}

To compare the ADHD-C and TDC groups, we compute the mean occupation time of each state per subject. Since the E-step treats each group's concatenated series as one sequence, the posterior of a subject's states conditions on all $15$ subjects of the group, and we define the summary of subject $u$ by restricting this group posterior to the subject's scans. Let $\mathcal I_u$ be the set of the indices $n$ of the $171$ intervals $(t_{n-1},t_n]$ between consecutive scans of subject $u$, which excludes the junctions. The occupation time of state $\ell$ for subject $u$ and its group mean are
\begin{align*}
\tau_{\ell,u}&=\sum_{n\in\mathcal I_u}\sum_{a,b=1}^k P(Z(t_{n-1})=a,Z(t_n)=b\mid Y_0^N;\hat\Theta_{\text{group}})\\
&\qquad\times\EE[\tau_\ell(h)\mid Z(0)=a,Z(h)=b;\hat Q_{\text{group}}],\\
\tau_{\ell,\text{group}}&=\frac{1}{15}\sum_{u=1}^{15}\tau_{\ell,u},
\end{align*}
where the group is either ADHD-C or TDC, $\hat\Theta_{\text{group}}$ denotes the estimated parameter with the group's transition rate matrix $\hat Q_{\text{group}}$, noise variance and initial distribution, $P(\cdot\mid Y_0^N;\hat\Theta_{\text{group}})$ is the posterior computed in the group's E-step, here and in Figure~\ref{fig:adhdgraph}, and the conditional expectation, computed with the transition rates of $\hat Q_{\text{group}}$ expressed per second, is given by~\eqref{eq:decomposetau} in Appendix~\ref{appendix:algorithm:details}. Thus $\tau_{\ell,u}$ is the part of the group's posterior expected occupation time~\eqref{eq:computetaul} that falls in the scans of subject $u$, and $\sum_{\ell}\tau_{\ell,u}=171h=342$ seconds. 

\subsection{Connectivity summaries}
\label{appendix:adhd-connectivity}

Table~\ref{tab:state_connectivity} summarizes the estimates of the off-diagonal $\theta_{ij}^\ell$ used for the connectomes. Each node represents one region of the $116$ AAL regions.  State~$1$ has the most nonzero off-diagonal estimates, $7075$ ($53.0\%$ of the $116\times115$ possible directed edges), followed by states~$2$ and~$3$ with $4961$ ($37.2\%$) and $1873$ ($14.0\%$), respectively. These counts include all nonzero estimates, rather than only the edges displayed in Figure~\ref{fig:adhdgraph}(a)--(c), where only the top $1\%$ of the $116\times115$ possible directed edges by magnitude are displayed, as reported in the last row of Table~\ref{tab:state_connectivity}.

\begin{table}[ht!]
\centering

\caption{Connectivity summaries for the estimates of $\theta_{ij}^\ell$ used in Figure~\ref{fig:adhdgraph}(a)--(c). All summaries exclude diagonal entries; there are $116\times115=13{,}340$ possible directed edges. Density is the proportion of nonzero off-diagonal estimates, and the mean absolute magnitude is computed over those nonzero estimates.}
\label{tab:state_connectivity}
\begin{tabular}{lccc}
\toprule
& State $1$ & State $2$ & State $3$ \\
\midrule
Nonzero estimates (density) & $7075\ (53.0\%)$ & $4961\ (37.2\%)$ & $1873\ (14.0\%)$ \\
Positive / negative nonzero estimates & $3325\,/\,3750$ & $2387\,/\,2574$ & $879\,/\,994$ \\
Total absolute magnitude & $66.88$ & $52.69$ & $14.82$ \\
Mean absolute magnitude & $0.0095$ & $0.0106$ & $0.0079$ \\
Edges displayed & $134$ & $133$ & $134$ \\
\bottomrule
\end{tabular}
\end{table}